\documentclass[12pt]{article}

\usepackage{amsthm,amsmath,amsfonts,amssymb}
\usepackage{mathrsfs}
\usepackage{bbm}

\usepackage{graphicx,psfrag,epsf}
\usepackage{enumerate}
\usepackage{subcaption}
\usepackage{algorithm}
\usepackage{algorithmic}
\usepackage{booktabs}
\usepackage{multirow}
\usepackage{verbatim}

\usepackage{natbib}
\PassOptionsToPackage{unicode}{hyperref}
\PassOptionsToPackage{hyphens}{url}
\PassOptionsToPackage{dvipsnames,svgnames,x11names}{xcolor}
\usepackage{xcolor}
\usepackage{hyperref}
\hypersetup{
  colorlinks=true,
  linkcolor={blue},
  filecolor={Maroon},
  citecolor={Blue},
  urlcolor={Blue}
}
\usepackage{cleveref}

\newtheorem{theorem}{Theorem}

\newtheorem{lemma}[theorem]{Lemma}
\newtheorem{condition}[theorem]{Condition}

\DeclareMathOperator*{\argmin}{argmin}
\theoremstyle{definition}

\newtheorem{assumption}{Assumption}

\newtheorem{proposition}{Proposition}

\newcommand{\anon}{0} 

\begin{document}

\def\spacingset#1{\renewcommand{\baselinestretch}%
{#1}\small\normalsize} \spacingset{1}


\if0\anon
{
  \title{\bf Group Permutation Testing in Linear Model: Sharp Validity, Power Improvement, and Extension Beyond Exchangeability}

\author{
  Zonghan Li$^{1}$, \quad
  Hongyi Zhou$^{1}$, \quad
  Zhiheng Zhang$^{2, 3}$\thanks{The author Zonghan Li and Hongyi Zhou contributed equally to this work. Correspondence should be addressed to Zonghan Li: \url{lizongha24@mails.tsinghua.edu.cn} and Zhiheng Zhang: \url{zhangzhiheng@mail.shufe.edu.cn}}
  \\ \vspace{0.2cm}
  \parbox{\textwidth}{\centering \small
    $^1$ Institute for Interdisciplinary Information Sciences, Tsinghua University, Beijing, China \\
    $^2$ School of Statistics and Data Science, Shanghai University of Finance and Economics, Shanghai 200433, P.R. China \\
    $^3$ Institute of Big Data Research, Shanghai University of Finance and Economics, Shanghai 200433, P.R. China
  }
}
\date{}
\maketitle
} \fi

\bigskip
\begin{abstract}

We study finite-sample inference for a single regression coefficient in the fixed-design linear model $Y=Z\beta+bX+\epsilon$, allowing $\epsilon$ to be dependent, heterogeneous, or non-exchangeable. We develop a group-permutation framework that unifies randomization-based regression testing and yields sharp finite-sample guarantees. Under exchangeable errors, we show that permutation-augmented regression tests admit an exact validity characterization: a grouped version of PALMRT controls Type I error at level at most $2\alpha$ for any permutation group, and this factor $2$ is unimprovable in general. We further connect Type II error to a design-dependent geometric separation between the target regressor and its permuted versions after removing nuisance effects, leading to a combinatorial optimization problem over permutation groups. Under mild Sub-Gaussian assumptions, we propose a constructive, design-adaptive strategy that is provably no worse than i.i.d. random permutations and often substantially more powerful. Finally, we extend cyclic and permutation-augmented regression tests beyond exchangeability by linking rank-based randomization with weighted conformal inference. The resulting weighted group tests retain finite-sample Type I control that degrades smoothly with total variation departures from group symmetry, recovering exact validity under exchangeability.

\end{abstract}

\noindent%
\vfill

\newpage
\spacingset{1.9} 
\section{Introduction}

Exact finite-sample inference for linear regression coefficients has become an increasingly important problem in modern statistics and applied data science~\citep{lei2021,wen2025,d2024robust,guan2024,spector2024mosaic}. In many scientific studies, including genomics~\citep{love2014moderated,smyth2004linear}, neuroscience~\citep{button2013power}, and small-sample clinical trials~\citep{julious2023sample}, the available sample size is often moderate~\citep{fan2008high}. In such regimes, classical asymptotic tests can suffer from poor Type I error control, making resampling-based procedures that exploit exact symmetries of the noise distribution particularly appealing. In this paper, we revisit the problem of testing the coefficient $b$ of a target regressor $X$ in the fixed-design linear model
\begin{equation}
\label{eq::model}
Y=Z \beta+b X+\epsilon, 
\end{equation}
{where we have observations $(Z,X,Y)\in \mathbb{R}^{n \times p}\times \mathbb{R}^{n} \times \mathbb{R}^{n}$ and $\epsilon:=(\epsilon_1,\cdots,\epsilon_n)^T$ is an $n$-dimensional noise vector. Our goal is to test the null hypothesis $H_0:b=0$ against the alternative hypothesis $H_1:b\neq 0.$} We develop a unified permutation framework that delivers finite-sample Type I error guarantees with provably sharp (i.e., unimprovable) constants, provides a principled and optimizable handle on Type II error, and extends permutation-based inference beyond the classical exchangeability assumption on the noise $\epsilon$. A recent line of work has addressed this problem by exploiting noise exchangeability through carefully designed permutation tests, typically built upon the following assumption:
 {\begin{assumption}(Exchangeable noise). 
\label{assump::exchanglable}
For $\epsilon:=(\epsilon_1,\cdots,\epsilon_n)^T$ and any permutation $\pi$ of indices $1,\cdots,n$,
$$(\epsilon_1,\cdots,\epsilon_n) \overset{d}{=} (\epsilon_{\pi(1)},\cdots,\epsilon_{\pi(n)})$$
\end{assumption}}

\citet{lei2021}'s cyclic permutation test (CPT) achieves exact finite-sample validity in fixed-design
regression by constructing linear statistics whose joint null distribution is invariant under a
structured left-shift permutation group; ranking these statistics yields an exact test under
exchangeable noise.
The explicit group structure provides a transparent invariance argument and sharp Type~I error
control.
However, CPT relies on solving a linear constraint system to eliminate nuisance effects, and the
existence of nontrivial solutions requires restrictive dimensional regimes, which can be prohibitive
when $p$ is moderate relative to $n$ or when finer randomization resolution is desired. \citet{wen2025}'s residual permutation test (RPT) addresses this feasibility issue by replacing
CPT's constraint-based construction with a projection-based reformulation.
By working in orthogonal complements of augmented design spaces, RPT substantially relaxes CPT's
dimensional requirements while preserving exact finite-sample validity under exchangeability.
The price of this increased applicability is a conservative aggregation of permutation evidence
through a minimum-type operator, which can attenuate power and obscures how the choice of
permutations influences the resulting Type~II error. \citet{guan2024}'s permutation-augmented linear model regression test (PALMRT) further modifies the
permutation scheme by replacing conservative aggregation with pairwise permutation comparisons.
This change can yield substantial empirical power improvements and admits an explicit finite-sample
Type~I error bound of at most $2\alpha$ under exchangeable errors.
At the same time, PALMRT samples permutations i.i.d.\ from the full symmetric group, discarding the
geometry of the transformation set and making it difficult to analyze or optimize power in a
design-dependent manner; moreover, the factor-$2$ constant in the Type~I bound is not known to be
improvable under the same assumptions.

These gaps are most salient in applications that demand both exact calibration and high power under
complex fixed designs.
For example, in small-$n$, moderate-$p$ regressions with strongly correlated nuisance covariates
(e.g., basis expansions or multi-omic features) and mild heteroskedasticity across experimental
batches, existing exact tests may require restrictive regimes, incur conservative Type~I constants,
or provide no principled mechanism for selecting permutations tailored to the observed design.
Accordingly, three questions remain open: (i) whether the factor-$2$ Type~I bound in PALMRT is
improvable without strengthening assumptions; (ii) how to control and optimize Type~II error while
maintaining finite-sample Type~I guarantees; and (iii) how far finite-sample inference can be pushed
beyond exchangeability using only structural information about $\epsilon$. These questions form the research gap addressed in this work. We propose a group permutation test framework that centers on the assumption that the chosen permutation matrices form a finite group.
{Let $\mathcal{P}_n$ denote the set of all permutation matrices in $\mathbb{R}^{n\times n}$. We impose the following assumption: \begin{assumption}(Group permutation)
\label{assump::grouppermutation}
The set of permutation matrices $P_K := \{P_0,\ldots,P_K\}$ is closed under matrix multiplication; that is, for any $P_i, P_j \in P_K,$ there exists $k \in \{0,\ldots,K\}$ such that $P_k = P_i P_j$.
\end{assumption}}

First, Assumption~\ref{assump::grouppermutation} strictly relaxes the left-shifting structure in CPT
while remaining sufficiently rich to sustain exact symmetry arguments.
Within this unified group-theoretic view, we revisit PALMRT and show that under exchangeable noise its
Type~I error is controlled at level $2\alpha$ for \emph{any} admissible permutation group.
We further construct explicit worst-case designs establishing sharpness: without additional assumptions,
the constant $2$ cannot be improved.

Second, the group perspective yields a principled route to power analysis.
For fixed designs $(X,Z)$, we link the power of group-based PALMRT to a design-dependent geometric
separation between permuted projections of $X$, which induces an optimization problem over permutation
groups.
Although combinatorial and non-convex, this formulation admits tractable bound characterizations of
Type~II error.
We propose a constructive, design-adaptive algorithm that selects permutation groups informed by the
observed geometry; under mild assumptions, the resulting groups achieve systematically tighter Type~II
error upper bounds than i.i.d.\ permutations (e.g., \citet{guan2024}), while preserving valid Type~I control.

Finally, we extend the framework beyond exchangeability by connecting rank-based randomization to
conformal inference under distribution shift.
Recasting CPT and PALMRT as weighted rank statistics yields weighted group permutation tests whose
Type~I error bounds degrade gracefully with total variation distances between $\epsilon$ and its
group-permuted versions, recovering exact validity under exchangeability and providing quantitative
robustness otherwise.
Together, the group-permutation viewpoint unifies exact validity, design-adaptive power, and robustness
within a single finite-sample framework. Our main contributions are summarized as follows:

The remainder of the paper is organized along the motivation–method–validation chain. \Cref{sec::exact_tests} reviews existing exact tests for linear models, emphasizing how their permutation structures and assumptions relate. \Cref{sec::The relax of matrix selections assumptions to group permutation} introduces the group permutation framework, analyzes Type I error for grouped PALMRT, develops the Type II optimization formulation, and presents our algorithm for constructing permutation groups together with numerical comparisons to random permutations. \Cref{sec::The relax of exchangeable to nonexchangeable} extends CPT and PALMRT with group permutations to the nonexchangeable setting, drawing explicit parallels to conformal prediction and deriving robustness guarantees. We conclude in Section~\ref{conclusion}.

{{\bf Notation.} Given any permutation matrix $P_i\in \mathcal{P}_n$ and any matrix $M\in \mathbb{R}^{n \times d}$, we write $M_{\pi_i}:=P_i M.$}


\section{Preliminaries}
\label{sec::exact_tests}

Recent literature on linear model inference \eqref{eq::model} leverages the exchangeability of noise $\epsilon$. Since $\epsilon$ is unobserved, the core challenge is to propagate its exchangeability to test statistics. We briefly review three representative methods that motivate our work. The details of these work can be found in \Cref{sec::exact_tests submateral}.

\subsection{Cyclic permutation test (CPT)}

\citet{lei2021} proposed the CPT, constructing statistics $S_j = Y^T \eta_j$ where $\eta_j$ are derived from a specific left-shifting permutation matrix $\mathcal{P}^L_{m}$, which satisfies \Cref{assump::leftshifting}. By imposing constraints (Condition \ref{con::lei}) to ensure $Z^T\eta_j$ is constant, the deterministic part of $Y$ is neutralized under $H_0$, allowing $\epsilon$'s exchangeability to control the Type I error: $P(R_0/(m+1) \leq \alpha) \leq \alpha$. However, CPT requires a restrictive sample size condition, $n > (1/\alpha - 1)p$, which often fails in high-dimensional settings.

\subsection{Residual permutation test (RPT)}

To relax this constraint, \citet{wen2025} introduced RPT. By projecting $Y$ onto the orthogonal complement of the column space spanned by $(Z, Z_{\pi_k})$, they eliminate the nuisance parameters without the strict $\eta_j$ construction. RPT only requires $p < n/2$. However, to maintain exactness, they employ a minimum operator in their statistic $\phi_1$, which makes the test conservative and reduces power (Type II error).
\subsection{Permutation-augmented linear model regression test (PALMRT)}

\citet{guan2024} addressed this conservativeness by replacing the rank-based statistic with a pairwise comparison $\phi_2$, using i.i.d. permutations in \Cref{assump::iidpermute}.

 \begin{assumption}(Independent and Identically Distributed permutations)
 \label{assump::iidpermute}
    $\mathcal{P}^{i.i.d}_K:=\{P_k\}_{k=0}^{K}$, $\forall k \in \{1,\cdots,K\}, P_k$ is generated independently and uniformly from $\mathcal{P}_n$, and let $P_0=I$.
\end{assumption}

While achieving $P(\phi_2 \leq \alpha) \leq 2\alpha$ under $p < n/2$, PALMRT relies on random permutations, which complicates the theoretical optimization of power.

This paper focuses on the \textit{group permutation} set $\mathcal{P}_K$ in  \Cref{assump::grouppermutation}. We begin with PALMRT~\citep{guan2024}. We consider changing the construction method of permutation matrices from 
  \Cref{assump::iidpermute} to \Cref{assump::grouppermutation}, maintaining the same statistic as $\phi_2$. Under exchangeability, we show that the group-permutation version satisfies the same Type I error bound as in (6). Moreover, the group structure provides enough insight to construct explicit examples showing that this bound is in fact tight. Compared with PALMRT, the group assumption allows us to formalize the Type II error control question as an optimization problem and provide a theoretically grounded construction method for the permutation group $\mathcal{P}_K$. Additionally, it is natural to relax the original CPT's~\citep{lei2021} assumption for constructing permutation matrices from \Cref{assump::leftshifting} to the weaker \Cref{assump::grouppermutation} and to reformulate the proof. To clarify why the group permutation assumption and the proof reformulation are useful, we extend CPT with group permutation beyond the exchangeable assumption of $\epsilon$. Using a similar idea, we also extend PALMRT with group permutation to a nonexchangeable scenario.
\section{From permutation-selection assumptions to group permutations}
\label{sec::The relax of matrix selections assumptions to group permutation}

As discussed in \Cref{sec::exact_tests}, CPT forms a permutation matrix set $\mathcal{P}_K^L:=\{P_k\}_{k=0}^K$ based on \Cref{assump::leftshifting}, and PALMRT forms a permutation matrix set $\mathcal{P}_K^{i.i.d}:=\{P_k\}_{k=0}^K$ based on \Cref{assump::iidpermute}. In this section, we relax both assumptions to the group permutation in \Cref{assump::grouppermutation} and present some new conclusions. 
\subsection{PALMRT with group permutation}
\label{Grouped PALMRT}
In this section, we restate the problem of interest: to test the null hypothesis $H_0:b=0$ against the alternative hypothesis $H_1:b\neq0$ under model \eqref{eq::model}, and the error term $\epsilon$ satisfies \Cref{assump::exchanglable} and is exchangeable. Correspondingly, $S_n$ denotes the permutation group on [n]. We construct our permutation group $\mathcal{P}_{K}:=\{P_0,P_1,P_2,\cdots,P_K\},P_i\in\mathcal{P}_n,i\in\{0,1,\cdots,K\}$, and $K$ is a hyperparameter, with $\mathcal{P}_K$ satisfying  \Cref{assump::grouppermutation}. Let $\pi_k$ denote the permutation corresponding to $P_k$, $Z_{\pi_k}:=P_k\cdot Z$, $k=0,\cdots,K.$ We will prove later that $I\in\mathcal{P}_K$ in \Cref{prop::Proposition for group permutation}; there we simply denote $P_0:=I.$ Then, we define the test statistic as follows:
   \begin{equation}
   \label{eq::PALMRTphi}
       \phi= \frac{1}{K+1}(1+\sum_{k=1}^{K}\mathbbm{1}\{X^{T}(I-H^{ZZ_{\pi_k}})Y\leq X_{\pi_k}^{T}(I-H^{ZZ_{\pi_k}})Y\}).
   \end{equation}
 Under the null hypothesis $H_0:b=0$, we define a bivariate function that also incorporates data $x$, $Z$, and  unobserved noise $\epsilon$ as fixed inputs:
   \begin{equation*}
       F(\pi_1,\pi_2;x,Z,\epsilon)=X_{\pi_1}^{T}(I-H^{Z_{\pi_1}Z_{\pi_2}})\epsilon.
   \end{equation*}
Substituting the definition of $Y$ in model \eqref{eq::model} under the null hypothesis $H_0:b=0$, we have that:
$$(X_{\pi_k})^T(I-H^{ZZ_{\pi_k}})Y=\underbrace{(X_{\pi_k})^T(I-H^{ZZ_{\pi_k}})Z \beta}_{\text{orthogonal part}}+\underbrace{(X_{\pi_k})^T(I-H^{ZZ_{\pi_k}})\epsilon }_{\text{stochastic part}}.$$

Since $Z$ lies in the column space of $(Z, Z_{\pi_k})$, the deterministic term vanishes. Hence, under $H_0$, $X^T(I-H_{ZZ_{\pi_k}})Y = X^T(I-H_{ZZ_{\pi_k}})\epsilon$. This function $F(\pi_1,\pi_2;x,Z,\epsilon)=X_{\pi_1}^{T}(I-H^{Z_{\pi_1}Z_{\pi_2}})\epsilon$ satisfies the proposition below:
 \begin{proposition}
 \label{prop::Finverse}
     For any permutation $\pi_1,\pi_2$ of $S_n$, the function $F(.,.;x,Z,\epsilon)$ satisfies 
     \begin{equation*}
         F(\pi_1,\pi_2;x,Z,\epsilon_{\sigma})=F(\sigma^{-1} \circ \pi_1,\sigma^{-1} \circ \pi_2;x,Z,\epsilon).
     \end{equation*}
 \end{proposition}
   Then, the statistic $\phi$ can be rewritten as ($\pi_0 := I$):
   \begin{equation}
   \label{eq::phinew}
       \phi= \frac{1}{K+1}(1+\sum_{k=1}^{K}\mathbbm{1}\{F(\pi_0,\pi_k;x,Z,\epsilon_{\sigma}) \leq F(\pi_k,\pi_0;x,Z,\epsilon_{\sigma})\}).
   \end{equation}
   We denote $r_{ab}=\mathbbm{1}\{F(\pi_a,\pi_b;x,Z,\epsilon) < F(\pi_b,\pi_a;x,Z,\epsilon)\}+\frac{1}{2}\mathbbm{1}\{F(\pi_a,\pi_b;x,Z,\epsilon) = F(\pi_b,\pi_a;x,Z,\epsilon)\}$, $R_{a}:=\frac{1}{K+1}\sum^{K}_{b=0} r_{ab}$, and $\forall a,b, \pi_a,\pi_b \in \mathcal{P}_K$.
\begin{theorem}
\label{Thm::PALMRT}
    Suppose $\mathcal{P}_K:=\{P_k:k=0,\cdots,K\}$ satisfies \Cref{assump::grouppermutation}. Suppose that $(X,Z,Y)$ is generated under the model \eqref{eq::model} with $p \leq n/2$ and that the noise $\epsilon$ satisfies \Cref{assump::exchanglable}. We define $S:=\{m:R_m\leq \alpha\}$. Under $H_0: b=0$, $\phi$ defined in \eqref{eq::PALMRTphi} satisfies $P(\phi\leq \alpha)\leq \frac{1}{K+1}E|S|$, and for a given $\epsilon$, $\frac{[\alpha(K+1)]}{K+1}\leq\frac{1}{K+1}|S|\leq\frac{[2\alpha(K+1)]}{K+1}$.
\end{theorem} If we define 
\begin{equation}
    \begin{aligned}
    \label{phi'}
        \phi'= &\frac{1}{K+1}(1+\sum_{k=1}^{K}\mathbbm{1}\{X^{T}(I-H^{ZZ_{\pi_k}})Y < X_{\pi_k}^{T}(I-H^{ZZ_{\pi_k}})Y\}\\ &+\frac{1}{2}(\mathbbm{1}\{X^{T}(I-H^{ZZ_{\pi_k}})Y = X_{\pi_k}^{T}(I-H^{ZZ_{\pi_k}})Y\})).
    \end{aligned}
\end{equation}
When $H_0: b=0$, from the proof of \Cref{Thm::PALMRT}, it is immediate that $\phi'$ also satisfies  \Cref{Thm::PALMRT}. The detailed proof shows that the bound $2\alpha$ in \Cref{Thm::PALMRT} is driven mainly by $\phi'$, and the Type I error of $\phi$ can be controlled by $\phi'$. We write the test statistic $\phi$ to align with the statistical format of the previous work \citet{wen2025}.


Not only the original test statistic in \eqref{eq::PALMRTphi}, but also the entire class of statistics of the form \eqref{eq::phinew}, satisfies \Cref{prop::Finverse} and hence falls within the scope of \Cref{Thm::PALMRT}. In particular, \eqref{eq::PALMRTphi} is a special case of \eqref{eq::phinew}. From \Cref{Thm::PALMRT}, we can obtain $P(\phi\leq \alpha)\leq 2\alpha$. It is interesting that the above bound for the rejection region $[0,\alpha]$ is $2\alpha$. As discussed in \Cref{sec::exact_tests}, previous studies based on rank statistics ensure $\alpha$-level error control for the rejection region $[0,\alpha]$ in \eqref{eq::lei type I error bound} and (\ref{eq::Wen type I error control}). The core idea in these works is to exploit exchangeability to construct several identically distributed statistics, whose ranks are also identically distributed. Therefore, it is provable that their normalized rank distribution stochastically dominates the uniform distribution $U[0,1]$, so that the $\alpha$ bound follows. Then, a natural question arises: if we refine our proof procedure, is it possible to obtain an error bound like
$P( \phi\leq \alpha)\leq \alpha$? Thus, the proposition below provides an application to this question and offers a method for constructing worst-case examples to see that, without adding any other assumptions, we cannot improve the error bound in PALMRT with the group permutation assumption; that is, the factor of $2$ in Theorem \ref{Thm::PALMRT} cannot be improved.
\begin{proposition}
\label{proposition of 2alpha}
    For any sample size $n\geq 2p$, there exist $(X,Z)\in \mathbb{R}^{n\times p}\times \mathbb{R}^{n}$, a permutation group $\mathcal{P}_K$ satisfying \Cref{assump::grouppermutation}, and a given $\alpha$, such that for any noise vector $\epsilon$, 
    $$P(\phi'\leq \alpha)=2\alpha.$$
\end{proposition}
The full proof is deferred to \Cref{Sec::section of proof of theorem 1}.

    \textit{Proof sketch of \Cref{Thm::PALMRT}}. We control the Type I error of $\phi$ through that of $\phi'$. Due to the orthogonality of $I-H^{ZZ_{\pi_k}}$, we change $\phi$ and $\phi'$ to a form like \eqref{eq::phinew} that satisfies \Cref{prop::Finverse}. Similarly, we denote $r_{ab}$, which forms the matrix $R:=(r_{ab})_{a,b\in\{0,1,\cdots,K\}}$, and $\phi'$ is determined by the sum over the 0-th row, namely $R_0$. Through the construction of the group permutation in \Cref{assump::grouppermutation}, we can prove that the sum of each row is actually identically distributed. Consequently, our problem turns to providing a bound for the set of ``strange'' rows defined by $S$, which can be controlled using the structural properties of the matrix $R$, where the two elements of the matrix $R:=(r_{ab})_{a,b\in\{0,1,\cdots,K\}}$ are complementary with respect to diagonal symmetry, and their sum is $1$.

Now we provide our numerical experiment to verify the validity of our PALMRT with group permutation. Under model~\eqref{eq::model} with $b=0$, we generate the entries of $X$ and $Z$ $i.i.d.$ from a distribution $\mathcal{D}_{data}$, and generate each the entries of $\epsilon$ $i.i.d.$ from a distribution $\mathcal{D}_{noise}$. We choose a permutation group with size $1+K=20$, and evaluate the Type I error at nominal levels $5\%, 10\%$ and $20\%$. 

\begin{table}[htbp]
\centering
\caption{Type I error of PALMRT with group permutation under different $(n,p)$ settings.}
\label{table::Type_I}

\begin{subtable}{0.32\textwidth}
\centering
\caption{$n=300, p=100$}
\label{tab:type1_n300_p100}
\resizebox{\linewidth}{!}{%
\begin{tabular}{ccccc} 
\toprule
\multirow{2}{*}{$\mathcal{D}_{data}$} & \multirow{2}{*}{$\mathcal{D}_{noise}$} & \multicolumn{3}{c}{Type I error} \\
\cmidrule(lr){3-5} 
 & & $5\%$ & $10\%$ & $20\%$ \\
\midrule
Gaussian & Gaussian & 0.69 & 2.37 & 9.12 \\
Gaussian & $t_1$    & 0.75 & 2.47 & 9.37 \\
Gaussian & $t_2$    & 0.72 & 2.52 & 9.37 \\
$t_1$    & Gaussian & 0.47 & 2.04 & 8.97 \\
$t_1$    & $t_1$    & 0.54 & 1.52 & 5.47 \\
$t_1$    & $t_2$    & 0.49 & 1.73 & 7.42 \\
$t_2$    & Gaussian & 0.65 & 2.43 & 8.99 \\
$t_2$    & $t_1$    & 0.84 & 2.26 & 7.56 \\
$t_2$    & $t_2$    & 0.78 & 2.38 & 8.71 \\
\bottomrule
\end{tabular}%
}
\end{subtable}\hfill
\begin{subtable}{0.32\textwidth}
\centering
\caption{$n=600, p=100$}
\label{tab:type1_n600_p100}
\resizebox{\linewidth}{!}{%
\begin{tabular}{ccccc} 
\toprule
\multirow{2}{*}{$\mathcal{D}_{data}$} & \multirow{2}{*}{$\mathcal{D}_{noise}$} & \multicolumn{3}{c}{Type I error} \\
\cmidrule(lr){3-5} 
 & & $5\%$ & $10\%$ & $20\%$ \\
\midrule
Gaussian & Gaussian & 2.70 & 6.24 & 15.6 \\
Gaussian & $t_1$    & 2.96 & 6.82 & 16.3 \\
Gaussian & $t_2$    & 2.94 & 6.71 & 15.6 \\
$t_1$    & Gaussian & 2.25 & 5.96 & 16.0 \\
$t_1$    & $t_1$    & 1.58 & 3.65 & 9.91 \\
$t_1$    & $t_2$    & 2.19 & 5.52 & 13.9 \\
$t_2$    & Gaussian & 2.35 & 6.27 & 15.6 \\
$t_2$    & $t_1$    & 2.46 & 5.39 & 13.0 \\
$t_2$    & $t_2$    & 2.57 & 5.94 & 14.6 \\
\bottomrule
\end{tabular}%
}
\end{subtable}\hfill
\begin{subtable}{0.32\textwidth}
\centering
\caption{$n=600, p=200$}
\label{tab:type1_n600_p200}
\resizebox{\linewidth}{!}{%
\begin{tabular}{ccccc} 
\toprule
\multirow{2}{*}{$\mathcal{D}_{data}$} & \multirow{2}{*}{$\mathcal{D}_{noise}$} & \multicolumn{3}{c}{Type I error} \\
\cmidrule(lr){3-5} 
 & & $5\%$ & $10\%$ & $20\%$ \\
\midrule
Gaussian & Gaussian & 0.74 & 2.44 & 9.20 \\
Gaussian & $t_1$    & 0.73 & 2.46 & 9.08 \\
Gaussian & $t_2$    & 0.69 & 2.32 & 8.94 \\
$t_1$    & Gaussian & 0.50 & 1.97 & 8.68 \\
$t_1$    & $t_1$    & 0.43 & 1.21 & 4.92 \\
$t_1$    & $t_2$    & 0.47 & 1.69 & 7.13 \\
$t_2$    & Gaussian & 0.69 & 2.29 & 9.18 \\
$t_2$    & $t_1$    & 0.84 & 2.17 & 7.41 \\
$t_2$    & $t_2$    & 0.87 & 2.42 & 8.58 \\
\bottomrule
\end{tabular}%
}
\end{subtable}

\end{table}
The simulation results are reported in Table \ref{table::Type_I}. In this simulation, we take both $\mathcal{D}_{data}$ and $\mathcal{D}_{noise}$ to be Gaussian, $t_1$, or $t_2$, and test over $(n,p)=(300,100),(600,100),(600,200)$. In each example, we test 50000 simulations and compute the overall Type I error.

\subsection{Type II error analysis for PALMRT with group permutation}

We have extended the permutation-test methodology to accommodate arbitrary groups under exchangeable noise, yielding finite-sample Type~I error control. In contrast, a systematic theoretical understanding of Type~II error (power) in such group-based settings remains comparatively limited. Existing results typically proceed under additional distributional or structural assumptions on $(X,Z)$. For instance, \citet{wen2025} analyze the Type~II behavior under a model of the form $X=h+\beta'Z+e$, where the coordinates of $e$ are independent with zero mean and bounded variance. Complementary to this line, other recent works—e.g., \citet{guan2024}—provide empirical evaluations that illustrate the practical performance of permutation-based procedures in representative regimes, albeit without offering general finite-sample power guarantees. Motivated by this gap between broad Type~I validity and the less-understood, design-dependent power properties, we aim to develop an optimization approach that, given $(X,Z)$, constructs a valid $\mathcal{P}_{K}$ while ensuring a provably controlled Type~II error (equivalently, a guaranteed separation in power) under minimal additional assumptions. 


\paragraph{Notation.}\label{notations in type II} We first clarify the notation used for Type II error analysis.
\begin{enumerate}
    \item Let $\{\pi_{0},\pi_{1},...,\pi_{K}\}$ be the permutation group of $\{1,2,...,n\}\to\{1,2,...,n\}$ corresponding to $\mathcal{P}_{K}$, with $\pi_{0}$ being the identity mapping.
    \item For any permutation $\pi$ of $\{1,2,...,n\}\to \{1,2,...,n\}$, let $P_{\pi}\in \mathbb{R}^{n\times n}$ be the permutation matrix corresponding to $\pi$. Furthermore, for any matrix $M\in\mathbb{R}^{n\times d}$, we write $M_{\pi}:=P_{\pi} M$.
    \item For vector $v\in\mathbb{R}^{n}$ and $i\in\{1,2,...,n\}$, we let $v_{\pi(i)}$ denote the $\pi(i)$th coordinate of $v$.
    \item Let $e_{i}\in \mathbb{R}^{n}$ be the unit vector whose $i$-th coordinate is $1$ and whose all other coordinates are $0$. Let $w_{i}=H^{Z}e_{i}$ the projection of $e_i$ onto the column space of $Z$. 
    \item Denote $\lambda_1(X,Z,\mathcal{P}_{K},t):= \inf \lambda: \frac{1}{1+K}\left[\sum_{k=0}^{K}\mathbbm{1}\left\{X^{T}H^{ZZ_{\pi_k}}X+X^{T}_{\pi_k}(I-H^{ZZ_{\pi_k}})X\geq\lambda\right\}\right]\leq t$ for any given $X,Z,\mathcal{P}_{K}$ and $t\in[0,1]$.
\end{enumerate}

\subsubsection{An optimization formulation for Type II error control}\label{subsubsection: optimization formulation}

Consider hypotheses $\mathcal{H}_{0}:b=0$ and $H_{1}:b\neq 0$ for the regression problem $Y=bX+\beta Z+\epsilon$, where $\epsilon$ is exchangeable. We use statistics $\phi_1,\phi_2$ to distinguish between $\mathcal{H}_0$ and $\mathcal{H}_{1}$:
\begin{align}\label{statistics for type II}
\phi_1=\frac{1}{1+K}\sum_{k=0}^{K}\mathbbm{1}\left\{ X^{T}(I-H^{ZZ_{\pi_k}})Y \leq X^{T}_{\pi_k}(I-H^{ZZ_{\pi_k}})Y      \right\}\,,~\phi_2 = 1-\phi_1.
\end{align}


Taking $\mathcal{H}_{0}$ to be accepted when $\phi_1,\phi_2\in(\alpha,1]$, the Type I error is at most $4\alpha$. This follows because when $b=0$, $\mathbb{P}[\phi_{1}\leq \alpha]\leq 2\alpha$, and for $\phi_2$, by substituting $-X$ for $X$, we likewise obtain $\mathbb{P}[\phi_2\leq \alpha]\leq 2\alpha$. Consequently, the Type I error is bounded by $\mathbb{P}[\phi_1\leq \alpha]+\mathbb{P}[\phi_2\leq \alpha]\leq 4\alpha$. On the other hand, since
$
X^{T}(I-H^{ZZ_{\pi_{k}}})Y-X^{T}_{\pi_{k}}(I-H^{ZZ_{\pi_{k}}})Y=(X^{T}-X_{\pi_{k}}^{T})(I-H^{ZZ_{\pi_{k}}})(bX+\epsilon)\,,
$
it suffices to consider \eqref{optimization: hypothesis} for rejecting $\mathcal{H}_{0}$.
\begin{align}\label{optimization: hypothesis}
\frac{1}{1+K}\sum_{k=0}^{K}\mathbbm{1}\left\{ b (X^{T}-X_{\pi_{k}}^{T})(I-H^{ZZ_{\pi_{k}}})X \geq (X_{\pi_{k}}^{T}-X^{T})(I-H^{ZZ_{\pi_{k}}})\epsilon    \right\}\notin [\alpha,1-\alpha]\,.
\end{align}
Since the noise term, $(X_{\pi_{k}}^{T}-X^{T})(I-H^{ZZ_{\pi_{k}}})\epsilon$, is hard to analyze exactly but is usually well bounded, we instead consider the following problem \eqref{eq: optimization1}, or, equivalently, $\min_{\mathcal{P}_{K}} \lambda_1(X,Z,\mathcal{P}_{K},\frac{1}{2}\alpha)$, which serves as a sufficient and necessary optimization problem for distinguishing $b\neq 0$ with a small absolute value (details are shown in~\ref{subsubsection: optimization: intuitive}).
\begin{align}\label{eq: optimization1}
\min_{\mathcal{P}_k} \lambda,\ s.t.\ \frac{1}{1+K}\sum_{k=0}^{K}1\left\{ X^{T}H^{ZZ_{\pi_k}}X+X^{T}_{\pi_k}(I-H^{ZZ_{\pi_k}}))X> \lambda \right\}\leq\frac{1}{2}\alpha\,.
\end{align}

The primary challenge now lies in solving \eqref{eq: optimization1}, which presents two principal difficulties. First, the projection matrices $H^{ZZ_{\pi_k}}$ defined over a general permutation group are extremely difficult to analyze. While optimizing over a single permutation may appear straightforward, the simultaneous consideration of all such projections introduces substantial complexity. Second, permutation groups that satisfy \Cref{assump::grouppermutation} are not easily adjustable, and there is no guarantee on how close a near-optimal solution is to the global optimum within such groups. Given the difficulty of optimizing directly over $H^{ZZ_{\pi_k}}$, we instead seek to establish a tractable bound on $\lambda$ in \eqref{eq: optimization1} that can be efficiently optimized. 

To address the problematic term $H^{ZZ_{\pi_k}}$, we present the following Lemma~\ref{lem: upper and lower of optimization objective informal} and Lemma~\ref{lem: decomposition}, which provide bounds for $X^{T}H^{ZZ_{\pi_k}}X$ and $X^{T}{\pi_k}(I-H^{ZZ{\pi_k}})X$, respectively.

\begin{lemma}[Upper and lower bounds of $X^{T}H^{ZZ_{\pi}}X$, informal]\label{lem: upper and lower of optimization objective informal}
$X^{T}H^{ZZ_{\pi}}X$ can be lower bounded as follows:
\begin{align}\label{lower bound term 1}
X^{T}H^{ZZ_{{\pi}}}X\geq \Vert H^{Z}X\Vert^{2}_{2}+\Vert H^{Z_{{\pi}}}(X-H^{Z}X)\Vert^{2}_{2}.    
\end{align}
Furthermore, if each coordinate of $Z$ is independent and $K$-subgaussian, with mean of $0$ and variance of $1$, then we can obtain the following upper and lower bounds of $X^{T}H^{ZZ_{\pi}}X$:
\begin{align*}
X^{T}H^{ZZ_{\pi}}X=\Vert H^{ZZ_{\pi}}X\Vert^{2}_{2}\leq \Vert H^{Z}X\Vert^{2}_{2}+\frac{1}{1-\frac{C(p+tr(P_{\pi}))}{n}}\Vert H^{Z_{\pi}}(X-H^{Z}X)\Vert^{2}_{2}\,,\forall X\,,
\end{align*}
where $C$ is a constant that depends only on $K$.
\end{lemma}

\begin{lemma}\label{lem: decomposition}
For any permutation $\pi$ and $X$, we have:
\begin{align}\label{lower and upper bound term 2}
     \left\vert X^{T}_{\pi}(I-H^{ZZ_{\pi}})X-\frac{1}{2} X^{T}_{\pi} (I-H^{Z_{\pi}})(I-H^{Z})X\right\vert \leq \frac{1}{2}\Vert (I-H^{Z})X\Vert^{2}_{2}\,.
\end{align}
\end{lemma}

Analogous to the definition of $\lambda_1(X,Z,\mathcal{P}_K,t)$, we define $$\lambda_2(X,Z,\mathcal{P}_{K},t)=\inf \lambda: \frac{1}{1+K}\sum_{k=0}^{K}\mathbbm{1}\left\{\frac{1}{2}v^{T}_{\pi_k}v+\Vert H^{Z_{\pi_k}}v\Vert^{2}_{2}\geq \lambda\right\}\leq t.$$
Based on the bounds of both $X^{T}H^{ZZ_{\pi}}X$ and $X^{T}_{\pi}(I-H^{ZZ_{\pi}})X$ according to Lemma~\ref{lem: upper and lower of optimization objective informal} and Lemma~\ref{lem: decomposition}, we establish the optimization problem represented in~\eqref{optimization objective for alg}, where we denote $v = (I - H^{Z})X$ for simplicity. In particular, we have the following comparison between $\lambda_1$ and $\lambda_2$, represented in~\eqref{upper of original optimization objective} and~\eqref{lower of original optimization objective}, where~\eqref{upper of original optimization objective} holds when $Z$ is independent of $\mathcal{P}_{K}$ and follows Lemma~\ref{lem: upper and lower of optimization objective informal}, and ~\eqref{upper of original optimization objective} is always true. Details on the validity of~\eqref{optimization objective for alg} are provided in Appendix \ref{paragraph: upper and lower of optimization objective}.

\begin{align}\label{upper of original optimization objective}
\lambda_1(X,Z,\mathcal{P}_{K},\frac{1}{2}\alpha)\leq \frac{n}{n-C(p+m)}\lambda_2(X,Z,\mathcal{P}_{K},\frac{1}{4}\alpha)+\frac{n}{2n-2C(p+m)}\Vert v\Vert^{2}_{2}+\Vert H^{Z}X\Vert^{2}_{2}\,.
\end{align}
\begin{align}\label{lower of original optimization objective}
\lambda_1(X,Z,\mathcal{P}_{K},\frac{1}{2}\alpha)\geq\lambda_{2}(X,Z,\mathcal{P}_{K},\frac{1}{2}\alpha)+\Vert H^{Z}X\Vert^{2}_{2}-\frac{1}{2}\Vert v\Vert^{2}_{2}\,.
\end{align}


\begin{align}\label{optimization objective for alg}
\min_{\mathcal{P}_{K}}\lambda_2(X,Z,\mathcal{P}_{K},\frac{1}{4}\alpha)\,,\,i.e.\,\min_{\mathcal{P}_{K}}\lambda,\, s.t.\, \frac{1}{1+K}\sum_{k=0}^{K}\mathbbm{1}\left\{\frac{1}{2}\cdot v^{T}_{\pi_k}v +\Vert H^{Z_{\pi_k}}v\Vert^{2}_{2}> \lambda\right\}\leq \frac{1}{4}\alpha. 
\end{align}
\paragraph{Remark} 
When dependence exists between $Z$ and $\mathcal{P}_{K}$, or $Z$ does not satisfy Lemma~\ref{lem: upper and lower of optimization objective informal}, $\lambda_2(X,Z,\mathcal{P}_{K},\frac{1}{4}\alpha)$ no longer yields a theoretically guaranteed upper estimate for $\lambda_1(X,Z,\mathcal{P}_{K},\frac{1}{2}\alpha)$. Consequently, its effectiveness in reducing Type II error remains an empirical matter. We leave a more refined analysis of Type II error control as an open problem.

\subsubsection{Algorithm design for optimizing \texorpdfstring{\eqref{optimization objective for alg}}{(17)}}
In this section, we propose an algorithm for solving \eqref{optimization objective for alg}. The core idea is to consider permutation groups with a structured decomposition: the group is constructed as the composition of multiple subgroups, where the permutations within each subgroup act only on a distinct subset of $\{1,2,\dots,n\}$. Under this structure, $\mathcal{P}_{K}$ admits a tractable bound on $\lambda_2$ as follows:$\left\vert \lambda_2(X,Z,\mathcal{P}_{K},\frac{1}{4}\alpha)-\mathbb{E}_{\pi_k}\left[\frac{1}{2}v^{T}_{\pi_k}v+\Vert H^{Z}v_{\pi_k}\Vert^{2}_{2}\right]\right\vert\leq O(\alpha)\Vert v\Vert^{2}_{2}$,
where $\pi_k$ is uniformly sampled from $\mathcal{P}_{K}$ (see Theorem~\ref{thm: value of optimization objective} and Proposition~\ref{proposition: high probability bound on group decomposition}). This reduces~\eqref{optimization objective for alg} to \eqref{optimization problem for alg}, for which we can design an efficient algorithm that finds a $\mathcal{P}_{K}$ with a valid $\lambda_2$. 


\paragraph{Notation for the algorithm design}
We now clarify the notation used in this part.
\begin{enumerate}
    \item Let $v=(I-H^{Z})X\in\mathbb{R}^{n}$, with $i$-th coordinate $v_{i}$. Let $\bar{v}=\frac{1}{n}\sum_{i=1}^{n}v_i$, $a_i=v_{i}-\bar{v}$.
    \item Let $\vec{1}\in\mathbb{R}^{n}$ be the all-ones vector, $e_i\in\mathbb{R}^{n}$ be the unit vector whose $i$-th coordinate is $1$ and whose other coordinates are $0$. Let $w_{i}=H^{Z}e_{i}$ be the projection of $e_{i}$ onto the column space of $Z$ and $v^{*}=\bar{v}H^{Z}\vec{1}$. 
    \item Let $b_{i}=\Vert w_{i}\Vert^{2}_{2}$, $\bar{b}=\frac{1}{n}\sum_{i=1}^{n}b_{i}$, $c_i=a^{2}_{i}$, $\bar{c}=\frac{1}{n}\sum_{i=1}^{n}c_i$.
\end{enumerate}

\paragraph{The group decomposition idea and its high probability guarantee for $\lambda_2$ in~\eqref{optimization objective for alg}}
We consider the permutation groups with the following structure: suppose that $\{1,2,...,n\}$ is partitioned as $\{1,2,...,n\}=S_{1}\cup S_{2}\cup ... \cup S_{k}$ and $S_i\cap S_j =\emptyset,\forall 1\leq i\neq j\leq k$. For each $i$, define $\mathcal{Q}_{i}=\{\sigma_{i}\vert \sigma_{i}(j)=j,\forall j\notin S_{i}, j\in\{1,2,\cdots,n\}\}$, the set of permutations of $\{1,2,...,n\}$ that act non-trivially only on indices within $S_i$. We then construct $\mathcal{P}_{K}$ as the set of permutation matrices corresponding to
$
\pi\in\{\sigma_{1}\circ \sigma_{2}\circ...\circ \sigma_{k}\vert \sigma_{i}\in \mathcal{Q}_{i}\}.
$
By this construction, $H^{Z}v_{\pi}$ admits the following form:
$$
H^{Z}v_{\pi} = \sum_{i=1}^{k}\sum_{j\in S_{i}}v_{j}w_{\sigma_{i}(j)}=\bar{v}\sum_{i=1}^{n}w_{i}+\sum_{i=1}^{k}\sum_{j\in S_{i}}(v_{j}-\bar{v})w_{\sigma_{i}(j)}:=v^{*}+\sum_{i=1}^{k}u_{i}\,,
$$
where $u_{i}=\sum_{j\in S_{i}}(v_{j}-\bar{v})w_{\sigma_{i}(j)}$. Thus, $H^{Z}v_{\pi}-v^{*}$ can be viewed as a summation of independent random variables when $\pi$ is chosen uniformly in $\{\sigma_{1}\circ \sigma_{2}\circ...\circ \sigma_{k}\vert \sigma_{i}\in \mathcal{Q}_{i}\}$.

According to Theorem~\ref{thm: informal for high prob bound}, this structured form of $H^{Z_{\pi}}v$ can be analyzed when $v$ is not distributed in a low dimensional subspace of $\mathbb{R}^{n}$. Consequently, for a permutation group $\mathcal{P}_{K}$ constructed as above and with $P{\pi_k}$ drawn uniformly from $\mathcal{P}_{K}$, we obtain:
\begin{align}\label{high probability bound of objective for alg}
\mathbb{P}\left[\frac{1}{2}v^{T}_{\pi_k}v+\Vert H^{Z_{\pi_k}}v\Vert^{2}_{2} \leq \mathbb{E}_{\pi_k}\left[\frac{1}{2}v^{T}_{\pi_k}v+\Vert H^{Z_{\pi_k}}v\Vert^{2}_{2}\right] +o(1)\Vert v\Vert^{2}_{2}\right]\to 1 \,(\text{as }n\to \infty)\,.
\end{align}
This implies that, for any $\alpha > 0$, the quantity $\lambda_2$ admits the following upper bound:\\$\lambda_2(X,Z,\mathcal{P}_{K},\frac{1}{4}\alpha)\leq\mathbb{E}_{\pi_k}\left[\frac{1}{2}v^{T}_{\pi_k}v+\Vert H^{Z_{\pi_k}}v\Vert^{2}_{2}\right] +o(1)\Vert v\Vert^{2}_{2}=\mathbb{E}_{\pi_k}\left[\frac{1}{2}v^{T}_{\pi_k}v+\Vert H^{Z}v_{\pi_k}\Vert^{2}_{2}\right] +o(1)\Vert v\Vert^{2}_{2}$. 
\begin{theorem}[Informal]\label{thm: informal for high prob bound}
Suppose $v_{i}\in\mathbb{R}^{n}(i=1,2,...,m)$ are independent with $\mathbb{E}[v_i]=0$, $\Vert v_i\Vert_2^{2}\leq o(1)\sum_{i=1}^{m}\mathbb{E}[\Vert v_{i}\Vert^{2}_{2}]$, and for any $w\in \mathcal{S}^{n-1}$, we have: $\sum_{i=1}^{n}\mathbb{E}\left[(w^{T}v_{i})^{2}\right]\in o(1)\sum_{i=1}^{m}\mathbb{E}[\Vert v_{i}\Vert^{2}_{2}]$. Then $\Vert \sum_{i=1}^{m}v_{i}\Vert^{2}_{2}\leq (1+o(1))\sum_{i=1}^{m}\mathbb{E}[\Vert v_{i}\Vert^{2}_{2}]$ holds with high probability.
\end{theorem}
\paragraph{Algorithm design and performance analysis}
We now propose Algorithm~\ref{alg structure} to solve the optimization problem~\eqref{optimization problem for alg}, where the permutation group $\mathcal{P}_{K}$ is endowed with the structure discussed above.
\begin{align}\label{optimization problem for alg}
\min_{\mathcal{P}_{K}}\mathbb{E}_{\pi_k}\left[\frac{1}{2}v^{T}_{\pi_k}v+\Vert H^{Z}v_{\pi_k}\Vert^{2}_{2}\right]
\end{align}
\begin{algorithm}[t]
	\caption{Construction of the permutation group} \label{alg structure}
	\begin{algorithmic}
		\STATE \textbf{Input:} Vectors $w_1,w_2,...,w_n\in\mathbb{R}^{n}$ $(w_{i}=H^{Z}e_{i})$, $\vec{v}=(v_{1},...,v_{n})\in\mathbb{R}^{n}$
        \STATE Compute $\bar{v}=\frac{1}{n}\sum_{i=1}^{n}v_{i}$, let $a_{i}=v_{i}-\bar{v}(i=1,2,...,n)$, $b_{i}=\Vert w_{i}\Vert^{2}_{2}(i=1,2,...,n)$
        \STATE Let $\bar{b}=\frac{1}{n}\sum_{i=1}^{n}b_{i}$,  $c_{i}=a^{2}_{i}$, $\bar{c}=\frac{1}{n}\sum_{i=1}^{n}c_i$, $M=\max_{i\in\{1,2,...,n\}}a^{2}_{i}$, $S=\sum_{i=1}^{n}a^{2}_{i}$.
        \STATE Let $I_{1}=\{i\vert (c_i-\bar{c})(b_i-\bar{b})\geq 0\}$,
        \STATE $I_2 = \{i\vert (c_i-\bar{c}<0)\land (b_{i}-\bar{b}>0)   \}$,
        \STATE $I_3=\{i\vert (c_i-\bar{c}>0)\land (b_i-\bar{b}<0)\}$.
        \STATE Call \textbf{Rearrange}~\ref{alg rearrange} on $I_{1},I_{2},I_{3}$, with vectors $(a_i,c_i)(i=1,2,...,n)$ and parameter $M$.
        \STATE Let the output above be $J_1,J_2,J_3$.
        \STATE Call \textbf{Partition}~\ref{alg partition} with parameter $M^{\frac{1}{3}}S^{\frac{2}{3}}$ and obtain subsets $S_{1},S_{2},...,S_{m}$ of $\{1,2,...,n\}$
        \STATE Let $\mathcal{Q}_{i}$ be the set of all permutations on $S_{i}$.
        \STATE Return $\mathcal{P}_{K}=\{P_{\pi}\vert \pi\in\sigma_1\circ...\circ\sigma_{k}, \sigma_{i}\in\mathcal{Q}_{i}\}$.
	\end{algorithmic}
\end{algorithm}
In Algorithm~\ref{alg structure}, the index set $\{1,2,\dots,n\}$ is first partitioned into three subsets $I_1$, $I_2$, and $I_3$ according to the values of $b_i - \bar{b}$ and $c_i - \bar{c}$. Subsequently, the \textbf{Rearrange} step~\ref{alg rearrange} transfers a small portion of elements from $I_2$ and $I_3$ into $I_1$, yielding three new subsets $J_1$, $J_2$, and $J_3$. This adjustment ensures that the expectation in \eqref{optimization problem for alg} can be effectively controlled. In the final step, \textbf{Partition}~\ref{alg partition}, we carefully prescribe the size of each subset $S_i$. This construction simultaneously achieves a small value of $\mathbb{E}_{\pi_k}\left[\frac{1}{2}v^{T}_{\pi_k}v+\Vert H^{Z}v_{\pi_k}\Vert^{2}_{2}\right]$ and satisfies the high‑probability conditions required in Lemma~\ref{lem: approximation of optimization objective}. A detailed description of Algorithm~\ref{alg structure} is provided in \Cref{paragraph: algorithm detail}.

We now evaluate $\lambda_2(X,Z,\mathcal{P}_{K},\frac{1}{4}\alpha)$ for the permutation group $\mathcal{P}_K$ produced by our algorithm, and compare it with the random permutation scheme described in Assumption~\ref{assump::iidpermute}.
On one hand, let $\pi'$ be drawn uniformly from the set of all permutations of $\{1,2,\dots,n\}$. Then, as shown in \ref{subsubsection: comparison with random permutation}, we obtain:
\begin{align*}
\mathbb{E}_{\pi_k}\left[\frac{1}{2}v^{T}_{\pi_{k}}v+\Vert H^{Z}v_{\pi_{k}}\Vert^{2}_{2}\right]&\leq \mathbb{E}_{\pi^{'}}\left[\frac{1}{2}v^{T}_{\pi^{'}}v+\Vert H^{Z}v_{\pi^{'}}\Vert^{2}_{2}\right]+\vert J_{2}\vert (\bar{b}_{2}-\bar{b})(\bar{c}_2-\bar{c})+\vert J_{3}\vert (\bar{b}_{3}-\bar{b})(\bar{c}_3-\bar{c})\\
&+o(1)\Vert v\Vert^{2}_{2}\,,
\end{align*}
where $\bar{b}_{k}=\frac{1}{\vert J_k\vert}\sum_{i\in J_k}b_{i}$, $\bar{c}_k=\frac{1}{\vert J_k\vert}\sum_{i\in J_k}c_i$ $(k=2,3)$, and both of $(\bar{b}_2-\bar{b})(\bar{c}_2-\bar{c})$ and $(\bar{b}_3-\bar{b})(\bar{c}_3-\bar{c})$ are negative. On the other hand, Theorem~\ref{thm: lower bound of lambda informal} provides a strict lower bound of both $\lambda_2(X,Z,\mathcal{P}_K,\frac{1}{4}\alpha)$ and $\lambda_2(X,Z,\mathcal{P}_n,\frac{1}{4}\alpha)$. Combining \eqref{high probability bound of objective for alg} with Theorem~\ref{thm: lower bound of lambda informal} yields Proposition~\ref{prop: gap between ours and random permutation}, which implies that, as $\alpha\to 0$ and neglecting the $o(1)\Vert v\Vert^{2}_{2}$ term, $\lambda_2(X,Z,\mathcal{P}_{K},\frac{1}{4}\alpha)$ is asymptotically no larger than $\lambda_2(X,Z,\mathcal{P}_{n},\frac{1}{2}\alpha)$ when $n\to\infty$, with a provable gap of $\vert J_{2}\vert (\bar{b}_{2}-\bar{b})(\bar{c}_2-\bar{c})+\vert J_{3}\vert (\bar{b}_{3}-\bar{b})(\bar{c}_3-\bar{c})$, which depends on $X,Z$.

\begin{theorem}[Informal lower bound for $\lambda_2$]\label{thm: lower bound of lambda informal}
For any $\mathcal{P}_{K}$ and $\alpha>0,$ the quantity $\lambda_2$ corresponding to $\mathcal{P}_{K}$ satisfies:
$$
\lambda_2(X,Z,\mathcal{P}_{K},\alpha)\geq \mathbb{E}\left[\frac{1}{2}v^{T}_{\pi_{k}}v+\Vert H^{Z}v_{\pi_{k}}\Vert^{2}_{2}\right]-O(\alpha)\Vert v\Vert^{2}_{2}.
$$
Furthermore, suppose $\pi_{1},...,\pi_{m}$ are sampled independently and uniformly from $\mathcal{P}_{K}$. Then if $m\geq \frac{1}{\alpha^{2}}$, for $\lambda$ such that $\frac{1}{m}\sum_{i=1}^{m}\mathbbm{1}\left\{ \frac{1}{2}v^{T}_{\pi_{i}}v+\Vert H^{Z_{\pi_i}}v\Vert^{2}_{2}  \leq \lambda  \right\}\geq 1-\frac{1}{4}\alpha$, with high probability we have: $\lambda\geq\mathbb{E}_{\pi_k}\left[\frac{1}{2}v^{T}_{\pi_{k}}v+\Vert H^{Z}v_{\pi_{k}}\Vert^{2}_{2}\right]-O(\alpha)\Vert v\Vert^{2}_{2}$.
\end{theorem}
\begin{proposition}\label{prop: gap between ours and random permutation}
Let $\mathcal{P}_K$ be obtained from Algorithm~\ref{alg structure}. Then we have:
\begin{align*}
\lambda_2(X,Z,\mathcal{P}_{K},\frac{1}{4}\alpha)\leq \lambda_2(X,Z,\mathcal{P}_n,\frac{1}{2}\alpha)+\vert J_{2}\vert (\bar{b}_{2}-\bar{b})(\bar{c}_2-\bar{c})+\vert J_{3}\vert (\bar{b}_{3}-\bar{b})(\bar{c}_3-\bar{c})+O(\alpha)\Vert v\Vert^{2}_{2}\,.
\end{align*}
Where $\bar{b}_{k}=\frac{1}{\vert J_k\vert}\sum_{i\in J_k}b_{i}$, $\bar{c}_k=\frac{1}{\vert J_k\vert}\sum_{i\in J_k}c_i$ $(k=2,3)$ and $\vert J_{2}\vert (\bar{b}_{2}-\bar{b})(\bar{c}_2-\bar{c})+\vert J_{3}\vert (\bar{b}_{3}-\bar{b})(\bar{c}_3-\bar{c})\leq 0$. This also implies that $\lambda_2(X,Z,\mathcal{P}_{K},\frac{1}{4}\alpha)\leq \lambda_2(X,Z,\mathcal{P}_n,\frac{1}{2}\alpha)+O(\alpha)\Vert v\Vert^{2}_{2}$.
\end{proposition}

\paragraph{Remark} 
We have established an efficient algorithm that optimizes over $\lambda_2(X,Z,\mathcal{P}_{K},\frac{1}{4}\alpha)$, a valid quantity to reduce the Type II error.
In particular, 
when $Z$ satisfies the constraints of Lemma \ref{lem: upper and lower of optimization objective informal} (Sub-Gaussian satisfies the constraints), the optimization objective $\lambda_2(X,Z,\mathcal{P}_{K},\frac{1}{4}\alpha)$ yields a theoretically guaranteed upper estimate for $\lambda_1(X,Z,\mathcal{P}_{K},\frac{1}{2}\alpha)$, where $\lambda_1(X,Z,\mathcal{P}_{K},\frac{1}{2}\alpha)$ is a sufficient and necessary optimization problem for Type II error control. 
As for the effectiveness of our algorithm, if we denote $\mathcal{P}_K^{algorithm}$ obtained from Algorithm~\ref{alg structure}, $\mathcal{P}_n$ denote the i.i.d. random permutation in Assumption~\ref{assump::iidpermute}. From \Cref{prop: gap between ours and random permutation}, we have $\lambda_2(X,Z,\mathcal{P}_{K}^{algorithm},\frac{1}{4}\alpha)\leq \lambda_2(X,Z,\mathcal{P}_n,\frac{1}{2}\alpha)+O(\alpha)\Vert v\Vert^{2}_{2}$, which shows that our algorithm can yield a permutation group that outperforms i.i.d. permutations.

\subsection{Numerical Experiment of Type II Error Control}
We compare the experimental results of our algorithm with those of the random permutation method proposed by \cite{guan2024}. We plot the Type II error curve as a function of $|b|$. Since $\vert\mathcal{P}_{K}\vert$ in our algorithm is prohibitively large and $\phi_1, \phi_2$ cannot be computed directly, we estimate $\phi_1$ and $\phi_2$ using $m \in \Omega(1/\alpha^2)$ samples drawn uniformly at random from $\mathcal{P}_{K}$. Formally, $\phi_1$ is estimated by
$$
\phi_1\approx \frac{1}{m}\sum_{k=1}^{m}\mathbbm{1}\left\{X^{T}H^{ZZ_{\pi_k}}Y\geq X^{T}_{\pi}H^{ZZ_{\pi_k}}Y\right\}
$$
with $P_{k}$ i.i.d. chosen from $\mathcal{P}_{K}$, and the estimation for $\phi_2$ is obtained by replacing ``$\geq$'' with ``$\leq$''. These statistics provide accurate approximations to $\phi_1,\phi_2$, while retaining well-controlled Type I and Type II errors:
\begin{lemma}\label{lem: type II error of finite sample}
	Suppose that we have $P_1,P_2,...,P_m$ sampled independently and uniformly from $P_{K}$, and let $\pi_i$ denote the permutation of $\{1,2,...,n\}\to\{1,2,...,n\} $ corresponding to $P_i$. Let 

$$
\phi^{'}_1=\frac{1}{m}\sum_{i=1}^{m}\mathbbm{1}\left\{X^{T}H^{ZZ_{\pi_{i}}}X\leq X^{T}H^{ZZ_{\pi_{i}}}Y \right\}\,,
$$
$$
\phi^{'}_2=\frac{1}{m}\sum_{i=1}^{m}\mathbbm{1}\left\{X^{T}H^{ZZ_{\pi_{i}}}X\geq X^{T}H^{ZZ_{\pi_{i}}}Y \right\}\,,
$$
and we accept $\mathcal{H}_0$ when $\phi^{'}_1,\phi^{'}_2\in[\alpha, 1]$. Then when $m\geq \Omega(\frac{1}{\alpha^{2}})$, for any constant $c$ we have:
\begin{align*}
\mathbb{P}\left[\mathcal{H}_{0}\text{ is rejected }\vert b=0 \right] \leq 4(1+c)\alpha + e^{-\Omega(\alpha^{-1})}\,,
\end{align*}
\begin{align*}
\mathbb{P}\left[\mathcal{H}_{0}\text{ is accepted } \right]\leq \mathbb{P}\left[ \min(\phi_1,\phi_2)\leq (1-c)\alpha \right]+e^{-\Omega(\alpha^{-1})}\,.
\end{align*}
This implies a well-controlled Type I error and a well-approximated Type II error compared with using $\phi_1,\phi_2$ corresponding to the full permutation group.
\end{lemma}

For the numerical results, we first consider the simplest case where both $X$ and $\epsilon$ follow an i.i.d. standard normal distribution.
In the first experiment, we evaluate both our method and the random permutation method of \cite{guan2024} with $n=200$ and varying $p$, and set $\alpha=0.1$. $Z$ is generated from several distributions, including Gaussian, $t_1$, and $t_2$ distributions. The experimental results are presented in Figure~\ref{fig:gaussian X and gaussian noise}.

For comparison, we also generate $X$ from various distributions. Figure~\ref{fig:t2 X and gaussian noise} illustrates the performance under $t_2$ i.i.d. $X$ and Gaussian noise. These results indicate that our algorithm performs at least as well as the random permutation method, and the performance gap depends on both the distribution of $Z$ and the data dimension $p$. In particular, we highlight a key observation: when $Z$ is heavy-tailed, the improvement of our method over that of \cite{guan2024} can be substantial, especially when $p/n$ is not small (e.g., $p \geq \frac{1}{4}n$). A detailed discussion of these findings is provided in \ref{subsubsection: comparison with random permutation}.

\begin{figure}[t]
    \centering
    \begin{subfigure}[b]{0.32\textwidth}
        \centering
        \includegraphics[width=\textwidth]{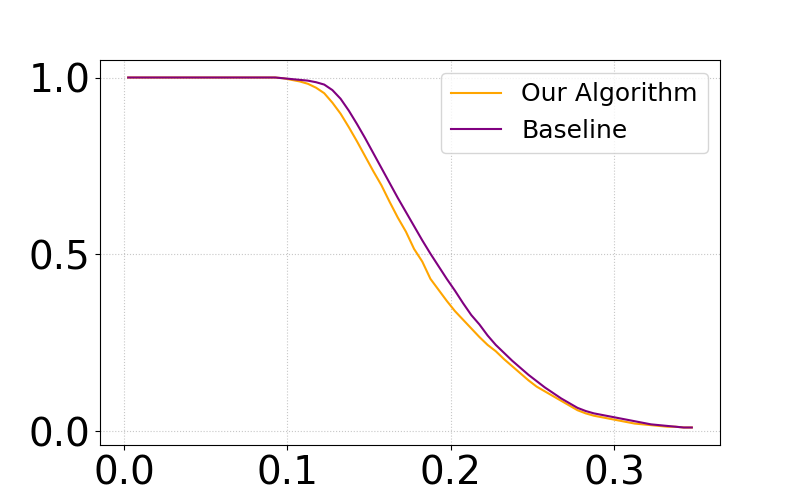}
        \caption{$p=40$, $Z\sim$ Gaussian}
        \label{fig:200_40_g_g_g}
    \end{subfigure}
    \hfill
    \begin{subfigure}[b]{0.32\textwidth}
        \centering
        \includegraphics[width=\textwidth]{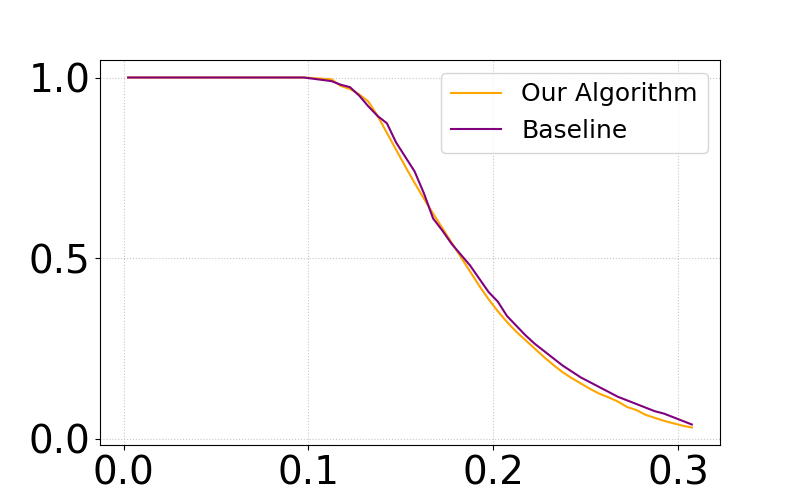}
        \caption{$p=40$, $Z \sim t_1$}
        \label{fig:200_40_g_t1_g}
    \end{subfigure}
    \hfill
    \begin{subfigure}[b]{0.32\textwidth}
        \centering
        \includegraphics[width=\textwidth]{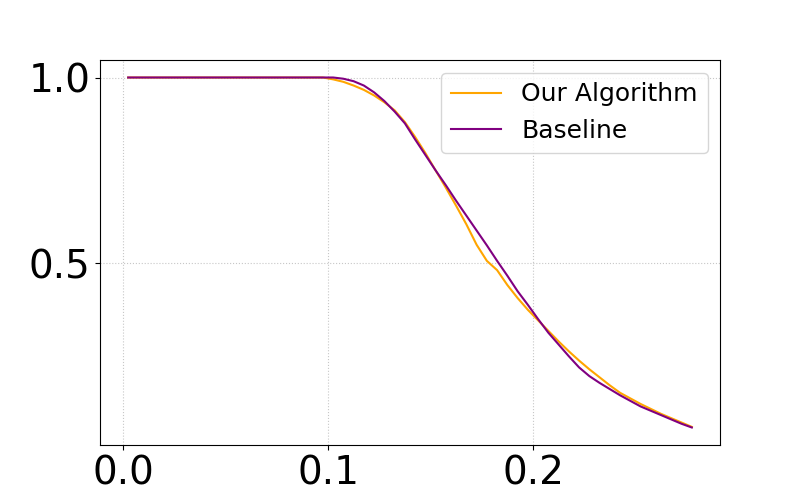}
        \caption{$p=40$, $Z\sim t_2$}
        \label{fig:200_40_g_t2_g}
    \end{subfigure}

    \vspace{0.5cm} 

    \begin{subfigure}[b]{0.32\textwidth}
        \centering
        \includegraphics[width=\textwidth]{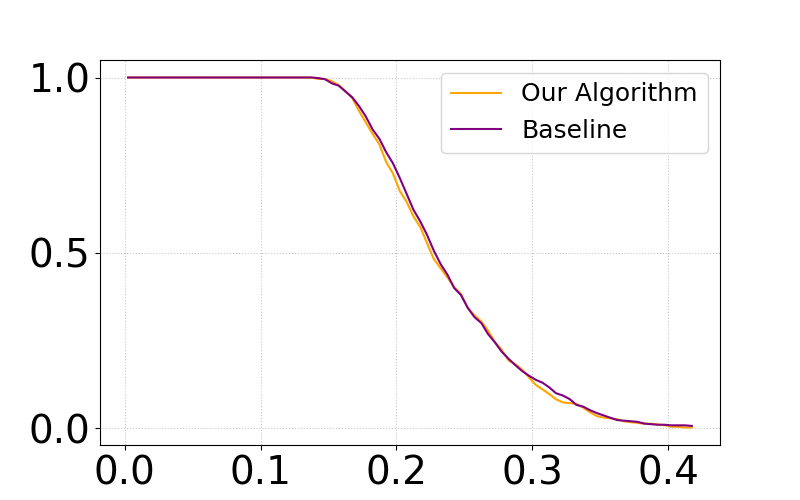}
        \caption{$p=60$, $Z\sim$ Gaussian}
        \label{fig:200_60_g_g_g}
    \end{subfigure}
    \hfill
    \begin{subfigure}[b]{0.32\textwidth}
        \centering
        \includegraphics[width=\textwidth]{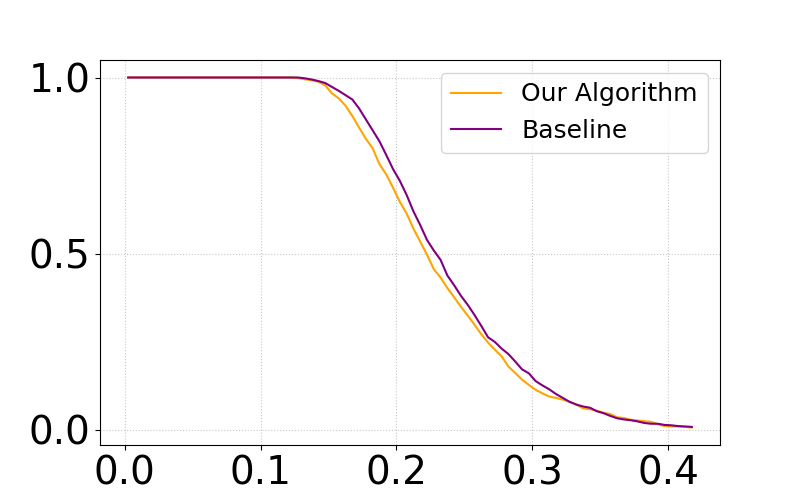}
        \caption{$p=60$, $Z\sim t_1$}
        \label{fig:200_g_t1_g}
    \end{subfigure}
    \hfill
    \begin{subfigure}[b]{0.32\textwidth}
        \centering
        \includegraphics[width=\textwidth]{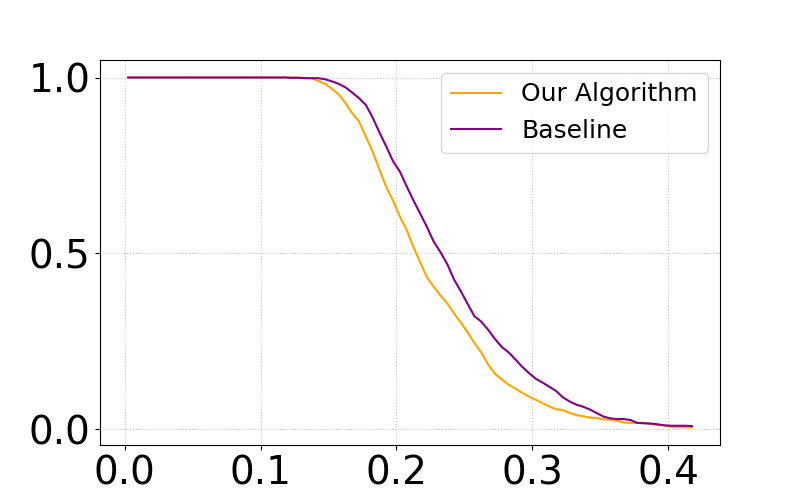}
        \caption{$p=60$, $Z\sim t_2$}
        \label{fig:200_60_g_t2_g}
    \end{subfigure}

    \vspace{0.5cm}

    \begin{subfigure}[b]{0.32\textwidth}
        \centering
        \includegraphics[width=\textwidth]{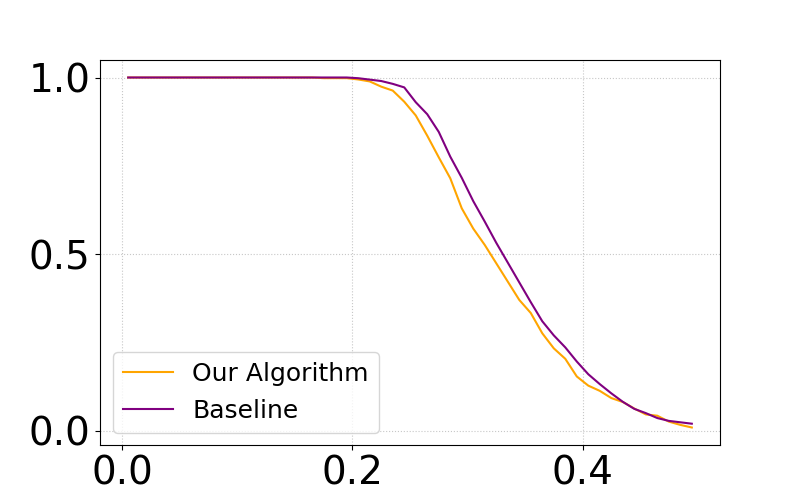}
        \caption{$p=80$, $Z\sim$ Gaussian}
        \label{fig:200_80_g_g_g}
    \end{subfigure}
    \hfill
    \begin{subfigure}[b]{0.32\textwidth}
        \centering
        \includegraphics[width=\textwidth]{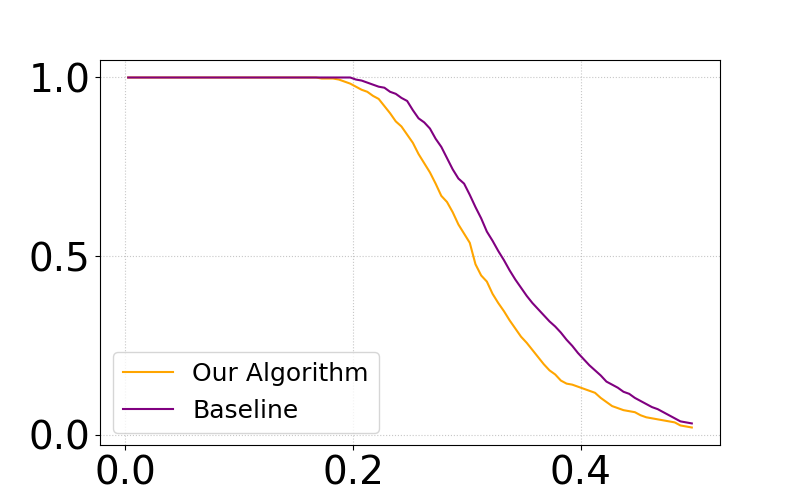}
        \caption{$p=80$, $Z\sim t_1$}
        \label{fig:200_80_g_t1_g}
    \end{subfigure}
    \hfill
    \begin{subfigure}[b]{0.32\textwidth}
        \centering
        \includegraphics[width=\textwidth]{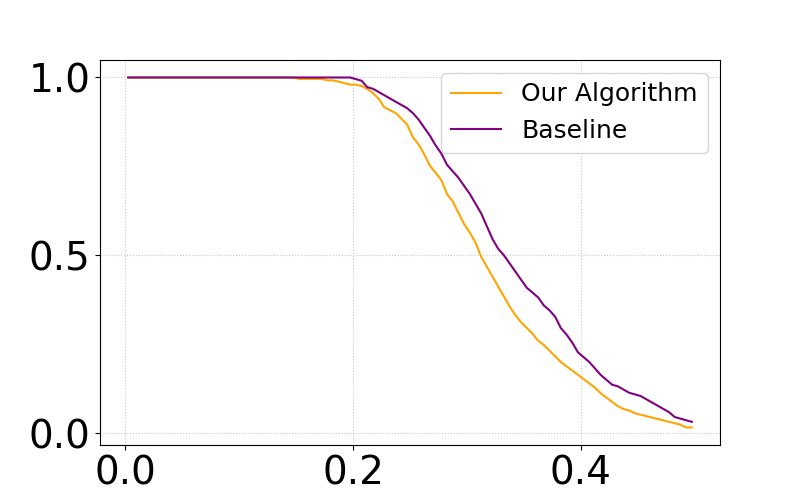}
        \caption{$p=80$, $Z\sim t_2$}
        \label{fig:200_80_g_t2_g}
    \end{subfigure}

    \caption{Type II Error for $X\sim$ Gaussian and $\epsilon\sim$ Gaussian, with $n=200$ and $\alpha=0.1$. Each row corresponds to a different dimension $p$, and each column shows a different distribution for $Z$.}
    \label{fig:gaussian X and gaussian noise}
\end{figure}

\begin{figure}[t]
    \centering
    \begin{subfigure}[b]{0.32\textwidth}
        \centering
        \includegraphics[width=\textwidth]{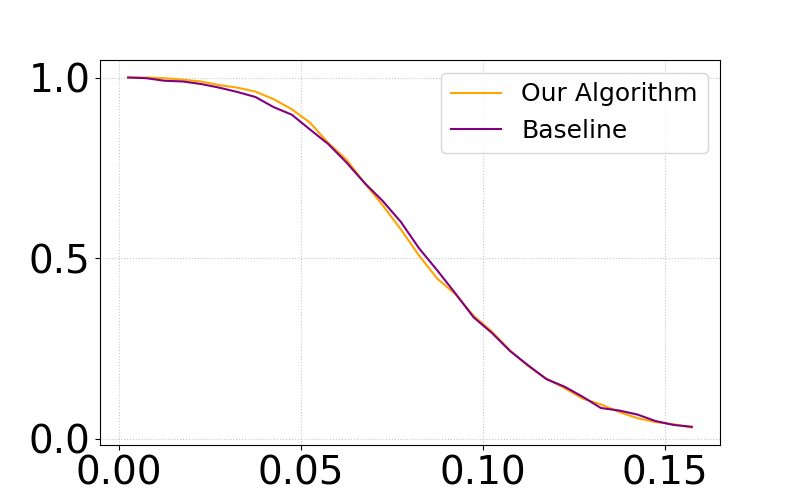}
        \caption{$p=50$, $Z\sim$ Gaussian}
        \label{fig:200_50_t2_g_g}
    \end{subfigure}
    \hfill
    \begin{subfigure}[b]{0.32\textwidth}
        \centering
        \includegraphics[width=\textwidth]{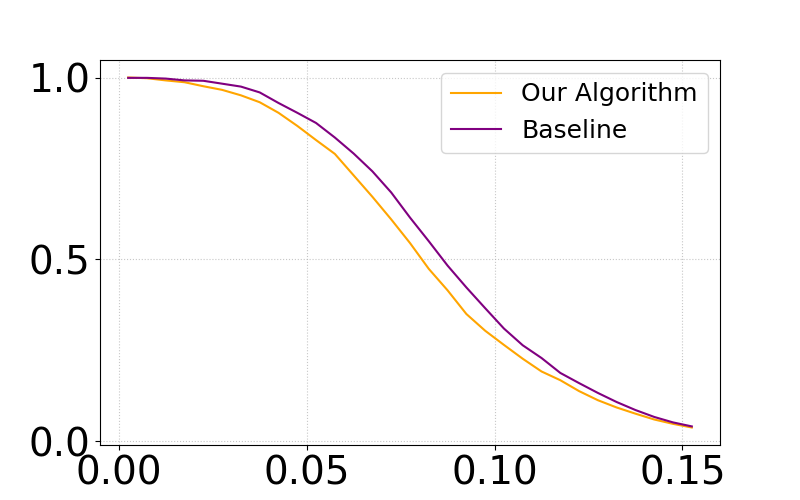}
        \caption{$p=50$, $Z\sim t_1$}
        \label{fig:200_50_t2_t1_g}
    \end{subfigure}
    \hfill
    \begin{subfigure}[b]{0.32\textwidth}
        \centering
        \includegraphics[width=\textwidth]{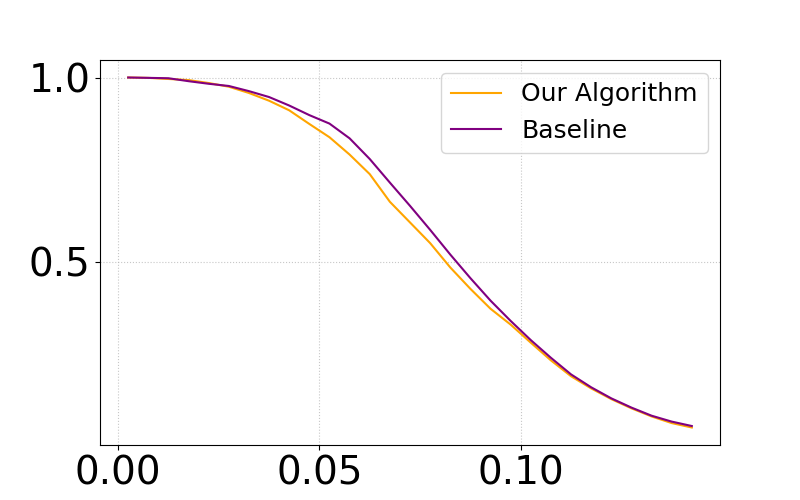}
        \caption{$p=50$, $Z\sim t_2$}
        \label{fig:200_50_g_t2_g}
    \end{subfigure}

    \vspace{0.4cm} 

    \begin{subfigure}[b]{0.32\textwidth}
        \centering
        \includegraphics[width=\textwidth]{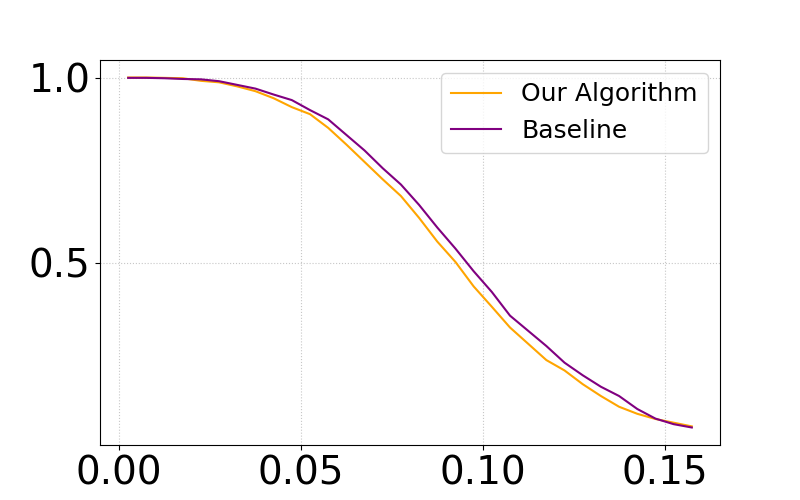}
        \caption{$p=60$, $Z\sim$ Gaussian}
        \label{fig:200_60_t2_g_g}
    \end{subfigure}
    \hfill
    \begin{subfigure}[b]{0.32\textwidth}
        \centering
        \includegraphics[width=\textwidth]{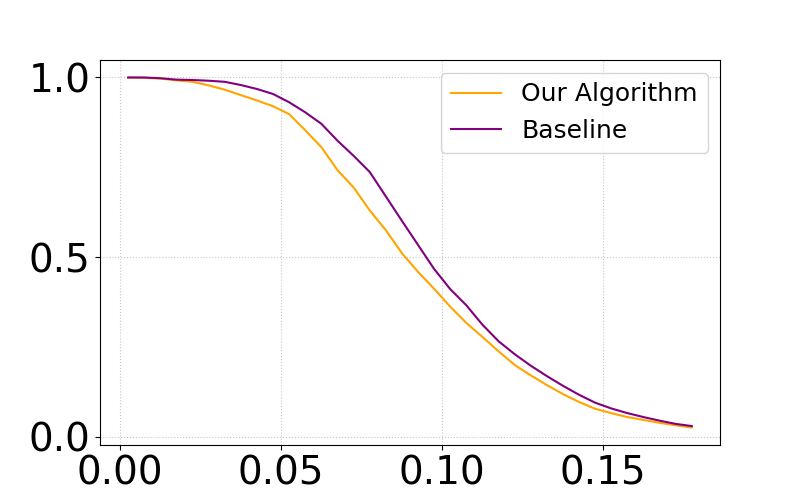}
        \caption{$p=60$, $Z\sim t_1$}
        \label{fig:200_t2_t1_g}
    \end{subfigure}
    \hfill
    \begin{subfigure}[b]{0.32\textwidth}
        \centering
        \includegraphics[width=\textwidth]{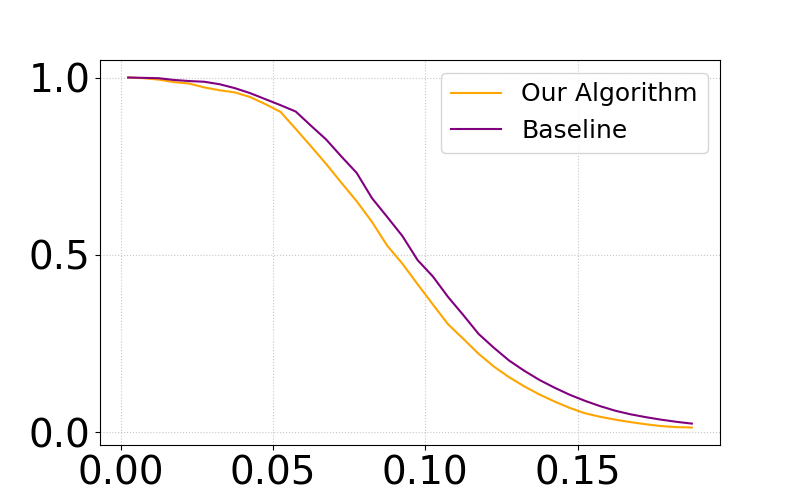}
        \caption{$p=60$, $Z\sim t_2$}
        \label{fig:200_60_t2_t2_g}
    \end{subfigure}

    \vspace{0.4cm}

    \begin{subfigure}[b]{0.32\textwidth}
        \centering
        \includegraphics[width=\textwidth]{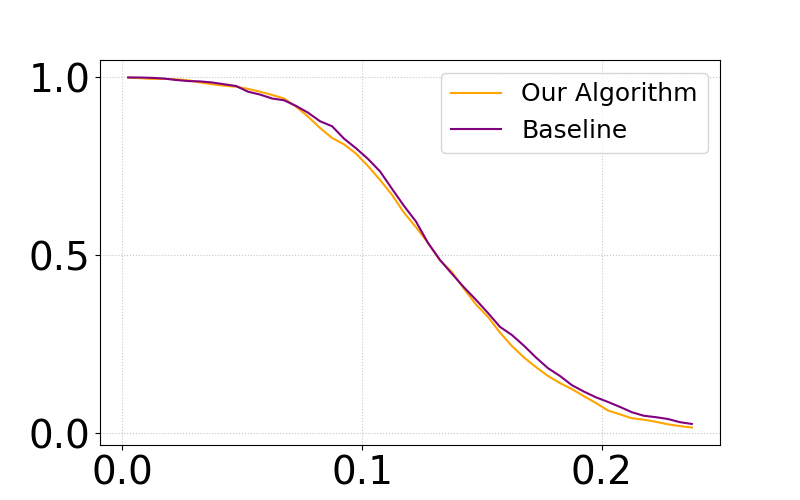}
        \caption{$p=80$, $Z\sim$ Gaussian}
        \label{fig:200_80_t2_g_g}
    \end{subfigure}
    \hfill
    \begin{subfigure}[b]{0.32\textwidth}
        \centering
        \includegraphics[width=\textwidth]{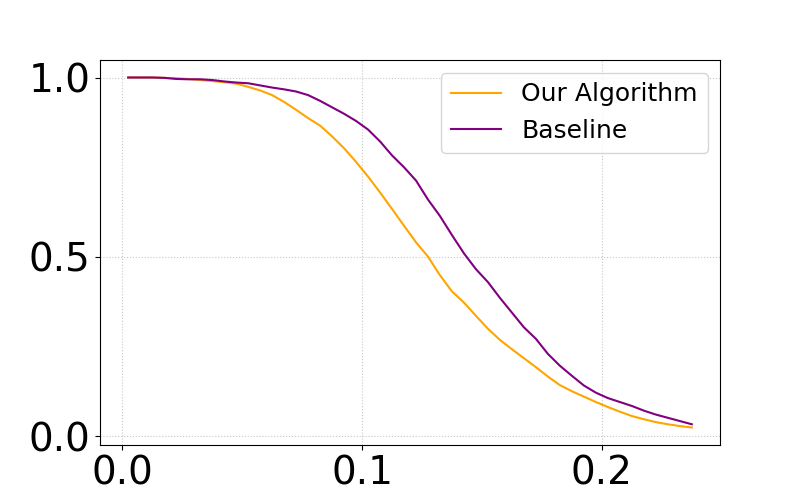}
        \caption{$p=80$, $Z\sim t_1$}
        \label{fig:200_80_t2_t1_g}
    \end{subfigure}
    \hfill
    \begin{subfigure}[b]{0.32\textwidth}
        \centering
        \includegraphics[width=\textwidth]{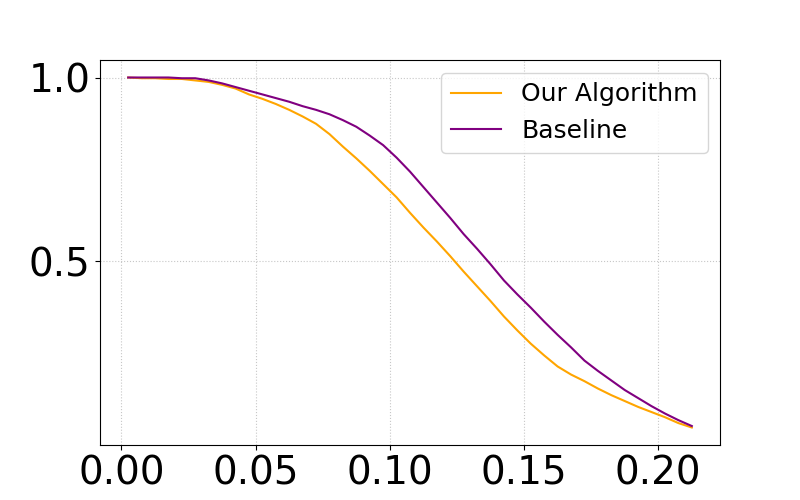}
        \caption{$p=80$, $Z\sim t_2$}
        \label{fig:200_80_t2_t2_g}
    \end{subfigure}

    \caption{Type II Error comparisons for $X\sim t_2$ and $\epsilon \sim$ Gaussian, with $n=200$ and $\alpha=0.1$. Each row represents a different dimension $p$, and columns represent different distributions of $Z$.}
    \label{fig:t2 X and gaussian noise}
\end{figure}

\subsection{CPT with group permutation}
\label{sec::CPT with group permutation}
In this section, we consider relaxing \Cref{assump::leftshifting} in CPT to \Cref{assump::grouppermutation}, and we reformulate the proof of the Type I error control of CPT under \Cref{assump::exchanglable}. The two aspects we address in this section provide insight into controlling the Type I error of CPT with \Cref{assump::grouppermutation} and without \Cref{assump::exchanglable} in \Cref{sec::CPTnonexchangeable}.

Initially, the CPT constructed from \citet{lei2021} relies on invariance under the left-shifting operator, as shown in \Cref{sec::exact_tests submateral}. In this section, we generalize CPT from the left-shifting group $P_K^L:=\{P_k\}_{k=0}^K$ to an arbitrary permutation group $\mathcal{P}_K$ satisfying \Cref{assump::grouppermutation}, where $K$ is the hyperparameter.

Assuming $\epsilon$ is exchangeable, we can also provide a Type I error bound. Through the group permutation set $\mathcal{P}_K$ and the corresponding pairwise exchangeability of $\epsilon$, with the same intuition as \cite{lei2021}, we can obtain a set of test statistics $\{Y^{T}\eta_0, \cdots,Y^T\eta_k,\cdots,Y^T\eta_{K}\}$, where $Y^{T}\eta_0=Y^{T}\eta^*$. These statistics satisfy constraints analogous to those in Condition \ref{con::lei}, 
\begin{condition}
\label{condition 2}
    There exist $\gamma\in \mathbb{R}^p$ such that 
    $Z^T\eta_{k}=\gamma$,where $k\in \{0,\cdots,K\}$.
\end{condition}
\begin{condition}
\label{condition 3}
    For $k\in\{0,\cdots,K\}$,
      $Y^T\eta_k=Y^T (P^T_k\eta^*), P_k\in \mathcal{P}_K.$
\end{condition}
As shown in \Cref{prop::Proposition for group permutation}, whenever $P_k\in \mathcal{P}_K$, we also have $P_k^T \in \mathcal{P}_K$.
Then, regarding the existence of the $(\eta^*,\gamma)$, since $\eta_j$ is a linear transformation of $\eta^*$, Condition \ref{condition 2} and Condition \ref{condition 3} imply that 
\begin{equation}
\label{eq::grouped CPT equation}\begin{pmatrix}
        -I_{p} & Z^TP_0 \\
        -I_{p} & Z^TP^T_1 \\
        \vdots & \vdots \\
        -I_{p} & Z^TP^T_K 
    \end{pmatrix}\begin{pmatrix}
        \gamma_{} \\
        \eta^*
    \end{pmatrix}=
    \begin{pmatrix}
        -I_{p} & Z^TP_0 \\
        -I_{p} & Z^TP_1 \\
        \vdots & \vdots \\
        -I_{p} & Z^TP_K 
    \end{pmatrix}
    \begin{pmatrix}
        \gamma_{} \\
        \eta^*
    \end{pmatrix}
    = 0,
\end{equation}
where the first equation follows from \Cref{prop::Proposition for group permutation}, since whenever $P_k\in \mathcal{P}_K$, we also have $P^T_k\in \mathcal{P}_K$. The above linear system has $(K+1)\cdot p$ equations and $n+p-1$ unknowns. Then, there always exists a non-zero solution if $(K+1)p<n+p-1.$

Also, from \Cref{prop::Proposition for group permutation}, we know that $I\in \mathcal{P}_K$, so we denote $P_0:=I$. The test statistic is defined as 
$$R_0=\frac{1}{K+1}\sum_{j=1}^K 1\{Y^T\eta^*\leq Y^T\eta_{j}\},$$
Then we can also define 
$$R_{k}=\frac{1}{K+1}\sum_{j=0,j\neq k}^K 1\{Y^T(P^T_k\eta)\leq Y^T\eta_{j}\},\,\,k=1,\cdots,K.$$
where $\eta_{k}=P^T_k\eta^*$, $\eta_0=\eta^*$.
\begin{theorem}
\label{thm::thorem of group CPT}
Suppose that $(X,Z,Y)$ is generated under model \eqref{eq::model}. Suppose $\mathcal{P}_K=\{P_k:k=0,\cdots,K\}$ satisfies \Cref{assump::grouppermutation}, and the noise $\epsilon$ satisfies \Cref{assump::exchanglable}.
\begin{enumerate}
    \item The linear system (\ref{eq::grouped CPT equation}) has a non-zero solution iff
     $n/p>K.$
     \item  Under $H_0$, and for any solution of $(\eta^*,\gamma)$, we have 
    $$P(R_0\leq \frac{\lceil(1-\alpha)(K+1)\rceil}{K+1})=P(R_0 \leq Q_{1-\alpha}(\sum_{k=0}^{K}\frac{1}{K+1}\cdot\delta_{R_k}))\geq 1-\alpha,$$
    where $Q_{\tau}(\cdot)$ denotes the $\tau$- quantile of its argument, $\delta_{a}$ denotes the point mass at a, and the first equation holds when there are no ties.
\end{enumerate}
    \end{theorem}

    \textbf{Remark}. From \Cref{thm::thorem of group CPT}, we have $P(R_0>\frac{\lceil(1-\alpha)(K+1)\rceil}{K+1})<\alpha.$ Compared with the original theorem \eqref{eq::lei type I error bound} in \citet{lei2021}, we find that the original work uses $[0,\alpha]$ as the rejection region for the hypothesis $H_0$, while in the above \Cref{thm::thorem of group CPT}, if we choose $[\frac{\lceil(1-\alpha)(K+1)\rceil}{K+1},1]$ as the rejection region for the hypothesis $H_0$, we can still obtain an $\alpha$-level Type I error bound. So under the group permutation, the probability density of the test statistic $R_0 \in [0,1]$ has very small tail probabilities for both sides.

The reason for changing the rejection region to the other side is that it provides more insight to extend the exchangeable case to the nonexchangeable case, as in Theorem \ref{thm::Theorem nonexchangeable grouped CPT}. Then, the discussion of Type II error in the original work can be naturally extended to the side we chose, and it is also natural to consider a two-sided rejection region. Based on the above discussion on $R_{0}$, we can slightly extend the algorithm by~\cite{lei2021} to obtain an optimization algorithm that solves for an $\eta^{*}$ with good Type II error performance for any permutation group $\mathcal{P}_{K}=\{P_{0}=I_{n},P_{1},...,P_{K}\}$. The optimization objective is as follows:

\begin{equation}
\begin{aligned}
&\max_{\eta\in R^{n}, \gamma\in R^{p},\Vert \eta\Vert_{2}=1,\gamma_1,\delta\in\mathbb{R}} \delta\quad \text{subject to} \\
& \begin{pmatrix}
        -I_{p} & Z^TP_0 \\
        -I_{p} & Z^TP_1 \\
        \vdots & \vdots \\
        -I_{p} & Z^TP_K 
    \end{pmatrix}
    \begin{pmatrix}
        \gamma_{} \\
        \eta
    \end{pmatrix}
    = 0,\
\begin{pmatrix}
X^{T}\\
X^{T}P_{1}\\
\vdots \\
X^{T}P_{K}
\end{pmatrix}
\eta 
=\begin{pmatrix}
\gamma_{1}+\delta\\
\gamma_1\\
\vdots \\
\gamma_1
\end{pmatrix}.
\end{aligned}
\end{equation}

\section{Relaxing exchangeability to nonexchangeability}
\label{sec::The relax of exchangeable to nonexchangeable}
As in model \eqref{eq::model}, we restate the setting where observations $(Z,X,Y)\in \mathbb{R}^{n \times p}\times \mathbb{R}^{n}\times \mathbb{R}^{n}$ satisfy the following model:
\begin{equation}
    Y=Z\beta+bX+\epsilon,
\end{equation}
where $\epsilon:=(\epsilon_1,\cdots,\epsilon_n)^T$ is an $n$-dimensional noise vector, and our goal is to test the null hypothesis $H_0: b=0$ against the alternative hypothesis $H_1: b \neq 0$.
  We also keep the permutation group $\mathcal{P}_K:=\{P_1,\cdots,P_K\}$ under the group permutation in \Cref{assump::grouppermutation};  $(\pi_k)_{k=1}^{K}$ is defined correspondingly to $\{P_k\}_{k=1}^K \in \mathcal{P}_K$.

  In this section, we want to relax the assumption of exchangeable noise in \Cref{assump::exchanglable}: for any permutation matrix $P\in \mathcal{P}$,
$$(\epsilon_1,\cdots,\epsilon_n) \overset{d}{=} P(\epsilon_1,\cdots,\epsilon_n).$$

As mentioned above in \Cref{sec::exact_tests}, previous work treats exchangeability as the core assumption. The core idea behind the methods in \cite{guan2024,lei2021,wen2025} is to extract exchangeability from $\epsilon$ to obtain identical distributions between their test statistics to ensure exact Type I error control. Can we still obtain Type I error control without the exchangeability assumption? We compare our setting with conformal prediction to illustrate the similarities and differences in \Cref{sec::Exchangeability in conformal prediction}. 

\textbf{Remark}: The insight is that in prediction tasks, for each statistic $i$, the current data point $i$ and the data point $n+1$ that needs to be predicted will have a higher \textbf{``status''} than other data points because what we ultimately want to obtain is a prediction of this data point $n+1$. Therefore, we only need exchangeability with respect to pairwise permutation matrices to measure the relationship between the predicted data point $n+1$ and the current data point $i$. For the inference problem, our goal is to estimate unknown parameters using all existing $n$ data points. Then, in this problem, we can consider that these $n$ data points have the same \textbf{``status''} for our research. For statistics, only the test statistic $R_0$ has a different weight because of the covariate shifting problem, while the other $R_k,k\in\{1,\cdots, K\}$ have the same weight. Therefore, we consider group permutation matrices to measure the overall relationship between all data points and these unknown parameters. The theoretical details can be referred to \Cref{sec::Exchangeability in conformal prediction}.

\subsection{CPT with group permutation beyond exchangeability}
\label{sec::CPTnonexchangeable}
We now turn to the nonexchangeable case of model \eqref{eq::model}, where $\epsilon$ no longer satisfies \Cref{assump::exchanglable}. We first consider CPT under group permutation. We retain a pair ($\gamma ,\eta^*$) satisfying \eqref{eq::grouped CPT equation} and the permutation group $\mathcal{P}_K$ satisfying \Cref{assump::grouppermutation}.
We also denote the test statistic $R_0$ as in \Cref{sec::CPT with group permutation}, which is given by 
$$R_0:=R_0(Y)=\frac{1}{K+1}\sum_{j=1}^K1\{Y^T\eta_0\leq Y^T\eta_j\},$$
and 
\begin{equation}
       \label{R_k}
       R_k:=R_k(Y)=\frac{1}{K+1}\sum_{j=0,j\neq k}^K1\{Y^T\eta_k\leq Y^T\eta_j\}, k=1,\cdots,K.
\end{equation}

where $\eta_{k}=P_k^T\eta^*,\eta_0=\eta^*$. Under the null hypothesis $H_0:b=0$, recall the reconstruction of  CPT with group permutation in \Cref{sec::CPT with group permutation}, the rejection region is constructed as 
$$ [Q_{1-\alpha}(\sum_{i=0}^{K}\frac{1}{K+1}\cdot \delta_{R_i(Y)}),1].$$

We randomly permute the data sequence in our dataset using our permutation group $\mathcal{P}_K$. First, draw a random index $\mathcal{K}\in [K]$ from the multinomial distribution that assigns probability $1/(K+1)$ to each index $i$:
\begin{equation}
\label{mathcalK}
    \mathcal{K} \sim \sum_{i=0}^{K} \frac{1}{K+1} \cdot \delta_{\{i\}}.
\end{equation}
In this case, define 
$$R_k(Y^{\mathcal{K}})=\frac{1}{K+1}\sum_{j=0,j\neq k'}^K1\{(P_{\mathcal{K}}Y)^T(P^T_k\eta^*)\leq (P_{\mathcal{K}}Y)^T\eta_j\}, k=1,\cdots,K.$$

Then, without the exchangeability of $\epsilon$, assuming only that the permutation group $\mathcal{P}_K$ satisfies \Cref{assump::grouppermutation} and the constraint $n/p \geq K$ obtained from \eqref{eq::grouped CPT equation}, we can achieve Type I error control as follows:
\begin{theorem}
\label{thm::Theorem nonexchangeable grouped CPT}
 Under $H_0$, if the permutation group $\mathcal{P}_K:=\{P_0,P_1,\cdots,P_K\}$ satisfies \Cref{assump::grouppermutation}
  then we obtain 
    \begin{equation*}
        \begin{aligned}
        P(R_0 \leq Q_{1-\alpha}
    (\sum_{k=0}^{K}\frac{1}{K+1}\cdot\delta_{R_k}))& \geq 1-\alpha-\sum_{k=1}^K\frac{1}{K+1}\,\,(d_{TV}(R(\epsilon),R(\epsilon^k)) \\
    & \geq 1-\alpha-\sum_{k=1}^K\frac{1}{K+1}\,\,(d_{TV}(\epsilon,\epsilon_{\pi_k})),
        \end{aligned}
    \end{equation*}
    where $Q_{\tau}(\cdot)$ denotes the $\tau$- quantile of its argument, $\delta_{a}$ denotes the point mass at $a$, and $R(\epsilon)$ is a $K+1$-dimensional vector that $(R(\epsilon))_{i}=R_{i-1}(\epsilon),i\in\{1,\cdots,K+1\}.$
    \end{theorem}

The above theorem explains why exchangeability yields the validity bound in \Cref{thm::thorem of group CPT}: exchangeability makes the total variation (TV) distance between $\epsilon$ and $\epsilon_{\pi_k}$ equal to $0$ for every $\pi_k\in \mathcal{P}_K$. In practice, achieving exchangeability may be hard; however, we can ensure that the TV distance described above is not too large so that it can still yield an acceptable Type I error bound. The detailed proof and a more generalized version of Lemma \ref{lem::Theorem nonexchangeable grouped CPT} can be found in Section
\ref{detailed proof of theorem nonexchangeable grouped CPT}.

\subsection{PALMRT with group permutation beyond exchangeability}
To overcome the constraint $n/p \geq K$ implied by \eqref{eq::grouped CPT equation}, we turn our attention to the Grouped PALMRT without exchangeability in Section \ref{Grouped PALMRT} which has the form
 \begin{equation}
       \phi= \frac{1}{K+1}\sum_{k=1}^{K}1\{X^{T}(I-H^{ZZ_{\pi_k}})Y > (X)_{\pi_k}^{T}(I-H^{ZZ_{\pi_k}})Y\}.
   \end{equation}

The permutation group $\mathcal{P}_K:=\{I(:=P_0),P_1,\cdots,P_K\}$ also satisfies \Cref{assump::grouppermutation}, since from \Cref{prop::Proposition for group permutation}, $I\in\mathcal{P}_K$, for simplicity we denote $P_0:=I$. We then define the matrix of core residual statistics $F(\epsilon)\in \mathbb{R}^{(K+1)\times(K+1)}$. 
 \begin{equation*}
 (F(\epsilon))_{i,j}:=F(\pi_{i-1},\pi_{j-1};x,Z,\epsilon)=X_{\pi_{i-1}}^{T}(I-H^{Z_{\pi_{i-1}}Z_{\pi_{j-1}}})\epsilon, i,j\in\{1,\cdots, K+1\}.
   \end{equation*}
 We directly observe that $\pi_k$ belongs to the corresponding set of permutations.
 $X^T(I-H^{ZZ_{\pi_k}})Y=X^T(I-H^{ZZ_{\pi_k}})\epsilon,\,(X)_{\pi_k}^T(I-H^{ZZ_{\pi_k}})Y=(X)_{\pi_k}^T(I-H^{ZZ_{\pi_k}})\epsilon.$
   \begin{theorem}
   \label{thm::Theorem nonexchangeable grouped PALMRT}
Under $H_0$, if $\mathcal{P}_K:=\{I(:=P_0),P_1,\cdots,P_K\}$ satisfies \Cref{assump::grouppermutation}, then we have 
 \begin{equation*}
     \begin{aligned}
         P(\sum_{k=1}^K \frac{1}{K+1} \cdot 1\{X^{T}(I-H^{ZZ_{\pi_k}})&Y > (X)_{\pi_k}^{T}(I-H^{ZZ_{\pi_k}})Y\}< 1-\alpha)\\&\geq 1-2\alpha-\sum_{k=1}^K\frac{1}{K+1}\cdot d_{TV}(T(\epsilon),T((\epsilon)^k))
         \\&\geq 1-2\alpha-\sum_{k=1}^K\frac{1}{K+1} \cdot d_{TV}(\epsilon,\epsilon_{\pi_k}).
     \end{aligned}
 \end{equation*}
    \end{theorem}

In summary, Theorem~\ref{thm::Theorem nonexchangeable grouped PALMRT} quantifies how the exact Type I validity of group-based PALMRT under exchangeability extends to nonexchangeable settings, with explicit error inflation controlled by total variation distances.

\section{Conclusion}\label{conclusion}

This paper develops a group-based perspective on permutation testing for linear models,
clarifying how finite-sample validity, power, and robustness are governed by the underlying
symmetry structure.
By formulating permutation-augmented regression tests within an explicit group framework,
we establish sharp Type I error bounds, provide a principled handle on Type II error through
group optimization, and show how exact-style inference can be extended beyond exchangeability
via stability guarantees inspired by weighted conformal inference.
Together, these results unify classical exact tests and their robust counterparts within a
single finite-sample framework.

Several limitations and directions for future work remain.
Designing permutation groups that are provably optimal for power under complex dependence
structures is largely open.
Extending to high-dimensional or multiple testing problems, and to sequential
or adaptive experimental designs, also presents substantial challenges.
Finally, while total variation distance yields transparent robustness bounds, developing less conservative discrepancy measures for structured nonexchangeability is also important for further research.


\clearpage


\bibliography{reference}

\newpage
\appendix
\begin{center}
    {\huge Supplementary material for ``Group Permutation Testing in Linear Models: Sharp Validity, Power Improvement, and Extension Beyond Exchangeability''}
\end{center}

This supplementary material provides the detailed theoretical proofs, algorithmic specifics, and expanded numerical analyses for the results presented in the main text. After the literature review (\Cref{sec:lit}), the contents are organized as follows:

\Cref{sec::Section A} contains the formal proofs for the exchangeable noise setting, specifically addressing the sharp Type I error bounds for PALMRT under group permutations (\Cref{Thm::PALMRT} in \Cref{Sec::section of proof of theorem 1}) and the construction of worst-case scenarios to demonstrate tightness (\Cref{proposition of 2alpha} in \Cref{sec::Proof of proposition factor 2}). It also includes the proof for the group-based Cyclic Permutation Test (CPT) under exchangeability (\Cref{thm::thorem of group CPT} in \Cref{sec::Proof of grouped CPT}).

\Cref{sec::Section B} provides a comprehensive analysis of Type II error control. It details the formulation of the optimization problem in \Cref{sec::optimization problem}, the combination of group decomposition with the optimization function (\eqref{optimization objective for alg}), together with a high-probability guarantee, in \Cref{sec::Group high prob guarantee}. The full specification and theoretical guarantees for Algorithm 1 and its subroutines (Algorithms 2–6) used for adaptive group construction in \Cref{sec::The algorithm for group decomposition} are also provided. Then, we carefully compare our algorithm with the random permutation in \Cref{subsubsection: comparison with random permutation}. The detailed proofs of the theorems and lemmas are shown in \Cref{sec::Type ii error proof of main results}.

Section~\ref{proof_nonex} extends the framework beyond the exchangeability assumption. It leverages connections to weighted conformal inference to prove the robustness guarantees for CPT (\Cref{thm::Theorem nonexchangeable grouped CPT} in \Cref{detailed proof of theorem nonexchangeable grouped CPT}) and PALMRT (\Cref{thm::Theorem nonexchangeable grouped PALMRT} in \Cref{sec::Proof of PALMRT beyond exchangeable}) in nonexchangeable settings, quantifying error inflation through total variation distances.

\section{Literature review}\label{sec:lit}

Exact finite-sample inference for regression coefficients has a long history but has recently seen renewed interest driven by modern regimes where the design is fixed, $n$ is moderate, and the noise law
may be complex or only partially characterized.
This paper concerns testing a single target coefficient $b$ in the fixed-design linear model
$Y = Z\beta + bX + \epsilon$ with nuisance $Z$ and unknown noise $\epsilon$.
Our review emphasizes three intertwined themes: (i) how exactness is obtained from symmetry,
(ii) how the chosen transformation set affects power, and (iii) how one can quantify the loss of calibration
when symmetry is only approximate.

\subsection{Randomization and permutation inference}


A convenient way to formalize finite-sample randomization inference is through an explicit
\emph{group action}. Let $\mathcal{Y}$ denote the sample space of the observed data object $Y$
(e.g., the response vector under a fixed design), and let $\mathcal{G}$ be a finite set of transformations
$g:\mathcal{Y}\to\mathcal{Y}$. We say that $\mathcal{G}$ acts on $\mathcal{Y}$ if it forms a group under
composition (contains the identity, is closed under composition, and contains inverses), so that
each $y\in\mathcal{Y}$ generates an \emph{orbit} $\mathcal{O}(y)=\{g\cdot y: g\in\mathcal{G}\}$.
A null model $H_0$ is \emph{$\mathcal{G}$-invariant} if the law of $Y$ satisfies
$Y \stackrel{d}{=} g\cdot Y$ for all $g\in\mathcal{G}$ (equivalently, the likelihood is constant on orbits);
see, e.g., \citet{eaton1989group,lehmann2005testing,good2005permutation,pesarin2010permutation}. Under $\mathcal{G}$-invariance, exact randomization $p$-values arise by comparing a test statistic
$T(Y)$ to its orbit $\{T(g\cdot Y): g\in\mathcal{G}\}$ when $g$ is drawn uniformly from $\mathcal{G}$.
The key point is that group structure makes the induced randomization distribution {well-defined} and {orbit-wise identical}: the conditional law of $\{T(g\cdot Y): g\in\mathcal{G}\}$ does not depend on which
representative of the orbit we start from, yielding finite-sample exactness (up to the usual discreteness/ties)
for level-$\alpha$ tests; see the classical symmetry arguments in \citet{hoeffding1952large} and modern
treatments in \citet{lehmann2005testing,edgington2007randomization}.

In many classical linear testing problems---e.g., ANOVA under exchangeable errors---the admissible
transformations are naturally a permutation group (often the full symmetric group, or a transitive subgroup)
acting on sample indices \citep{fisher1935design,pesarin2010permutation}. In regression with nuisance
parameters or data-dependent preprocessing, however, the transformations actually used in practice
(e.g., residual permutation schemes, studentization, or permutation after fitting) may no longer be closed under
composition or may depend on the data through estimated nuisance quantities. In such cases, the effective
collection of transformations fails to form a group, and the exact orbit-invariance argument can break down
\citep{anderson2001permutation,freedman1983nonstochastic,winkler2014permutation}.
This tension between finite-sample exactness and model complexity motivates methods that either (i) recover an
explicit group action by construction or (ii) relax exact invariance to approximate symmetry with explicit error
control \citep{janssen1997studentized,chung2013exact,canay2017randomization,diciccio2017robust}.

\subsection{Exact finite-sample tests for fixed-design regression under exchangeable errors}

Building on the group-invariance perspective discussed above
\citep{eaton1989group,lehmann2005testing,hoeffding1952large},
a growing literature studies \emph{finite-sample exact} hypothesis testing for regression
coefficients under a fixed design by exploiting the exchangeability of the regression errors.
Rather than relying on asymptotic normality, these approaches aim to construct test
statistics whose null distribution is invariant---exactly or conservatively---under
collections of transformations that preserve the joint law of the error vector
\citep{freedman1983nonstochastic,pesarin2010permutation,edgington2007randomization}.

First, a representative example is the cyclic permutation test (CPT) of \citet{lei2021}, which
achieves exactness by explicitly encoding exchangeability through a structured cyclic
permutation group.
By imposing linear constraints that eliminate nuisance coefficients, CPT constructs linear
statistics whose joint distribution under $H_0$ is invariant to cyclic shifts.
This yields finite-sample exact tests under minimal distributional assumptions on the errors,
fully aligned with the classical group-invariance paradigm
\citep{eaton1989group,lehmann2005testing}.
From a regression perspective, CPT can be interpreted as testing the independence between the
target coefficient and the response after projecting out nuisance effects.
A key limitation is that the associated linear constraint system admits nontrivial solutions
only under restrictive dimensional regimes; these constraints become increasingly severe when
finer randomization resolution or richer fixed designs are desired. Second, to relax the algebraic constraints required by CPT, \citet{wen2025} propose the residual
permutation test (RPT).
Instead of enforcing an explicit group action, RPT elegantly projects the response onto orthogonal
complements of augmented design spaces before applying permutation inference, in the spirit
of residual-based permutation schemes for regression
\citep{freedman1983nonstochastic,winkler2014permutation}.
This relaxation substantially broadens applicability while preserving finite-sample Type~I
error control under exchangeable errors.
Viewed through the lens of coefficient independence testing, RPT aggregates evidence across
permutations via a conservative minimum operator.
While such aggregation guarantees validity even in the absence of exact group closure, it
typically induces conservativeness; consequently, the power properties of the test---and
their dependence on the chosen permutation set---are nontrivial to characterize sharply
\citep{janssen1997studentized,chung2013exact}.

More recently, \citet{guan2024} proposes permutation-augmented linear model regression tests
(PALMRT), which replace conservative aggregation by pairwise permutation comparisons.
This modification can yield improved empirical power and admits explicit nonasymptotic
Type~I error guarantees under exchangeable errors.
However, the resulting Type~I bounds are generally not sharp when
viewed through the classical orbit-invariance lens
\citep{lehmann2005testing,chung2013exact}.
Moreover, sampling permutations from the full symmetric group obscures the geometry of the
transformation set, a phenomenon already noted in the broader permutation-testing literature
\citep{pesarin2010permutation,edgington2007randomization}, limiting interpretability and
principled power optimization for coefficient independence testing.

Taken together, these works illustrate a fundamental trade-off between explicit invariance,
algebraic feasibility, and power characterization in finite-sample regression testing under
exchangeable errors.

\subsection{Beyond exchangeability: robustness under approximate symmetry}

The preceding sections highlight that finite-sample exactness in randomization and
permutation inference is fundamentally tied to explicit symmetry or group-invariance
assumptions \citep{eaton1989group,lehmann2005testing,hoeffding1952large}.
In fixed-design regression, recent advances demonstrate that such symmetry can be
operationalized under exchangeable errors to yield exact tests for regression coefficients
\citep{lei2021,wen2025,guan2024}.
At the same time, these constructions also make clear that exact validity hinges on the
correct specification of the invariance structure and may fail under even mild departures
from exchangeability.

In practice, exact symmetry assumptions are frequently violated by heteroskedasticity,
dependence, batch effects, or other forms of structured noise.
A substantial classical literature therefore studies permutation procedures that remain
\emph{asymptotically} valid under non-i.i.d.\ sampling by means of studentization,
self-normalization, or related devices; see, for example,
\citet{janssen1997studentized}, \citet{chung2013exact}, and \citet{romano2005stepdown}.
While these approaches recover asymptotic size control under broad conditions, their validity
rests on large-sample approximations and typically does not admit explicit nonasymptotic
calibration guarantees indexed by a measurable deviation from exchangeability or symmetry.

A complementary and more quantitative perspective emerges from conformal prediction.
Under exact exchangeability, conformal methods are known to achieve exact finite-sample
coverage without distributional assumptions
\citep{vovk2005algorithmic,lei2018distribution,shafer2008tutorial}.
More recently, \citet{barber2023conformal} show that when exchangeability is violated,
weighted conformal procedures retain coverage up to an explicit degradation term controlled
by total variation distances between the data distribution and its swapped counterparts.
Related stability-based guarantees appear in
\citet{fannjiang2022conformal,kivaranovic2024conformal}, and build on conservative
importance-weighting arguments developed in the covariate-shift and domain-adaptation
literature \citep{harrison2012importance,sugiyama2007covariate}.
Collectively, these results replace exact invariance with a quantitative notion of
\emph{approximate symmetry}, yielding finite-sample guarantees that degrade gracefully with
the degree of symmetry violation.

Importantly for our setting, this stability-based conformal viewpoint aligns naturally with group-structured permutation designs. When the admissible transformations form a finite and interpretable collection of group
actions---as in the randomization and regression settings discussed above---departures from
exact invariance can be localized to a small, structured family of transformations.
This localization renders stability measures both computable and meaningful.
By contrast, in settings involving unrestricted permutations, complex dependence graphs, or
sequential data, the number of admissible swaps grows combinatorially, causing
total-variation-based stability bounds to become either vacuous or analytically intractable
\citep{pesarin2010permutation,edgington2007randomization}.
The alignment between group-structured transformations and stability-aware conformal
calibration therefore potentially provides a natural and principled bridge between exact randomization
inference and robust finite-sample guarantees under approximate symmetry.
\subsection{Preliminaries for linear exact test with permutation methods}
\label{sec::exact_tests submateral}


Recently, several works have focused on inference for regression parameters in the linear model~\eqref{eq::model}, leveraging permutation methods under exchangeable noise. Since the noise vector $\epsilon$ is unobserved and only enters the data through the response $Y$, a common theme across this literature is to transfer (or ``propagate'') the exchangeability of $\epsilon$ to suitably constructed test statistics. This guiding principle is made transparent by the discussion below. \citet{lei2021} introduce the cyclic permutation test (CPT).
    Based on the model \eqref{eq::model}, they construct the linear statistics as $
      S_j=Y^T\eta_{j}, j=0,1,\cdots,m,
 $
  where $m \leq n,m\in \mathbb{N}^*$ and  $\eta_j$ satisfy the following two conditions,
  \begin{condition}[Conditions on the vectors $\eta_j$]\label{con::lei}
    (i)
     $(Z^T\eta_j)^T=C,j=0,1,\cdots,m$.
    (ii)
    $\eta_j= \eta^*_{\pi_j}, j=0,1\cdots,m,$~where permutation operators  $\{\pi_k\}_{k=1}^n$ satisfy Assumption~\ref{assump::leftshifting} below.
  \end{condition}

  \begin{assumption}[Left shifting permutation]
\label{assump::leftshifting}
We construct a set of permutation matrices $\mathcal{P}^L_{m}:=\{P_k\}_{k=0}^{m}, P_k\in \mathcal{P}_n$, corresponding to the operator $\{\pi_k\}_{k=0}^{m}$. 

For each $j\in \{0,1,\cdots,m\}$, the permutation matrix $P_j$ is defined as follows: Let the block size be $t=\lfloor n/(m+1)\rfloor$. The entry $(P_j)_{rs}$ of the matrix $P_j$, for row and column indices $r,s\in\{0,1,\cdots,n-1\}$, is given by:
\begin{equation*}
  (P_j)_{rs} = 
  \begin{cases} 
    1 & \text{if } r, s < (m+1)t \text{ and } \lfloor s/t \rfloor = (\lfloor r/t \rfloor + j) \pmod{m+1} \text{ and } s \pmod{t} \\ 
    1 & \text{if } r, s \ge (m+1)t \text{ and } r = s \\ 
    0 & \text{otherwise} 
  \end{cases}
\end{equation*}
For simplicity, for $\eta \in \mathbb{R}^{n}$, let  $t = \lfloor n / (m+1) \rfloor$, we have 

\[
\eta
= \Bigl(
\underbrace{\eta_{1},\ldots,\eta_{t}}_{\text{block }1},
\ldots,
\underbrace{\eta_{mt+1},\ldots,\eta_{(m+1)t}}_{\text{block }(m+1)},
\underbrace{\eta_{(m+1)t+1},\ldots,\eta_{n}}_{\text{tail}}
\Bigr)^{\top}.
\]

\[P_j \eta = \Bigl( \underbrace{\eta_{jt+1}, \dots, \eta_{(j+1)t}}_{\text{block~} (j+1)}, \dots, \underbrace{\eta_{mt+1}, \dots, \eta_{(m+1)t}}_{\text{block~} (m+1)}, \underbrace{\eta_{1}, \dots, \eta_{t}}_{\text{block 1}}, \dots, \underbrace{\eta_{(j-1)t+1}, \dots, \eta_{jt}}_{\text{block j}}, \underbrace{\eta_{(m+1)t+1}, \dots, \eta_{n}}_{\text{tail part}} \Bigr)^{\top}.\]

\end{assumption}

 The definition of $P_j$ in~\Cref{assump::leftshifting} is the permutation matrix version of the definition of the left shifting operator $\pi_L^{j}$ in \citet[C~2]{lei2021}. The two constraints for $\eta_j$ in Condition \ref{con::lei} allow us to establish a system of linear equations to determine the value of $(\eta^*,C).$ Then, their test statistic is the rank of $S_0$ among $\{S_j\}_{j=0}^{m}$ in descending order, denoted as $R_0$.
The idea behind their method is that they extract the exchangeability from the noise $\epsilon$ through the parameter $\eta$. Under $H_0: b=0$, their method divides the $Y^T\eta_j$ into two parts:
  $$Y^T\eta_j=\underbrace{(Z^T\eta_j)^T\beta}_{\text{deterministic part}}+\underbrace{\epsilon^T\eta_j }_{\text{stochastic part}},$$
where $\eta_j, j=1,\cdots,m$ is determined through the two constraints above. Then, $\epsilon$ satisfies \Cref{assump::exchanglable}, which ensures the  control of Type I error.
  \begin{equation}
  \label{eq::lei type I error bound}
      P(\frac{R_0}{m+1}\leq \alpha)\leq \alpha.
  \end{equation}                                
For Type II error, under $H_1$ the invariance of S no longer holds. The idea is therefore to make $S_0$ sufficiently separated from the other $S_j$, $j = 1,\ldots,m,$ and then formulate an optimization problem to control the theoretical Type II error. However, the system of linear equations formed by Condition \ref{con::lei} has a non-zero solution for $(\eta^*,C)$ if
  \begin{equation}
   \label{eq::lei constraint}
      n>(\frac{1}{\alpha}-1)p\,
  \end{equation} which is difficult to obtain in large dimensions. In order to relax the constraint \eqref{eq::lei constraint}, \citet{wen2025} introduce the residual permutation method (RPT) to improve the extraction of exchangeability from $\epsilon$, thereby avoiding the use of the two constraints for $\eta_j$. They introduce the permutation matrices $\mathcal{P}_K:=\{P_k\}_{k=1}^K$ that satisfy \Cref{assump::grouppermutation}. They construct their statistics as: 
  \begin{equation}
  S_k=X^T(I-H^{ZZ_{\pi_k}})Y ;~~S_k'=X^T(I-H^{ZZ_{\pi_k}})Y_{\pi_k},
  \end{equation}
  where  $H^{*}$ denotes the projection matrix onto the column space of its arguments, and suppose $(Z,Z_{\pi_k})$ has full column rank.
  Under $H_0:\,\,b=0$, their method can be simplified as: 
$$X^T(I-H^{ZZ_{\pi_k}})Y=\underbrace{X^T(I-H^{ZZ_{\pi_k}})Z\beta}_{\text{orthogonal part}}+\underbrace{X^T(I-H^{ZZ_{\pi_k}})\epsilon }_{\text{stochastic part}},$$
and 
$$X^T(I-H^{ZZ_{\pi_k}})Y_{\pi_k}=\underbrace{X^T(I-H^{ZZ_{\pi_k}})Z_{\pi_k}\beta}_{\text{orthogonal part}}+\underbrace{X^T(I-H^{ZZ_{\pi_k}})\epsilon_{\pi_k} }_{\text{stochastic part}}$$
Since both $Z$ and $Z_{\pi_k}$ lie in the column space onto which $H_{ZZ_{\pi_k}}$ projects, the orthogonal terms vanish: $X^T (I - H_{ZZ_{\pi_k}}) Z \beta = 0$ and $X^T (I - H_{ZZ_{\pi_k}}) Z_{\pi_k} \beta = 0$. Through the above construction, they transfer the exchangeability from noise $\epsilon$ and $\epsilon_{\pi_k}$ to their statistics $S_k$ and $S'_k$. To preserve the stochastic-dominance property of the statistics relative to the uniform distribution (see Lei and Bickel [2021, Proposition 1]), they introduce the minimum operator and retain a rank-based test statistic $\phi_1$:

\begin{equation*}
       \phi_1:=\frac{1}{K+1}(1+\sum_{k=1}^K\mathbbm{1}\{\min_{\widetilde{H}\in\{H^{ZZ_{\pi_1}},\cdots,H^{ZZ_{\pi_{K}}}\}}X^T(I-H^{ZZ_{\pi_k}})Y\leq X^T(I-H^{ZZ_{\pi_k}})Y_{\pi_k}\}).
   \end{equation*}
By projecting onto the complementary space, they transfer the exchangeability to the statistics $S_k,S_k'$ without using Condition \ref{con::lei}. They can also ensure Type I error control.
 \begin{equation}
     \label{eq::Wen type I error control}
     P(\phi_1\leq \alpha)\leq \alpha,
 \end{equation}
 under the constraints $p<\frac{n}{2}$, which is a substantial relaxation of the previous constraints \eqref{eq::lei constraint}. 

 However, to maintain the Type I constraint in (\ref{eq::Wen type I error control}), \citet{wen2025} introduce the minimum operator, which makes this test conservative and affects the Type II error. \citet{guan2024} aims to remove the minimum operator. Their method is also built on the exchangeability of $\epsilon$ and on permutation-based inference. However, they introduce \Cref{assump::iidpermute} to construct the permutation matrices.
 $(\pi_k)_{k=0}^{K}$ is defined in  correspondence with permutation matrices $\{P_0,\cdots, P_K\} \in \mathcal{P}_n$. After constructing the permutation matrices, they define their test statistic as follows:
 $$\begin{aligned}
     \phi_2 =\frac{1}{K+1}(1+\sum_{k=1}^K\mathbbm{1}\{X^T(I-H^{ZZ_{\pi_k}})Y\leq X_{\pi_k}^T(I-H^{ZZ_{\pi_k}})Y\}).
 \end{aligned}$$
The test uses pairwise comparisons rather than a rank-based construction, leveraging both exchangeability of $\epsilon$ and Assumption \ref{assump::iidpermute}. Also, with the constraint $p<\frac{n}{2}$, they ensure  Type I error control.
 \begin{equation}
 \label{guan Type I error}
     P(\phi_2\leq \alpha)\leq 2\alpha,
 \end{equation} 
Their empirical results show that, while maintaining valid Type I error control, the procedure achieves better Type II error performance than \cite{wen2025}’s statistic $\phi_1$.

\subsection{Exchangeability in conformal prediction}
\label{sec::Exchangeability in conformal prediction}

Conformal prediction is now a widely used method for model-free prediction, and it also relies heavily on the exchangeability assumption. 
Suppose we have i.i.d. training data $Z_i=(X_i,Y_i)\in\mathbb{R}^d\times\mathbb{R},i=1,\cdots,n,$ (we denote $Z_{1:n}:=\{Z_1,\cdots,Z_n\}$ and $Z_{-i}=Z_{1:n}\backslash\{Z_i\}$) and a new test point $X_{n+1},Y_{n+1}$ drawn independently from the same distribution $P$. The problem of conformal prediction is as follows: if we observe training data $\{(X_i,Y_i)\}_{i\in\{1,\cdots,n\}}$ and are given a new feature vector $X_{n+1}$ for the new test point, can we form a prediction interval $C$ depending on $\{(X_i,Y_i)\}_{i\in\{1,\cdots,n\}}$ and $X_{n+1}$, denoted as $\hat{C}_n(X_{n+1})$, to guarantee that, without assumptions on the data distribution $P$, $$P\{Y_{n+1}\in \hat{C}_n(X_{n+1})\}\geq 1-\alpha,$$
for some target coverage level $1-\alpha$.

We focus on classical full conformal prediction as in~\citet{vovk2005algorithmic}, \citet{lei2018distribution}. The core idea behind full conformal prediction is to choose a \textit{nonconformity score} $S((x,y),Z)$. Informally, the nonconformity score $S$ is used to verify the correlation between $(x,y)$ and Z, and it measures how well the point $(x,y)$ conforms to $Z$. A high value of $S((x,y),Z)$ indicates that $(x,y)$ is atypical relative to the points in $Z$.

In its basic form, the full conformal prediction uses the score function $$S(y,z)=|y-\hat{\mu}(x)|,$$
where $\hat{\mu}$ represents a regression model $\hat{\mu}:\mathbb{R}^d\xrightarrow[]{}\mathbb{R}$ that was fitted using the training data $(y,x)$ and $Z$.
Then, conformity scores for each data point are calculated $$V^{(x,y)}_i=S(Z_i,Z_{-i}\cup\{(x,y)\}),i=1,\cdots,n \,\, \text{and } V^{(x,y)}_{n+1}=S((x,y),Z_{1:n}),$$ and the prediction interval is constructed using the former conformity scores. \begin{equation}
\label{conformal prediction set}
    \hat{C}_n(X)=\{y\in \mathbb{R}:V_{n+1}^{(x,y)}\leq \text{Quantile}(1-\alpha;V^{(x,y)}_{1:n}\cup\{\infty\})\}. 
\end{equation}
\begin{assumption}(Pairwise exchangeable)
\label{assumption for pairwise exchangeable}
We write $Z_i=(X_i,Y_i)$ to denote the $i$-th data point and we write $$Z=(Z_1,\cdots,Z_n,Z_{n+1}),$$
to denote the whole data set (training and test). Consider $n$ permutation matrices $P^{pairwise}_i\in \mathcal{P}$, for each $i\in\{1,\cdots,n\},$ 
$$P^{pairwise}_i(Z_1,\cdots,Z_n,Z_{n+1})=(Z_1,\cdots,Z_{i-1},Z_{n+1},Z_{i+1},\cdots,Z_n,Z_{i}),$$
 $\mathcal{P}^{pairwise}_n:=\{P^{pairwise}_i\}_{i=1}^n$ and pairwise exchangeability means that $$Z\overset{d}{=}P^{pairwise}_{i}Z,i\in\{1,\cdots,n\},$$
we say that $Z$ is pairwise exchangeable under $\{I,P^{pairwise}_1,\cdots,P^{pairwise}_n\}.$
\end{assumption}
Then, from \cite{vovk2005algorithmic} and \cite{lei2018distribution}, the full conformal prediction has the following theoretical guarantee: for any $\alpha\in(0,1)$ and $(X_i,Y_i)\in\mathbb{R}^d\times\mathbb{R},i=1,\cdots,n+1$ satisfying Assumption $\ref{assumption for pairwise exchangeable}$ and $\hat{C}_n$ formed as (\ref{conformal prediction set}), then we obtain \begin{equation}
\label{Therom bound for confromal prediction}
    P\{Y_{n+1}\in\hat{C}_n(X_{n+1})\}\geq 1-\alpha,
\end{equation}

\textit{Proof sketch}: In the original statement of the theorem in \cite{vovk2005algorithmic,lei2018distribution}, it is necessary for $(X_i,Y_i)\in\mathbb{R}^d\times\mathbb{R},i=1,\cdots,n+1$ to be exchangeable; however, through a simple modification, we can weaken the exchangeable assumption to the pairwise-exchangeable setting in Assumption \ref{assumption for pairwise exchangeable} to obtain the same error bound as in (\ref{Therom bound for confromal prediction}). The reason for weakening this assumption is to clarify the contribution of exchangeable (pairwise exchangeable) to ensuring prediction accuracy, making it easier to shift to discussions on the nonexchangeable situation.
The core idea of the proof is that 
$$Y_{n+1}\in\hat{C}_n(X_{n+1}) \Longleftrightarrow V^{(X_{n+1},Y_{n+1})}_{n+1}\leq\text{Quantile}(1-\alpha;V^{(X_{n+1},Y_{n+1})}_{1:n}\cup\{\infty\})$$

which yields identical distributions for $V^{(X_{n+1},Y_{n+1})}_{i},i=1,\cdots,n,n+1$
under the assumption of pairwise exchangeability.

There is also a growing literature on relaxing the exchangeability assumption. \citet{cauchois2024robust} and \citet{tibshirani2019conformal} consider the shift of the distributions; through their theoretical proof, they need some further conditions and assumptions about the distribution of $X_i$, which are beyond the scope of our discussion. \cite{barber2023conformal} consider full conformal prediction beyond exchangeability. Under nonexchangeability, the identical distribution of $V^{(X_{n+1},Y_{n+1})}_{i},i=1,\cdots,n,n+1$ defined above will not exist anymore; they define the weights $\tilde{w}_i$ as
\begin{equation}
\label{tildew_i}
    \tilde{w}_{i}=\frac{w_{i}}{w_{1}+\cdots+w_{n}+1}, \quad i=1, \ldots, n \quad \text { and } \quad \tilde{w}_{n+1}=\frac{1}{w_{1}+\cdots+w_{n}+1},
\end{equation}
where $w_i$s are user-defined hyper-parameters.  
$\tilde{w}_i$s are used as weights for each unconformity score $V^{(X_{n+1},Y_{n+1})}_{i}$, so that reweighting the prediction interval is (informally)  
\begin{equation}
\label{reweighting conformal prediction set}
    \hat{C}_n(X)=\{y\in \mathbb{R}:V_{n}^{(x,y)}\leq \text{Quantile}(1-\alpha;\sum_{i=1}^n\tilde{w}_i\delta _{V^{(x,y)}_{i}}\}. 
\end{equation}
Then, from Theorem 2 in \cite{barber2023conformal}, informally, the nonexchangeable full conformal method defined in (\ref{reweighting conformal prediction set}) satisfies 
\begin{equation}
    \label{conformal prediction beyond exchangeable}
    P(Y_{n+1}\in \hat{C}_n(X_{n+1}))\geq 1-\alpha-\sum_{i=1}^n\tilde{w}_i\cdot d_{TV}(Z,P^{pairwise}_iZ),
\end{equation}
where $P^{pairwise}_i\in\mathcal{P}$ is defined in Assumption \ref{assumption for pairwise exchangeable}.

To connect this with our setting, the covering rate of the prediction interval in the full conformal prediction method is determined by the relationship between $Z$ and its transforms $P^{pairwise}_iZ,i=1,\cdots,n$, which are determined by the pairwise permutation matrices $\{I,P_1^{pairwise},\cdots,P_n^{pairwise}\}$, specifically, by the TV distances between $Z$ and $P^{pairwise}_iZ$, $i=\{1,\cdots,n\}$. When the exchangeability holds, the TV distance becomes zero, so that the conclusion (\ref{Therom bound for confromal prediction}) holds.

 The brief introduction of full conformal prediction under exchangeability and nonexchangeability provides the key insight for studying parameter inference under a linear model beyond exchangeability. Having discussed this above, the guarantee of prediction coverage is determined by the relationship amongst $Z$ and its transforms $P^{pairwise}_iZ$, which are determined by the matrices $\mathcal{P}^{pairwise}_{n}=\{I,P_1^{pairwise},\cdots,P_n^{pairwise}\}$. 
For the parameter inference under a linear model in \eqref{eq::model}, we consider the permutation matrices $\mathcal{P}_K:=\{I,P^{group}_1,\cdots,P^{group}_n\}$ that satisfy \Cref{assump::grouppermutation}. For CPT with group permutation discussed in \Cref{sec::CPT with group permutation}, we can prove below in \Cref{thm::Theorem nonexchangeable grouped CPT} that without exchangeability,  
$$P(R_0 \leq Q_{1-\alpha}
    (\sum_{k=0}^{K}w_k\cdot\delta_{R_k})) \geq 1-\alpha-\sum_{i=1}^Kw_k\cdot d_{TV}(\epsilon,P_i^{group}\epsilon),$$
where we defined 
\begin{equation}
    w_i:=\frac{1-w_0}{K}, w_0\geq w_i ,i\in\{1,\cdots,K\},
\end{equation}
$w_0$ is the user-specified hyper-parameter. Compared with equation (\ref{conformal prediction beyond exchangeable}), we see that the Type I error bound of the parameter inference under the linear model is determined by the TV distance between $\epsilon$ and $P_i^{group}\epsilon, i\in \{1,\cdots,n\}$.

The permutation matrices used in the prediction task are naturally pairwise permutation matrices since $\{I,P_1^{pairwise},\cdots,P_n^{pairwise}\}$ satisfies Assumption \ref{assumption for pairwise exchangeable}, while the inference task considers  group permutation matrices since $\{I,P_1^{group},\cdots,P_n^{group}\}$ satisfies \Cref{assump::grouppermutation}.

\section{Proofs for the exchangeable case}
\label{sec::Section A}
\subsection{Proof of \texorpdfstring{\Cref{Thm::PALMRT}}{Theorem 2}}
\label{Sec::section of proof of theorem 1}
\subsubsection{Proof of Proposition \ref{prop::Finverse}}
Initially, for the statistic $X_{\pi_1}^{T}(I-H^{Z_{\pi_1}Z_{\pi_2}})\epsilon_\sigma$, we consider the following regression model,
$$\epsilon_\sigma= Z_{\pi_1}\beta_1+Z_{\pi_2}\beta_2+\eta_1,
$$
where $\eta_1$ is the error term. If we use least squares to estimate $(\hat{\beta_1},\hat{\beta_2})$ for the coefficients $(\beta_1,\beta_2)$, since $(Z_{\pi_1},Z_{\pi_2})$ is full column rank, the residual $\bar{\eta}$ can be written as:
$$\bar{\eta}_1=\epsilon_\sigma-(Z_{\pi_1},Z_{\pi_2})((\hat{\beta_1})^T,(\hat{\beta_2})^T)^T=(I-H^{Z_{\pi_1}Z_{\pi_2}})\epsilon_\sigma,$$
the vector $(I-H^{Z_{\pi_1}Z_{\pi_2}})\epsilon_\sigma$ is the least-squares residual, where $\epsilon_\sigma$ is treated as the dependent variable, and $(Z_{\pi_1},Z_{\pi_2})$ as the regressors. Then, we rewrite it as,
\begin{equation*}
    P_{\sigma}\epsilon=P_{1}Z\hat{\beta_1}+P_{2}Z\hat{\beta_2}+\bar{\eta}_1,
\end{equation*}
since, by definition, $\epsilon_\sigma=P_{\sigma}\epsilon,$ $Z_{\pi_1}=P_{1}Z,$ $Z_{\pi_2}=P_{2}Z$, and hence, after a simple transformation,
\begin{equation}
\label{residual epsilon}
\epsilon=P^T_{\sigma}P_{1}Z\hat{\beta_1}+P^T_{\sigma}P_{2}Z\hat{\beta_2}+P^T_{\sigma}\bar{\eta}_1,
\end{equation}
Denote $(P_{1}Z,P_{2}Z):=Z*$. Next, we show that $(\hat{\beta_1},\hat{\beta_2})=((Z*)^T Z*)^{-1}(Z*)^T\epsilon_{\sigma}$ is still the least squares estimator for $(\beta_3,\beta_4)$ in the following model:
$$\epsilon=P^T_{\sigma}P_1Z\beta_3+P^T_{\sigma}P_2Z\beta_4+\eta_2,$$
where $\eta_2$ is the error term. Through the definition, since $P_{\sigma}$ is a permutation matrix,
\begin{equation*}
\begin{aligned}
(\hat{\beta_1},\hat{\beta_2})&=((Z*)^T Z*)^{-1}(Z*)^TP_{\sigma}\epsilon\\
&=[(P_{\sigma}^TZ*)^T(P_{\sigma}^TZ*)]^{-1}(P_{\sigma}^TZ*)^TP_{\sigma}^TP_{\sigma}\epsilon\\
&=[(P_{\sigma}^TZ*)^T(P_{\sigma}^TZ*)]^{-1}(P_{\sigma}^TZ*)^T\epsilon,
\end{aligned}
\end{equation*}
where $P_{\sigma}^TZ*=(P^T_{\sigma}P_{1}Z,P^T_{\sigma}P_{2}Z)$, so we obtain
\begin{equation}
\label{residual epsilon2}
\epsilon=P^T_{\sigma}P_{1}Z\hat{\beta_1}+P^T_{\sigma}P_{2}Z\hat{\beta_2}+\bar{\eta}_2,
\end{equation}
where $\bar{\eta}_2$ is the residual that results from regressing $\epsilon$ onto the column space of $Z*$, $\bar{\eta}_2=(I-H^{Z_{\sigma^{-1}\circ\pi_1}Z_{\sigma^{-1}\circ\pi_2}})\epsilon,$ where $Z_{\sigma^{-1}\circ\pi_1}=P^T_{\sigma}P_1Z,$ $Z_{\sigma^{-1}\circ\pi_2}=P^T_{\sigma}P_2Z.$ Comparing (\ref{residual epsilon}) and (\ref{residual epsilon2}), we obtain
\begin{equation*}
    P^T_{\sigma}(I-H^{Z_{\pi_1}Z_{\pi_2}})\epsilon=P^T_{\sigma}\bar{\eta}=\bar{\eta}_2=(I-H^{Z_{\sigma^{-1}\circ\pi_1}Z_{\sigma^{-1}\circ\pi_2}})\epsilon
\end{equation*}
This implies 
\begin{equation*}
    X_{\pi_1}^{T}(I-H^{Z_{\pi_1}Z_{\pi_2}})\epsilon_\sigma=X_{\sigma^{-1}\circ\pi_1}^{T}(I-H^{Z_{\sigma^{-1}\circ\pi_1}Z_{\sigma^{-1}\circ\pi_2}})\epsilon
\end{equation*}
\begin{proposition}
\label{prop::Proposition for group permutation}
    For the set of permutation matrices $\mathcal{P}_K:=\{P_0,\cdots,P_K\}$ satisfying \Cref{assump::grouppermutation}, we can define the function $\mathbbm{P}_{\pi_k}(\pi_j):=P_jP_k,k=0,\cdots,K,j=1,\cdots,K$, for any fixed $k$, $\mathbbm{P}_{\pi_k}:\mathcal{P}_K \mapsto \mathcal{P}_K$, then
    \begin{enumerate}
        \item For fixed $P_k \in \{P_0,\cdots,P_K\}$, $\mathbbm{P}_{\pi_k}$ is a bijection.
        \item $I \in \mathcal{P}_K$, and for any $P_k \in \mathcal{P}_K$, $P_k^T \in \mathcal{P}_K$.
    \end{enumerate}
\end{proposition}
 \textit{Proof}.

From the definition, we know that for fixed k, $\mathbbm{P}_{\pi_k}:\mathcal{P}_K \mapsto \mathcal{P}_K$. Then for any $j\in\{0,\cdots,K\}$, $\mathbbm{P}_{\pi_k}(\pi_j)=P_jP_k$. Since $P_k,P_j \in \mathcal{P}_K$, by \Cref{assump::grouppermutation}, $\mathbbm{P}_{\pi_k}(\pi_j)=P_jP_k \in \mathcal{P}_K$, $\mathbbm{P}_{\pi_k}$ is injective.

For any $i,j \in \{0,\cdots,K\}$ and $i \neq j$, if $\mathbbm{P}_{\pi_k}(\pi_i)=\mathbbm{P}_{\pi_k}(\pi_j)$, then $P_iP_k=P_jP_k$
\begin{equation*}
    P_iP_kP_k^T=P_jP_kP_k^T \Rightarrow P_i=P_j,
\end{equation*}
Since $P_k$ is invertible, this implies $P_i=P_j$, and hence $i=j$. Since $P_K$ is finite, injectivity implies surjectivity. Hence $P_{\pi_k}$ is a bijection.

For any fixed $k \in \{0,\cdots,K\}$, we traverse the entire set $\{0,\cdots,K\}$ and calculate
\begin{equation*}
  \mathbbm{P}_{\pi_k}(\pi_i)=P_iP_k,\,\,\,\, i\in \{1,\cdots,K\},
\end{equation*}
Since $\mathbbm{P}_{\pi_k}$ is a bijection from $\mathcal{P}_K$ to $\mathcal{P}_K$ , then there must exist $i' \in \{0,\cdots,K\}$, s.t. $\mathbbm{P}_{\pi_k}(\pi_{i'})=P_k$. It equals $P_{i'}P_k=P_k$, since $P_k$ is a permutation matrix, hence invertible, then we obtain $P_{i'}=I$, so that $I \in \mathcal{P}_K$.

Also, for any fixed $k \in \{0,\cdots,K\}$, we traverse the entire set $\{0,\cdots,K\}$ and calculate
\begin{equation*}
  \mathbbm{P}_{\pi_k}(\pi_i)=P_iP_k,\,\,\,\, i\in \{0,\cdots,K\},
\end{equation*}
Since $\mathbbm{P}_{\pi_k}$ is a bijection from $\mathcal{P}_K$ to $\mathcal{P}_K$, and $I \in \mathcal{P}_K$, then there must exist $i'' \in \{0,\cdots,K\}$, s.t. $\mathbbm{P}_{\pi_k}(\pi_{i''})=I$. It equals $P_kP_{i''}=I$, $P_{i''}=P_k^T$, so that $P_k^T \in \mathcal{P}_K$.
\begin{proposition}
\label{proposition for E|S|}
    For K identically distributed rank statistics $R_1(\epsilon),\cdots,R_K(\epsilon)$, define the random subset $S(\epsilon) \subseteq \{R_1,\cdots,R_K\}$, then the probability $$\sum_{k=1}^K P(R_k \in S) = E|S|.$$
\end{proposition}
\textit{Proof}.
\begin{equation*}
    \sum^{K}_{k=1} P(R_k(\epsilon) \in S(\epsilon)) =\sum^{K}_{a=1}\sum^{K}_{k=1} P(R_k \in S||S|=a)P(|S|=a) 
\end{equation*}
For any event $\{R_k \in S\}$, it can be decomposed into the following disjoint unions:
\begin{equation*}
    \begin{aligned}
        \{R_k \in S\}=& \{R_k \in S || \forall i \neq k, R_i \notin S\}\\
        &\cup\{R_k \in S || \exists k' \neq k , R_{k'} \in S; \forall i \neq k,k', R_i \notin S\}\\
        &\cup\{R_k \in S || \exists k',k'' \neq k , R_{k'} \in S, R_{k''} \in S; \forall i \notin \{k,k',k''\}, R_i \notin S\}\\
        &\cup \cdots\\
        &\cup \{\forall i ,R_i \in S\}
    \end{aligned}
\end{equation*}
then,
$$P(\{R_k \in S\}||S|=a)=P(\{R_k \in S, k' \in N, R_{k'} \in S; \forall i \notin \{k,N\} , R_i \notin S\}),$$
where $N \subset \{1,\cdots,K\}/\{k\}, |N|=a-1$.

Also, we can decompose $\{|S|=a\}$,
$$\{|S|=a\}=\{\exists k_1,\cdots,k_a, R_{k_1},\cdots,R_{k_a} \in S, \forall i\neq \{k_1,\cdots,k_a\}, R_{i}\notin S \}.$$
For specific $k'_1,\cdots,k'_a$,  the event $\{R_{k'_1},\cdots,R_{k'_a}\in S, \forall i\neq \{k'_1,\cdots,k'_a\}, R_{i}\notin S\}$ can be covered in events 
$$\{R_{j} \in S, k' \in N, R_{k'} \in S; \forall i \notin \{k,N\} , R_i \notin S\},j=k'_1,\cdots,k'_a,$$
which is counted exactly $a$ times.

Therefore,
\begin{equation*}
    \begin{aligned}
         \sum^{K}_{k=1} P(R_k(\epsilon) \in S(\epsilon)) &=\sum^{K}_{a=1}\sum^{K}_{k=1} P(R_k \in S||S|=a)P(|S|=a) \\
         &= a \times P(|S|=a)= E(|S(\epsilon)|)
    \end{aligned}
\end{equation*}
\subsubsection*{Proof of \Cref{Thm::PALMRT}}
We restate our test statistic
$$\phi= \frac{1}{K+1}(1+\sum_{k=1}^{K}\mathbbm{1}\{X^{T}(I-H^{ZZ_{\pi_k}})Y\leq X_{\pi_k}^{T}(I-H^{ZZ_{\pi_k}})Y\}),$$
as noted above, $I-H^{ZZ_{\pi_k}}$ is orthogonal to the column space of $(Z,Z_{\pi_k}).$ Hence, $\phi$ can be expressed as 
$$\phi= \frac{1}{K+1}(1+\sum_{k=1}^{K}\mathbbm{1}\{X^{T}(I-H^{ZZ_{\pi_k}})\epsilon\leq X_{\pi_k}^{T}(I-H^{ZZ_{\pi_k}})\epsilon\}).$$
Let $r_{ab}=\mathbbm{1}\{F(\pi_a,\pi_b;x,Z,\epsilon) < F(\pi_b,\pi_a;x,Z,\epsilon)\}+\frac{1}{2}\mathbbm{1}\{F(\pi_a,\pi_b;x,Z,\epsilon) = F(\pi_b,\pi_a;x,Z,\epsilon)\}$ and denote $R_{a}:=\frac{1}{K+1}\sum^{K}_{b=0} r_{ab}$, $\forall a,b, \pi_a,\pi_b \in \mathcal{P}_K$, and $\mathcal{P}_K$ is generated through \Cref{assump::grouppermutation}.

Recall \Cref{prop::Finverse} and the fact that $I \in \mathcal{P}_K$by \Cref{prop::Proposition for group permutation}. Without loss of generality, write $\pi_0:=I$. Our test statistic can then be written as

$$       \phi= \frac{1}{K+1}(1+\sum_{k=1}^{K}\mathbbm{1}\{F(\pi_0,\pi_k;x,Z,\epsilon_{\sigma}) \leq F(\pi_k,\pi_0;x,Z,\epsilon_{\sigma})\})\geq R_{0},$$
so we change our attention from $\phi$ to $R_{0}$ through
$$P(\phi\leq \alpha)\leq P(R_0\leq\alpha),$$
so we turn to analyze the property of $R_0$. For fixed $m\in \{0,1,\cdots,K\}$, $P_m\in\mathcal{P}_K$ is a permutation matrix, then $P^T_m\in \mathcal{P}_K$ through \Cref{prop::Proposition for group permutation}. From \Cref{assump::exchanglable}, $\epsilon \overset{d}{=} P_m^T\epsilon:=\epsilon_{\pi_m^{-1}}$.  Using Proposition \ref{prop::Finverse}, we have 
\begin{equation*}
\begin{aligned}
    R_0&=\frac{1}{K+1}\sum^{K}_{k=0} \mathbbm{1}\{T(\pi_0,\pi_k;\epsilon) <  T(\pi_k,\pi_0;\epsilon)\} +\frac{1}{2}\mathbbm{1}\{T(\pi_0,\pi_k;\epsilon) = T(\pi_k,\pi_0;\epsilon)\}\\
    &\overset{d}{=} \frac{1}{K+1}\sum^{K}_{k=0} \mathbbm{1}\{T(\pi_0,\pi_k;\epsilon_{\pi_m^{-1}}) < T(\pi_k,\pi_0;\epsilon_{\pi_m^{-1}})\}+\frac{1}{2}\mathbbm{1}\{T(\pi_0,\pi_k;\epsilon_{\pi_m^{-1}}) =  T(\pi_k,\pi_0;\epsilon_{\pi_m^{-1}})\}\\
    &=\frac{1}{K+1}\sum^{K}_{k=0} \mathbbm{1}\{T(\pi_m \circ \pi_0,\pi_m \circ\pi_k;\epsilon) < T(\pi_m \circ\pi_k,\pi_m \circ\pi_0;\epsilon)\}\\
    &\,\,\,\,\,\,\,\,+\frac{1}{2}\mathbbm{1}\{T(\pi_m \circ \pi_0,\pi_m \circ\pi_k;\epsilon) = T(\pi_m \circ\pi_k,\pi_m \circ\pi_0;\epsilon)\}\\
    &=\frac{1}{K+1}\sum^{K}_{j=0} \mathbbm{1}\{T(\pi_m, \pi_j;\epsilon) < T(\pi_j,\pi_m;\epsilon)\}+\frac{1}{2}\mathbbm{1}\{T(\pi_m ,\pi_j;\epsilon)=T(\pi_j,\pi_m;\epsilon)\}\\
    &=R_m,
\end{aligned}
\end{equation*}
where the last equation uses the bijection property of $\mathbbm{P}_{\pi_m}$ established in \Cref{prop::Proposition for group permutation}. This implies that $R_m \overset{d}{=} R_0, \forall m \in \{1,\cdots,K\}$.

Then as above, we define the comparison  $S:=\{m:R_m\leq \alpha \}$, from Proposition \ref{proposition for E|S|}, we have 
\begin{equation*}
    \begin{aligned}
        P(R_0\leq \alpha) \overset{d}{=}\frac{1}{K+1}\sum^{K}_{k=0} P(R_k\leq \alpha)= \frac{1}{K+1}\sum^{K}_{k=0} P(R_k(\epsilon) \in S(\epsilon)) =\frac{1}{K+1} E|S|.
    \end{aligned}
\end{equation*}
For every fixed $k\in\{0,1,\cdots,K\}$, in $R_0$, there exists an item 
$$\mathbbm{1}\{T(\pi_0,\pi_k;\epsilon) < T(\pi_k,\pi_0;\epsilon)\}+\frac{1}{2}\mathbbm{1}\{T(\pi_0,\pi_k;\epsilon) = T(\pi_k,\pi_0;\epsilon)\},$$  correspondingly, since Proposition \ref{prop::Proposition for group permutation} that for any $P_k \in \mathcal{P}_K$, $P_k^T \in \mathcal{P}_K$, there exists an entry in $R_k$ that
\begin{equation*}
\begin{aligned}
     \mathbbm{1}\{T( \pi_k \circ \pi_0,\pi_k \circ \pi^T_k;\epsilon) < &T( \pi_k \circ \pi^T_k,\pi_k \circ \pi_0;\epsilon)\}\\&+\frac{1}{2}\mathbbm{1}\{T( \pi_k \circ \pi_0,\pi_k \circ \pi^T_k;\epsilon) = T( \pi_k \circ \pi^T_k,\pi_k \circ \pi_0;\epsilon)\}\\
\end{aligned}
\end{equation*}

This implies for any $a,b \in \{0,1,\cdots,K\}, a\neq b$, $r_{ab}+r_{ba}= 1$,$r_{aa}=1, r_{ab} \geq 0$. The two elements of matrix $R=(r_{ab})_{a,b\in[0,\cdots,K]}$ are complementary with respect to diagonal symmetry, and their sum is 1, and $R_i, i\in\{0,1,\cdots ,K\}$ can be viewed as a sum of the i-th row in matrix $R$.

To bound $|S|$, we first derive a lower bound. Since $S:=\{m: R_m\leq \alpha\}$, the lower bound is obtained by considering the smallest possible number of values $R_m, m=1,\cdots,K$, that belong to this set, in other words, the value of $R_m$ is as large as possible. Consider an implementation where $R_{m'}$ takes the maximum value $1$, because for any $m''\in \{0,1,\cdots,K\}, m''\neq m'$, $r_{m''m'}+r_{m'm''}= 1$, and $r_{m'm''}=1\in R_{m'}$, $r_{m''m'}\in R_{m''}$, the second largest value takes $\frac{K}{K+1}$, and so on, we get a sequence of equal differences from $1$ to $\frac{1}{K+1}$,
hence $|S|\geq [\alpha (K+1)] $.

On the other hand, since $S:=\{m: R_m\leq \alpha\}$, then
\begin{equation*}
    \begin{aligned}
        \alpha |S| \geq \sum_{m \in S} R_m \geq \frac{1}{K+1} \sum_{\{m,k\}\in S} r_{mk} =\frac{1}{2(K+1)}|S|(|S|+1),
    \end{aligned}
\end{equation*}
where the last step uses the fact that $a,b \in \{0,1,\cdots,K\}, a\neq b$, $r_{ab}+r_{ba}= 1$, $r_{aa}=1$. As a result, $|S|\leq max(0,2\alpha(K+1)-1)\leq[2\alpha(K+1)].$
Therefore,
$$P(\phi \leq \alpha)\leq P(R_0\leq \alpha)=\frac{1}{K+1} E|S|,$$
and 
$$\frac{[\alpha(K+1)]}{K+1}\leq\frac{1}{K+1}|S|\leq\frac{[2\alpha(K+1)]}{K+1}.$$
\subsubsection*{Proof of Proposition \ref{proposition of 2alpha}}
\label{sec::Proof of proposition factor 2}
For the statistic $\phi'$ as defined in (\ref{phi'}), we restate it as 
\begin{equation*}
    \begin{aligned}
        \phi'= &\frac{1}{K+1}(1+\sum_{k=1}^{K}\mathbbm{1}\{X^{T}(I-H^{ZZ_{\pi_k}})\epsilon < X_{\pi_k}^{T}(I-H^{ZZ_{\pi_k}})\epsilon\}\\ &+\frac{1}{2}(\mathbbm{1}\{X^{T}(I-H^{ZZ_{\pi_k}})\epsilon = X_{\pi_k}^{T}(I-H^{ZZ_{\pi_k}})\epsilon\})).
    \end{aligned}
\end{equation*}
In this case, we construct a special example, in which there exists $\alpha\in[0,1]$, such that $P(\phi'\leq\alpha)=2\alpha.$
For the sake of concise presentation, we consider the situation $n=5$ initially. We consider the special construction of $(X,Z)$, since in the former theorem we treat $(X,Z)$ as fixed. We consider 
\begin{equation}
z=
\begin{bmatrix}
  1 & 0 \\
  0 & 1 \\
  0 &0\\
  0&0\\
  0&0\\
\end{bmatrix}
\end{equation}
in order to simplify its generated column space. And we consider the permutation group as the down shifting permutation group $\mathcal{P}_5$ in $n=5$, $\mathcal{P}_5=\{P_0,P_1,\cdots,P_4\},P_0=I$ for $x=(x_1,x_2,\cdots,x_5)^T,$ this group satisfies 
$$P_ix=(x_{5-i+1},\cdots,x_5,x_1,\cdots,x_{5-i})^T, i=\{1,2,3,4,5\}.$$
We design this down shifting permutation group in order to simplify the column space generated by $Z_{\pi_k}$ ($\pi_k$ corresponds to $P_k,k=\{0,1,2,3,4\}$ defined above). We also simplify the joint column space generated by $(Z,Z_{\pi_k}),k=0,\cdots,4.$
We then calculate each term $f(k,X,\epsilon)$ in $\phi'$  to see whether it has a lower bound,
$$f(k,X,\epsilon):=\mathbbm{1}\{X^{T}(I-H^{ZZ_{\pi_k}})\epsilon < X_{\pi_k}^{T}(I-H^{ZZ_{\pi_k}})\epsilon\}+\frac{1}{2}(\mathbbm{1}\{X^{T}(I-H^{ZZ_{\pi_k}})\epsilon = X_{\pi_k}^{T}(I-H^{ZZ_{\pi_k}})\epsilon\}).$$
When $k=0$, here $\operatorname{span}(Z,Z_{\pi_0})=\operatorname{span}(Z)$, so that,
\begin{equation*}
I-H^{ZZ_{\pi_0}}=
\begin{bmatrix}
  0 &  &\cdots&&0\\
   & 0 & &0 &\\
  \vdots & &1&&\vdots  \\
  & 0&    &1&\\
  0&&\cdots&&1\\
\end{bmatrix}, \,\,\,\,f(0,X,\epsilon)=\frac{1}{2}.
\end{equation*}

When $k=1$,
\begin{equation*}
I-H^{ZZ_{\pi_1}}=
\begin{bmatrix}
  0 &  &\cdots&&0\\
   & 0 & &0 &\\
  \vdots & &0&&\vdots  \\
  & 0&    &1&\\
  0&&\cdots&&1\\
\end{bmatrix}, 
\end{equation*}
\begin{equation*}
     f(1,X,\epsilon)=\mathbbm{1}\{X_4\epsilon_4+X_5\epsilon_5< X_3\epsilon_3+X_4\epsilon_5\}+\frac{1}{2}\mathbbm{1}\{X_4\epsilon_4+X_5\epsilon_5= X_3\epsilon_3+X_4\epsilon_5\}.
\end{equation*}
When k=2, from the same discussion as above,
\begin{equation*}
    I-H^{ZZ_{\pi_2}}=
\begin{bmatrix}
  0 &  &\cdots&&0\\
   & 0 & &0 &\\
  \vdots & &0&&\vdots  \\
  & 0&    &0&\\
  0&&\cdots&&1\\
\end{bmatrix}, \,\,\,\,
    f(2,X,\epsilon)=\mathbbm{1}\{X_5\epsilon_5< X_3\epsilon_5\}+\frac{1}{2}\mathbbm{1}\{X_5\epsilon_5= X_3\epsilon_5\}.
\end{equation*}

When k=3, we know that 
\begin{equation*}
    I-H^{ZZ_{\pi_3}}=
\begin{bmatrix}
  0 &  &\cdots&&0\\
   & 0 & &0 &\\
  \vdots & &1&&\vdots  \\
  & 0&    &0&\\
  0&&\cdots&&0\\
\end{bmatrix}, \,\,\,\,
    f(3,X,\epsilon)=\mathbbm{1}\{X_3\epsilon_3< X_5\epsilon_3\}+\frac{1}{2}\mathbbm{1}\{X_3\epsilon_3= X_5\epsilon_3\}.
\end{equation*}

When k=4, we know that 
\begin{equation*}
    I-H^{ZZ_{\pi_4}}=
\begin{bmatrix}
  0 &  &\cdots&&0\\
   & 0 & &0 &\\
  \vdots & &1&&\vdots  \\
  & 0&    &1&\\
  0&&\cdots&&0\\
\end{bmatrix}, \,\,\,\,
\end{equation*}
\begin{equation*}
    f(4,X,\epsilon)=\mathbbm{1}\{X_3\epsilon_3+X_4\epsilon_4< X_4\epsilon_3+X_5\epsilon_4\}+\frac{1}{2}\mathbbm{1}\{X_3\epsilon_3+X_4\epsilon_4= X_4\epsilon_3+X_5\epsilon_4\}.
\end{equation*}
Since $\epsilon$ is exchangeable in \Cref{assump::exchanglable}, we set $X_1=X_2=X_3=X_4=X_5$, then $\phi=\frac{1}{2}$. Setting $\alpha=\frac{1}{2},$
$$P(\phi\leq\frac{1}{2})=1,$$
so that the bound $2\alpha$ in \Cref{Thm::PALMRT} is tight.

\subsection{Proof of \Cref{thm::thorem of group CPT}}
\label{sec::Proof of grouped CPT}
\begin{proof}
    Initially, we recall the lemmas about quantiles. From the definition of $R_i,i=0,\cdots,K$, and the assumptions of $\mathcal{P}_K$ and $\epsilon$, it is easy to see that $$R_0\overset{d}{=}R_1\overset{d}{=}\cdots\overset{d}{=}R_K.$$ Recall that the quantile function $\hat{Q}_{1-\alpha}$ of random variables $R_0,\cdots,R_K$ is the quantile function with respect to the empirical CDF $\hat{F}_n(r):=\frac{1}{K+1}\sum_{k=0}^K \mathbbm{1}\{R_i\leq r\}$. It has the explicit formula,
    $$Q_{1-\alpha}(\frac{1}{K+1}\sum^K_{k=0}\delta_{R_k})=R_{(\lceil \alpha(K+1)\rceil )},$$ where $R_{(k)}$ denotes the $k$-th smallest value in $R_0,\cdots, R_K$.

    Define the set of ``strange'' points 
    $$S(R)=\{i\in\{0,1,\cdots,K\}:R_i> Q_{1-\alpha}(\sum_{k=0}^{K}\frac{1}{K+1}\cdot\delta_{R_k}))\}$$
    That is, an index $i$ corresponds to a ``strange''
 point if its rank statistic $R_i$ is one of the $[\alpha(K+1)]$ largest elements of the list $R_1,\cdots,R_{K+1}.$ By definition, this can include at most $\alpha(K+1)$ entries of the list, that is 
 $$|S(R)|\leq \alpha(K+1).$$ 
 \begin{equation*}
 \begin{aligned}
      R_0&=\frac{1}{K+1}\sum_{j=0,j\neq0}^K 1\{Y^T\eta\leq Y^T\eta_{j}\}\\&=\frac{1}{K+1}
      \sum_{j=0,j\neq 0}^K 1\{\epsilon^T\eta_0\leq \epsilon^T\eta_{j}\}=\frac{1}{K+1}\sum_{j=1}^{K}1\{(P_0\epsilon)^T\eta\leq (P_j \epsilon)^T\eta)\}\\
      &\overset{d}{=}\frac{1}{K+1}\sum_{j=1}^{K}1\{(P_0P_k\epsilon)^T\eta\leq (P_j P_k\epsilon)^T\eta)=\frac{1}{K+1}\sum_{j=0,j\neq k'}^K 1\{(P_kY)^T\eta\leq (P_j^TY)^T\eta\}(*)\\
      &=\frac{1}{K+1}\sum_{j=0,j\neq k}^K 1\{Y^T(P^T_k\eta)\leq Y^T\eta_j\}=R_k,\,\, k=1,\cdots, K.
 \end{aligned}
 \end{equation*}
 Here $P_k \in \mathcal{P}_K$, and (*) holds since $\mathcal{P}_K$ meets \Cref{assump::grouppermutation} and $\epsilon$ is exchangeable. Therefore, we have 
 \begin{equation}
     \begin{aligned}
         P(R_0 >Q_{1-\alpha}(\sum_{k=0}^{K}\frac{1}{K+1}\cdot\delta_{R_k}))&=P(K+1 \in S(R))=\frac{1}{K+1}\sum_{k=0}^{K}P(k \in S(R))\\
         &=\frac{1}{K+1}\mathbb{E}[\sum_{k=0}^{K}1{k\in S(R)}]=\frac{1}{K+1}E[|S(R)|]\\
         &=\frac{1}{K+1}\cdot \alpha (K+1)=\alpha.
     \end{aligned}
\end{equation}
\end{proof}

\section{Analysis of Type II error control}
\label{sec::Section B}
In this section, we complete our analyses of Type II error control. We first provide detailed analyses omitted from the main text, and then we complete the proof of our main theorems, lemmas, etc.

\subsection{Supplementary analysis of our Type II error control}\label{subsection: supplementary of Type II}
In this part, we complete the detailed analyses on our Type II error, and we first complete the overall analysis for the statistics $\phi_1,\phi_2$ in~\eqref{statistics for type II}.

Generally, about half of the cases where $\vert (X_{\pi_{k}}^{T}-X^{T})(I-H^{ZZ_{\pi_{k}}})\epsilon\vert>\vert b (X^{T}-X_{\pi_{k}}^{T})(I-H^{ZZ_{\pi_{k}}})X\vert$ will result in $b (X^{T}-X_{\pi_{k}}^{T})(I-H^{ZZ_{\pi_{k}}})X < (X_{\pi_{k}}^{T}-X^{T})(I-H^{ZZ_{\pi_{k}}})\epsilon $. Consequently, \eqref{optimization: hypothesis} requires that 
\begin{align}\label{necessary for optimization}
\frac{1}{1+K}\sum_{k=0}^{K}\mathbbm{1}\left\{\left\vert (X^{T}-X^{T}_{\pi_k})(I-H^{ZZ_{\pi_k}})X\right\vert\geq\vert b\vert^{-1}\left\vert(X_{\pi_{k}}^{T}-X^{T})(I-H^{ZZ_{\pi_{k}}})\epsilon\right\vert\right\}\geq 1-2\alpha
\end{align}
hold with high probability under the randomness of $\epsilon$, in order to ensure a small Type II error. In particular, when $\epsilon\overset{d}{=}-\epsilon$, condition \eqref{necessary for optimization} becomes a strict necessary condition if we require a well-controlled Type II error under all possible exchangeable noise (see \ref{subsubsection: optimization: intuitive}).
Furthermore, the absolute value notation can be removed, since $$\frac{1}{1+K}\sum_{k=0}^{K}\mathbbm{1}\left\{ (X^{T}-X^{T}_{\pi_k})(I-H^{ZZ_{\pi_k}})X\leq-\vert b\vert^{-1}\left\vert(X_{\pi_{k}}^{T}-X^{T})(I-H^{ZZ_{\pi_{k}}})\epsilon\right\vert\right\}\geq 1-2\alpha$$  
is generally suboptimal in terms of Type II error control (also see \ref{subsubsection: optimization: intuitive}). This suggests that
$$
\frac{1}{1+K}\sum_{k=0}^{K}\mathbbm{1}\left\{ (X^{T}-X^{T}_{\pi_k})(I-H^{ZZ_{\pi_k}})X\geq\vert b\vert^{-1}\left\vert(X_{\pi_{k}}^{T}-X^{T})(I-H^{ZZ_{\pi_{k}}})\epsilon\right\vert\right\}\geq 1-2\alpha
$$
is a necessary condition for~\eqref{optimization: hypothesis}. On the other hand, although the distribution of the noise term $(X^{T}-X^{T}_{\pi_{k}})(I-H^{ZZ_{\pi_k}})\epsilon$ is difficult to analyze, previous studies such as \citep{wen2025} have demonstrated that under certain mild conditions, this noise term is bounded by
$
\left\vert\ (X^{T}-X^{T}_{\pi_k})(I-H^{ZZ_{\pi_k}})\epsilon\right\vert\leq o(X^{T}X)
$
with high probability (see \ref{paragraph: Type II guaranty}). This yields a sufficient condition for controlling the Type II error: if for some $\lambda_0 \in \Omega(X^{T}X)$ (as $n,p \to \infty$) we have
$$
\frac{1}{1+K}\left[\sum_{k=0}^{K}\mathbbm{1}\left\{ X^{T}(I-H^{ZZ_{\pi_k}})X-X^{T}_{\pi_k}(I-H^{ZZ_{\pi_k}})X\geq \lambda_0  \right\}\right]\geq 1-\frac{1}{2}\alpha\,,
$$
then for any nonzero constant $b$, the Type II error converges to $0$. Synthesizing both perspectives, the solution to the optimization problem~\eqref{optimization: intuitive} furnishes both necessary and sufficient conditions for Type II error control:
\begin{align}\label{optimization: intuitive}
\max_{\mathcal{P}_{K}} \lambda, s.t. \frac{1}{K+1}\left[\sum_{k=0}^{K}\mathbbm{1}\{X^{T}(I-H^{ZZ_{\pi_{k}}})X-X^{T}_{\pi_k}(I-H^{ZZ_{\pi_k}})X\geq \lambda\}\right]\geq 1-\frac{1} {2}\alpha\,,
\end{align}
which is equivalent to solving problem~\ref{eq: optimization1} due to the fact that
$\forall k\in[1,K]$,
$$
X^{T}(I-H^{ZZ_{\pi_k}})X-X^{T}_{\pi_k}(I-H^{ZZ_{\pi_k}})X = X^{T}X-\left[X^{T}H^{ZZ_{\pi_k}}X+X^{T}_{\pi_k}(I-H^{ZZ_{\pi_k}}))X   \right]\,,
$$
where the first component $X^{T}X$ is independent of the choice of $\mathcal{P}_K$. 

In fact, we do not care about the exact constant factor preceding $\alpha$, as our primary emphasis is on Type II error control under a general $\alpha>0$, particularly in the regime $\alpha\to 0_{+}$. Accordingly, we adopt the optimization problem~\eqref{optimization: intuitive} as the foundational objective of our Type II error control procedures.

\subsubsection{The optimization problem \texorpdfstring{\eqref{optimization: intuitive}}{(36)}}

\label{subsubsection: optimization: intuitive}
\label{sec::optimization problem}
\paragraph{Necessary condition for $\lambda$ in~\eqref{optimization: intuitive}}
Suppose $\alpha < \frac{1}{4}$ and that we aim to achieve a small Type II error under any exchangeable noise $\epsilon$. We first demonstrate the necessity of
\begin{align*}
\frac{1}{1+K}\left[\sum_{k=0}^{K}\mathbbm{1}\left\{(X^{T}-X^{T}_{\pi_k})(I-H^{ZZ_{\pi_k}})X\geq\vert b\vert^{-1}\left\vert(X_{\pi_{k}}^{T}-X^{T})(I-H^{ZZ_{\pi_{k}}})\epsilon\right\vert\right\}\right]
\geq 1-2\alpha
\end{align*}
being satisfied with high probability over the distribution of $\epsilon$.

First, consider the case where $\vert b\vert$ is not too small and satisfies
\begin{align*}
&\frac{1}{1+K}\left[\sum_{k=0}^{K}\mathbbm{1}\left\{\left\vert (X^{T}-X^{T}_{\pi_k})(I-H^{ZZ_{\pi_k}})X\right\vert\geq\vert b\vert^{-1}\left\vert(X_{\pi_{k}}^{T}-X^{T})(I-H^{ZZ_{\pi_{k}}})\epsilon\right\vert\right\}\right]\\
&> 2\alpha\,, w.h.p.
\end{align*}
For such $b$ we prove that, in order to control the Type II error under any exchangeable noise $\epsilon$, at least one of the following must hold with high probability:
\begin{align}\label{eq:positive}
\frac{1}{1+K}\left[\sum_{k=0}^{K}\mathbbm{1}\left\{(X^{T}-X^{T}_{\pi_k})(I-H^{ZZ_{\pi_k}})X\geq\vert b\vert^{-1}\left\vert(X_{\pi_{k}}^{T}-X^{T})(I-H^{ZZ_{\pi_{k}}})\epsilon\right\vert\right\}\right]\geq 1-2\alpha
\end{align}
\begin{align}\label{eq:negative}
\frac{1}{1+K}\left[\sum_{k=0}^{K}\mathbbm{1}\left\{(X^{T}-X^{T}_{\pi_k})(I-H^{ZZ_{\pi_k}})X\leq-\vert b\vert^{-1}\left\vert(X_{\pi_{k}}^{T}-X^{T})(I-H^{ZZ_{\pi_{k}}})\epsilon\right\vert\right\}\right]\geq1-2\alpha\,.
\end{align}

We consider a special case where $\epsilon\overset{d}{=}-\epsilon $. In this case, for any $k$ such that\\ $\left\vert(X^{T}-X^{T}_{\pi_k})(I-H^{ZZ_{\pi_k}})X\right\vert<\vert b\vert^{-1}\left\vert(X_{\pi_{k}}^{T}-X^{T})(I-H^{ZZ_{\pi_{k}}})\epsilon\right\vert$, we must have:
\begin{align*}
&\mathbbm{1}\left\{b (X^{T}-X_{\pi_{k}}^{T})(I-H^{ZZ_{\pi_{k}}})X \geq (X_{\pi_{k}}^{T}-X^{T})(I-H^{ZZ_{\pi_{k}}})\epsilon \right\}+\\
&\mathbbm{1}\left\{b (X^{T}-X_{\pi_{k}}^{T})(I-H^{ZZ_{\pi_{k}}})X \geq (X_{\pi_{k}}^{T}-X^{T})(I-H^{ZZ_{\pi_{k}}})(-\epsilon) \right\}=1\,.
\end{align*}
Let $A = \frac{1}{1+K}\left[\sum_{k=0}^{K}\mathbbm{1}\left\{(X^{T}-X^{T}_{\pi_k})(I-H^{ZZ_{\pi_k}})X\geq\vert b\vert^{-1}\left\vert(X_{\pi_{k}}^{T}-X^{T})(I-H^{ZZ_{\pi_{k}}})\epsilon\right\vert\right\}\right]$,\\
$B=\frac{1}{1+K}\left[\sum_{k=0}^{K}\mathbbm{1}\left\{(X^{T}-X^{T}_{\pi_k})(I-H^{ZZ_{\pi_k}})X\leq-\vert b\vert^{-1}\left\vert(X_{\pi_{k}}^{T}-X^{T})(I-H^{ZZ_{\pi_{k}}})\epsilon\right\vert\right\}\right]$, 
and denote $\phi_1(\mathcal{P}_{K},X,Z,\epsilon)=\frac{1}{1+K}\sum_{k=0}^{K}1\left\{ X^{T}(I-H^{ZZ_{\pi_k}})Y \leq X^{T}_{\pi_k}(I-H^{ZZ_{\pi_k}})Y      \right\}$, where $Y$ follows~\eqref{eq::model}.
For any $\epsilon$ such that $A,B< 1-2\alpha$ and $A+B>2\alpha$, we have:

(1): $\vert \phi_1(\mathcal{P}_{K},X,Z,\epsilon)-\phi_1(\mathcal{P}_{K},X,Z,-\epsilon)\vert\leq 1-A-B$. This is because for any $k$ such that $\left\vert(X^{T}-X^{T}_{\pi_k})(I-H^{ZZ_{\pi_k}})X\right\vert<\vert b\vert^{-1}\left\vert(X_{\pi_{k}}^{T}-X^{T})(I-H^{ZZ_{\pi_{k}}})\epsilon\right\vert$, \\
$\mathbbm{1}\left\{b (X^{T}-X_{\pi_{k}}^{T})(I-H^{ZZ_{\pi_{k}}})X \geq (X_{\pi_{k}}^{T}-X^{T})(I-H^{ZZ_{\pi_{k}}})\epsilon \right\}$ does not change when replacing $\epsilon$ by $-\epsilon$.

(2): $\vert \phi_1(\mathcal{P}_{K},X,Z,\epsilon)+\phi_1(\mathcal{P}_{K},X,Z,-\epsilon)-1\vert=A-B$. This holds because summing up $k$ such that $\left\vert(X^{T}-X^{T}_{\pi_k})(I-H^{ZZ_{\pi_k}})X\right\vert<\vert b\vert^{-1}\left\vert(X_{\pi_{k}}^{T}-X^{T})(I-H^{ZZ_{\pi_{k}}})\epsilon\right\vert$ contributes to $1-A-B$, and the remaining $k$ contributes to either $2A$ or $2B$, which depends on whether $b>0$ or not. 

Therefore, when $A,B<1-2\alpha$ and $A\geq B$, we have: 
\begin{align*}
&\min(\phi_1(\mathcal{P}_{K},X,Z,\epsilon),\phi_1(\mathcal{P}_{K},X,Z,-\epsilon))\\&\geq \frac{1}{2}\left[\vert\phi_1(\mathcal{P}_{K},X,Z,\epsilon)+\phi_1(\mathcal{P}_{K},X,Z,-\epsilon)\vert-\vert \phi_1(\mathcal{P}_{K},X,Z,\epsilon)-\phi_1(\mathcal{P}_{K},X,Z,-\epsilon)\vert\right]\\
&\geq \frac{1}{2}[(1-A+B)-(1-A-B)]\\
&\geq B\,,
\end{align*}
and $\min(\phi_1(\mathcal{P}_{K},X,Z,\epsilon),\phi_1(\mathcal{P}_{K},X,Z,-\epsilon))\leq \frac{1}{2}(\phi_1(\mathcal{P}_{K},X,Z,\epsilon)+\phi_1(\mathcal{P}_{K},X,Z,-\epsilon))\leq \frac{1}{2}(1+A-B)$.
Thus, we have:  $\min(\phi_1(\mathcal{P}_{K},X,Z,\epsilon),\phi_1(\mathcal{P}_{K},X,Z,-\epsilon))<1-\alpha$.
On the other hand, 
\begin{align*}
&\max(\phi_1(\mathcal{P}_{K},X,Z,\epsilon),\phi_1(\mathcal{P}_{K},X,Z,-\epsilon))\\&\leq \min(\phi_1(\mathcal{P}_{K},X,Z,\epsilon),\phi_1(\mathcal{P}_{K},X,Z,-\epsilon))+\vert \phi_1(\mathcal{P}_{K},X,Z,\epsilon)-\phi_1(\mathcal{P}_{K},X,Z,-\epsilon)\vert\\
&\leq \min(\phi_1(\mathcal{P}_{K},X,Z,\epsilon),\phi_1(\mathcal{P}_{K},X,Z,-\epsilon))+(1-A-B)\,.
\end{align*}
If
$\min(\phi_1(\mathcal{P}_{K},X,Z,\epsilon),\phi_1(\mathcal{P}_{K},X,Z,-\epsilon))\leq \alpha$, then \\$\max(\phi_1(\mathcal{P}_{K},X,Z,\epsilon),\phi_1(\mathcal{P}_{K},X,Z,-\epsilon))\leq \alpha+1-A-B< 1-\alpha$, while, at the same time, \\
$\max(\phi_1(\mathcal{P}_{K},X,Z,\epsilon),\phi_1(\mathcal{P}_{K},X,Z,-\epsilon))\geq \frac{1}{2}(\phi_1(\mathcal{P}_{K},X,Z,\epsilon)+\phi_1(\mathcal{P}_{K},X,Z,-\epsilon))\geq \frac{1}{2}(1-A+B)>\alpha$.

Thus, $\mathcal{H}_0$ will be accepted for at least one of $\epsilon$ and $-\epsilon$. This conclusion also holds when $A,B\leq 1-2\alpha$ and  $A\leq B$ due to a similar calculation.
This implies that, when $\epsilon\overset{d}{=}-\epsilon$
\begin{align*}
\mathbb{P}\left[\mathcal{H}_{0} \text{ is accepted}  \right]&\geq \frac{1}{2}\mathbb{P}\left[\phi_1(\mathcal{P}_{K},X,Z,\epsilon)\in(\alpha,1-\alpha)\right]+\frac{1}{2}\mathbb{P}\left[\phi_1(\mathcal{P}_{K},X,Z,-\epsilon)\in(\alpha,1-\alpha)\right]\\
&\geq \frac{1}{2}\mathbb{P}\left[A,B<1-2\alpha\cap A+B>2\alpha\right]\\
&\geq \frac{1}{2}\mathbb{P}[A,B<1-2\alpha]-\frac{1}{2}\mathbb{P}[A+B\leq2\alpha]\,.
\end{align*}
Therefore, when $\alpha<\frac{1}{4}$ and $\epsilon\overset{d}{=}-\epsilon$, a small Type II error requires a small $\mathbb{P}[A,B<1-2\alpha]-\mathbb{P}[A+B\leq2\alpha]$. 

For the case where $A+B\leq 2\alpha$, i.e., the regression residual $b(X^{T}-X^{T}_{\pi})(I-H^{ZZ_{\pi}})X$ has a smaller absolute value compared with the noise term. In this case, the Type II error depends on the signs of each dimension of $(X-X_{\pi_k})(I-H^{ZZ_{\pi_k}})\epsilon$, which is impossible to  be well guaranteed unless we know the exact distribution of $\epsilon$.  Consequently, to control the Type II error in the absence of such distributional knowledge, both $\mathbb{P}[A, B \leq 1 - 2\alpha]$ and $\mathbb{P}[A + B \leq 2\alpha]$ must be sufficiently small.

Since our goal is to control the Type II error under any exchangeable noise, the Type II error must also be controlled when $\epsilon\overset{d}{=}-\epsilon$. Thus, we can conclude that $\max(A,B)\geq 1-2\alpha$ must hold with high probability.

Finally, we explain why we do not take $B\geq 1-2\alpha$ into consideration. Consider 
\begin{align*}
(X^{T}-X^{T}_{\pi_k})(I-H^{ZZ_{\pi_k}})X&=(X-X_{\pi_k})^{T}X-(X-X_{\pi_k})H^{ZZ_{\pi_k}}X\\
&=\frac{1}{2}(X^{T}X-2X^{T}_{\pi_k}X+X^{T}_{\pi_{k}}X_{\pi_k})-(X-X_{\pi_k})H^{ZZ_{\pi_k}}X
\\&=\frac{1}{2}\Vert X-X_{\pi_k}\Vert^{2}_{2}-(X-X_{\pi_k})H^{ZZ_{\pi_k}}X\\
&\geq -\frac{1}{2}\Vert H^{ZZ_{\pi_k}}X\Vert^{2}_{2}\,,
\end{align*}
where the last inequality is usually not achievable. Moreover, $\Vert H^{ZZ_{\pi_k}}X\Vert^{2}_{2}$ is not guaranteed to reach $\Omega(X^{T}X)$ when $n/p$ increases. This implies that for $b\neq0$ and $X,Z$ to satisfy $B\geq 1-2\alpha$, which causes $\mathcal{H}_{0}$ to be rejected, $\vert b\vert$ must be sufficiently large. Therefore, we only consider optimizing over $A\geq 1-2\alpha$.

\paragraph{Type II guarantee of $\lambda$ in~\eqref{optimization: intuitive}}\label{paragraph: Type II guaranty}
We demonstrate that a sufficiently large $\lambda$ in \eqref{optimization: intuitive} guarantees $\phi_1, \phi_2 \notin [\alpha, 1-\alpha]$ with high probability.

Suppose that for some $\lambda_0$, we have
$$
\frac{1}{1+K}\left[\sum_{k=0}^{K}1\{X^{T}(I-H^{ZZ_{\pi_{k}}})X-X^{T}_{\pi_k}(I-H^{ZZ_{\pi_k}})X\geq \lambda_0\}\right]\geq 1-\frac{1} {2}\alpha\,,
$$
and at the same time, the noise term can be bounded by $g(X,\gamma)$ such that for every permutation $\pi$, 
$$
\mathbb{P}\left[\left\vert X^{T}(I-H^{ZZ_{\pi}})\epsilon -X^{T}_{\pi}(I-H^{ZZ_{\pi}})\epsilon\right\vert\geq g(X,\gamma)\right]\leq \gamma\,.
$$
It is then typically the case that, for any $\gamma > 0$, $g(X,\gamma) \in o(X^T X)$. As an illustration, suppose each coordinate $X_i$ of $X$ has variance bounded below by a positive constant $c_e > 0$, and that $X^T X \geq \frac{1}{2} c_e n$ with high probability—a reasonable assumption commonly adopted in the literature (e.g., \citep{wen2025}). Furthermore, if the $\epsilon_i$ are independent, satisfy $\mathbb{E}[\epsilon_i] = 0$, and $\mathbb{E}[|\epsilon_i|^{1+t}] \leq C_e$ for some $t \in (0,1]$, then it can be shown that, for any permutation $\pi$ and whenever $X^T X \geq \frac{1}{2} c_e n$, the following holds:
\begin{align}\label{upper bound of noise}
\mathbb{P}\left[\vert (X-X_{\pi})^{T}(I-H^{ZZ_{\pi}})\epsilon\vert \geq (X^{T}X)^{\frac{1+\frac{1}{2}t}{1+t}} \right]\leq O((X^{T}X)^{-\frac{1}{2}t})\,.
\end{align}
This means for any constant $\gamma>0$, $g(X,\gamma)\leq o(1)\max(X^{T}X,\frac{1}{2}c_e\cdot n)$.
On the other hand, when $\vert b\vert \geq \lambda_0^{-1}g(X,\gamma)$, we can upper bound $\mathbb{P}[\phi_1,\phi_2\geq \alpha]$ by:
\begin{align*}
\mathbb{P}[\phi_1,\phi_2\geq \alpha]\leq \mathbb{P}\left[\sum_{k=0}^{K}\mathbbm{1}\left\{\left\vert (X^{T}-X^{T}_{\pi_k})(I-H^{ZZ_{\pi_k}})\epsilon\right\vert\geq g(X,\gamma)\right\}\geq \frac{1}{2}\alpha(1+K)\right]\leq\frac{2\gamma}{\alpha}\,.
\end{align*}
Therefore, when $\vert b\vert\geq \lambda^{-1}_0g(X,\frac{1}{2}\alpha\cdot \theta)$, the Type II error is at most $\theta$. Based on this result, if $\lambda_0\geq \Omega(X^{T}X)$, then we can bound the Type II error for any nonzero constant $b$.

However, $g(X,\gamma)$ is difficult to estimate accurately unless the exact distribution of $\epsilon$ is known. Therefore, we focus on maximizing $\lambda$ by selecting an appropriate $\mathcal{P}_{K}$ that satisfies \Cref{assump::grouppermutation}.

We finally prove~\eqref{upper bound of noise}. We first provide the following inequality:
\begin{align*}
\vert x+y\vert^{1+t}\leq \vert x\vert^{1+t}+\vert y\vert^{1+t}+(1+t)\vert x\vert^{t}y\,,\,\forall x\geq 0\,.
\end{align*}
\begin{proof}
The above inequality is obvious when $x=0$. Now consider $x>0$. Let $f(y)=\vert x\vert^{1+t}+\vert y\vert^{1+t}+(1+t)\vert x\vert^{t}y-\vert x+y\vert^{1+t}$. When $y\geq 0$, we have: $f^{'}(y)=(1+t)y^{t}+(1+t)x^{t}-(1+t)(x+y)^{t}$. Since $0<t\leq 1$, $(x+y)^{t}\leq x^{t}+y^{t}$, so that $f(y)\geq 0$, implying that $f(y)\geq f(0)=0,\forall y\geq 0$. \\   
For $y<0$, we let $y=-ax$, and it suffices to show that
\begin{align*}
g(a)=1+a^{1+t}-(1+t)a-\vert 1-a\vert^{1+t}\geq 0
\end{align*}
When $a\in[0,1)$, $g^{'}(a)=(1+t)(a^{t}+(1-a)^{t}-1)\geq 0$. $g(1)=1-t\geq 0$.\\
When $a\geq 1$, $g(a)=1+a^{1+t}-(1+t)a-(a-1)^{1+t}$, $g^{'}(a)=(1+t)(a^{t}-1-(a-1)^{t})\geq 0$, implying that $g(a)\geq g(1)\geq0$. Therefore, $g(a)\geq 0,\forall a\geq 0$.
\end{proof}
By the inequality we know that, if $x,y$ are independent and $\mathbb{E}[\vert x\vert^{1+t}],\mathbb{E}[\vert y\vert^{1+t}]<\infty$, $\mathbb{E}[x]=\mathbb{E}[y]=0$, then $\mathbb{E}[\vert x+y\vert^{1+t}]\leq\mathbb{E}[\vert x\vert^{1+t}]+\mathbb{E}[\vert y\vert^{1+t}] $. Thus, for any $a_1,a_2,...,a_n\in\mathbb{R}$, we have 
\begin{align*}
\mathbb{E}\left[\vert\sum_{i=1}^{n}a_{i}\epsilon_{i} \vert^{1+t}\right]\leq \mathbb{E}[\vert a_1\epsilon_1\vert^{1+t}]+\mathbb{E}\left[\vert\sum_{i=2}^{n}a_{i}\epsilon_{i} \vert^{1+t}\right]\leq...\leq\sum_{i=1}^{n}\vert a_{i}\vert^{1+t}\mathbb{E}[\vert \epsilon_i\vert^{1+t}]\,.
\end{align*}
Let $a_{i}=[(I-H^{ZZ_{\pi}})(X-X_{\pi})]_i$. Then $\sum_{i=1}^{n}a^{2}_{i}\leq 4X^{T}X$, and when $X^{T}X\geq \frac{1}{2}c_en$, since $\vert a_i\vert^{1+t}\leq 1+a^{2}_i$, there exists some constant $c^{'}>0$ such that $\sum_{i=1}^{n}\vert a_{i}\vert^{1+t}\leq c^{'}(X^{T}X)$. This means for some constant $C>0$, $\mathbb{E}\left[\vert (X-X_{\pi})^{T}(I-H^{ZZ_{\pi}})\epsilon \vert^{1+t}\right]\leq CX^{T}X$, which leads to the probability bound:
\begin{align*}
\mathbb{P}\left[\vert(X-X_{\pi})^{T}(I-H^{ZZ_{\pi}})\epsilon\vert^{1+t}\geq (X^{T}X)^{1+\frac{1}{2}t}  \right]\leq O((X^{T}X)^{-\frac{1}{2}t})\to 0\,.
\end{align*}

\paragraph{The effectiveness of the optimization problem~\eqref{optimization objective for alg}}\label{paragraph: upper and lower of optimization objective}

We first show the formal form of Lemma~\ref{lem: upper and lower of optimization objective informal}

\begin{lemma}[Upper and lower bounds of $X^{T}H^{ZZ_{\pi}}X$, formal]\label{lem: upper and lower of optimization objective formal}
Suppose the rows of $Z$ are independent, and the $r$-th row $z_r$ follows a distribution $f_r(z_1,z_2,...,z_{p-1})$, with $\mathbb{E}_{f_r}[z_{i}]=0$. If there exist some constants $c_1, c_2>0$ such that, for any $f_r (r=1,2,...,n)$ and $v\in\mathbb{R}^{p-1}$, $\Vert v\Vert^{2}_{2}=1$, we have:
\begin{align*}
\mathbb{E}\left[(z^{T}v)^{2}\right]\geq c_1,\qquad \mathbb{E}\left[e^{tz^{T}v}\right]\leq e^{c_2 t^{2}},\forall t\in \mathbb{R}\,.
\end{align*}
Then there exists a constant $C>0$ depending only on $c_1,c_2$ such that, for any permutation $\pi$ with probability at least $1-40e^{-p}$ over the distribution of $Z$, we have:
\begin{align*}
X^{T}H^{ZZ_{\pi}}X=\Vert H^{ZZ_{\pi}}X\Vert^{2}_{2}\leq \Vert H^{Z}X\Vert^{2}_{2}+\frac{1}{1-\frac{C(p+tr(P_{\pi}))}{n}}\Vert H^{Z_{\pi}}(X-H^{Z}X)\Vert^{2}_{2}\,,\forall X\,.
\end{align*}

On the other hand, we have a strict lower bound of $X^{T}H^{ZZ_{\pi}}X$:
\begin{align*}
X^{T}H^{ZZ_{{\pi}}}X\geq \Vert H^{Z}X\Vert^{2}_{2}+\Vert H^{Z_{{\pi}}}(X-H^{Z}X)\Vert^{2}_{2}.    
\end{align*}
\end{lemma}
\paragraph{Remark} In the special case that $z_1,z_2,...,z_{p-1}$ are also independent, with $\mathbb{E}_{f_r}[z_{i}z_{j}]=\delta_{ij},\forall r$, we simply have $c_1 = 1$. If we further know that all $z_{i}$'s are subgaussian, that is, for some $K$, we have
$$
\mathbb{E}[e^{tz_{i}}]\leq e^{K^{2}t^{2}},\forall i=1,2,...,p-1\,,
$$
then we obtain:
$$
\mathbb{E}(e^{tz^{T}v}) = \prod_{i=1}^{p-1}\mathbb{E}[e^{tz_{i}v_{i}}]\leq\prod_{i=1}^{p-1}e^{K^{2}t^{2}v^{2}_{i}} = e^{K^{2}t^{2}}\,.
$$ 
This implies a valid $c_2=K^{2}$.

Denote $v = (I - H^{Z})X$ for simplicity, and assume $Z$ is as given in Lemma~\ref{lem: upper and lower of optimization objective informal}. We first establish the comparison between $\lambda_1(X,Z,\mathcal{P}_{K},t)$ and $\lambda_2(X,Z,\mathcal{P}_{K},t)$.

On one hand, for any given $X$, $Z$, $\mathcal{P}_K$, and $\alpha$, the quantity $\lambda_1$ admits the following lower bound:
\begin{align}
\lambda_1(X,Z,\mathcal{P}_{K},\frac{1}{2}\alpha)\geq\lambda_{2}(X,Z,\mathcal{P}_{K},\frac{1}{2}\alpha)+\Vert H^{Z}X\Vert^{2}_{2}-\frac{1}{2}\Vert v\Vert^{2}_{2}\,.
\end{align}
On the other hand, for any given $\mathcal{P}_{K}$, if $\alpha\geq 8e^{-\frac{1}{2}p}$ and there exists some $m$ such that $\vert \{\pi_k \vert tr(P_{\pi_k})\leq m\}\vert\geq (1-\frac{1}{8}\alpha)(1+K)$, then, with probability at least $1 - 40e^{-\frac{1}{2}p}$ over the distribution of $Z$, from  Lemma \ref{lem: upper and lower of optimization objective informal}, we have (detailed proof is provided in \ref{paragraph: upper and lower of optimization objective}):
\begin{align}
\lambda_1(X,Z,\mathcal{P}_{K},\frac{1}{2}\alpha)\leq \frac{n}{n-C(p+m)}\lambda_2(X,Z,\mathcal{P}_{K},\frac{1}{4}\alpha)+\frac{n}{2n-2C(p+m)}\Vert v\Vert^{2}_{2}+\Vert H^{Z}X\Vert^{2}_{2}\,.
\end{align}
When $\mathcal{P}_{K}$ is independent of $Z$ ($\mathcal{P}_{K}$ may depend on $X$), and $tr(P_{\pi})\in o(n)$ for at least $1-\frac{1}{8}\alpha$ fraction of $\pi\in\mathcal{P}_{K}$, then \eqref{upper of original optimization objective} provides an upper estimate of $\lambda_1(X,Z,\mathcal{P}_{K},\frac{1}{2}\alpha)$. Specifically, as $n/p\to\infty$, \eqref{upper of original optimization objective} implies that $\lambda_1(X,Z,\mathcal{P}_{K},\frac{1}{2}\alpha)\leq (1+o(1))\lambda_2(X,Z,\mathcal{P}_{K},\frac{1}{4}\alpha)+(\frac{1}{2}+o(1))\Vert v\Vert^{2}_{2}+\Vert H^{Z}X\Vert^{2}_{2}$, assuming $Z$ follows the conditions in Lemma~\ref{lem: upper and lower of optimization objective informal}. 
In summary, \eqref{lower of original optimization objective} and \eqref{upper of original optimization objective} demonstrate the consistency between the two families of statistics $\lambda_1(X,Z,\mathcal{P}_{K},t)$ and $\lambda_2(X,Z,\mathcal{P}_{K},t)$ (for $t \in (0,1)$), with a strict lower bound that holds for all $\mathcal{P}_K$ and $Z$, and a high-probability upper bound valid when $Z$ satisfies Lemma~\ref{lem: upper and lower of optimization objective informal} and $\mathcal{P}_{K}$ is independent of $Z$.

Therefore, we consider the statistic $\lambda_2(X,Z,\mathcal{P}_{K},t)$ for all $t\in(0,1)$, with particular interest in the regime $t \to 0^{+}$. Since our objective is to identify a permutation group $\mathcal{P}_{K}$ that performs well simultaneously for any constant $t>0$, we adopt $\lambda_2(X,Z,\mathcal{P}_{K},\frac{1}{4}\alpha)$  as the optimization target and construct $\mathcal{P}_{K}$ for \eqref{optimization objective for alg}  over an arbitrary given $\alpha > 0$, without concern for the constant factor preceding $\alpha$.

\paragraph{Theoretical guarantee with $Z$ in Lemma~\ref{lem: upper and lower of optimization objective informal}}
Given that $\pi$ is independent of $Z$ and that $Z$ satisfies the conditions in Lemma~\ref{lem: upper and lower of optimization objective informal}, we now verify both the upper and lower bounds of $\lambda_1(X,Z,\mathcal{P}_{K},\frac{1}{2}\alpha)$.

Lemma~\ref{lem: upper and lower of optimization objective informal} together with \eqref{lower and upper bound term 2} yields both a high-probability upper bound and a strict lower bound for $X^{T}H^{ZZ_{\pi}}X+X^{T}_{\pi}(I-H^{ZZ_{\pi}})X$: 
\begin{align*}
&\frac{1}{2}X^{T}_{\pi}(I-H^{Z_{\pi}})(I-H^{Z})X+\Vert H^{Z}X\Vert^{2}_{2}+\frac{n}{n-C(p+tr(P_{\pi}))}\Vert H^{Z_{\pi}}(X-H^{Z}X)\Vert^{2}_{2}+\frac{1}{2}\Vert (I-H^{Z})X\Vert^{2}_{2}\\
&\geq X^{T}H^{ZZ_{\pi}}X+X^{T}_{\pi}(I-H^{ZZ\pi})X\\
&\geq \frac{1}{2}X^{T}_{\pi}(I-H^{Z_{\pi}})(I-H^{Z})X+\Vert H^{Z}X\Vert^{2}_{2}+\Vert H^{Z_{\pi}}(X-H^{Z}X)\Vert^{2}_{2}-\frac{1}{2}\Vert (I-H^{Z})X\Vert^{2}_{2}\,.
\end{align*}
Where the first inequality holds with probability $O(e^{-p})$, and the second always holds.

Let $\mathcal{P}_{K}$ be the permutation matrices of any permutation group and $\mathcal{K}=\{k\vert \frac{1}{2}v^{T}_{\pi_k}v+\Vert H^{Z_{\pi_k}}v\Vert^{2}_{2} \geq \lambda_2(X,Z,\mathcal{P}_{K},\frac{1}{2}\alpha)\}$. Then for $k\in \mathcal{K}$ we must have:
$$
X^{T}H^{ZZ_{\pi_k}}X+X^{T}_{\pi_k}(I-H^{ZZ\pi_k})X\geq \lambda_{2}(X,Z,\mathcal{P}_{K},\frac{1}{2}\alpha)+\Vert H^{Z}X\Vert^{2}_{2}-\frac{1}{2}\Vert (I-H^{Z})X\Vert^{2}_{2}
$$
Therefore, we have
\begin{align*}
&\frac{1}{1+K}\sum_{k=0}^{K}\mathbbm{1}\left\{ X^{T}H^{ZZ_{\pi_k}}X+X^{T}_{\pi_k}(I-H^{ZZ\pi_k})X \geq \lambda_2(X,Z,\mathcal{P}_{K},\frac{1}{2}\alpha)+\Vert H^{Z}X\Vert^{2}_{2}-\frac{1}{2}\Vert (I-H^{Z})X\Vert^{2}_{2} \right\}\\
&\geq \frac{1}{1+K}\sum_{k=0}^{K}\mathbbm{1}\left\{ \frac{1}{2}v^{T}_{\pi_k}v+\Vert H^{Z_{\pi_k}}v\Vert^{2}_{2} \geq \lambda_2(X,Z,\mathcal{P}_{K},\frac{1}{2}\alpha) \right\}\\
&>\frac{1}{2}\alpha\,.
\end{align*}
This means $\lambda_1(X,Z,\mathcal{P}_{K},\frac{1}{2}\alpha)\geq\lambda_2(X,Z,\mathcal{P}_{K},\frac{1}{2}\alpha)+\Vert H^{Z}X\Vert^{2}_{2}-\frac{1}{2}\Vert (I-H^{Z})X\Vert^{2}_{2}$.
On the other hand, it holds with probability $1-40e^{-p}$ that
\begin{align*}
&X^{T}H^{ZZ_{\pi}}X+X^{T}_{\pi}(I-H^{ZZ\pi})X\leq \frac{1}{2}\Vert (I-H^{Z})X\Vert^{2}_{2}+\Vert H^{Z}X\Vert^{2}_{2}\\
&+\frac{n}{n-C(p+tr(P_{\pi_k}))}-\frac{C(p+tr(P_{\pi}))}{2n-2C(p+tr(P_{\pi_k}))}v^{T}_{\pi_k}v\\
&\leq \Vert H^{Z}X\Vert^{2}_{2}+\frac{n}{n-C(p+tr(P_{\pi_k}))}\left[\frac{1}{2}v^{T}_{\pi_k}v+\Vert H^{Z_{\pi_k}}v\Vert^{2}_{2}\right]+\frac{n}{2n-2C(p+tr(P_{\pi_k}))}\Vert v\Vert^{2}_{2}\,,
\end{align*}
where the last inequality is because $v^{T}_{\pi_k}v\geq-\Vert v\Vert^{2}_{2}$.
Now, suppose that $\vert\{\pi_{k}\vert tr(P_{\pi_k})\leq m\}\vert\geq (1-\frac{1}{8}\alpha)(1+K)$ for some $m\leq o(n)$. Let $\mathcal{K}^{'}=\{k\vert tr(P_{\pi_k})\leq m, \frac{1}{2}v^{T}_{\pi_k}v+\Vert H^{Z_{\pi_k}}v\Vert^{2}_{2}\leq \lambda_2(X,Z,\mathcal{P}_{K},\frac{1}{4}\alpha)\} $

Then $\vert \mathcal{K}^{'}\vert\geq (1-\frac{1}{4}\alpha)(1+K)-\frac{1}{8}\alpha(1+K)=(1-\frac{3}{8}\alpha)(1+K)$. On the other hand, by Lemma~\ref{lem: upper and lower of optimization objective informal}, for any $k\in\mathcal{K^{'}}$, 
\begin{align*}
&X^{T}H^{ZZ_{\pi_k}}X+X^{T}_{\pi_k}(I-H^{ZZ\pi_k})X\leq\\&\frac{n}{n-C(p+m)}\lambda_2(X,Z,\mathcal{P}_{K},\frac{1}{2}\alpha)+\Vert H^{Z}X\Vert^{2}_{2}+\frac{n}{2n-2C(p+m)}\Vert v\Vert^{2}_{2}
\end{align*}
holds with probability at least $1-40e^{-p}$. This implies that

\begin{align*}
&\mathbb{P}\left[\lambda_1(X,Z,\mathcal{P}_K,\frac{1}{2}\alpha)>\frac{n}{n-C(p+m)}\lambda_2(X,Z,\mathcal{P}_{K},\frac{1}{2}\alpha)+\Vert H^{Z}X\Vert^{2}_{2}+\frac{n}{2n-2C(p+m)}\Vert v\Vert^{2}_{2}\right]\\
&\leq \frac{(1+K)\cdot 40e^{-p}}{\frac{1}{8}\alpha(1+K)}\\
&\leq 40e^{-\frac{1}{2}p}\,.
\end{align*}


\subsubsection{Group decomposition for high probability guarantee}
\label{sec::Group high prob guarantee}
In the optimization problem~\eqref{optimization objective for alg}, we first consider the second term $\Vert H^{Z_{\pi}}v\Vert^{2}_{2}$. This is equivalent to analyzing $\Vert H^{Z}v_{\pi}\Vert^{2}_{2}$, since $\Vert H^{Z_{\pi}}v\Vert^{2}_{2}=\Vert H^{Z}v_{\pi^{-1}}\Vert^{2}_{2}$, and $P_{\pi^{-1}}\in\mathcal{P_{K}}$ if and only if $P_{\pi}\in \mathcal{P}_{K}$. In our algorithm, the set $\{\pi\vert P_{\pi}\in\mathcal{P}_{K}\}$ is constructed as $\mathcal{Q}_{1}\circ\mathcal{Q}_{2}\circ ...\circ\mathcal{Q}_{k}$, and $H^{Z}v_{\pi}$ admits the decomposition $H^{Z}v_{\pi}=\sum_{i=1}^{k}u_{i}$. Leveraging this structure, we present Theorem~\ref{thm: high prob bound of optimization objective}—the formal counterpart of Theorem~\ref{thm: informal for high prob bound}—along with Proposition~\ref{proposition: high probability bound on group decomposition}, which together specify the requirements on the group decomposition and the distributions of $X$ and $Z$.

\begin{theorem}\label{thm: high prob bound of optimization objective}
Suppose that $v_{i}(i=1,2,...,m)$ is a sequence of independent vectors with  $\mathbb{E}(\Vert v_{i}\Vert^{2}_{2})=t_i$ and let $S$ be any parameter such that $S\geq\sum_{i=1}^{m}t_i$. Also, suppose these vectors satisfy the following conditions:

(1): $\mathbb{E}(v_{i})=0$.

(2): $\Vert v_{i}\Vert_{2}\leq a$, with $a^{2}\leq \frac{S}{\ln^{4}n}$.

(3): $\forall w\in\mathcal{S}^{n-1}$, we have: $S_{w}:=\sum_{i=1}^{m}\mathbb{E}\left[(v^{T}_{i}w)^{2}\right]\leq \frac{S}{\ln^{4}n}$.

Then we have: for any $c,k>0$, there exists $C$ such that 

$$
\mathbb{P}\left[\Vert \sum_{i=1}^{m}v_{i}\Vert^{2}_{2}\geq \sum_{i=1}^{m}t_i+cS\right]\leq Cn^{-k}.
$$
\end{theorem}
Following this, if we regard $u_i' = u_i - \mathbb{E}[u_i]$ as a sequence of independent random variables, then we obtain the following proposition:
\begin{proposition}\label{proposition: high probability bound on group decomposition}
Suppose that a partition of $S_{i}(i=1,2,...,k)$ satisfies:

(1): 
\begin{align}\label{condition for group decomposition}
\max_{i\in[1,k]}\frac{1}{\vert S_{i}\vert}\vert \sum_{j\in S_{i}}(v_{j}-\bar{v})\vert \in o\left(\sqrt{\frac{\sum_{i=1}^{n}(v_{i}-\bar{v})^{2}}{n}}\right)\,,
\end{align}


(2) For any $j$ we have: $(v_{j}-\bar{v})^{2}\leq\frac{1}{\ln^{6}n}\sum_{i=1}^{n}(v_{i}-\bar{v})^{2}$

(3): For any $i$, we have: $\frac{\sum_{j\in S_{i}}(v_{j}-\bar{v})^{2}}{\sum_{i=1}^{n}(v_{i}-\bar{v})^{2}}\leq \frac{1}{\ln^{4}n}$

(4): $\sum_{i=1}^{n}(v_{i}-\bar{v})^{2}\to+\infty$\,.

Let $c$ be any constant such that $c>0$; then if $\sum_{i=1}^{k}\mathbb{E}[\Vert u_{i}-\mathbb{E}(u_{i})\Vert^{2}_{2}]\geq \Omega(\sum_{i=1}^{n}(v_{i}-\bar{v})^{2})$, we have:
\begin{align}
\lim_{n\to \infty}\mathbb{P}\left[\Vert H^{Z}v_{\pi}\Vert^{2}_{2}\leq(1+c)\mathbb{E}[\Vert H^{Z}v_{\pi}\Vert^{2}_{2}]\right]\to 1\,,
\end{align}
where the probability is defined on the randomness of $\pi$ and the given $n$, $\pi$ is sampled uniformly from $\mathcal{P}_{K}$.
Otherwise, we apply another bound
\begin{align}
\lim_{n\to \infty}\mathbb{P}\left[\Vert H^{Z}v_{\pi}\Vert^{2}_{2}\leq\mathbb{E}[\Vert H^{Z}v_{\pi}\Vert^{2}_{2}]+c\sum_{i=1}^{n}(v_{i}-\bar{v})^{2}\right]\to 1\,.
\end{align}
On the other hand, when~\eqref{condition for group decomposition} and condition (3) hold we also have:
\begin{align}\label{upper bound of vv_pi}
 v^{T}_{\pi}v \leq \mathbb{E}[ v^{T}_{\pi}v]+o(1)\sum_{i=1}^{n}(v_{i}-\bar{v})^{2}\,, w.h.p.
\end{align}
\end{proposition}
Here, condition $(1)$ implies that $\Vert \sum_{i=1}^{k}\mathbb{E}(u_{i})\Vert^{2}_{2}\in o(\sum_{i=1}^{n}(v_{i}-\bar{v})^{2})$, which in turn yields the conclusion via Theorem~\ref{thm: high prob bound of optimization objective}. Note that conditions (2) and (4) pertain to the properties of the residual $X - H^{Z}X$. Specifically, condition (2) ensures that this residual does not exhibit excessively heavy tails; both (2) and (4) are commonly satisfied in real‑world data.

\subsubsection{The algorithm for group decomposition}
\label{sec::The algorithm for group decomposition}

In this section, we complete the analysis of our algorithm for Type II error control, including both how Algorithm~\ref{alg structure} guarantees a small $\lambda_2$ in \eqref{optimization objective for alg} and the detailed implementation of Algorithm~\ref{alg structure}.

Firstly, we discuss how $\{1,2,...,n\}$ is partitioned. When the conditions in Proposition~\ref{proposition: high probability bound on group decomposition} are satisfied, the following Lemma~\ref{lem: approximation of optimization objective} demonstrates that
\begin{align*}
\sum_{i=1}^{k}\left(\sum_{j\in S_{i}}\frac{1}{\vert S_{i}\vert}a_{j}^{2}\right)\left( \sum_{j\in S_{i}}b_j \right)
\end{align*}
can be used to effectively bound $\mathbb{E}[\Vert H^{Z}v_{\pi}\Vert^{2}_{2}]$
under the additional constraints that 
$\vert \sum_{j\in S_{i}}a_{j}\vert\leq O(\sqrt{\sum_{i=1}^{n}a^{2}_{i}})$ and $\vert S_{i}\vert\geq n^{0.55}$.

\begin{lemma}\label{lem: approximation of optimization objective}
Suppose that we can divide $\{1,2,...,n\}$ into $S_{1},...,S_{k}$ such that conditions (2), and (3) in Proposition~\ref{proposition: high probability bound on group decomposition} hold, and
$$
\vert \sum_{j\in S_{i}}a_{j}\vert\leq O(1)\sqrt{\sum_{i=1}^{n}a^{2}_{i}}\,,\,\,\,\,\,\, \vert S_{i}\vert\geq n^{0.55},\forall i\in\{1,2,...,k\}
$$
Then we have:
\begin{align}
\left\vert \left[ \frac{1}{2}n\bar{v}^{2}+\sum_{i=1}^{k}\left(\sum_{j\in S_{i}}\frac{1}{\vert S_{i}\vert}a_{j}^{2}\right)\left( \sum_{j\in S_{i}}b_j \right) +\Vert v^{*}\Vert^{2}_{2} \right] -\mathbb{E}\left[\frac{1}{2} v^{T}_{\pi}v+\Vert H^{Z}v_{\pi}\Vert^{2}_{2}  \right]    \right\vert\leq o(1)\sum_{i=1}^{n}a^{2}_{i}\,.
\end{align}
\end{lemma}
This provides an alternative optimization problem as in~\eqref{optimization objective approximated}, which is applied to our algorithm design.

\begin{align}\label{optimization objective approximated}
\min_{\mathcal{P}_{K}}\sum_{i=1}^{k}\left(\sum_{j\in S_{i}}\frac{1}{\vert S_{i}\vert}a_{j}^{2}\right)\left( \sum_{j\in S_{i}}b_j \right)
\end{align}

We summarize our constraints and targets of Algorithm~\ref{alg structure} for finding $S_{1},...,S_{k}$ as follows:
\begin{enumerate}
    \item For each $S_i$, control $\vert\sum_{j\in S_i}a_j\vert \leq O(\sqrt{\sum_{i=1}^{n}a^{2}_{i}}) $ and $\vert S_{i}\vert\geq n^{0.55}$ so that $\mathbb{E}(\vert v^{T}_{\pi}v\vert)$ is near optimal  and $ \sum_{i=1}^{k}\mathbb{E}(\Vert u_i\Vert_2^{2}) $ can be well approximated by an easier optimization problem of~\eqref{optimization objective approximated}. 
    \item Control $\sum_{j\in S_{i}}a^{2}_{j}\in o(\sum_{j=1}^{n}a^{2}_{j})$ for all $S_{i}$, satisfying (3) in Proposition~\ref{proposition: high probability bound on group decomposition}, to guarantee a high probability bound of $\Vert H^{Z}v_{\pi}\Vert^{2}_{2}\leq (1+o(1))\mathbb{E}(\Vert H^{Z}v_{\pi}\Vert^{2}_{2})=(1+o(1))[\Vert v^{*}\Vert^{2}_{2}+\sum_{i=1}^{k}\mathbb{E}(\Vert u_{i}\Vert^{2}_{2})]$. This is the key approach for removing the indicator function in \eqref{optimization objective for alg}.
    \item Given the above constraints and guarantee, optimize over \eqref{optimization objective approximated}.
\end{enumerate}

\paragraph{Algorithm design for finding permutation group}\label{paragraph: algorithm detail}
We now demonstrate how to obtain a valid solution to \eqref{optimization objective approximated} under the constraints specified above. For simplicity, denote $c_{i}=a^{2}_{i}$, $\bar{b}=\frac{1}{n}\sum_{i=1}^{n}b_{i}$, and $\bar{c}=\frac{1}{n}\sum_{i=1}^{n}c_{i}$. Then we have:
\begin{align*}
\sum_{i=1}^{k}\left(\sum_{j\in S_{i}}\frac{1}{\vert S_{i}\vert}a_{j}^{2}\right)\left( \sum_{j\in S_{i}}w^{T}_{j}w_{j} \right)
&=\sum_{i=1}^{k}\frac{1}{\vert S_{i}\vert}\left[\sum_{j\in S_{i}}(c_{j}-\bar{c})\right]\sum_{j\in S_{i}}b_{j}+\sum_{i=1}^{k}\bar{c}\sum_{j\in S_{i}}b_{j}\\
&=\sum_{i=1}^{k}\frac{1}{\vert S_{i}\vert}\left[\sum_{j\in S_{i}}(c_{j}-\bar{c})\right]\left[\sum_{j\in S_{i}}(b_{j}-\bar{b})\right]+n\bar{b}\bar{c}\,.
\end{align*}
Then it suffices to minimize $\sum_{i=1}^{k}\frac{1}{\vert S_{i}\vert}\left[\sum_{j\in S_{i}}(c_{j}-\bar{c})\right]\left[\sum_{j\in S_{i}}(b_{j}-\bar{b})\right]$. To accomplish this, a key observation is that if either all the $j\in S_{i}$ satisfy $b_{j}\leq \bar{b}, c_{j}\geq \bar{c}$ or all of them satisfy $b_{j}\geq \bar{b},c_j\leq \bar{c}$, then $\frac{1}{\vert S_{i}\vert}\left[\sum_{j\in S_{i}}(c_{j}-\bar{c})\right]\left[\sum_{j\in S_{i}}(b_{j}-\bar{b})\right]\leq 0$. Conversely, if there exist $b_{j_1}\leq \bar{b}, c_{j_1}\leq \bar{c}$ and $b_{j_2}\geq \bar{b}, c_{j_2}\geq \bar{c}$,  then placing both $j_1$ and $j_2$ in the same subset $S_i$ reduces the objective.

Based on this insight, we develop Algorithm~\ref{alg structure}, which partitions the $n$ indices into three subsets $J_1, J_2, J_3$. This is achieved by first categorizing the elements into three groups according to the signs of $b_i - \bar{b}$ and $c_i - \bar{c}$, followed by an adjustment procedure (Algorithm~\ref{alg rearrange}) that ensures a bounded $$\sum_{S_{i}\in J_{1}}\frac{1}{\vert S_{i}\vert}\left[\sum_{j\in S_{i}}(c_{j}-\bar{c})\right]\left[\sum_{j\in S_{i}}(b_{j}-\bar{b})\right]\,,$$ 
while the remaining $J_{2},J_{3}$ contribute to reducing the overall sum
\begin{align}\label{optimization objective: final form}
\sum_{i=1}^{k}\frac{1}{\vert S_{i}\vert}\left[\sum_{j\in S_{i}}(c_{j}-\bar{c})\right]\left[\sum_{j\in S_{i}}(b_{j}-\bar{b})\right]\,.
\end{align}

\begin{algorithm}[t]
	
	\begin{algorithmic}
    \caption{Rearrange} \label{alg rearrange}
		\STATE \textbf{Input:} $(a_1,c_1),(a_2,c_2),...,(a_n,c_n)$, 3 subsets $I_1,I_2,I_3$ of $\{1,2,...,n\}$ and parameter $M$.
        \STATE Let $S=\sum_{i=1}^{n} a^{2}_{i}$, $\bar{c}=\frac{1}{n}\sum_{i=1}^{n}c_{i}$

        \STATE Determine $I^{'}_{1}\bigcup I^{'}_2\bigcup I^{'}_{3}=\{1,2,...,n\}$ in the following steps:
        (i) Find a subset $I^{'}_2\subseteq I_2$ with the maximum number of elements such that $\vert \sum_{i\in I^{'}_{2}} a_i \vert\leq \sqrt{S} $.
        (ii) Find a subset $I^{'}_3\subseteq I_3$ with the maximum number of elements such that $\vert \sum_{i\in I^{'}_{3}} a_i \vert\leq \sqrt{S} $.
        (iii) Let $J=\{1,2,...,n\}\setminus(I^{'}_2\bigcup I^{'}_3)$, check whether $\vert \sum_{i\in J} a_{i}\vert^{2}+\frac{1}{M}\vert \sum_{i\in J} (c_{i}-\bar{c}) \vert^{2}\leq 8S $. 
        \STATE If so, let $I^{'}_1=J$. Otherwise add elements to $J$ under the following two cases:
        \IF{$\sum_{i\in J}(c_{i}-\bar{c})>0$}
        \STATE Call \textbf{Remove}$(I^{'}_2,\sum_{i\in J}(c_{i}-\bar{c}))$ and obtain output $I^{"}_2$, update $I^{'}_1\leftarrow J\bigcup I^{"}_2$, $I^{'}_2\leftarrow I^{'}_2\setminus I^{"}_2$.
        \ELSE
        \STATE Call \textbf{Remove}$(I^{'}_3,\sum_{i\in J}(c_{i}-\bar{c}))$ and obtain output $I^{"}_3$, update $I^{'}_1\leftarrow J\bigcup I^{"}_3$, $I^{'}_3\leftarrow I^{'}_3\setminus I^{"}_3$.
        \ENDIF
        \STATE Check whether $\vert I_1^{'}\vert\geq n^{0.9}$, if not, at each step move an index $i\in I^{'}_{2}\bigcup I^{'}_{3}$ to $I^{'}_1$ until $\vert I^{'}_{1}\vert\geq n^{0.55}$, with $i= \argmin_{i\in I^{'}_2\bigcup I^{'}_3}\left\{\sum_{\mathscr{A} = I^{'}_2,I^{'}_3}(\sum_{j\in \mathscr{A}\setminus i}a_j)^{2}+\frac{1}{M}\sum_{\mathscr{A}=I^{'}_2,I^{'}_3}[\sum_{j\in \mathscr{A}\setminus i}(c_j-\bar{c})]^{2}\right\}$. Let $J_1=\{(a_{i},c_i)\vert i\in I^{'}_1\}$, $J_2=\{(a_i,c_i)\vert i\in I^{'}_2 \}$, $J_3=\{(a_i,c_i)\vert i\in I^{'}_3\}$ be the sets of vectors corresponding to $I^{'}_1$, $I^{'}_2$, $I^{'}_3$. For $J_1,J_2,J_3$, call Scale~\ref{alg:scale}  with parameter $(1,\sqrt{\frac{1}{M}})$.
      \IF{$\vert J_2\vert,\vert J_3\vert\geq n^{0.9}$}
        \STATE Return $J_1,J_2,J_3$
        \ELSE
        \STATE Return $J_1\bigcup J_2\bigcup J_3,\emptyset,\emptyset$.
        \ENDIF  
    \end{algorithmic}
    
\end{algorithm}

On the other hand, we need to ensure that $\vert \sum_{j\in S_{i}}a_{i}\vert \leq O(\sum_{i=1}^{n}a^{2}_{i})$, which is to make sure that $\mathbb{E}_{\pi_k}\left[v^{T}_{\pi_k}v\right]\leq n\bar{v}^{2}+o(1)\sum_{i=1}^{n}a^{2}_{i}$ as well as 
$$
\mathbb{E}_{\pi_k\in\mathcal{P}_{K}}\left[\Vert H^{Z}v_{\pi_k}\Vert^{2}_{2}\right]\leq \sum_{i=1}^{k}\frac{1}{\vert S_{i}\vert}\left[\sum_{j\in S_{i}}(c_{j}-\bar{c})\right]\left[\sum_{j\in S_{i}}(b_{j}-\bar{b})\right] + o(1)\sum_{i=1}^{n}a^{2}_{i}\,.
$$
Therefore, we need to upper bound $\vert \sum_{(a_i,c_i)\in J_{t}}a_{i}\vert(t=1,2,3)$, in order to guarantee the existence of a valid partition into subsets $S_1, \dots, S_k$. This is achieved via Algorithm~\ref{alg rearrange}, which adjusts the vectors based on the initial grouping into $I_1, I_2, I_3$. The core of this adjustment is presented in Algorithm~\ref{alg:remove}, which efficiently removes from $I_2$ and $I_3$ those vectors whose component sums are close to prescribed values. A subsequent scaling step, described in Algorithm~\ref{alg:scale}, is then applied to facilitate the final subset assignment, with the complete partitioning procedure detailed in Algorithm~\ref{alg partition}.
\begin{algorithm}[h]
\caption{Remove}\label{alg:remove}
    \begin{algorithmic}
        \STATE \textbf{Input:} $(a_1,b_1),(a_2,b_2),...,(a_n,b_n)$, parameter $S\leq \sum_{i=1}^{n}\vert b_{i}\vert$
        \STATE \textbf{Output}: $I\subseteq \{1,2,...,n\}$ such that $\vert \sum_{i\notin I}a_i\vert\leq \max\{\vert \sum_{i\in I} a_i\vert,\max_{i}(\vert a_{i}\vert)\}$ and $ \vert\sum_{i\in I}\vert b_i\vert -S\vert\leq \max_{i}\vert b_{i}\vert$.
        \STATE Initial: $I=\emptyset$, $I_1=\{i\vert a_i\geq 0$\}, $I_{2}=\{i\vert a_{i}<0\}$, $sum(a)=\sum_{i=1}^{n}a_i$, $sum(b)=0$
        \FOR{t=1,2,...,n}
            \IF{$I_1=\emptyset$ or $I_2=\emptyset$}
            \STATE Randomly choose $i\in I$.
            \STATE Update $I\leftarrow I\cup\{i\}$, $I_1\leftarrow I_{1}\setminus\{i\}$, $sum(a)\leftarrow sum(a)-a_i$, $sum(b)\leftarrow sum(b)+\vert b_i\vert$.
            \ENDIF
            \IF{$sum(a)>0$}
                \STATE Randomly choose $i\in I_1$.
                \STATE Update $I\leftarrow I\cup\{i\}$, $I_1\leftarrow I_{1}\setminus\{i\}$, $sum(a)\leftarrow sum(a)-a_i$, $sum(b)\leftarrow sum(b)+\vert b_i\vert$.
            \ELSE
                \STATE Randomly choose $i\in I_2$.
                \STATE Update $I\leftarrow I\cup\{i\}$, $I_2\leftarrow I_{2}\setminus\{i\}$, $sum(a)\leftarrow sum(a)-a_i$, $sum(b)\leftarrow sum(b)+\vert b_i\vert$.
            \ENDIF
            \IF{$sum(b)\geq S$}
                \STATE Return $I$.
            \ENDIF
        \ENDFOR
    \STATE Return $I$.
    \end{algorithmic}
\end{algorithm}

\begin{algorithm}[h]
\caption{Scale}\label{alg:scale}
    \begin{algorithmic}
        \STATE \textbf{Input:} $(a_1,b_1),(a_2,b_2),...,(a_n,b_n)$, parameter $(a,b)$
        \STATE Compute $\bar{a}=\frac{1}{n}\sum_{i=1}^{n}a_i$, $\bar{b}=\frac{1}{n}\sum_{i=1}^{n}b_i$.
        \STATE $a_{i}\leftarrow a(a_i-\bar{a})$, $b_i\leftarrow b(b_{i}-\bar{b})$.
    \end{algorithmic}
\end{algorithm}

Based on this, Algorithm~\ref{alg partition} further divides each subset $J_i$ ($i = 1,2,3$) into smaller subsets subject to the following requirements, thereby ensuring that the optimization objective $c\vert v^{T}v_{\pi}\vert +\Vert H^{Z}v_{\pi}\Vert^{2}_{2}$ is well controlled.
\begin{enumerate}
    \item For each subset $S_{i}$, $\mathbb{E}[\Vert u_{i}\Vert^{2}_{2}]\in O(\sum_{i=j}^{n}a^{2}_{j})$, and $\frac{1}{\vert S_i\vert}\sum_{j\in S_{i}} a^{2}_{j}$ is close to $\frac{1}{\vert J\vert}\sum_{i\in J}a^{2}_{i}$, so that $\mathbb{E}\left[\sum_{i=1}^{k}\Vert u_{i}\Vert^{2}_{2}\right]$ can be controlled.
    \item Each $\vert S_{i}\vert$ is also not too small so that Lemma~\ref{lem: approximation of optimization objective} holds.
\end{enumerate}
The core idea of Algorithm~\ref{alg partition} is that we construct a sequence $i_1,i_2,...,i_m$ such that $\Vert\sum_{j=1}^{l}(a_{i_j},b_{i_j})\Vert_2$ is small for all $l$, thereby facilitating the formation of the subsets $S_i$.
\begin{algorithm}[h]
\caption{Partitioning Set}
\label{alg partition}
    \begin{algorithmic}
        \STATE \textbf{Input:} $(a_1,b_1),(a_2,b_2),...,(a_k,b_k)$, parameter $M$.
        \STATE \textbf{Target:}
        Partition $\{1,2,...,k\}$ into subsets $I_1,I_2,...,I_m$ such that 
        \STATE(1): For all $1\leq j\leq m$, $(\sum_{i\in I_{1},I_2,...,I_{j}}a_i)^{2}+(\sum_{i\in I_{1},I_2...,I_{j}}b_i)^{2}\leq \sum_{i=1}^{k}a_i^{2}+\sum_{i=1}^{k}b^{2}_{i}$.
        \STATE(2): For any $j\in[1,m]$, $M\leq \sum_{i\in I_{j}}(a_i^{2}+b_{i}^{2})\leq 2M+\max_{i}(a^{2}_{i}+b^{2}_{i})$.
        \STATE \textbf{Step 1:} Construct a stream $s$ of indices $\{1,2,...,k\}$ as follows:
        \STATE \textbf{Initial:} Let $a^{2}_{i_1}+b^{2}_{i_1}=\max_{i\in[1,k]}a^{2}_i+b^{2}_i$, $s=(i_1)$.
        \FOR{$j=1,2,...,k-1$}
            \STATE \textbf{If} $(\sum_{l=1}^{j}a_{i_l})+(\sum_{l=1}^{j}b_{i_l})^{2}\leq \sum_{i=1}^{k}(a^{2}_{i}+b^{2}_{i})$:
            \STATE Find $a^{2}_{i_{j+1}}+b^{2}_{i_{j+1}}=\max_{i\in[1,k],i\notin s}a^{2}_{i}+b^{2}_{i}$
            \STATE Update $s=(i_1,i_2,...,i_{j+1})$.
            \STATE \textbf{Else:}
            \STATE Let $u=(\sum_{l=1}^{j}a_{i_l},\sum_{l=1}^{j}b_{i_l})$, find $i_{j+1}\notin s$  such that $\Vert u+(a_{i_{j+1}},b_{i_{j+1}})\Vert^{2}_{2}\leq \Vert u\Vert^{2}_{2}-\Vert(a_{i_{j+1}},b_{i_{j+1}})\Vert^{2}_{2}$.
            \STATE Update $s=(i_1,...,i_{j+1})$.
        \ENDFOR      
        \STATE \textbf{Step 2:} Partition $s$ into $I_1,I_2,...,I_m$
        \STATE \textbf{Initial:} $j_0=0$
        \FOR{$l=1,2,...$; $j_{l-1}<k$}
        \STATE Find the smallest integer $j_l>j_{l-1}$ such that 
        \STATE $\sum_{t=j_{l-1}+1}^{j_l}(a_{i_t}^{2}+b^{2}_{i_t})\geq M$
        \STATE If such $j_l$ does not exist then simply let $j_l=k$.
        \ENDFOR
        \STATE Let $m=l$, $I_t=\{i_{j_{t-1}+1},...,i_{t}\}$ for $t=1,2,...,l$.
        \STATE \textbf{If} $\sum_{i\in I_{m}}(a^{2}_{i}+b^{2}_{i})<M$:
        \STATE $m\leftarrow m-1$, $I_{m-1}\leftarrow I_{m-1}\bigcup I_m$
    \end{algorithmic}
\end{algorithm}
\paragraph{Discussion on our optimization algorithm} We now summarize our algorithms and discuss the alternative implementation. In Algorithm~\ref{alg structure}, we introduced Algorithm~\ref{alg rearrange} and Algorithm~\ref{alg:remove} to reorganize the three initial subsets $I_1, I_2, I_3$ into a new partition $J_1,J_2,J_3$ of $\{1,2,...,n\}$ so that an upper bound of $\lambda_2$ is guaranteed. The subroutine Algorithm~\ref{alg partition} further partitions each of $J_1, J_2, J_3$ into small subsets while imposing strict control on the sum of the $a$-component within each subset. In fact, the only condition required for the quantity in \eqref{optimization objective: final form} to serve as a valid approximation of the objective in \eqref{optimization objective for alg} is that $\sum_{i}\frac{1}{\vert S_{i}\vert}(\sum_{j\in S_i}a_j)^{2}$ must be well bounded. An alternative implementation of Algorithm~\ref{alg partition} proceeds by randomly selecting $\Omega(n^{0.5+\epsilon})$ elements from each of $J_1, J_2, J_3$ to form each $S_i$. Although this approach lacks theoretical guarantees, it performs adequately in practice when the tail of $X - H^{Z}X$ is not excessively heavy. The details of this randomized variant are provided in Algorithm~\ref{alg partition random}.

\begin{algorithm}[h]
\caption{Partitioning Set (Alternative)}
\label{alg partition random}
    \begin{algorithmic}
        \STATE \textbf{Input:} $(a_1,b_1),(a_2,b_2),...,(a_k,b_k)$, $N=\left\lfloor\frac{k}{n^{\frac{1}{2}+\epsilon}}\right\rfloor$ is the target number of subsets.
        \STATE \textbf{Initial:} $S_{i}\leftarrow \emptyset (i=1,2,...,k)$.
        \FOR{t=1,2,...,k}
            \STATE Sample $i$ uniformly from $\{1,2,...,N\}$.
            \STATE $S_{i}\leftarrow S_{i}\cup \{(a_{t},b_{t})\}$
        \ENDFOR
    \end{algorithmic}
\end{algorithm}
\subsubsection{Comparison with random permutation}\label{subsubsection: comparison with random permutation}
In this section, we demonstrate that the value of $\lambda_2$ in \eqref{optimization objective for alg}, obtained under the permutation group constructed by our algorithm, is no worse than that achieved by uniformly sampling permutations from the full symmetric group—and, in most cases, is strictly better.

Let $\pi^{'}$ be sampled uniformly from all the permutations $[1,n]\to [1,n]$.
We first compute $\mathbb{E}[v^{T}_{\pi^{'}}v]$ and $\mathbb{E}[\Vert H^{Z}v_{\pi^{'}}\Vert^{2}_{2}]$ respectively.

\begin{align*}
\mathbb{E}[v^{T}_{\pi^{'}}v]&=\mathbb{E}\left[\sum_{i=1}^{n}v_{\pi^{'}(i)}v_{i}\right]\\
&=\mathbb{E}\left[\sum_{i=1}^{n}\left(\frac{1}{n}\sum_{j=1}^{n}v_{j}\right)v_{i}\right]\\
&=\frac{1}{n}\left(\sum_{i=1}^{n}v_{i}\right)^{2}\\
&=n\bar{v}^{2}
\end{align*}

\begin{align*}
\mathbb{E}[\Vert H^{Z}v_{\pi^{'}}\Vert^{2}_{2}]
&=\mathbb{E}\left[\left\Vert \bar{v}H^{Z}\vec{1}+\sum_{i=1}^{n}(v_i-\bar{v})w_{\pi^{'}(i)}\right\Vert^{2}_{2} \right]\\
&=\Vert v^{*}\Vert^{2}_{2}+\mathbb{E}\left[\left\Vert \sum_{i=1}^{n}(v_i-\bar{v})w_{\pi^{'}(i)}\right\Vert^{2}_{2} \right]\\
&=\Vert v^{*}\Vert^{2}_{2}+\sum_{i,j=1}^{n} a_{i}a_{j}\mathbb{E}[w^{T}_{\pi^{'}(i)}w_{\pi^{'}(j)}]\\
&=\Vert v^{*}\Vert^{2}_{2}+\frac{1}{n}\sum_{i=1}^{n}a^{2}_{i}\sum_{i=1}^{n}\Vert w_{i}\Vert^{2}_{2}+\frac{1}{n(n-1)}\sum_{1\leq i\neq j\leq n}a_{i}a_{j}\sum_{1\leq i\neq j\leq n}w^{T}_{i}w_{j}
\end{align*}

Similar to the derivation for $\mathbb{E}[\Vert H^{Z}v_{\pi}\Vert^{2}_{2}]$, we have:
$$
\left\vert\frac{1}{n(n-1)}(\sum_{1\leq i\neq j\leq n}a_ia_j)(\sum_{1\leq i\neq j\leq n}w^{T}_{i}w_{j})\right\vert\leq \frac{1}{n-1}\sum_{i=1}^{n}a^{2}_{i}\,.
$$
Therefore, we obtain:
\begin{align}
\mathbb{E}[\Vert H^{Z}v_{\pi^{'}}\Vert^{2}_{2}]\geq \Vert v^{*}\Vert^{2}_2+n\bar{b}\bar{c}-\frac{1}{n-1}\sum_{i=1}^{n}a^{2}_{i}
\end{align}
Combining this with Lemma~\ref{lem: approximation of optimization objective}, we finally obtain that:
\begin{align}
&\mathbb{E}_{\pi_k}[\frac{1}{2} v^{T}_{\pi_k}v+\Vert H^{Z}v_{\pi_k}\Vert^{2}_{2}]-\mathbb{E}_{\pi^{'}}[\frac{1}{2} v^{T}_{\pi^{'}}v+\Vert H^{Z}v_{\pi^{'}}\Vert^{2}_{2}]\\
&\leq \vert J_{2}\vert(\bar{b}_2-\bar{b})(\bar{c}_2-\bar{c})+\vert J_3\vert(\bar{b}_3-\bar{b})(\bar{c}_3-\bar{c})+o(1)\sum_{i=1}^{n}a^{2}_{i}\,,
\end{align}
where both $\vert J_{2}\vert(\bar{b}_2-\bar{b})(\bar{c}_2-\bar{c})$ and $\vert J_3\vert(\bar{b}_3-\bar{b})(\bar{c}_3-\bar{c})$ are smaller than 0, and their absolute values can be as large as $\Omega(\sum_{i=1}^{n}a^{2}_{i})$, with $\sum_{i=1}^{n}a^{2}_{i}=\Vert v\Vert^{2}_{2}-n\bar{v}^{2}$.

Finally, combining Theorem~\ref{thm: value of optimization objective}, which provides a lower bound for $\lambda_2(X,Z,\mathcal{P}_n,\frac{1}{2}\alpha)$, with Proposition~\ref{proposition: high probability bound on group decomposition}, which yields an upper bound for $\lambda_2(X,Z,\mathcal{P}_{K},\frac{1}{4}\alpha)$, we establish the following comparison between our algorithm and the random permutation method in the regime $\alpha \to 0_+$:
\begin{enumerate}
    \item When the conditions in Proposition~\ref{proposition: high probability bound on group decomposition} hold, and additionally $\max_{i}\{(v_{i}-\bar{v})^{2}\}\leq \frac{1}{poly\ln(n)}\sum_{i=1}^{n}(v_i-\bar{v})^{2}$ for some polynomial, our algorithm has a solution with $\lambda$ in~(\ref{optimization objective for alg}) provably not worse than uniformly sampling permutations.
    \item The provable gap in which our solution surpasses the random permutation is presented by \\$\left\vert \vert J_{2}\vert(\bar{b}_2-\bar{b})(\bar{c}_2-\bar{c})+\vert J_{3}\vert(\bar{b}_3-\bar{b})(\bar{c}_3-\bar{c})\right\vert-o(1)\Vert v\Vert^{2}_{2}$, which depends on the exact $X,Z$.
\end{enumerate}
Now we explain how different distribution of $Z$ as well as $n,p$ influences the gap $\vert \vert J_{2}\vert(\bar{b}_2-\bar{b})(\bar{c}_2-\bar{c})+\vert J_{3}\vert(\bar{b}_3-\bar{b})(\bar{c}_3-\bar{c})\vert$ of $\lambda_2$ in~\eqref{optimization objective for alg}.
\paragraph{The impact of $p$.}
Recall the optimization problem~\eqref{optimization: intuitive} where we finally provide a lower bound:
\begin{align*}
&X^{T}(I-H^{ZZ_{\pi_k}})X-X^{T}_{\pi_k}(I-H^{ZZ_{\pi_k}})X\\
&=X^TX-[X^{T}H^{ZZ_{\pi_{k}}}X+X^{T}_{\pi_{k}}(I-H^{ZZ_{\pi_k}})X]\\
&\leq \left[X^{T}X-\Vert H^{Z}X\Vert^{2}_{2}+\frac{1}{2}\Vert (I-H^{Z})X\Vert^{2}_{2}\right]-\left[\Vert H^{Z_{\pi_k}}(X-H^{Z}X)\Vert^{2}_{2}+\frac{1}{2}X^{T}_{\pi_{k}}(I-H^{Z_{\pi_k}})(I-H^{Z})X\right]
\\&
=\frac{3}{2}\Vert v\Vert^{2}_{2}-\left[\Vert H^{Z_{\pi_k}}v\Vert^{2}_{2}+\frac{1}{2} v^{T}_{\pi_k}v  \right]\,,
\end{align*}
where in the final expression, the first term corresponds to the projection residual $X - H^{Z}X$, which is independent of any permutation, while the second term constitutes our optimization objective. We first conclude that, as $p/n$ increases, the maximum achievable gap between our algorithm and uniformly random permutation also widens.

First, $\sum_{i=1}^{n}b_{i}=tr((H^{Z})^{2})=rank(Z)\leq p-1$, which implies $\bar{b}\leq \frac{p-1}{n}$. On the other hand, we have:
\begin{align*}
\vert J_{2}\vert \vert \bar{b}_{2}-\bar{b}\vert&=\left\vert \sum_{j\in J_{2}}b_{j}-\bar{b}\vert J_{2}\vert\right\vert\\
&=\max(\bar{b}\vert J_{2}\vert-\sum_{j\in J_{2}}b_{j},\sum_{j\in J_{2}}b_{j}-\bar{b}\vert J_{2}\vert)\\
&\leq \max\left\{ \bar{b}\vert J_{2}\vert-\max(0,n\bar{b}-\vert J_{2}\vert), \min(\vert J_{2}\vert,n\bar{b})-\bar{b}\vert J_{2}\vert \right\}\\
&\leq n\bar{b}(1-\bar{b})
\end{align*}
Combining with the fact that $\bar{c}_{2}\geq 0$, we obtain:
\begin{align*}
\vert J_{2}\vert \vert \bar{b}_{2}-\bar{b}\vert\vert \bar{c}_{2}-\bar{c}\vert\leq n\bar{b}(1-\bar{b})\bar{c}\leq \frac{p}{n}(1-\frac{p}{n})\sum_{i=1}^{n}a^{2}_{i}\,.
\end{align*}
For $J_{3}$, we use the fact that $\vert \bar{b}_{3}-\bar{b}\vert\leq \bar{b}$ and obtain:
\begin{align*}
\vert J_{3}\vert\vert\bar{b}_{3}-\bar{b}\vert\vert \bar{c}_{3}-\bar{c}\vert\leq \bar{b}\sum_{i=1}^{n}c_{i}\leq \frac{p}{n}\sum_{i=1}^{n}a^{2}_{i}.
\end{align*}
Combining with these two upper bounds, we have:
\begin{align*}
\vert \vert J_{2}\vert (\bar{b}_{2}-\bar{b})(\bar{c}_2-\bar{c})+\vert J_{3}\vert (\bar{b}_{3}-\bar{b})(\bar{c}_3-\bar{c})  \vert\leq \frac{p}{n}(2-\frac{p}{n})\sum_{i=1}^{n}a^{2}_{i}.
\end{align*}
Therefore, we conclude that as $p/n$ increases (under the standing assumption that $n \geq 2p$), the potential gap between our algorithm and the uniformly random permutation widens. Intuitively, this suggests that the value of $\lambda$ in \eqref{optimization: intuitive} achieved by our algorithm is expected to yield a greater improvement over the permutation group consisting of all permutations.

\paragraph{The impact of $Z$.} 
In this section, we demonstrate how a heavy-tailed distribution of $Z$ can enlarge the performance gap. For the term $\vert J_{2}\vert\vert \bar{b}_{2}-\bar{b}\vert\vert \bar{c}_{2}-\bar{c}\vert$, we have the following bound:
\begin{align*}
\vert J_{2}\vert  \vert \bar{b}_{2}-\bar{b}\vert\cdot \vert \bar{c}_{2}-\bar{c}\vert \leq \vert \bar{b}_2-\bar{b}\vert\cdot \frac{\vert J_{2}\vert}{n} \sum_{i=1}^{n}a^{2}_{i}.
\end{align*}
On the other hand, we also have:
\begin{align*}
\vert J_{3}\vert \vert b_{3}-\bar{b}\vert \vert c_{3}-\bar{c}\vert \leq \vert b_{3}-\bar{b}\vert\cdot \sum_{i=1}^{n} a^{2}_{i}\,,
\end{align*}
with $\vert b_{3}-\bar{b}\vert\leq \frac{p-1}{n}$.
This indicates that when the values $b_i$ are concentrated around $\bar{b}$, the performance of our algorithm is substantially limited. Conversely, when the $b_i$ are more likely to deviate from $\bar{b}$, our algorithm is expected to achieve better performance.

We now return to explain the experimental results shown in Figures~\ref{fig:gaussian X and gaussian noise} and \ref{fig:t2 X and gaussian noise}. To this end, we compute the probability density of $\Vert H^{Z}e_{i}\Vert^{2}_{2}$ for the values of $n$ and $p$ tested previously, where $e_i \in \mathbb{R}^n$ denotes the unit vector with a $1$ in the $i$-th coordinate.
Figure~\ref{fig:distribution of HZ} shows that a heavier tail in the distribution of $Z$ leads to greater variance in $\Vert H^{Z} e_{i}\Vert^{2}_{2}$. This, in turn, increases both $\vert \bar{b}_{2}-\bar{b}\vert$ and $\vert \bar{b}_{3}-\bar{b}\vert$. In particular, a heavy tail prevents the projected $\ell_2$ norm of each standard basis vector from concentrating around its expectation, thereby enabling a larger performance gap.

\begin{figure}[ht]
    \centering
    \begin{subfigure}[b]{0.32\textwidth}
        \includegraphics[width=\textwidth]{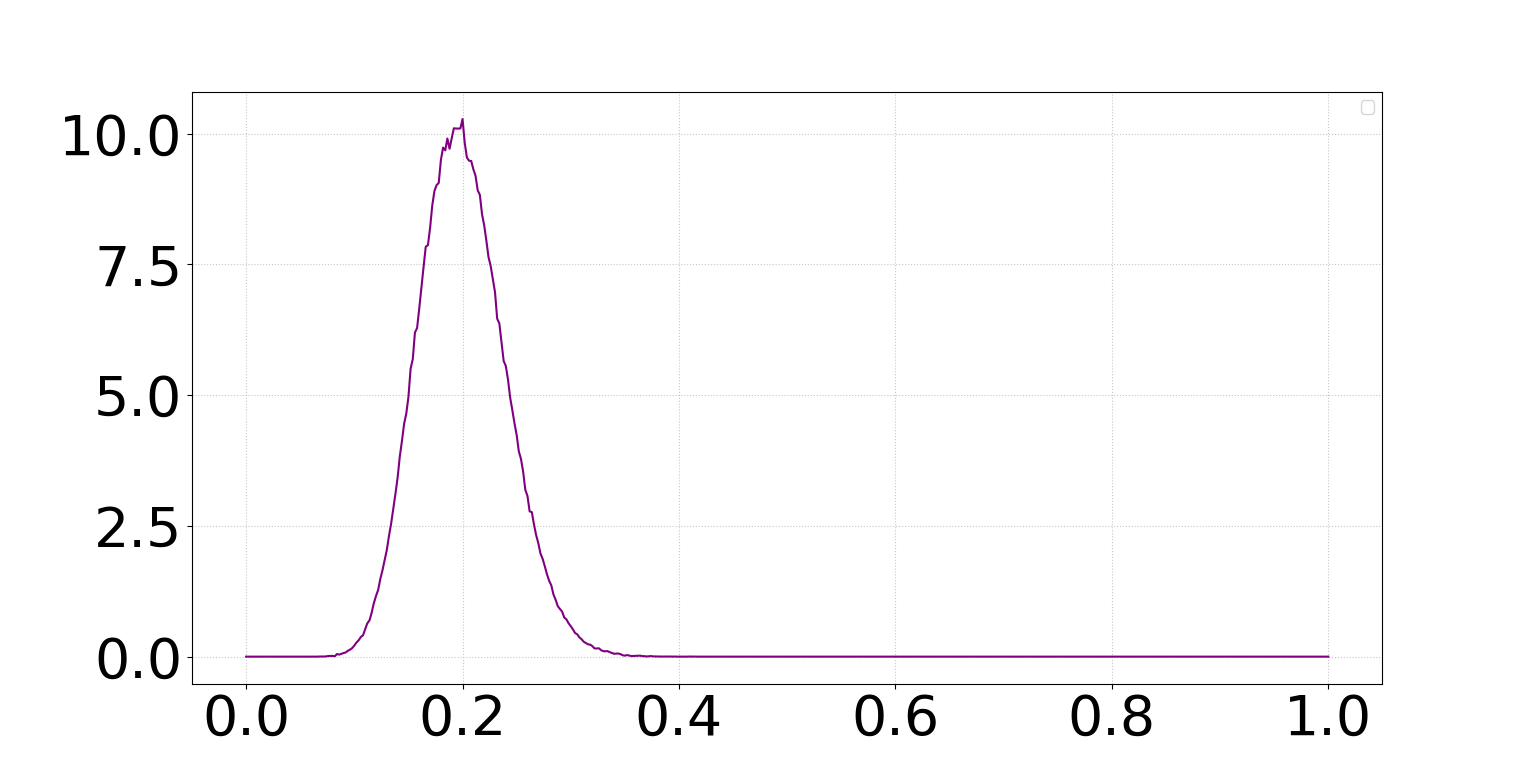}
        \caption{n=200,p=40, Gaussian Z}
        \label{fig:HZ_200_40_g_p40}
    \end{subfigure}
    \hfill
    \begin{subfigure}[b]{0.32\textwidth}
        \includegraphics[width=\textwidth]{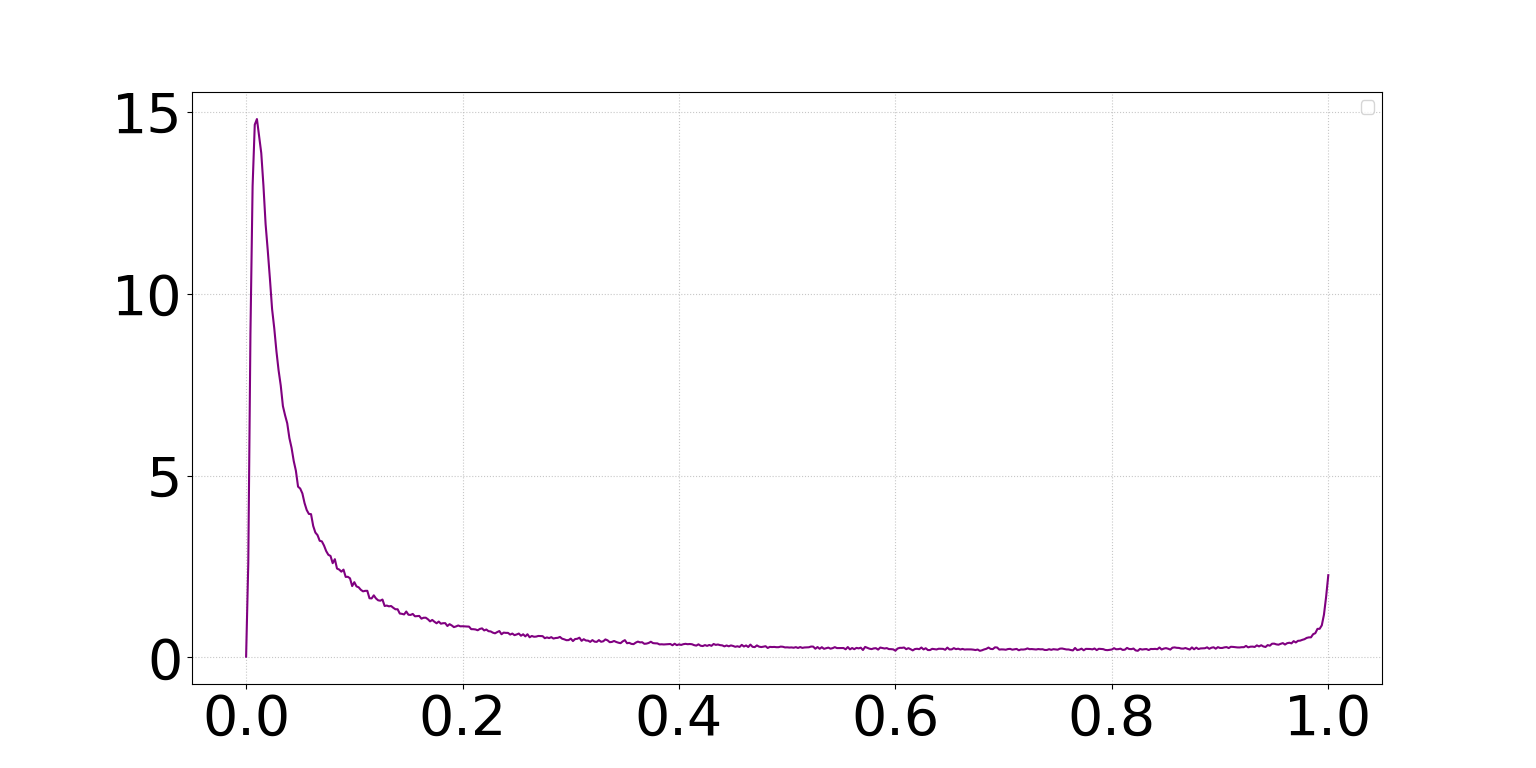}
        \caption{n=200,p=40, $t_1$ Z}
        \label{fig:HZ_200_40_t1}
    \end{subfigure}
    \hfill
    \begin{subfigure}[b]{0.32\textwidth}
        \includegraphics[width=\textwidth]{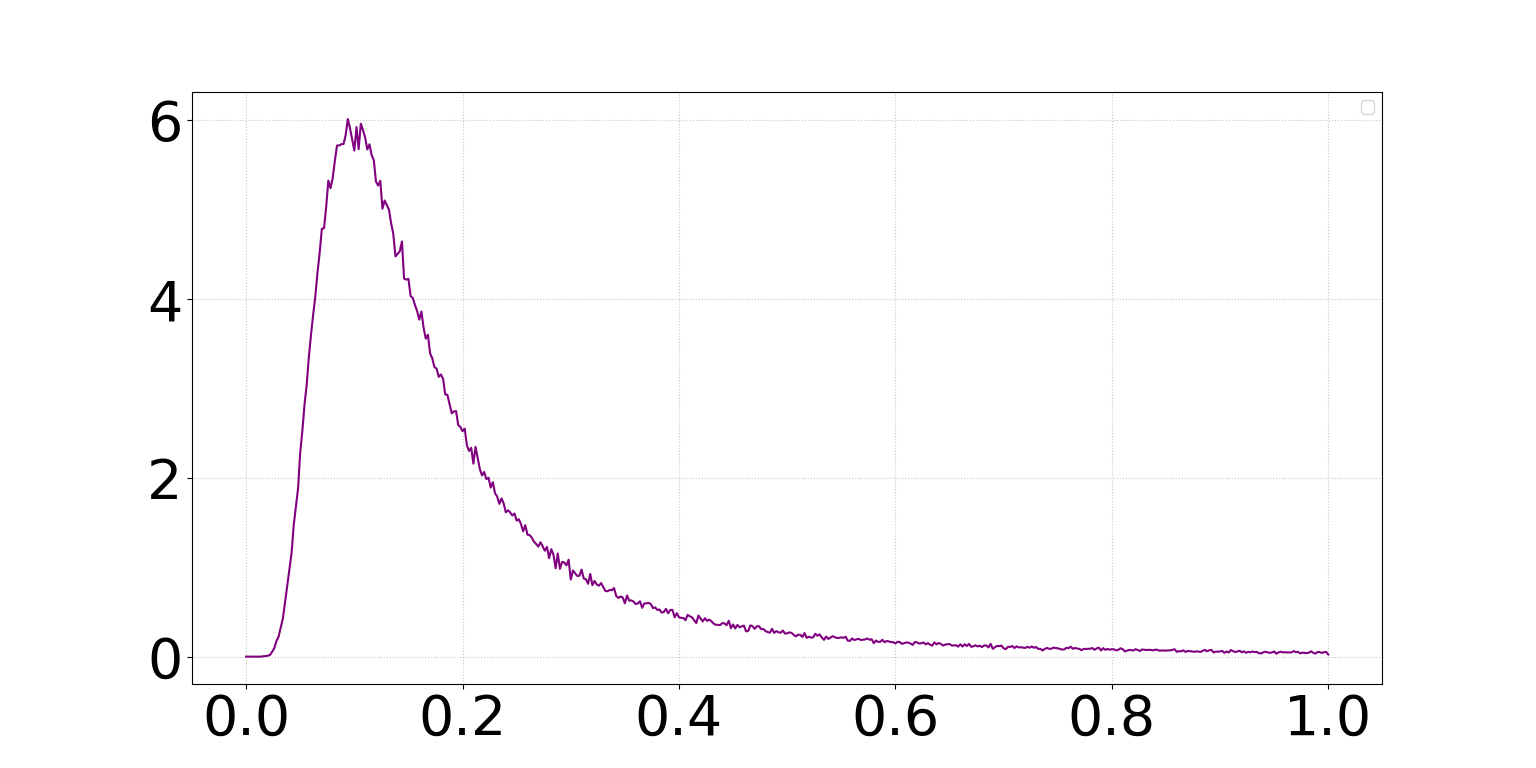}
        \caption{n=200,p=40, $t_2$ Z}
        \label{fig:HZ_200_40_t2}
    \end{subfigure}
    
    \vfill
     \begin{subfigure}[b]{0.32\textwidth}
        \includegraphics[width=\textwidth]{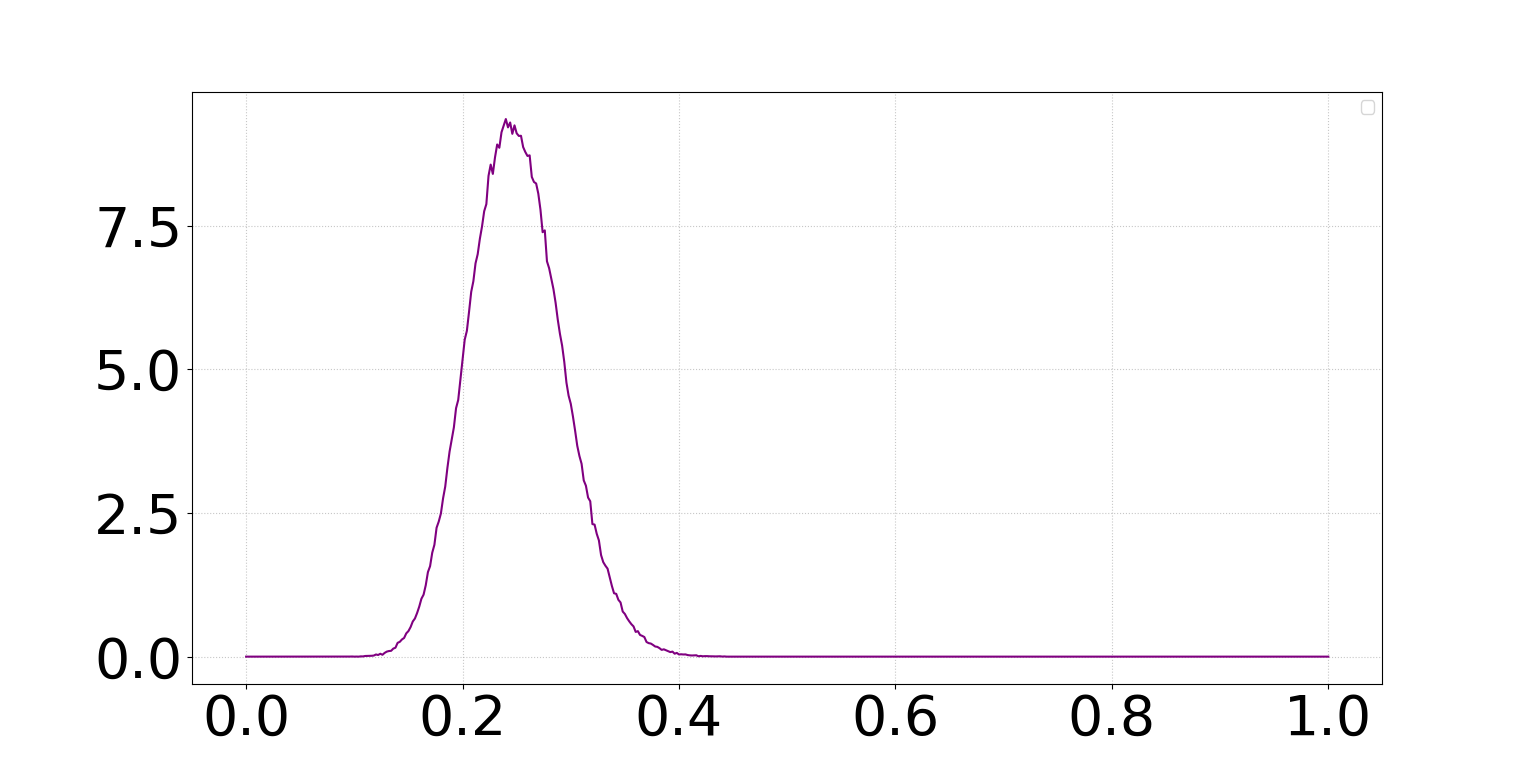}
        \caption{n=200,p=50, Gaussian Z}
        \label{fig:HZ_200_40_g_p50}
    \end{subfigure}
    \hfill
    \begin{subfigure}[b]{0.32\textwidth}
        \includegraphics[width=\textwidth]{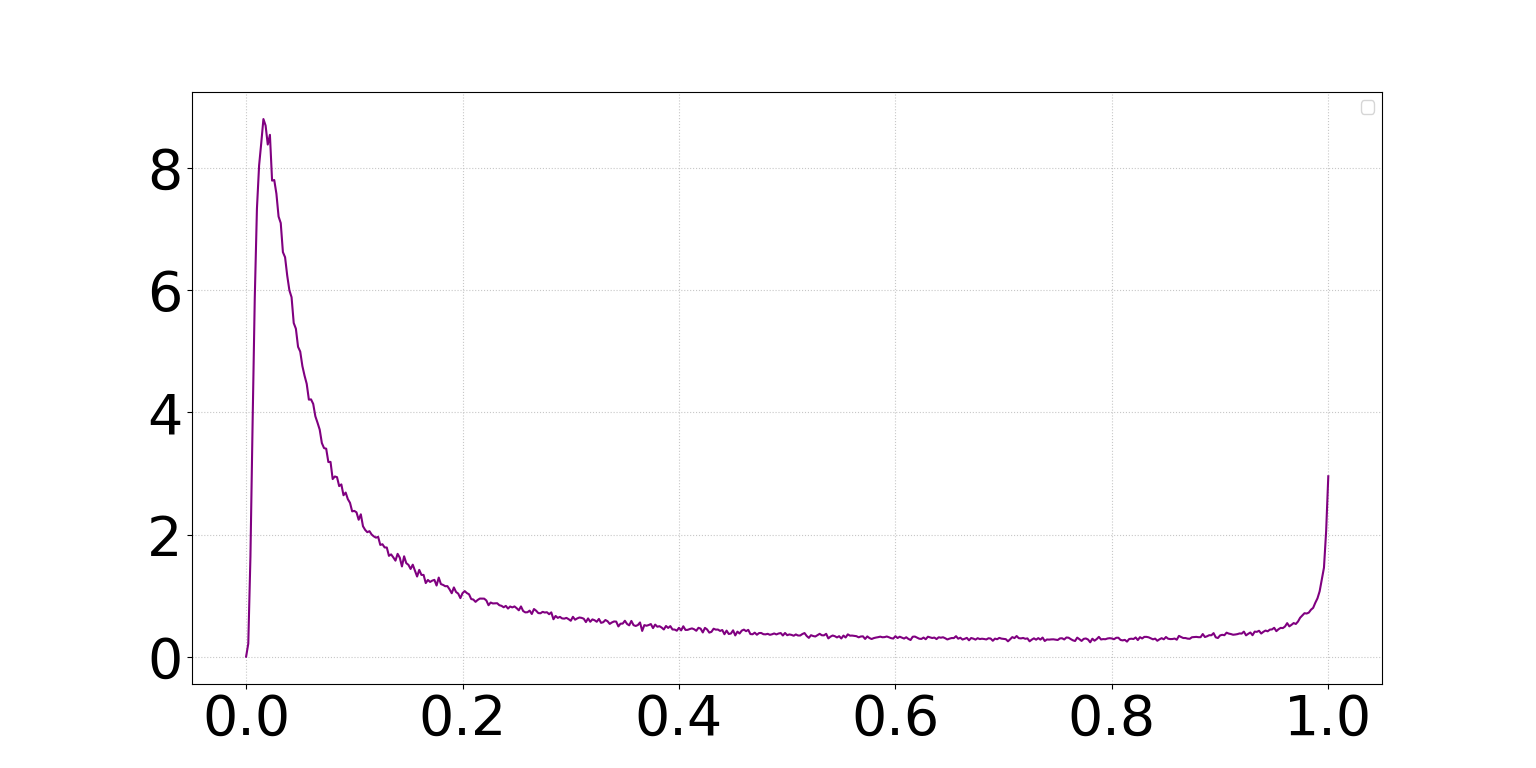}
        \caption{n=200,p=50, $t_1$ Z}
        \label{fig:HZ_200_50_t1}
    \end{subfigure}
    \hfill
    \begin{subfigure}[b]{0.32\textwidth}
        \includegraphics[width=\textwidth]{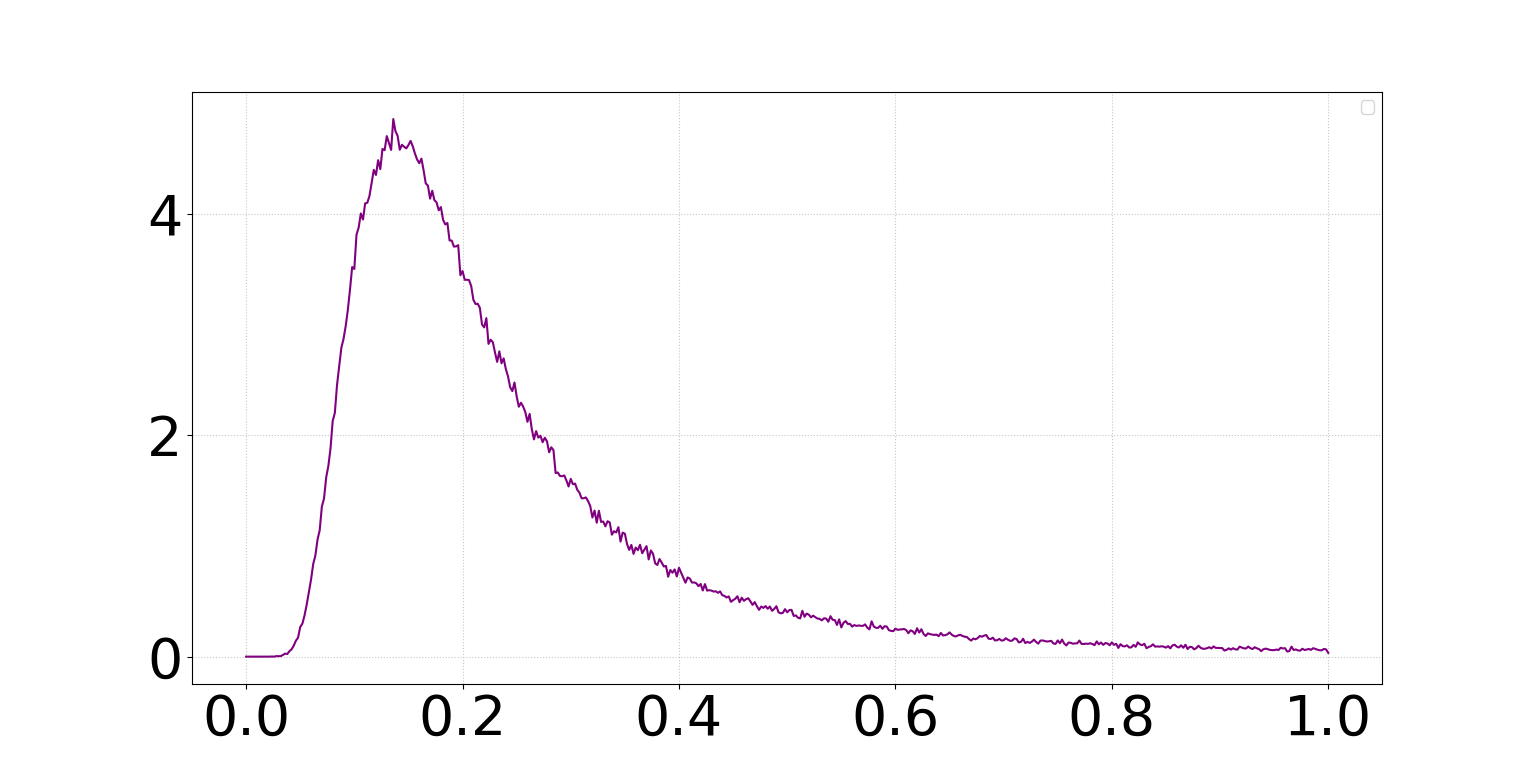}
        \caption{n=200,p=50, $t_2$ Z}
        \label{fig:HZ_200_50_t2}
    \end{subfigure}
    
    \vfill
    
    \begin{subfigure}[b]{0.32\textwidth}
        \includegraphics[width=\textwidth]{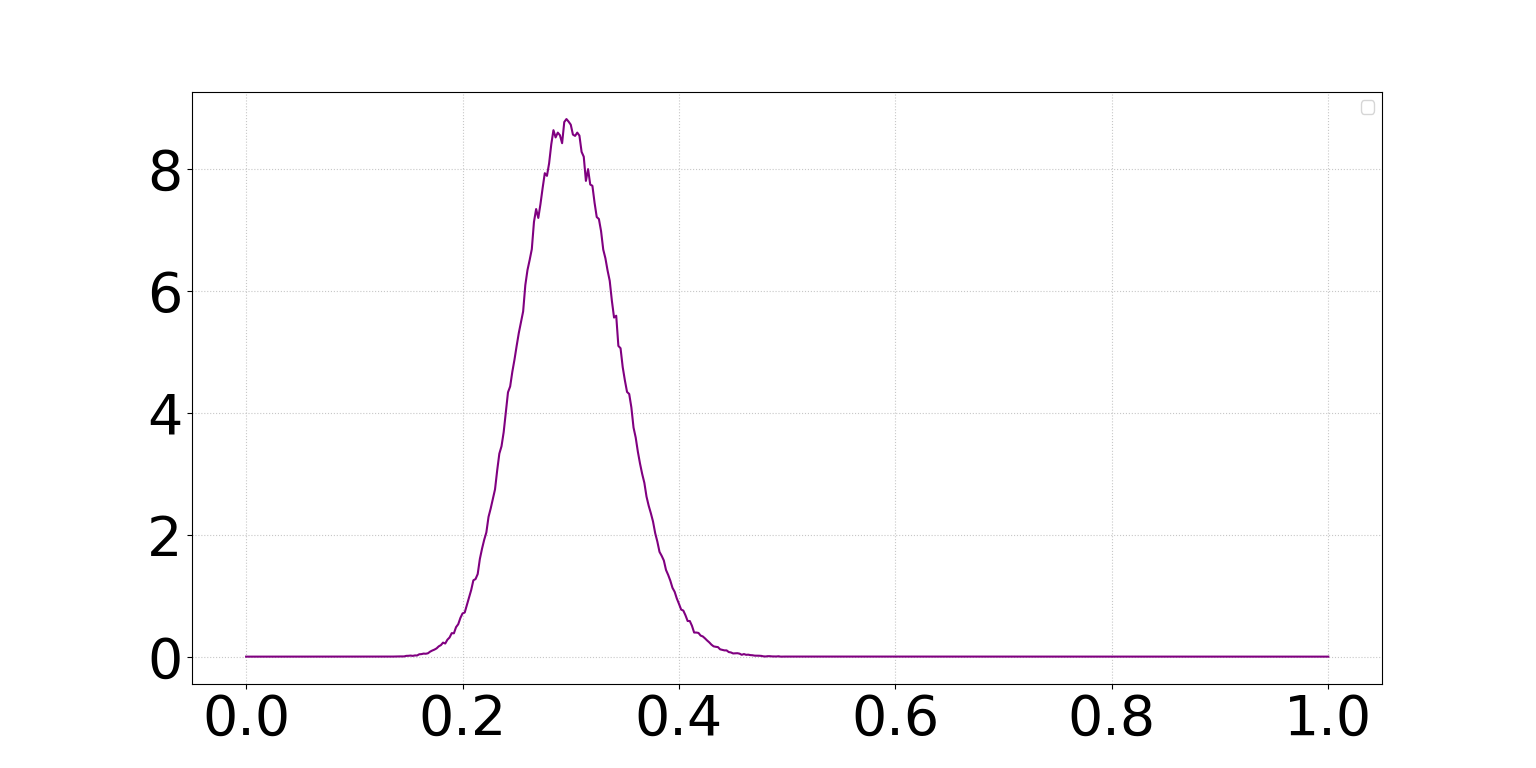}
        \caption{n=200,p=60, Gaussian Z}
        \label{fig:HZ_200_60_g}
    \end{subfigure}
    \hfill
    \begin{subfigure}[b]{0.32\textwidth}
        \includegraphics[width=\textwidth]{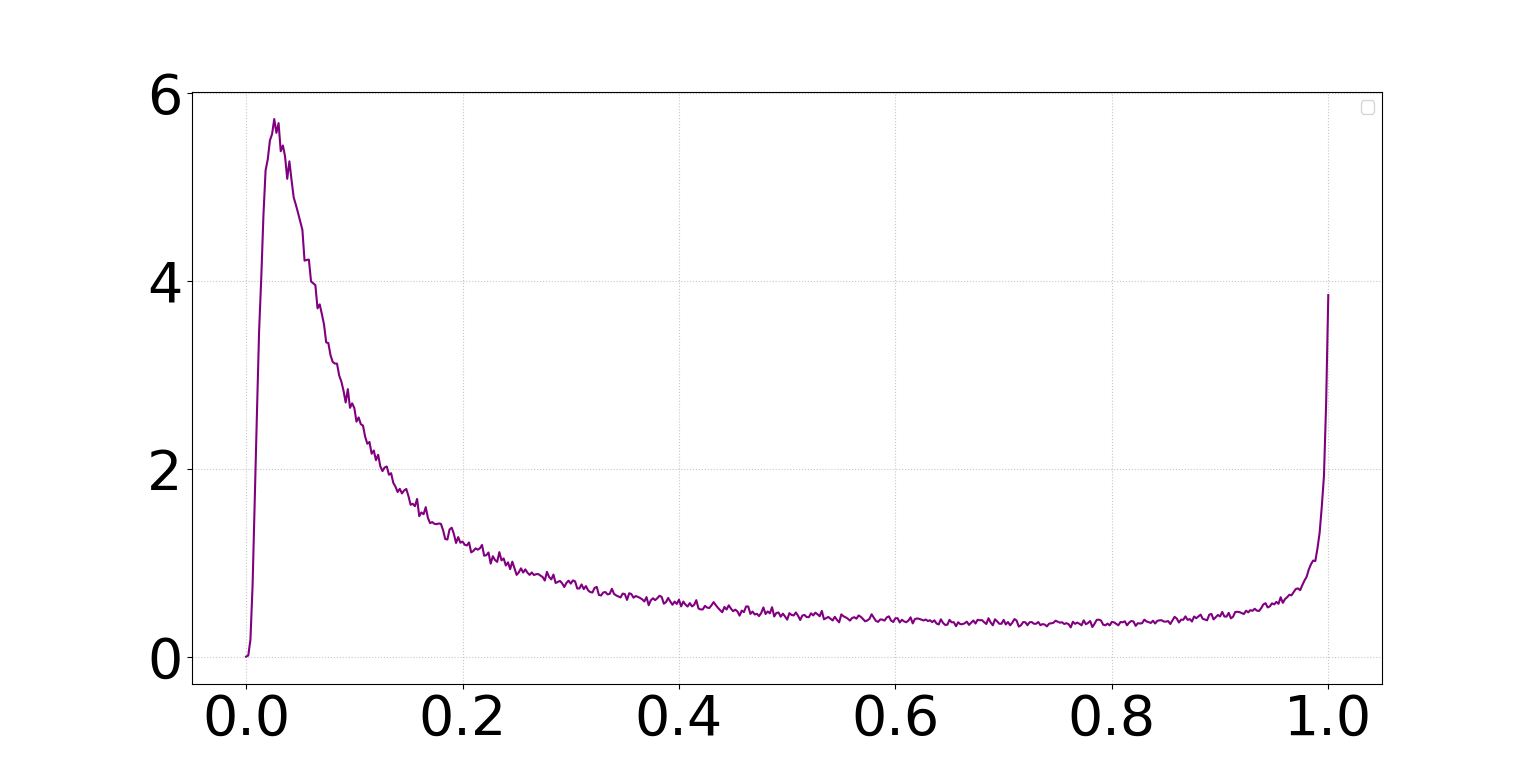}
        \caption{n=200,p=60, $t_1$ Z}
        \label{fig:HZ_200_t1}
    \end{subfigure}
    \hfill
    \begin{subfigure}[b]{0.32\textwidth}
        \includegraphics[width=\textwidth]{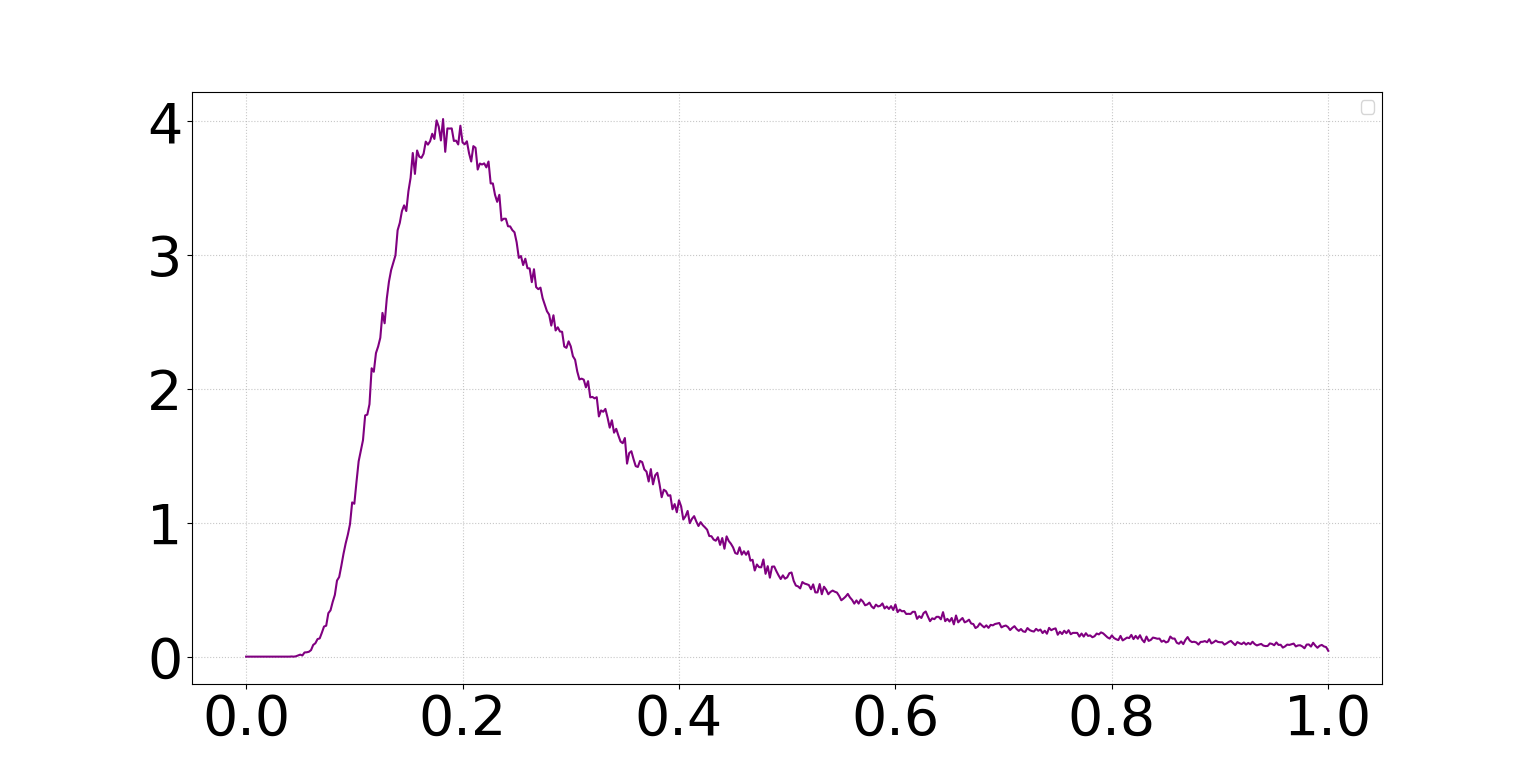}
        \caption{n=200,p=60, $t_2$ Z}
        \label{fig:HZ_200_60_t2}
    \end{subfigure}
    
    \vfill
    
    \begin{subfigure}[b]{0.32\textwidth}
        \includegraphics[width=\textwidth]{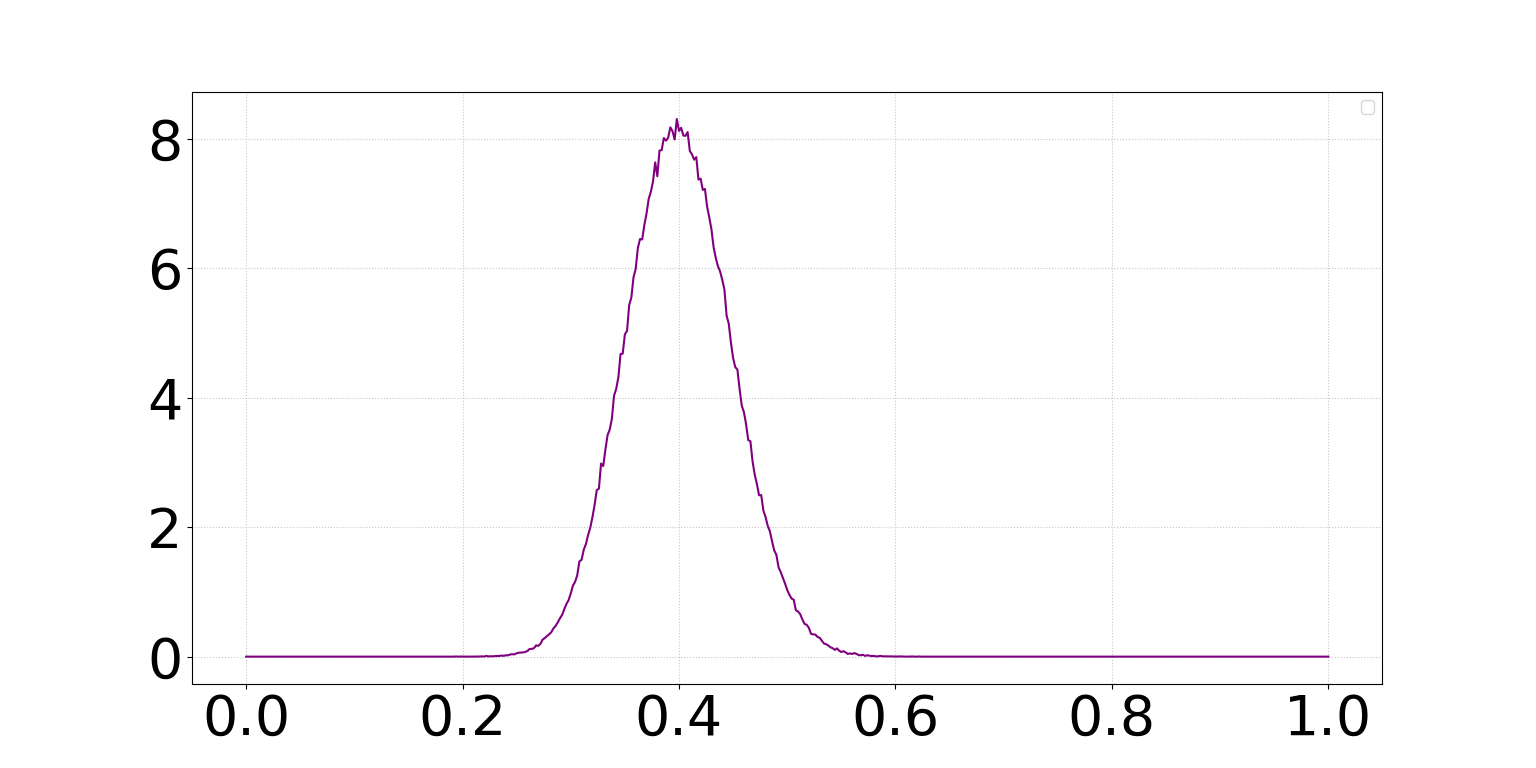}
        \caption{n=200,p=80, Gaussian Z}
        \label{fig:HZ_200_80_g}
    \end{subfigure}
    \hfill
    \begin{subfigure}[b]{0.32\textwidth}
        \includegraphics[width=\textwidth]{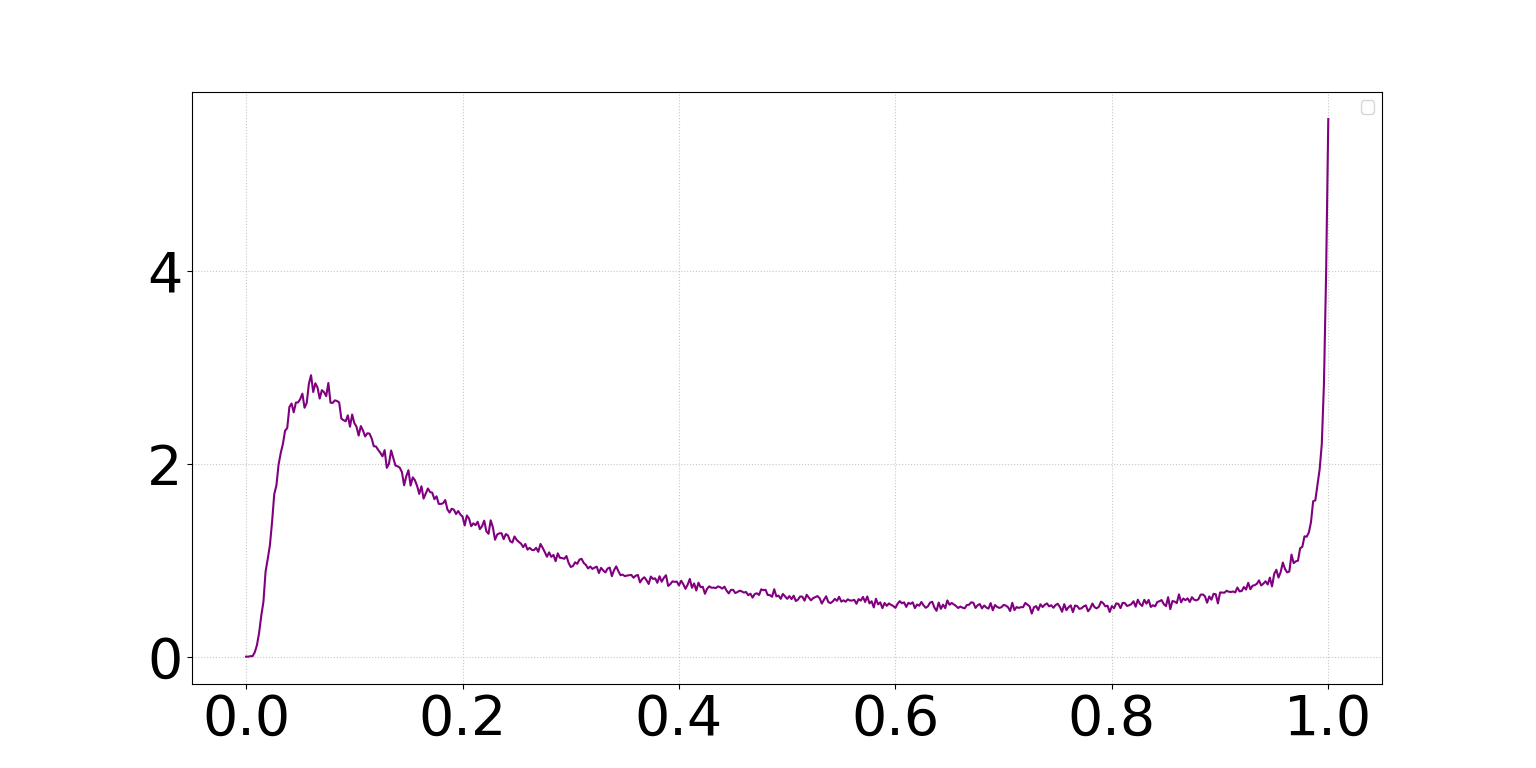}
        \caption{n=200,p=80, $t_1$ Z}
        \label{fig:HZ_200_80_t1}
    \end{subfigure}
    \hfill
    \begin{subfigure}[b]{0.32\textwidth}
        \includegraphics[width=\textwidth]{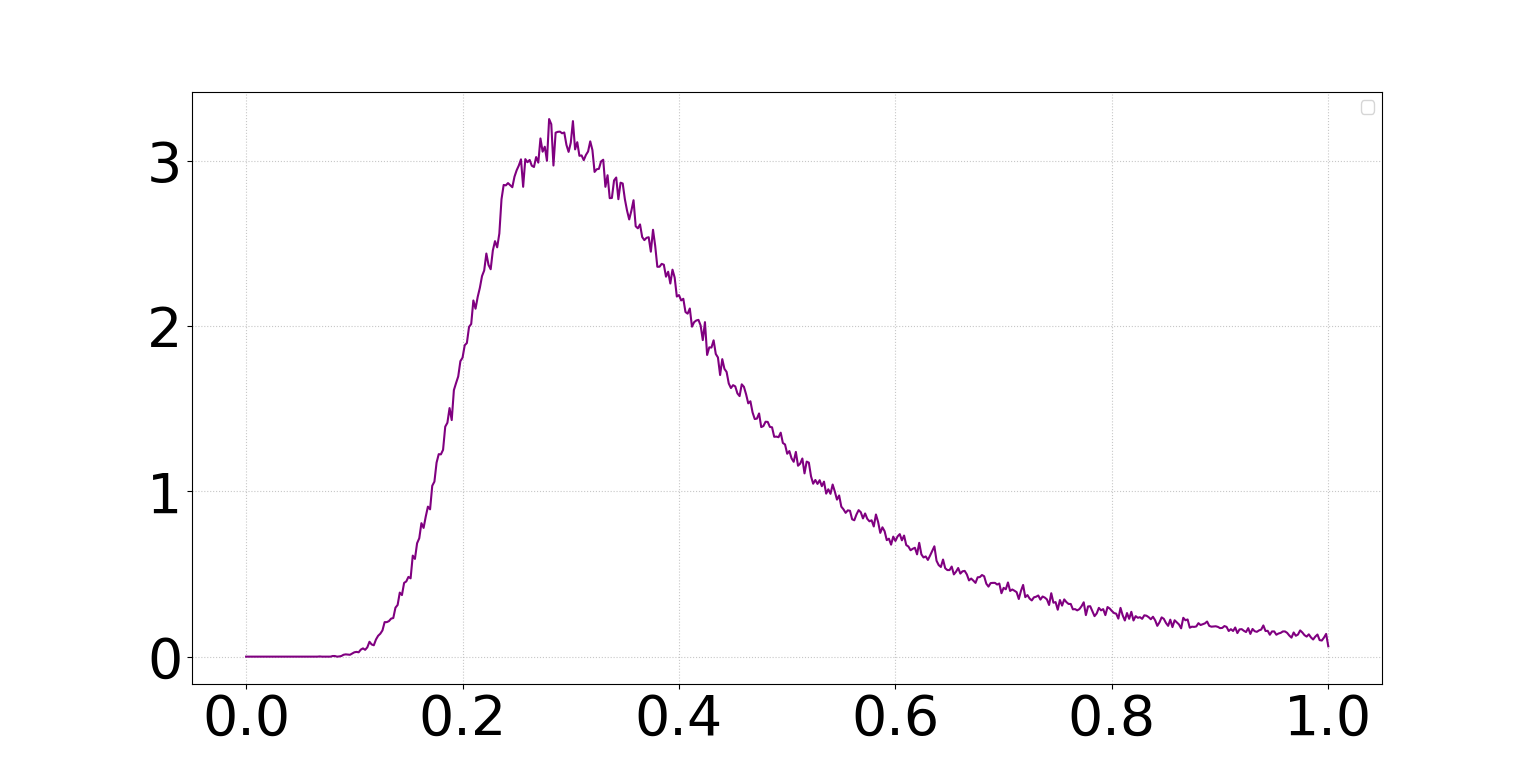}
        \caption{n=200,p=80, $t_2$ Z}
        \label{fig:HZ200_80_t2}
    \end{subfigure}
    
    \caption{Probability density of $\Vert H^{Z}e_{i}\Vert^{2}_{2}$}
    \label{fig:distribution of HZ}
\end{figure}

In addition to the theoretical guarantees for the optimization objective, we provide a numerical example to offer further insight into why our algorithm outperforms the uniform random selection of permutations.

\paragraph{A naive example.} Consider $Z=I_{p}$, $X_{i}\sim\mathcal{N}(0,I_{n})$, $\epsilon\sim\mathcal{N}(0,I_n)$ (for simplicity). 

\textbf{Case 1:} For $\pi$ uniformly sampled from the set of all permutations, the vector $(I - H^{ZZ_{\pi}})X$ retains the coordinates of $X$ indexed by $I=\{i\vert i>p, \pi(j)\neq i,\,\forall 1\leq j\leq p\}$. It can be shown that $\mathbb{E}[\vert I\vert]=(n-p)(1-\frac{p}{n})$. Now consider the inequality $X^{T}(I-H^{ZZ\pi})Y\leq X^{T}_{\pi}(I-H^{ZZ_{\pi}})Y$. If $b>0$, we have:
\begin{align*}
\mathbbm{1}\left\{X^{T}(I-H^{ZZ\pi})Y\leq X^{T}_{\pi}(I-H^{ZZ_{\pi}})Y \right\}= \mathbbm{1}\left\{b(X-X_{\pi})^{T}(I-H^{ZZ_{\pi}})X\leq(X-X_{\pi})^{T}(I-H^{ZZ_{\pi}})\epsilon\right\}\,.
\end{align*}
By taking expectation, we obtain:
\begin{align*}
&\mathbb{E}\left[1\left\{X^{T}(I-H^{ZZ\pi})Y\leq X^{T}_{\pi}(I-H^{ZZ_{\pi}})Y \right\}\right]\\
&=\mathbb{E}_{X}\mathbb{E}_{\epsilon}\left[1\left\{b(X-X_{\pi})^{T}(I-H^{ZZ_{\pi}})X\leq(X-X_{\pi})^{T}(I-H^{ZZ_{\pi}})\epsilon\right\}\right]\\
&=\mathbb{E}_{X}\left[\Phi(\frac{\Vert (I-H^{ZZ_{\pi}})X\Vert_{2}}{b(X-X_{\pi})^{T}X}\right]\,,
\end{align*}
where $\Phi(x)=\frac{1}{\sqrt{2\pi}}\int_{x}^{\infty}e^{-\frac{1}{2}t^{2}}dt$.
As $n,p\to +\infty$, $n>2p$, for any $c>0$, $$\mathbb{P}\left[\frac{\Vert (I-H^{ZZ_{\pi}})X\Vert^{2}_{2}}{(n-p)(1-p/n)}\in[1-c,1+c]\right]\to 1\,,$$
$$
\mathbb{P}\left[\frac{(X-X_{\pi})^{T}X}{(n-p)(1-p/n)}\in[1-c,1+c]\right]\to 1\,.
$$
\textbf{Case 2:} For $\pi$ chosen in our algorithm, we notice that $J_1=\emptyset$ because for any $i$ we have $(b_i-\bar{b})(c_i-\bar{c})<0$, $J_{2}$ contains the first $p$ components and $J_3$ contains the last $n-p$ components. Then we have:
$$
\mathbb{E}\left[1\left\{X^{T}(I-H^{ZZ\pi})Y\leq X^{T}_{\pi}(I-H^{ZZ_{\pi}})Y \right\}\right]=\mathbb{E}_{X}\left[\Phi(\frac{\Vert (I-H^{ZZ_{\pi}})X\Vert_{2}}{b(X-X_{\pi})^{T}X}\right]\,.
$$
But differently, as $n,p\to +\infty$, $n>2p$, for any $c>0$, $$\mathbb{P}\left[\frac{\Vert (I-H^{ZZ_{\pi}})X\Vert^{2}_{2}}{n-p}\in[1-c,1+c]\right]\to 1\,,$$
$$
\mathbb{P}\left[\frac{(X-X_{\pi})^{T}X}{n-p}\in[1-c,1+c]\right]\to 1\,.
$$
This is because our algorithm divides $\{1,2,...,p\}$ and $\{p+1,...,n\}$ into distinct subsets, ensuring that for all $i \in [1,p]$, we have $\pi(i) \in [1,p]$. This structural difference implies that, to achieve the same Type II error, Case 1 requires a larger signal level, namely $b = \sqrt{n/(n-p)}\, b'$, compared with $b'$ under our algorithm. The case for $b < 0$ is symmetric and is addressed analogously.
In this example, our algorithm provably attains better performance in detecting a nonzero $b$. This aligns with the behavior of our optimization framework, where $\mathbb{E}[\frac{1}{2}\vert v^{T}_{\pi}v\vert +\Vert H^{Z}v_{\pi}\Vert^{2}_{2}]$ is strictly smaller under our algorithm, with a provable gap of $\frac{p}{n}(1-\frac{p}{n})X^{T}X$ compared to uniformly random permutations.

\subsection{Proof of main results}
\label{sec::Type ii error proof of main results}
In this part, we derive the proof of our main results about Type II error control, including the main theorems, lemmas, and detailed computations.

\subsubsection{Proof of Lemma~\ref{lem: upper and lower of optimization objective formal}}\label{subsubsection: proof of thm: upper lower of optimization objective}
For \eqref{lower bound term 1}, consider two orthonormal bases $\{u_1,...,u_{m_1}\}$, $\{u_1,...,u_{m_1},v_{1},...,v_{m_2}\}$ where $\{u_1,...,u_{m_1}\}$ is an orthonormal basis of $H^Z$, $\{u_1,...,u_{m_1},v_1,...,v_{m_2}\}$ is an orthonormal basis of $H^{ZZ_{\pi}}$. Decompose $X$ as $X=u+v+w$ with $u\in \operatorname{span}(u_1,...,u_{m_1})$, $v\in \operatorname{span}(v_1,...,v_{m_2})$, and $H^{ZZ_{\pi_k}}w=0$.
Then we have $H^{Z} X = u$ and $H^{ZZ_{\pi_k}} X = u + v$, so that $\Vert H^{ZZ_{\pi_k}}X\Vert^{2}_{2}=\Vert u+v\Vert^{2}_{2}=\Vert u\Vert^{2}_{2}+\Vert v\Vert^{2}_{2}$. 

On the other hand,
$\Vert H^{Z_{\pi}}(X-H^{Z}X)\Vert^{2}_{2}=\Vert H^{Z_{\pi}}(v+w)\Vert^{2}_{2}=\Vert H^{Z_{\pi}}v\Vert^{2}_{2}\leq \Vert v\Vert^{2}_{2}\,,
$
which implies \eqref{lower bound term 1}.

We now prove the converse direction of Lemma~\ref{lem: upper and lower of optimization objective formal}. To this end, we first establish Lemma~\ref{lem: upper bound of HZZpiX}, which demonstrates that it suffices to bound $\sup_{u\in span(Z),v\in span(Z_{\pi})}\frac{u^{T}v}{\Vert u\Vert_{2}\Vert v\Vert_{2}}$.

\begin{lemma}\label{lem: upper bound of HZZpiX}
Suppose that we have $\sup_{u\in span(Z),v\in span(Z_{\pi})}\frac{u^{T}v}{\Vert u\Vert_{2}\Vert v\Vert_{2}}\leq t$, then we have:
\begin{align}\label{eq: upper bound of HZZpiX}
\Vert H^{ZZ_{\pi}}X\Vert^{2}_{2}\leq \Vert H^{Z}X\Vert^{2}_{2}+\frac{1}{1-t^{2}}\Vert H^{Z_{\pi}}(X-H^{Z}X)\Vert^{2}_{2}\,.
\end{align}
\end{lemma}

\begin{proof}
Let $\{u_1,...,u_{m_1}\}$ be an orthogonal basis of $Z$, and $\{u_1,...,u_{m_1},v_1,...,v_{m_2}\}$ be an orthogonal basis of $ZZ_{\pi}$. Then we can let $X=u+v+w$ with $u\in span(\{u_1,...,u_{m_1}\})$, $v\in span(\{v_1,...,v_{m_2}\})$. Then $H^{Z}X=u$, $H^{ZZ_{\pi}}X=u+v$, $H^{Z_{\pi}}(X-H^{Z}X)=H^{Z_{\pi}}v$.

Let $v = u^{'}+v^{'}$ with $u^{'}\in span(\{u_1,...,u_{m_1}\})=span(Z)$, $v^{'}\in span(Z_{\pi})$. Then $(u^{'})^{T}v=0$, $\vert(u^{'})^{T}v^{'}\vert\leq t\Vert u^{'}\Vert_{2}\Vert v^{'}\Vert_{2}$, $\Vert v\Vert^{2}_{2}=v^{T}v=v^{T}(u^{'}+v^{'})=v^{T}v^{'}$ and
$$\Vert H^{Z_{\pi}}v\Vert^{2}_{2}=v^{T}H^{Z_{\pi}}v=v^{T}H^{Z_{\pi}}(u^{'}+v^{'})=v^{T}v^{'}+v^{T}H^{Z_{\pi}}u^{'}\geq \Vert v\Vert^{2}_{2}-t\Vert u^{'}\Vert_2\Vert H^{Z_{\pi}}v\Vert_{2}\,.$$

Now we upper bound $\Vert u^{'}\Vert_{2}$. Since $\Vert v\Vert_{2}^{2}=\Vert v^{'}\Vert^{2}_{2}+\Vert u^{'}\Vert^{2}_{2}+2(u^{'})^{T}v^{'}\geq \Vert v^{'}\Vert^{2}_{2}+\Vert u^{'}\Vert^{2}_{2}-2t\Vert u^{'}\Vert_{2}\Vert v^{'}\Vert_{2}$, and $\Vert v\Vert^{2}_{2}=\Vert v^{'}\Vert^{2}_{2}-\Vert u^{'}\Vert^{2}_{2}$, we obtain $\Vert u^{'}\Vert \leq t\Vert v^{'}\Vert_{2}$, implying that $\Vert u^{'}\Vert_{2}\leq\sqrt{\frac{t^{2}}{1-t^{2}}}\Vert v\Vert_{2}$.

Finally, let $\Vert H^{Z_\pi}v\Vert_2=x\Vert v\Vert_2$, $x$ must satisfy:
\begin{align*}
x^{2}\geq 1-\frac{t^{2}}{\sqrt{1-t^{2}}}x\,.
\end{align*}
Since $x\geq 0$, by solving the above inequality we obtain: $x\geq \sqrt{1-t^{2}}$. This implies $\Vert H^{Z_\pi}(X-H^{Z}X)\Vert^{2}_{2}\geq (1-t^{2})\left[\Vert H^{ZZ_{\pi}}X\Vert^{2}_{2}-\Vert H^{Z}X\Vert^{2}_{2} \right] $.
\end{proof}
Next we provide Lemma~\ref{lem: upper of linearity of Z} to upper bound $\sup_{u\in span(Z),v\in span(Z_{\pi})}\frac{u^{T}v}{\Vert u\Vert_{2}\Vert v\Vert_{2}}$.
\begin{lemma}\label{lem: upper of linearity of Z}
Suppose $Z$ satisfies the condition in Lemma~\ref{lem: upper and lower of optimization objective formal}.
Then there exist some constants $C>0$, such that for any permutation $\pi:[1,n]\to[1,n]$,
\begin{align*}
\mathbb{P}\left[\sup_{u\in span(Z), v\in span(Z_{\pi}),u,v\neq 0}\frac{u^{T}v}{\Vert u\Vert_{2}\Vert v\Vert_{2}}\geq \sqrt{\frac{C(p+tr(P_{\pi}))}{n}}   \right]\leq 40e^{-p}\,.
\end{align*}
\end{lemma}
\begin{proof}
First, we consider $u,v\in\mathcal{S}^{p-2}$ and upper bound $\vert (Zu)^{T}Z_{\pi}v\vert$. Let row $i$ of $Z$ be $Z_{i}\in\mathbb{R}^{1\times (p-1)}$, so that row $\pi(i)$ of $Z_{\pi}$ is equal to $Z_i$. For simplicity, denote $x_{i}=Z_{i}u$, $y_{i}=Z_{i}v$. Then $Zu=[x_1,x_2,...,x_n]^{T}$ and $Z_{\pi}v=[y_{\pi^{-1}(1)},y_{\pi^{-1}(2)},...,y_{\pi^{-1}(n)}]^{T}$, with both $x_1,...,x_n$ and  $y_1,...,y_n$ forming two sets of i.i.d. random variables. Now we consider

\begin{align*}
(Zu)^{T}(Z_{\pi}v) = \sum_{i=1}^{n}x_{\pi(i)}y_{i}\,.
\end{align*}
When $i = \pi(i)$, we apply the bound $\vert x_{\pi(i)}y_{i}\vert \leq \vert x_{\pi(i)}\vert\cdot \vert y_{i}\vert$. For the remaining indices, let $I=\{i\vert \pi(i)\neq i  \}$, and define $\pi_{I}:I\to I, \pi_{I}(i)=\pi(i),\forall i\in I$. Then $\pi_{I}$ can be decomposed into cycles, each containing at least 2 elements. Since every such cycle admits a proper 3‑coloring such that adjacent vertices receive distinct colors, the set $I$ can be partitioned into three subsets $I_1, I_2, I_3$ with the property that for each $j = 1,2,3$ and every $i \in I_j$, we have $\pi_I(i) \notin I_j$. Consequently, for each $j$, the collection $\{x_{\pi(i)},y_{i}\vert i\in I_{j} \}$ consists of mutually independent elements.

Now we construct a maximal $\epsilon-$net $\mathcal{T}_{\epsilon}$ of the unit sphere $\mathcal{S}^{p-2}$, meaning that any two distinct points in $\mathcal{T}\epsilon$ are at distance at least $\epsilon$, and for every $u \in \mathcal{S}^{p-2}$, there exists some $v \in \mathcal{T}_\epsilon$ such that $\Vert u-v\Vert_{2}\leq \epsilon$. Since the balls of radius $\frac{1}{2}\epsilon$ centered at the points of $\mathcal{T}_\epsilon$ are pairwise disjoint and all lie within a ball of radius $1 + \frac{1}{2}\epsilon$, a standard volume argument yields $\vert \mathcal{T}_{\epsilon}\vert \leq (1+\frac{2}{\epsilon})^{p-1}$.

We first show that $\sup_{u,v\in\mathcal{T}_{\epsilon}}\frac{\vert(Zu)^{T}(Z_{\pi}v)\vert}{\Vert Zu\Vert_2\Vert Z_{\pi}v\Vert_{2}}$ is a good approximation to $\sup_{u,v\in\mathcal{S}^{p-2}}\frac{\vert(Zu)^{T}(Z_{\pi}v)\vert }{\Vert Zu\Vert_2\Vert Z_{\pi}v\Vert_{2}}$ (which is exactly $\sup_{u^{'}\in span(Z),v^{'}\in span(Z_{pi}),u^{'},v^{'}\neq 0}\frac{(u^{'})^{T}v^{'}}{\Vert u^{'}\Vert_{2}\Vert v^{'}\Vert_{2}}$). Let $\sup_{u,v\in\mathcal{S}^{p-2}}\vert (Zu)^{T}(Z_{\pi}v)\vert:=A$ and \\
$\sup_{u,v\in \mathcal{T}_{\epsilon}}\vert (Zu)^{T}(Z_{\pi}v)\vert:=B$, then $A,B$ has the following relationship:
\begin{align*}
\vert (Zu)^{T}(Z_\pi v)\vert &\leq \vert (Zu_i)^{T}(Z_\pi v)\vert +\Vert u_i-u\Vert_{2}\sup_{u^{'}\in\mathcal{S}^{p-1}}\vert (Zu^{'})^{T}(Zv)\vert\\&
\leq \vert (Zu_i)^{T}(Z_{\pi}v_j)\vert +\Vert v_j-v\Vert_{2}\sup_{v^{'}\in\mathcal{S}^{p-1}}\vert (Zu_i)^{T}(Z_{\pi}v^{'})\vert+\Vert u_i-u\Vert_{2}\sup_{u^{'}\in\mathcal{S}^{p-1}}\vert (Zu^{'})^{T}(Zv)\vert\\
&\leq \vert (Zu_i)^{T}(Z_{\pi}v_j)\vert+2\epsilon A \\
&\leq B+2\epsilon A\,,
\end{align*}
indicating that $A\leq \frac{1}{1-2\epsilon}B$.
Now if we let $D:=\sup_{u\in\mathcal{S}^{p-2}\Vert Zu\Vert_{2}}$, $E:=\inf_{u\in\mathcal{S}^{p-2}\Vert Zu\Vert_{2}}$, we can derive that
\begin{align*}
\vert(Zu)^{T}(Z_{\pi}v)\vert\leq \vert (Zu_i)^{T}Z_{\pi}v_i\vert+2\epsilon A\,,\quad \Vert Zu\Vert_{2}\geq  \Vert Zu_{i}\Vert_{2}-\epsilon D\,,\quad \Vert Z_{\pi}v\Vert_{2}\geq \Vert Z_{\pi}v_j\Vert_{2}-\epsilon D\,.
\end{align*}
This means that for all $u,v\in \mathcal{S}^{p-2}$, we have:
\begin{align}\label{upper bound of ep-net}
\frac{\vert(Zu)^{T}(Z_{\pi}v)\vert }{\Vert Zu\Vert_2\Vert Z_{\pi}v\Vert_{2}}\leq \sup_{u_i,v_j\in\mathcal{T}_{\epsilon}}\frac{\vert(Zu_i)^{T}(Z_{\pi}v_j)\vert }{\Vert Zu_i\Vert_2\Vert Z_{\pi}v_j\Vert_{2}}\left(1+\frac{\epsilon D}{E}\right)^{2}+\frac{2\epsilon B}{(1-2\epsilon)E^{2}}\,.
\end{align}

\paragraph{Upper bound $\vert (Zu)^{T}(Z_{\pi}v)\vert$}
We first upper bound $\sum_{i\in I_{j}}x_{\pi(i)}y_{i}$ for $j=1,2,3$ and all $u,v\in\mathcal{T}_{\epsilon}$. 
For any $\lambda$ and $\pi(i)\neq i$ we have:
\begin{align*}
\mathbb{E}[e^{\lambda x_{\pi(i)}y_{i}}\vert x_{\pi(i)}]\leq \mathbb{E}[e^{c_2\lambda^{2} x_{\pi(i)}^{2}}]\,.
\end{align*}
Therefore, for $j=1,2,3$ and any $t\geq 0$, we have:
\begin{align*}
\mathbb{P}\left[\vert \sum_{i\in I_j}x_{\pi(i)}y_{i}\vert\geq t\right]\leq 2 e^{c_2\lambda^{2}\sum_{i\in I_j}x^{2}_{\pi(i)}-\lambda t}\,, \forall \lambda\geq 0\,.
\end{align*}
	
Let $t = \sqrt{4c_2(p+2\ln(\vert \mathcal{T}_{\epsilon}\vert))\sum_{i\in I_j}x_{\pi(i)}^{2}}$ and $\lambda=\frac{t}{2c_2\sum_{i\in I_{j}}x^{2}_{\pi(i)}}$, we obtain:
\begin{align*}
\mathbb{P}\left[\vert \sum_{i\in I_j}x_{\pi(i)}y_{i}\vert\geq \sqrt{4c_2(p+2\ln(\vert \mathcal{T}_{\epsilon}\vert))\sum_{i\in I_j}x_{\pi(i)}^{2}}\right]\leq 2\vert \mathcal{T}_{\epsilon}\vert^{-2}e^{-p}\,.
\end{align*}
Due to the symmetry we also have:
\begin{align*}
\mathbb{P}\left[\vert \sum_{i\in I_j}x_{\pi(i)}y_{i}\vert\geq \sqrt{4c_2(p+2\ln(\vert \mathcal{T}_{\epsilon}\vert))\sum_{i\in I_j}y_{i}^{2}}\right]\leq 2\vert \mathcal{T}_{\epsilon}\vert^{-2}e^{-p}\,.
\end{align*}
Therefore we can obtain:
\begin{align*}
\mathbb{P}\left[\vert \sum_{i\in I_j}x_{\pi(i)}y_{i}\vert\geq \sqrt{4c_2(p+2\ln(\vert \mathcal{T}_{\epsilon}\vert))\min\left(\sum_{i\in I_j}x_{\pi(i)}^{2},\sum_{i\in I_j}y_{i}^{2}\right)}\right]\leq 4\vert \mathcal{T}_{\epsilon}\vert^{-2}e^{-p}\,.
\end{align*}

By summing up over $j=1,2,3$ we obtain:
\begin{align*}
\mathbb{P}\left[\vert \sum_{i\in I_j}x_{\pi(i)}y_{i}\vert\geq \sqrt{12c_2(p+2\ln(\vert \mathcal{T}_{\epsilon}\vert))\min\left(\sum_{i\in I}x_{\pi(i)}^{2},\sum_{i\in I}y_{i}^{2}\right)}\right]\leq 12\vert \mathcal{T}_{\epsilon}\vert^{-2}e^{-p}\,,
\end{align*}
where in this step we use Cauchy's inequality that $\sqrt{a}+\sqrt{b}+\sqrt{c}\leq \sqrt{3(a+b+c)}$. This implies a union bound for all $u,v\in\mathcal{T}_{\epsilon}$:
\begin{align*}
\mathbb{P}\left[\vert (Zu)^{T}Z_{\pi}v\vert\leq \sqrt{12c_2(p+2\ln(\vert \mathcal{T}_{\epsilon}\vert))\min\left(\sum_{i\in I}x_{\pi(i)}^{2},\sum_{i\in I}y_{i}^{2}\right)}+\vert \sum_{i\notin I} x_{i}y_{i}\vert \,, \forall u,v\in\mathcal{T}_{\epsilon}\right]\leq 12e^{-p}\,,
\end{align*}
where $x_i=(Zu)_i$ and $y_i=(Zv)_i$.

For $\vert\sum_{i\notin I}x_iy_i\vert$, we simply upper bound it by 
\begin{align*}
\vert\sum_{i\notin I}x_iy_i\vert\leq \sqrt{\sum_{i\notin I}x^{2}_{i}}\sqrt{\sum_{i\notin I}y^{2}_{i}}\,.
\end{align*}
This implies a probability bound of $\frac{\vert (Zu)^{T}(Z_{\pi}v)\vert}{\Vert Zu\Vert_2\Vert Z_{\pi}v\Vert_2}$: $\forall u,v\in\mathcal{T}_{\epsilon}$,
\begin{align*}
\mathbb{P}\left[\frac{\vert (Zu)^{T}(Z_{\pi}v)\vert}{\Vert Zu\Vert_2\Vert Z_{\pi}v\Vert_2}\geq \sqrt{\frac{12c_2(p+2\ln(\vert \mathcal{T}_{\epsilon}\vert))}{\max(\Vert Zu\Vert^{2}_{2},\Vert Z_{\pi}v\Vert^{2}_2)}}+\sqrt{\frac{\sum_{i\notin I} x^{2}_{i}}{\sum_{i=1}^{n}x^{2}_{i}}}\sqrt{\frac{\sum_{i\notin I} y^{2}_{i}}{\sum_{i=1}^{n}y^{2}_{i}}}  \right]\leq 12\vert \mathcal{T}_{\epsilon}\vert^{-2}e^{-p}
\end{align*}

Now it suffices to bound $\Vert Zu\Vert^{2}_{2}$. We first consider all $u\in\mathcal{T}_{\epsilon}$.
Since $\mathbb{E}[e^{tx_i}]\leq e^{c_2t^{2}}$ for all t, by setting $t = \frac{x}{2c_2}$ and $t = -\frac{x}{2c_2}$ we obtain:
\begin{align*}
\mathbb{P}[\vert x_i\vert \geq x]\leq 2 e^{-\frac{x^{2}}{4c_2}}\,.
\end{align*}
Therefore, when $\lambda<\frac{1}{4c_2}$, we can upper bound $\mathbb{E}[e^{\lambda x^{2}_i}]$ by
\begin{align*}
\mathbb{E}[e^{\lambda x^{2}_i}]&= \int_{t\geq 0}\mathbb{P}[e^{\lambda x^{2}_i}\geq t]dt\\
& =1+\int_{t\geq 1}\mathbb{P}[x^{2}_{i}\geq \frac{1}{\lambda}\ln t]dt\\
& \leq 1+\int_{t\geq 1}2e^{-\frac{1}{4c_2\lambda}\ln t}dt\\
& \leq 1+\int_{t\geq 1}2t^{-\frac{1}{4c_{2}\lambda}}dt\,.\\
& = 1+\frac{2}{\frac{1}{4c_{2}\lambda}-1}\\
& = 1+\frac{8c_{2}\lambda}{1-4c_{2}\lambda}\,.
\end{align*}
Let $\lambda = \frac{1}{10c_2}$, then $\mathbb{E}[e^{\lambda x^{2}_{i}}]\leq e$. Therefore, for any $t\geq 0$, we have:
\begin{align*}
\mathbb{P}\left[\sum_{i\notin I}x^{2}_{i}\geq t\right]\leq e^{\frac{1}{10c_2} (tr(P_{\pi})-t)}\,.
\end{align*}
This leads to an upper bound of $\sum_{i\notin I}x^{2}_{i}$:
\begin{align*}
\mathbb{P}\left[\sum_{i\notin I}x^{2}_{i}\leq 10c_2(tr(P_{\pi})+\ln (\vert\mathcal{T}_{\epsilon}\vert)+p)\right]\leq \vert \mathcal{T}_{\epsilon}\vert^{-1}e^{-p}\,.
\end{align*}
Now we lower bound $\sum_{i=1}^{n}x^{2}_{i}$. Notice that $e^{-x}\leq 1-x+\frac{1}{2}x^{2}(x\geq 0)$, for any $\lambda\geq 0$ we obtain that $\mathbb{E}[e^{-\lambda x^{2}_{i}}]\leq 1-\lambda\mathbb{E}[x_i^{2}]+\frac{1}{2}\lambda^{2}\mathbb{E}[x_i^{4}]\leq 1-c_1\lambda+\frac{1}{2}\lambda^{2}\mathbb{E}[x^{4}_{i}]$.
On the other hand,
\begin{align*}
\mathbb{E}[x^{4}_i]&=\int_{t\geq 0}\mathbb{P}[x^{4}_i\geq t]dt\\
&\leq \int_{t\geq 0}2e^{-\frac{1}{4c_2}\sqrt{{t}}}dt\\
&=\int_{x\geq 0} 4x e^{-\frac{1}{4c_2}x}dx\\
&=64c^{2}_{2}\,.
\end{align*} 
Thus, for $\lambda\leq \frac{c_1}{96c^{2}_2}$ we obtain: $\mathbb{E}[e^{-\lambda x^{2}_{i}}]\leq 1-\frac{2}{3}c_1\lambda\leq e^{-\frac{2}{3}c_1\lambda}$. Therefore, we conclude that for some constant $c_3$ depending on $c_1$ and $c_2$ ($c_3 = \frac{c_1}{96 c_2^2}$ is a valid choice and may be taken larger), we have:
\begin{align*}
\mathbb{P}\left[\sum_{i=1}^{n}x^{2}_{i}\leq \frac{2}{3}c_1 n-\frac{1}{c_3}(p+\ln(\vert \mathcal{T}_\epsilon\vert))\right]\leq \vert \mathcal{T}_{\epsilon}\vert^{-1}e^{-p}\,.
\end{align*}
This indicates that
\begin{align*}
\mathbb{P}\left[\inf_{u\in\mathcal{T}_{\epsilon}}\Vert Zu\Vert^{2}_{2}\leq \frac{2}{3}c_1 n-\frac{1}{c_3}(p+\ln(\vert \mathcal{T}_\epsilon\vert)) \right]\leq e^{-p}\,.
\end{align*}

Combining all the above results and applying a union bound over the probability estimates, we obtain the following lower bound: $\frac{\vert (Zu)^{T}(Z_{\pi}v)\vert}{\Vert Zu\Vert_{2}\Vert Z_{\pi}v\Vert_{2}}(u,v\in\mathcal{T}_{\epsilon})$:
\begin{align}\label{upper bound inner product in epsilon net}
\mathbb{P}\left[\sup_{u,v\in\mathcal{T}_{\epsilon}}\frac{\vert (Zu)^{T}(Z_{\pi}v)\vert}{\Vert Zu\Vert_{2}\Vert Z_{\pi}v\Vert_{2}}\geq \sqrt{\frac{12c_2(p+2\ln(\vert \mathcal{T}_{\epsilon}\vert))}{\frac{2}{3}c_1 n-\frac{1}{c_3}(p+\ln(\vert \mathcal{T}_\epsilon\vert))}}+\frac{10c_2(tr(P_{\pi})+\ln (\vert\mathcal{T}_{\epsilon}\vert)+p)}{\frac{2}{3}c_1 n-\frac{1}{c_3}(p+\ln(\vert \mathcal{T}_\epsilon\vert))}\right]\leq 14e^{-p}\,.
\end{align}

Since $\ln\vert \mathcal{T}_{\epsilon}\vert\leq p\ln(1+\frac{2}{\epsilon})$, we can conclude that for some $C_1>0$ ($C_1$ depends on $c_1,c_2$), we have:

\begin{align}
\mathbb{P}\left[\sup_{u,v\in\mathcal{T}_{\epsilon}}\frac{\vert (Zu)^{T}(Z_{\pi}v)\vert}{\Vert Zu\Vert_{2}\Vert Z_{\pi}v\Vert_{2}}\leq \sqrt{\frac{C_1(tr(P_{\pi})+\ln(1+\frac{2}{\epsilon})p)}{n}}\right]\geq 1-14e^{-p}\,.
\end{align}

Finally, we upper bound $B:=\sup_{u,v\in\mathcal{T}_{\epsilon}}\vert (Zu)^{T}(Z_{\pi}v)\vert$, $D:=\sup_{u\in\mathcal{S}^{p-2}}\Vert Zu\Vert_{2}$ and lower bound  $E:=\inf_{u\in\mathcal{S}^{p-2}}\Vert Zu\Vert_{2}$ to complete the proof. 

For $B$, by an argument analogous to the upper bound for $\sum_{i \notin I} x_i^2$, we obtain
\begin{align*}
\mathbb{P}\left[\sum_{i\in I}x^{2}_{i}\geq 10c_2(n-tr(P_{\pi})+p+\ln(\vert \mathcal{T}_{\epsilon}\vert))\right]\leq \vert \mathcal{T}_{\epsilon}\vert^{-1}e^{-p}\,,
\end{align*}
Combining this with the upper bounds for $\sum_{i \notin I} x_i^2$ and $\sum_{i \notin I} y_i^2$ yields:
\begin{align*}
&\mathbb{P}\left[\sup_{u,v\in\mathcal{T}_{\epsilon}}\vert (Zu)^{T}Z_\pi v\vert\geq \sqrt{12c_2(p+2\ln(\vert \mathcal{T}_{\epsilon}\vert))[10c_2(n+p+\ln(\vert \mathcal{T}_{\epsilon}\vert))]}+10c_2(tr(P_{\pi})+\ln (\vert\mathcal{T}_{\epsilon}\vert)+p)\right]\\&\leq 16 e^{-p}\,.
\end{align*}
Here we use the upper bound that $\sum_{i\notin I}\vert x_iy_i\vert\leq \sqrt{\sum_{i\notin I}x^{2}_i\sum_{i\notin I}y^{2}_i }$. Therefore,
\begin{align*}
\mathbb{P}\left[B\geq \sqrt{12c_2(p+2\ln(\vert \mathcal{T}_{\epsilon}\vert))[10c_2(n+p+\ln(\vert \mathcal{T}_{\epsilon}\vert))]}+10c_2(tr(P_{\pi})+\ln (\vert\mathcal{T}_{\epsilon}\vert)+p)\right]\leq 16e^{-p}\,.
\end{align*}

For $D$, since for any $u$ we have:
\begin{align*}
\mathbb{P}\left[\sum_{i=1}^{n}x^{2}_{i}\geq 10c_2(n+\ln (\vert\mathcal{T}_{\epsilon}\vert)+p)\right]\leq \vert \mathcal{T}_{\epsilon}\vert^{-1}e^{-p}\,,
\end{align*}
We can then obtain the following bound on $\mathcal{T}_\epsilon$: 
$\mathbb{P}\left[\sup_{u\in\mathcal{T}_{\epsilon}}\Vert Zu\Vert^{2}_{2}\geq 10c_2(n+\ln (\vert\mathcal{T}_{\epsilon}\vert)+p)\right]\leq e^{-p}$.
Moreover, note that for any $u \in \mathcal{S}^{p-2}$, $\Vert Zu\Vert_2\leq \sup_{u^{'}\in\mathcal{T}_{\epsilon}}\Vert Zu^{'}\Vert_{2}+\epsilon \sup_{u^{''}\in \mathcal{S}^{p-2}}\Vert Zu^{''}\Vert_2 $, which implies
$D\leq \frac{1}{1-\epsilon}\sup_{u\in\mathcal{T}_{\epsilon}}\Vert Zu\Vert^{2}_{2}$.
Consequently, we obtain:
\begin{align*}
\mathbb{P}\left[D\leq \frac{1}{1-\epsilon}\sqrt{10c_2(n+\ln (\vert\mathcal{T}_{\epsilon}\vert)+p)}\right]\geq 1-e^{-p}\,.
\end{align*}
On the other hand, we also have: $E\geq \inf_{u\in\mathcal{T}_{\epsilon}}\Vert Zu\Vert_2-\epsilon D$. This implies a lower bound of $E$:
\begin{align*}
\mathbb{P}\left[E\geq \sqrt{\frac{2}{3}c_1 n-\frac{1}{c_3}(p+\ln (\vert\mathcal{T}_\epsilon\vert))}-\frac{\epsilon}{1-\epsilon}\sqrt{10c_2(n+\ln (\vert\mathcal{T}_{\epsilon}\vert)+p)}\right]\geq 1-2e^{-p}\,.
\end{align*}
Combining all these probabilistic bounds on $B$, $D$, and $E$ with \eqref{upper bound of ep-net} and \eqref{upper bound inner product in epsilon net}, there exist constants $C', c' > 0$ and $\epsilon_0 > 0$ such that, for $\epsilon = \epsilon_0$ and any $n \geq C' p$, we have:
\begin{align}
\mathbb{P}\left[\sup_{u,v\in\mathcal{S}^{p-2}}\frac{\vert (Zu)^{T}Zv\vert}{\Vert Zu\Vert_2 \Vert Zv\Vert_2}\geq c^{'}\sqrt{\frac{p+tr(P_{\pi})}{n}}\right]\leq 40e^{-p}\,.
\end{align}

Since $\sup_{u,v\in\mathcal{S}^{p-2}}\frac{\vert (Zu)^{T}Zv\vert}{\Vert Zu\Vert_2 \Vert Zv\Vert_2}\leq 1$, we can find a constant $C>0$ satisfying:
\begin{align*}
\forall n>0\,,\,\mathbb{P}\left[\sup_{u,v\in\mathcal{S}^{p-2}}\frac{\vert (Zu)^{T}Zv\vert}{\Vert Zu\Vert_2 \Vert Zv\Vert_2}\geq C\sqrt{\frac{p+tr(P_{\pi})}{n}}\right]\leq 40e^{-p}\,.
\end{align*}
\end{proof}

\subsubsection{Proof of Lemma~\ref{lem: decomposition}}
\paragraph{Lemma~\ref{lem: decomposition}}
For all permutation $\pi$ and $X$ we have:
\begin{align*}
     \left\vert X^{T}_{\pi}(I-H^{ZZ_{\pi}})X-\frac{1}{2} X^{T}_{\pi} (I-H^{Z_{\pi}})(I-H^{Z})X\right\vert \leq \frac{1}{2}\Vert (I-H^{Z})X\Vert^{2}_{2}\,.
\end{align*}
\begin{proof}
Suppose $\{u_{1},...,u_{m}\}$ forms an orthogonal basis for $I-H^{ZZ_{\pi}}$, and let 
$\{u_{1},...,u_{m},v_1,...,v_{m_1}\}$ be a basis for $I-H^{Z}$ such that $\{v_1,...,v_{m_1}\}$ is orthogonal and each $u_i$ is orthogonal to every $v_j$. Similarly, we define $\{u_{1},...,u_{m},w_{1},...,w_{m_2}\}$ as a basis for $I-H^{Z_{\pi}}$ with $\{w_{1},...,w_{m_2}\}$ orthogonal and $u_{i}^{T}w_{j}=0$ for all $i, j$.

Now, let $(I-H^{ZZ_{\pi}})X=X^{u}$, $(I-H^{ZZ\pi})X_{\pi}=X_{\pi}^{u}$. Furthermore, write $(I-H^{Z})X=X^{u}+v$, $(I-H^{Z_{\pi}})X_{\pi}=X_{\pi}^{u}+w$, where $v \in \operatorname{span}(v_1,...,v_{m_1})$ and $w \in \operatorname{span}(w_1,...,w_{m_2})$. Then we have: 

$$X^{T}_{\pi}(I-H^{ZZ_{\pi}})X=(X_{\pi}^{u})^{T}X^{u}$$,
$$X_{\pi}^{T}(I-H^{Z_{\pi}})(I-H^{Z})X=(X^{u}_{\pi}+w)^{T}(X^{u}+v)=(X^{u}_{\pi})^{T}X^{u}+w^{T}v$$

By Cauchy's inequality, we have:

\begin{align}
(\vert (X^{u}_{\pi})^{T}X^{u}\vert+\vert w^{T}v\vert)^{2}\leq \left(\Vert X^{u}_{\pi}\Vert^{2}_{2}+\Vert w\Vert^{2}_{2}\right)\left(\Vert X^{u}\Vert^{2}_{2}+\Vert v\Vert^{2}_{2}\right)\leq \Vert (I-H^{Z})X\Vert^{4}_{2}\,.
\end{align}

Therefore, we obtain that $\vert (X^{u}_{\pi})^{T}X^{u}\vert+\vert w^{T}v\vert\leq\Vert  (I-H^{Z})X\Vert^{2}_{2}$. This implies
$$
\left\vert \frac{1}{2}X^{T}_{\pi}(I-H^{Z_{\pi}})(I-H^{Z})-X^{T}_{\pi}(I-H^{ZZ_{\pi}})X\right\vert\leq \frac{1}{2}\Vert (I-H^{Z})X\Vert^{2}_{2}\,.
$$
\end{proof}
\subsubsection{Formal version of Theorem~\ref{thm: informal for high prob bound}}
We provide Theorem~\ref{thm: value of optimization objective}, the formal version  of Theorem~\ref{thm: informal for high prob bound} and complete its proof here.
\begin{theorem}\label{thm: value of optimization objective}
For any $\alpha\in(0,\frac{1}{2}]$ and $\mathcal{P}_{K}$, $\lambda_2(X,Z,\mathcal{P}_{K},\alpha)$ satisfies
$$
\lambda_2(X,Z,\mathcal{P}_{K},\alpha)\geq \mathbb{E}\left[ \frac{1}{2}v^{T}_{\pi_k}v+\Vert H^{Z}v_{\pi_k}\Vert^{2}_{2}\right]-2\alpha\Vert v\Vert^{2}_{2}\,.
$$
Furthermore, suppose that $\pi_i(i=1,2,...,m)$ are sampled $i.i.d.$ and uniformly from $\mathcal{P}_{K}$, with $\lambda^{'}$ satisfying:
\begin{align*}
\lambda^{'}=\inf \lambda:\frac{1}{m}\sum_{i=1}^{m}\mathbbm{1}\left\{ \frac{1}{2}v^{T}_{\pi_i}v+\Vert H^{Z_{\pi_i}}v\Vert^{2}_{2}>\lambda    \right\}\leq \frac{1}{4}\alpha\,.
\end{align*}
Then when $m\geq \frac{1}{\alpha^{2}}$, we have:
\begin{align*}
\mathbb{P}\left[\lambda^{'}\geq \mathbb{E}_{\pi_k}\left[ \frac{1}{2}v^{T}_{\pi_k}v+\Vert H^{Z}v_{\pi_k}\Vert^{2}_{2}\right]-\alpha\Vert v\Vert^{2}_{2}    \right]\geq 1-e^{-0.05\alpha^{-1}}
\end{align*}
\end{theorem}
\begin{proof}

Firstly, we observe that $\left\vert\frac{1}{2}v^{T}_{\pi_k}v+\Vert H^{Z_{\pi_k}}v\Vert^{2}_{2}\right\vert \leq \frac{1}{2}\Vert v\Vert^{2}_{2}+\Vert v\Vert^{2}_{2}=\frac{3}{2}\Vert v\Vert^{2}_{2}$, we also have: \\$ \frac{1}{2}v^{T}_{\pi_k}v+\Vert H^{Z_{\pi_k}}v\Vert^{2}_{2}\geq -\frac{1}{2}\Vert v\Vert^{2}_{2}$.
Therefore, we obtain:
\begin{align*}
\mathbb{E}\left[\frac{1}{2}v^{T}_{\pi_k}v+\Vert H^{Z_{\pi_k}}v\Vert^{2}_{2}\right]&\leq \lambda_2(X,Z,\mathcal{P}_{K},\alpha)\cdot(1-\alpha)+\alpha\cdot (2\Vert v\Vert^{2}_{2}+\lambda_2(X,Z,\mathcal{P}_{K},\alpha))\\
&\leq \lambda_2(X,Z,\mathcal{P}_{K},\alpha)+2\alpha\Vert v\Vert^{2}_{2}\,.
\end{align*}

When $\pi_i$ is sampled uniformly randomly from $\mathcal{P}_{K}$, we have:
\begin{align*}
\mathbb{P}\left[\frac{1}{2}v_{\pi_i}^{T}v+\Vert H^{Z_{\pi_i}}v\Vert^{2}_{2}<\mathbb{E}_{\pi_k}\left[ \frac{1}{2}v^{T}_{\pi_k}v+\Vert H^{Z}v_{\pi_k}\Vert^{2}_{2}\right]-\alpha\Vert v\Vert^{2}_{2}\right]\leq 1-\frac{1}{2}\alpha\,,
\end{align*}
which is because $\lambda_2(X,Z,\mathcal{P}_{K},\frac{1}{2}\alpha)\geq\mathbb{E}_{\pi_k}\left[ \frac{1}{2}v^{T}_{\pi_k}v+\Vert H^{Z}v_{\pi_k}\Vert^{2}_{2}\right]-\alpha\Vert v\Vert^{2}_{2}$.

Now we let $a_i = \mathbbm{1}\left\{ \frac{1}{2}v^{T}_{\pi_i}v+\Vert H^{Z_{\pi_i}}v\Vert^{2}_{2} \geq \mathbb{E}\left[\frac{1}{2}v^{T}_{\pi_i}v+\Vert H^{Z}v_{\pi_i}\Vert^{2}_{2}\right]-\alpha\Vert v\Vert^{2}_{2}    \right\}$, we have $\mathbb{E}[a_i]\geq \frac{1}{2}\alpha$.

Therefore, we can upper bound $\lambda^{'}$ by:
\begin{align*}
\mathbb{P}\left[ \lambda^{'}\geq \mathbb{E}\left[\frac{1}{2}v^{T}_{\pi_k}v+\Vert H^{Z}v_{\pi_k}\Vert^{2}_{2}\right]-\alpha\Vert v\Vert^{2}_{2}    \right]&\geq 1-\mathbb{P}\left[\sum_{i=1}^{m}a_i\leq \frac{1}{4}\alpha\cdot m\right]\\
&\geq 1-e^{\frac{1}{4}\alpha \cdot m}\mathbb{E}[e^{-\sum_{i=1}^{m}a_i}]\\
&\geq 1-e^{\frac{1}{4}\alpha^{-1}}(1-\frac{1}{2}\alpha+\frac{1}{2}\alpha \cdot e^{-1})^{-m}\\
&\geq 1-e^{-\alpha^{-1}(\frac{1}{2}(1-e^{-1})-\frac{1}{4})}\\
&\geq 1-e^{-0.05\alpha^{-1}}
\end{align*}
\end{proof}
\subsubsection{Proof of Theorem~\ref{thm: high prob bound of optimization objective}}
Define $X_{i}=\sum_{j=1}^{i}v_{j}$. $Y_{i}=\Vert X_{i}\Vert^{2}_{2}-\sum_{j=1}^{i}\Vert v_{i}\Vert^{2}_{2}$. Then we have the recurrence $Y_{i+1}=Y_{i}+2X_{i}^{T}v_{i+1}$.
We first present a lemma that provides an upper bound on $X_{i}$.

\begin{lemma}\label{lem: bound inner product}
Suppose that conditions (1),(2), and (3) are satisfied. Then, for any $w\in\mathcal{S}^{n-1}$, we have:
\begin{align*}
\mathbb{P}\left[\sup_{i}\vert X^{T}_{i}w\vert \leq \max\left(\sqrt{4S_w\ln\frac{2}{\delta}},2a\ln\frac{2}{\delta}\right) \right]\geq 1-\delta
\end{align*}
\end{lemma}
\paragraph{Proof of Lemma~\ref{lem: bound inner product}.}
Let $x_{i}=v_{i}^{T}w$, with $\vert x_{i}\vert \leq a$ and $\mathbb{E}(x_{i})=0$. Consider $\mathbb{E}(e^{\lambda x_{i}})$. Note that for all $|\lambda| \leq \frac{1}{a}$, we have $\mathbb{E}(e^{\lambda x_{i}})\leq 1+\lambda^{2}\mathbb{E}[(x^{T}_{i}w)^{2}]$.
This follows from the inequality $e^x \leq 1 + x + x^2$ for $x \in [-1,1]$.

Therefore, for all $\lambda\leq \frac{1}{a}$, we have: $\mathbb{E}(e^{\lambda}(\sum_{i=1}^{m}x_i))\leq e^{\lambda^{2}S_w}$.

On the other hand, note that for any $\lambda$, the process $e^{\lambda X_i^T w}$ is a submartingale. Hence, for any $t > 0$, we have
\begin{align*}
\mathbb{P}\left[\sup_{i}e^{\lambda X^{T}_{i}w}\geq e^{\vert \lambda\vert t}\right]\leq e^{-\vert \lambda\vert t}\mathbb{E}e^{\lambda X^{T}_{n}w}\leq e^{-\vert\lambda\vert t+\lambda^{2}S_{w}}\,.
\end{align*}

If $\frac{t}{2S_w} \leq \frac{1}{a}$, i.e., $t \leq \frac{2S_w}{a}$, then by choosing $\lambda = \frac{t}{2S_w}$ or $\lambda = -\frac{t}{2S_w}$, we obtain:
\begin{align*}
\mathbb{P}\left[\sup_{i}\vert X^{T}_{i}w\vert\geq t   \right]\leq 2e^{-\frac{t^{2}}{4S_w}}\,.
\end{align*}

If $t \geq \frac{2S_w}{a}$, then by choosing $\lambda = \frac{1}{a}$ or $\lambda = -\frac{1}{a}$, we obtain
\begin{align*}
\mathbb{P}\left[\sup_{i}\vert X^{T}_{i}w\vert\geq t   \right]\leq 2e^{-\frac{t}{a}+\frac{S_w}{a^{2}}}\leq 2e^{-\frac{t}{2a}}\,.
\end{align*}
Thus, we conclude that:
\begin{align*}
\mathbb{P}\left[\sup_{i}\vert X^{T}_{i}w\vert\geq \max\left(\sqrt{4S_w\ln\frac{2}{\delta}},2a\ln\frac{2}{\delta}\right)   \right]\leq \delta
\end{align*}

\begin{lemma}\label{lem:bound norm}
Suppose that condition (1),(2),(3) of Theorem~\ref{thm: high prob bound of optimization objective} are satisfied. Then for any $k$ and $\epsilon>0$ there exists a constant $C$ such that $$
\mathbb{P}\left[\Vert X_{i}\Vert_{2}\geq\frac{1}{2}\sqrt{S}\ln^{\frac{1}{2}+\epsilon} n   \right]\leq Cn^{-k},\forall i
$$
\end{lemma}
\begin{proof}
Let $\lambda = \frac{\ln^{0.5} n}{2a+\sqrt{S}}$. We consider $Z_{i}=e^{\lambda\Vert X_{i}\Vert_{2}}$.
We distinguish between two cases:

\textbf{Case 1}: $\Vert X_{i}\Vert_2\leq a+\sqrt{S}$. In this case, $\Vert X_{i+1}\Vert_{2}\leq 2a+\sqrt{S}$, implying that $Z_{i+1}\leq n$. Thus, we have $\mathbb{E}(Z_{i+1})\leq n+Z_{i}$.

\textbf{Case 2}: $\Vert X_{i}\Vert_2\geq a+\sqrt{S}$. In this case, we have
\begin{align*}
\Vert X_{i+1}\Vert^{2}_{2}&=\Vert X_{i}\Vert^{2}+2X_{i}^{T}v_{i+1}+\Vert v_{i+1}\Vert^{2}_{2}\\
&=\left[\Vert X_{i}\Vert_{2}+\frac{X^{T}_{i}v_{i+1}}{\Vert X_{i}\Vert_2}\right]^{2}+\Vert v_{i+1}\Vert^{2}_{2}-\frac{(X^{T}_{i}v_{i+1})^{2}}{\Vert X_{i}\Vert^{2}_2}\,.
\end{align*}
By applying the inequality that for $a > 0$ and $b \geq 0$, $\sqrt{a+b}\leq \sqrt{a}+\frac{b}{2\sqrt{a}}$, and noting that $\Vert X_{i}\Vert_{2}+\frac{X^{T}_{i}v_{i+1}}{\Vert X_{i}\Vert_2}\geq \sqrt{S}$, we obtain:
$$
\Vert X_{i+1}\Vert \leq \Vert X_{i}\Vert_{2}+\frac{X^{T}_{i}v_{i+1}}{\Vert X_{i}\Vert_2}+\frac{1}{2\sqrt{S}}\left( \Vert v_{i+1}\Vert^{2}_{2}-\frac{(X^{T}_{i}v_{i+1})^{2}}{\Vert X_{i}\Vert^{2}_2}    \right)\,.
$$

Since $\lambda=\frac{\ln^{0.5} n}{2a+\sqrt{S}}\leq \frac{1}{a}$, and $\vert\Vert X_{i+1}\Vert_2-\Vert X_i\Vert_2\vert\leq \Vert X_{i+1}-X_i\Vert_2=\Vert v_{i+1}\Vert_2\leq a$, we obtain:
\begin{align*}
\frac{Z_{i+1}}{Z_{i}}&=e^{\lambda(\Vert X_{i+1}\Vert_{2}-\Vert X_{i}\Vert_{2})}\\
&\leq 1+\lambda(\Vert X_{i+1}\Vert_{2}-\Vert X_{i}\Vert_{2})+\lambda^{2}(\Vert X_{i+1}\Vert_{2}-\Vert X_{i}\Vert_{2})^{2}\\
&\leq 1+\lambda(\Vert X_{i+1}\Vert_{2}-\Vert X_{i}\Vert_{2})+\lambda^{2}\Vert v_{i+1}\Vert^{2}_{2}\\
&\leq 1+\lambda\left[\frac{X^{T}_{i}v_{i+1}}{\Vert X_{i}\Vert_2}+\frac{1}{2\sqrt{S}}\left( \Vert v_{i+1}\Vert^{2}_{2}-\frac{(X^{T}_{i}v_{i+1})^{2}}{\Vert X_{i}\Vert^{2}_2}\right)   \right]+\lambda^{2}\Vert v_{i+1}\Vert^{2}_{2}\\
&\leq 1+\lambda\frac{X^{T}_{i}v_{i+1}}{\Vert X_{i}\Vert_2}+\frac{\lambda}{2\sqrt{S}}\Vert v_{i+1}\Vert^{2}_{2}+\lambda^{2}\Vert v_{i+1}\Vert^{2}_{2}\\
&\leq 1+\lambda\frac{X^{T}_{i}v_{i+1}}{\Vert X_{i}\Vert_2}+2\lambda^{2}\Vert v_{i+1}\Vert^{2}_{2}\,.
\end{align*}
This implies that
$
\mathbb{E}(Z_{i+1}\vert Z_{i})\leq Z_{i}(1+2\lambda^{2}t_{i+1})
$
since $\mathbb{E}\left[\frac{X^{T}_{i}v_{i+1}}{\Vert X_{i}\Vert_{2}}\right]=0 $.

Combining the two cases, we obtain
\begin{align}
\mathbb{E}(Z_{i+1}\vert Z_{i})\leq (1+2\lambda^{2} t_{i+1})Z_{i}+n\,.
\end{align}
This implies
$
\mathbb{E}[Z_{i}\prod_{j=1}^{i}(1+2\lambda^{2}t_{j})^{-1}]\leq Z_{0}+i\cdot n = 1+i\cdot n\,.
$

Thus, for any $i$, 
\begin{align*}
\mathbb{P}\left[\Vert X_{i}\Vert_{2}\geq\frac{1}{2}\sqrt{S}\ln^{\frac{1}{2}+\epsilon} n \right]
&=\mathbb{P}\left[Z_{i}\geq e^{\frac{\sqrt{S}}{4a+2\sqrt{S}}\ln^{1+\epsilon} n}\right]\\
&\leq e^{-\frac{\sqrt{S}}{4a+2\sqrt{S}}\ln^{1+\epsilon} n}\mathbb{E}\left[Z_{i}\right]\\
&\leq C_0 e^{-\frac{1}{3}\ln^{1+\epsilon}n}\cdot \prod_{j=1}^{n}(1+2\lambda^{2}t_j)(1+n^{2})\\
&\leq C_0 e^{-\frac{1}{3}\ln^{1+\epsilon}n}\cdot e^{2\lambda^{2}S}(1+n^{2})\\
&\leq 2C_{0}n^{4}e^{-\frac{1}{3}\ln^{1+\epsilon}n}\,.
\end{align*}
where $C_0$ is a constant independent of $i$ and $n$.
\end{proof}
We now derive an upper bound for $Y_i$, where $Y_{i}=Y_{i-1}+2X^{T}_{i-1}v_{i}$, $\mathbb{E}\left[X^{T}_{i-1}v_i\right]=0$.
Let $V_{i}=\mathbb{E}_{v_{i}}[(X^{T}_{i-1}v_i)^{2}]=X^{T}_{i-1}Q_{i}X_{i-1}$, where $Q_i = \mathbb{E}[v_i v_i^T]$ is symmetric and positive semidefinite, with $tr(Q_{i})=t_{i}$.
Now consider a random variable $W_i$, defined over $v_1, v_2, \dots, v_m$, as follows:
$$
W_{i} = e^{\lambda Y_{i}}\mathbbm{1}\left\{ \vert 2X^{T}_{j}v_{j+1}\vert \leq a\sqrt{S}\ln^{0.75} n\text{ and }V_{i+1}\leq \frac{St_{i+1}}{\ln^{2}n} ,\forall j=1,2,...,i-1        \right\}\,.
$$
For $\vert \lambda\vert\leq\frac{1}{a\sqrt{S}\ln^{0.75} n}$, we have:
\begin{align*}
\mathbb{E}[W_{i+1}\vert W_{i}]\leq W_{i}\left[ 1+\lambda^{2}\mathbb{E}\left[(2X^{T}_{i}v_{i+1})^{2}\right]\right]\leq W_{i}(1+4\lambda^{2}\frac{St_{i+1}}{\ln^{2}n})\,.
\end{align*}

We now derive an upper bound for $\mathbb{P}\left[V_{i}\geq \frac{St_i}{\ln^{2}n}\right]$.
Let the eigenvectors of $Q_i$ be ${w_1^{(i)}, \dots, w_n^{(i)}}$ with corresponding eigenvalues $\lambda_1, \lambda_2, \dots, \lambda_n$. Then we have:
$$
X^{T}_{i-1}Q_{i}X_{i-1}=\sum_{j=1}^{n}\lambda_{i}(X^{T}_{i-1}w^{i}_{j})^{2}\leq t_{i}\sup_{j}(X^{T}_{i-1}w^{i}_{j})^{2}\,.
$$
By applying Lemma~\ref{lem: bound inner product} with $\ln\frac{2}{\delta} = \frac{1}{4}\ln^2 n$, we obtain:
\begin{align*}
\mathbb{P}\left[V_{i}\geq \frac{St_i}{\ln^{2}n}\right]\leq ne^{-\frac{1}{4}\ln^{2}n}\,.
\end{align*}
Let $\lambda = \frac{\ln^{1.25} n}{S}$, by setting $\epsilon=0.25$ in Lemma~\ref{lem:bound norm} we can obtain that:
\begin{align*}
\mathbb{P}\left[Y_{n}\geq \frac{1}{2}cS\right]&\leq \frac{\mathbb{E}[W_{n}]}{e^{\frac{1}{2}cS\lambda}}+\mathbb{P}\left[\exists i,s.t. \vert  X^{T}_{i}v_{i+1}\vert\geq \frac{1}{2}a\sqrt{S}\ln^{0.75} n\right]+\mathbb{P}\left[\exists i,s.t. V_{i}\geq\frac{St_i}{\ln^{2}n}\right]\\
&\leq 2n^{4} e^{-\frac{1}{2}c\ln^{1.25}n}+\mathbb{P}\left[\sup_{i}\Vert X_{i}\Vert_2 \geq \frac{1}{2}\sqrt{S}\ln^{0.75} n\right]\\
&\leq  2n^{4} e^{-\frac{1}{2}c\ln^{1.25}n}+C_{0}n^{5}e^{-\frac{1}{3}\ln^{1.25}n}\,.
\end{align*}
Finally, we derive an upper bound for $\mathbb{P}\left[\sum_{i=1}^{n}\Vert v_{i}\Vert^{2}_{2}\geq \sum_{i=1}^{n}t_i+\frac{1}{2}cS\right]$. 
Consider $x_i=e^{\lambda\sum_{j=1}^{n}\Vert v_{i}\Vert^{2}}$ with $\lambda \leq  \frac{1}{a^{2}}$. Then we obtain:
\begin{align*}
\mathbb{E}[x_{i+1}\vert x_i]&\leq x_{i}(1+\lambda\mathbb{E}(\Vert v_{i+1}\Vert^{2}_2)+\lambda^{2}\mathbb{E}[\Vert v_{i+1}\Vert^{4}_{2}]\\
&\leq x_{i}(1+\lambda t_{i+1}+\lambda^{2}a^{2}t_{i+1})\\
&\leq x_i e^{\lambda t_{i+1}+\lambda^{2}a^{2}t_{i+1}}\,.
\end{align*}
Since $x_0=1$, we obtain $\mathbb{E}[x_{m}]\leq e^{\lambda \sum_{i=1}^{m}t_{i}(1+a^{2}\lambda)}$. Therefore, 
\begin{align*}
\mathbb{P}\left[\sum_{i=1}^{m}\Vert v_{i}\Vert^{2}_{2}\geq \sum_{i=1}^{m}t_{i}+\frac{1}{2}cS\right]\leq e^{(-\frac{1}{2}c+a^{2}\lambda)\lambda S}\,. 
\end{align*}
Choosing $\lambda = \min(\frac{c}{4a^{2}},\frac{\ln^{1.25}n}{S})$, we obtain
$$
\mathbb{P}\left[\sum_{i=1}^{m}\Vert v_{i}\Vert^{2}_{2}\geq \sum_{i=1}^{m}t_{i}+\frac{1}{2}cS\right]\leq e^{-\min(\frac{c^{2}S}{4a^{2}},\frac{1}{4}\ln^{1.25}n)}
$$

Combining with the upper bound of $Y_{i}$ we finally obtain Theorem~\ref{thm: high prob bound of optimization objective}.

\subsubsection{Proof of Proposition~\ref{proposition: high probability bound on group decomposition}}

\begin{proof}
Let $v^* = \bar{v} H^{Z} \vec{1}$. We first examine the properties of $\mathbb{E}[u_i]$. By writing $w_i = H^{Z} e_i$, we obtain:
$$
\mathbb{E}(u_{i})=\mathbb{E}\left[\sum_{j\in S_{i}}(v_{j}-\bar{v})w_{\pi(j)}\right]=\frac{1}{\vert S_{i}\vert}\left[\sum_{j\in S_{i}}(v_{j}-\bar{v})\right]\left[\sum_{j\in S_{i}}w_{j}\right]\,.
$$
When condition (1) holds, since $\Vert \sum_{j\in S_{i}} w_{j}\Vert^{2}_{2}=\Vert H^{Z}\sum_{j\in S_{i}}e_{j}\Vert^{2}_{2}\leq \vert S_{j}\vert$, we obtain:
\begin{align*}
&\Vert \mathbb{E}(u_{i})\Vert^{2}_{2}\leq \frac{1}{\vert S_{i}\vert}\left[\sum_{j\in S_{i}}(v_j-\bar{v})\right]^{2}\in o\left(\frac{\vert S_{i}\vert}{n}\right)\sum_{i=1}^{n}(v_{i}-\bar{v})^{2}\,,\\
&\Vert \sum_{i=1}^{k}\mathbb{E}(u_i)\Vert^{2}_{2}=\left\Vert H^{Z}\sum_{i=1}^{k}\frac{1}{\vert S_{i}\vert}\sum_{j\in S_{i}}(v_j-\bar{v})\sum_{j\in S_{i}}e_{j} \right\Vert^{2}_{2}\leq \sum_{i=1}^{k}\frac{1}{\vert S_{i}\vert}\left[\sum_{j\in S_{i}}(v_j-\bar{v})\right]^{2}\in o\left(\sum_{i=1}^{n}(v_{i}-\bar{v})^{2}\right)\,.
\end{align*}

Now we compute $\Vert H^{Z}v_{\pi}\Vert^{2}_{2}$ by:
\begin{align*}
\Vert H^{Z}v_{\pi}\Vert^{2}_{2}&=\Vert v^{*}+\sum_{i=1}^{k}u_{i}\Vert^{2}_{2}\\
&=\Vert v^{*}+\sum_{i=1}^{k}\mathbb{E}(u_{i})\Vert^{2}_{2}+\Vert \sum_{i=1}^{k}[u_{i}-\mathbb{E}(u_{i})]\Vert^{2}_{2}\\
&+2\langle v^{*}+\sum_{i=1}^{k}\mathbb{E}(u_{i}),\sum_{i=1}^{k}[u_{i}-\mathbb{E}(u_{i})]\rangle\,,    
\end{align*}
where we directly obtain $\mathbb{E}[\Vert H^{Z}v_{\pi}\Vert^{2}_{2}]=\Vert v^{*}+\sum_{i=1}^{k}\mathbb{E}(u_{i})\Vert^{2}_{2}+\sum_{i=1}^{k}\mathbb{E}[\Vert u_{i}-\mathbb{E}(u_i)   \Vert^{2}_{2}]$. Furthermore, the first term can be bounded by

$\Vert v^{*}+\sum_{i=1}^{k}\mathbb{E}(u_{i})\Vert^{2}_{2}\leq (1+\frac{1}{3}c) \Vert v^{*}\Vert^{2}_{2}+(1+\frac{3}{c})\Vert \sum_{i=1}^{k}\mathbb{E}(u_i)\Vert^{2}_{2}\leq (1+c) \Vert v^{*}\Vert^{2}_{2}+o(1)\sum_{i=1}^{n}(v_i-\bar{v})^{2}$.
Here we use the fact that $\Vert u+v\Vert^{2}_{2}=\Vert u\Vert^{2}_{2}+\Vert v\Vert^{2}_{2}+2u^{T}v\leq \Vert u\Vert^{2}+\Vert v\Vert^{2}+(\frac{1}{t}\Vert u\Vert^{2}_{2}+t\Vert v\Vert^{2}_{2})$.

The second term is bounded by $\sum_{i=1}^{k}\mathbb{E}[\Vert u_{i}-\mathbb{E}(u_i)\Vert^{2}_{2}]+\frac{1}{3}c\sum_{i=1}^{n}(v_{i}-\bar{v})^{2}$, w.h.p.$(n\to\infty)$ by applying Theorem~\ref{thm: high prob bound of optimization objective} with $S=\sum_{i=1}^{n}(v_{i}-\bar{v})^{2}$. It can be easily verified that for any $i$, $\Vert u_{i}-\mathbb{E}(u_i)\Vert_{2}^{2}\leq \sum_{j\in S_{i}}\left[(v_j-\bar{v})-\frac{1}{\vert S_{i}\vert}\sum_{l\in S_{i}}(v_l-\bar{v})\right]^{2}\leq\sum_{j\in S_{i}}(v_{j}-\bar{v})^{2}\leq \frac{1}{\ln^{4}n}S$.
Now we show that for any $w\in \mathcal{S}^{n-1}$, we have $\sum_{i=1}^{n}(w^{T}_{i}w)^{2}\leq 1$.
Consider any $a_{1},...,a_{n}\in\mathbb{R}$, we have: 
$$
w^{T}\sum_{i=1}^{n}a_{i}w_{i}=\sum_{i=1}^{n}a_{i}\cdot w^{T}w_{i}.
$$
It can also be obtained that
$$
w^{T}\sum_{i=1}^{n}a_{i}w_{i}=(H^{Z}w)^{T}\sum_{i=1}^{n}a_{i}e_{i}\leq\sqrt{\sum_{i=1}^{n}a^{2}_{i}}
$$
Let $a_{i}=w^{T}w_{i}$, we must have:
$$
\sum_{i=1}^{n}(w^{T}w_{i})^{2}\leq\sqrt{\sum_{i=1}^{n}(w^{T}w_{i})^{2}}\,,
$$
implying that $\sum_{i=1}^{n}(w^{T}w_{i})^{2}\leq 1$. 

Now we upper bound $S_{w}$ as follows:
\begin{align*}
S_{w}&=\sum_{i=1}^{k}\mathbb{E}[(w^{T}(u_{i}-\mathbb{E}(u_i)))^{2}]\\
&\leq \sum_{i=1}^{k}\left[\sum_{j\in S_{i}}(v_j-\frac{1}{\vert S_{i}\vert}\sum_{l\in S_{i}}v_l)^{2}\right]\left[\sum_{j\in S_{i}}(w^{T}_{j}w)^{2}\right]\\
&\leq \max_{i}\sum_{j\in S_{i}}(v_{j}-\bar{v})^{2}\\
&\leq\frac{1}{\ln^{4}n}S\,.
\end{align*}
Here, for the second step, we use the decomposition $u_i-\mathbb{E}(u_i)=\sum_{j\in S_i}v_jw_{\sigma_i(j)}-\frac{1}{\vert S_i\vert}\sum_{l\in S_i}v_l\sum_{l\in S_i}w_l$ and apply the Cauchy–Schwarz inequality: $(\sum_{i=1}^{m}a_{i}b_{i})^{2}\leq (\sum_{i=1}^{m} a^{2}_{i})(\sum_{i=1}^{m}b^{2}_{i})$. 

The third step follows from the fact that $\sum_{i=1}^n (w_i^T w)^2 \leq 1$.
Thus, the conditions of Theorem~\ref{thm: high prob bound of optimization objective} are satisfied.

It now remains to upper bound $2\langle v^{*}+\sum_{i=1}^{k}\mathbb{E}(u_{i}),\sum_{i=1}^{k}[u_{i}-\mathbb{E}(u_{i})]\rangle$.
Let $w^{*}=v^{*}+\sum_{i=1}^{k}\mathbb{E}(u_i)$, $r_i=(w^{*})^{T}(u_{i}-\mathbb{E}(u_{i}))$. Then we have $\mathbb{E}(r_i) = 0$, and 
$$\vert r_{i}\vert \leq\sqrt{\sum_{j\in S_i}(v_{j}-\frac{1}{\vert S_i\vert}\sum_{l\in S_i}v_l)^{2}}\Vert w^{*}\Vert_{2}\leq \sqrt{\sum_{j\in S_i}(v_{j}-\bar{v})^{2}}\Vert w^{*}\Vert_{2}\,,
$$
$$\sum_{i=1}^{k}\mathbb{E}(r^{2}_{i})\leq \Vert w^{*}\Vert^{2}_{2}\cdot\sup_{w\in \mathcal{S}^{n-1}}S_{w}\leq \frac{1}{\ln^{4}n}\Vert w^{*}\Vert^{2}_{2}\sum_{i=1}^{n}(v_{i}-\bar{v})^{2}\,.$$

Now, we show that, for any constant $c^{'}>0$,
$$
\mathbb{P}\left[\vert \sum_{i=1}^{k}r_{i}\vert\geq c^{'}\sqrt{\sum_{i=1}^{n}(v_{i}-\bar{v})^{2}}\Vert w^{*}\Vert_{2}\right]\to 0\,.
$$
We can bound $\sum_{i=1}^{k}r_{i}$ as follows:
For $\vert \lambda\vert \leq \frac{1}{\max_{i}(\sqrt{\sum_{j\in S_{i}}(v_j-\bar{v})^{2}})\Vert w^{*}\Vert}$, we have:
\begin{align*}
\mathbb{E}(e^{\lambda\sum_{i=1}^{k}r_{i}})&= \prod_{i=1}^{k} \mathbb{E}(e^{\lambda r_{i}})\\
&\leq \prod_{i=1}^{k}(1+\lambda^{2}\mathbb{E}(r^{2}_{i}))\\
&\leq e^{\frac{1}{\ln^{4}n}\lambda^{2}\Vert w^{*}\Vert^{2}_2\sum_{i=1}^{n}(v_i-\bar{v})^{2}}
\end{align*}

Therefore, by rewriting $\lambda = \frac{\mu}{\Vert w^{*}\Vert_2\sqrt{\sum_{i=1}^{n}(v_{i}-\bar{v})^{2}}}$, we obtain:
\begin{align*}
\mathbb{P}\left[\vert \sum_{i=1}^{k}r_{i}\vert\geq t\Vert w^{*}\Vert_2\sqrt{\sum_{i=1}^{n}(v_{i}-\bar{v})^{2}}\right]\leq 2e^{\frac{1}{\ln^{4} n}\mu^{2}-\mu t}
\end{align*}
Where $\mu$ can range over
$$
\mu\leq \frac{\sqrt{\sum_{i=1}^{n}(v_{i}-\bar{v})^{2}}}{\max_{i}(\sqrt{\sum_{j\in S_{i}}(v_j-\bar{v})^{2}})}\,.
$$
By condition (3), we can simply let $\mu = \ln^{1.9} n$ and $t=\ln^{-0.5}n$, and we obtain:
\begin{align*}
\mathbb{P}\left[\vert \sum_{i=1}^{k} r_{i}\vert \geq 2\ln^{-0.5}(n)\Vert w^{*}\Vert_{2} \sqrt{\sum_{i=1}^{n}(v_i-\bar{v})^{2}}   \right]\leq O(e^{-\ln^{1.4}(n)})\,.
\end{align*}

This implies that, as $n\to \infty$,
\begin{align*}
2\langle v^{*}+\sum_{i=1}^{k}\mathbb{E}(u_{i}),\sum_{i=1}^{k}[u_{i}-\mathbb{E}(u_{i})]\rangle \leq \frac{1}{3}c\left[\Vert v^{*}+\sum_{i=1}^{k}\mathbb{E}(u_i)\Vert^{2}_{2}+\sum_{i=1}^{n}(v_i-\bar{v})^{2}\right], w.h.p.
\end{align*}
Combining this with the previously established result that, with high probability,
$$\Vert \sum_{i=1}^{k}(u_{i}-\mathbb{E}(u_{i}))\Vert^{2}_{2}\leq \sum_{i=1}^{k}\Vert u_{i}-\mathbb{E}(u_{i})\Vert^{2}_{2}+\frac{1}{3}c\sum_{i=1}^{n}(v_{i}-\bar{v})^{2}\,,$$
we obtain
\begin{align*}
\Vert H^{Z}v_{\pi}\Vert^{2}_{2}&\leq (1+\frac{1}{3}c)\left[\Vert v^{*}+\sum_{i=1}^{k}\mathbb{E}(u_{i})\Vert^{2}_{2}+\sum_{i=1}^{k}\mathbb{E}[\Vert u_{i}-\mathbb{E}(u_{i})\Vert^{2}_{2}]  \right]+c\sum_{i=1}^{n}(v_i-\bar{v})^{2}\\
&=(1+\frac{1}{3}c)\mathbb{E}[\Vert H^{Z}v_{\pi}\Vert^{2}_{2}]+c\sum_{i=1}^{n}(v_{i}-\bar{v})^{2}, w.h.p.
\end{align*}
This directly implies the conclusion in Proposition~\ref{proposition: high probability bound on group decomposition}.

Finally, we upper bound $v^{T}_{\pi_k}v$, which can be directly expanded by
\begin{align*}
 v^{T}_{\pi}v  =  \sum_{i=1}^{n}v_{\pi(i)}(v_i-\bar{v})+\bar{v}\sum_{i=1}^{n}v_{\pi(i)}=  (v_{\pi}-\bar{v}\cdot\vec{1})^{T}(v-\bar{v}\cdot \vec{1})+n\bar{v}^{2}\,.
\end{align*}
The expectation of the first term, $(v_{\pi}-\bar{v}\cdot\vec{1})^{T}(v-\bar{v}\cdot \vec{1})$, can be further rewritten as follows:
$$
\mathbb{E}[(v_{\pi}-\bar{v}\cdot\vec{1})^{T}(v-\bar{v}\cdot \vec{1})]=\sum_{i=1}^{k}\frac{1}{\vert S_{i}\vert}\left[\sum_{j\in S_{i}}(v_{j}-\bar{v})\right]^{2}\,.
$$
Therefore, condition~\eqref{condition for group decomposition} guarantees a near-optimal value of $\mathbb{E}[v_{\pi}^T v]$. Furthermore, it can be shown that conditions (1) and (3) also imply the following upper bound:
\begin{align}
 v^{T}_{\pi}v \leq \mathbb{E}[ v^{T}_{\pi}v]+o(1)\sum_{i=1}^{n}(v_{i}-\bar{v})^{2}\,, w.h.p.
\end{align}
This is because, for all $i$ we always have: $\vert \sum_{j\in S_{i}}(v_{\pi(j)}-\bar{v})(v_j-\bar{v})\vert\leq\sum_{j\in S_{i}}(v_j-\bar{v})^{2}$. By setting $A_i = \sum_{j\in S_{i}}(v_{\pi(j)}-\bar{v})(v_j-\bar{v})$, $\lambda=\frac{\ln^{2}n}{\sum_{i=1}^{n}(v_{i}-\bar{v})^{2}}$, we have:
\begin{align*}
\mathbb{E}[e^{\lambda A_i}]\leq 1+\lambda\mathbb{E}[A_i]+\lambda^{2}\mathbb{E}[A^{2}_{i}]\leq 1+\lambda\mathbb{E}[A_i]+\frac{\sum_{j\in S_{i}}(v_j-\bar{v})^{2}}{\sum_{i=1}^{n}(v_i-\bar{v})^{2}}
\end{align*}
This implies that 
\begin{align*}
\mathbb{P}\left[v^{T}_{\pi}v\geq t\right]&=\mathbb{P}\left[\sum_{i=1}^{k}A_i\geq t\right]\\
&\leq e^{-\lambda t}\mathbb{E}[e^{\lambda\sum_{i=1}^{k}A_{i}}]\\
&\leq e^{-\lambda t}\prod_{i=1}^{k}1+\lambda \mathbb{E}[A_{i}]+{\sum_{i=1}^{n}(v_i-\bar{v})^{2}}\\
&\leq e^{-\lambda t}\prod_{i=1}^{k}e^{\lambda \mathbb{E}[A_{i}]+{\sum_{i=1}^{n}(v_i-\bar{v})^{2}}}\\
&=exp\left(\lambda(-t+\sum_{i=1}^{k}\mathbb{E}[A_i]+1)\right)\,.
\end{align*}
Therefore, we obtain:
\begin{align*}
\mathbb{P}\left[v^{T}_{\pi}v\geq \mathbb{E}[v^{T}_{\pi}v]+\frac{\sum_{i=1}^{n}(v_i-\bar{v})^{2}}{\ln^{0.5}n}\right]\leq e^{1-\ln^{-1.5}(n)}
\end{align*}
and similarly we can also obtain
\begin{align*}
\mathbb{P}\left[v^{T}_{\pi}v\leq \mathbb{E}[v^{T}_{\pi}v]-\frac{\sum_{i=1}^{n}(v_i-\bar{v})^{2}}{\ln^{0.5}n}\right]\leq e^{1-\ln^{-1.5}(n)}
\end{align*}
\end{proof}
We now examine for which distributions of $v = (I - H^{Z})X$ the conditions in Proposition~\ref{proposition: high probability bound on group decomposition} are satisfied. Regarding the first condition, since
\begin{align*}
\mathbb{E}(u_{i})=\mathbb{E}\left[\sum_{j\in S_{i}}(v_{\pi(j)}-\bar{v})w_{j}\right]=\frac{1}{\vert S_{i}\vert}\left[\sum_{j\in S_i}(v_{j}-\bar{v})\right]\left[ \sum_{j\in S_{i}}w_{j}\right]\,,
\end{align*}
by denoting $s_{i}=\sum_{j\in S_i}(v_j-\bar{v})$, and $\vec{1}_{i}=\sum_{j\in S_i}e_{j}$, we obtain:
\begin{align*}
\left\Vert \sum_{i=1}^{k}\mathbb{E}(u_i) \right\Vert^{2}_{2}&=\left\Vert\sum_{i=1}^{k}\frac{1}{\vert S_i\vert}s_iH^{Z}\vec{1}_{i}  \right\Vert_{2}^{2}\leq\left\Vert\sum_{i=1}^{k}\frac{1}{\vert S_i\vert}s_i\vec{1}_{i}  \right\Vert_{2}^{2}=\sum_{i=1}^{k}\frac{1}{\vert S_i\vert}s^{2}_{i}.
\end{align*}

To ensure that $\Vert \sum_{i=1}^{k}\mathbb{E}(u_i)\Vert_{2}\leq o(1)\sqrt{\sum_{i=1}^{n}(v_i-\bar{v})^{2}}$, a sufficient set of constraints on ${s_i}$ and ${S_i}$ is that $s^{2}_{i}\in O(\sum_{i=1}^{n}(v_{i}-\bar{v})^{2})$ and $\vert S_{i}\vert\geq n^{0.55}$. These constraints are incorporated into our algorithm design.

Conditions (2) and (4) are mild and commonly satisfied for typical data distributions; in particular, they hold when the residual $X - H^{Z}X$ is not concentrated on a small number of coordinates.

Condition (3) imposes an additional requirement on the subsets $S_i$, which we satisfy by imposing an upper bound on $\sum_{j \in S_i} (v_j - \bar{v})^2$ in order to guarantee this condition.
\subsubsection{Proof of Lemma~\ref{lem: approximation of optimization objective}}
Firstly, we show that 
$$
\left\vert \sum_{i=1}^{k}\frac{1}{\vert S_{i}\vert}\left[\sum_{j\in S_{i}}a^{2}_{j}\right]\left[\sum_{j\in S_{i}}\Vert w_{j}\Vert^{2}_{2}\right]  -\sum_{i=1}^{k}\mathbb{E}[\Vert u_{i}\Vert^{2}_{2}] \right\vert\leq o(1)\sum_{i=1}^{n}a^{2}_{i},
$$
where $u_{i}=\sum_{j\in S_{i}}a_{j}w_{\sigma_{i}(j)}$.
Since $\mathbb{E}[\Vert u_{i}\Vert^{2}_{2}]$ can be represented by:
\begin{align*}
\mathbb{E}\left[\Vert u_{i}\Vert^{2}_{2}\right]&=\left(\sum_{j\in S_{i}}\frac{1}{\vert S_{i}\vert}a_{j}^{2}\right)\left( \sum_{j\in S_{i}}w^{T}_{j}w_{j} \right)\\
&+\left(\frac{1}{\vert S_{i}\vert(\vert S_{i}\vert-1)}\sum_{j,l\in S_{i},j\neq l}a_{j}a_{l}\right)\left(\sum_{j,l\in S_{i},j\neq l}w_{j}^{T}w_{l}\right)\,.
\end{align*}

It suffices to show that 
\begin{align*}
\left\vert \sum_i \frac{1}{\vert S_i\vert(\vert S_i\vert-1)}\left(\sum_{j,k\in S_i,j\neq k}a_j a_k\right)\left(\sum_{j,k\in S_i,j\neq k}w^{T}_jw_k\right)\right\vert \leq o(\sum_{l=1}^{n}a^{2}_l)\,.
\end{align*}
On one hand, 
\begin{align*}
\left\vert\sum_{j,k\in S_i,j\neq k}a_j a_k\right\vert=\left\vert\left(\sum_{j\in S_i}a_j\right)^{2}-\sum_{j\in S_i}a^{2}_{j}\right\vert\leq O(1)\sum_{i=1}^{n}a^{2}_{i}\,.
\end{align*}
On the other hand,
\begin{align*}
\left\vert \sum_{j,k\in S_i,j\neq k}w^{T}_jw_k\right\vert &=\left\vert\left\Vert \sum_{j\in S_i}w_j\right\Vert^{2}_2-\sum_{j\in S_i}\Vert w_j\Vert^{2}_{2}\right\vert\\
&\leq \max\left(\left\Vert\sum_{j\in S_{i}}w_j\right\Vert^{2}_{2},\sum_{j\in S_{i}}\Vert w_{j}\Vert^{2}_{2}\right)
\\&\leq \vert S_i\vert
\end{align*}
Therefore, we obtain:
\begin{align*}
&\left\vert \sum_i \frac{1}{\vert S_i\vert(\vert S_i\vert-1)}\left(\sum_{j,k\in S_i,j\neq k}a_j a_k\right)\left(\sum_{j,k\in S_i,j\neq k}w^{T}_jw_k\right)\right\vert\\
&\leq O(1)\sum_{i=1}^{k}\frac{1}{\vert S_{i}\vert(\vert S_{i}\vert-1)}\cdot \left(\sum_{j\in S_{i}}a^{2}_{j}\right)\cdot \vert S_{i}\vert \\
&\leq O(1)\cdot \left(\sum_{i=1}^{k}\frac{1}{\vert S_{i}\vert-1}\right)\sum_{i=1}^{n}a^{2}_{i}\\
&\leq o(1)\sum_{i=1}^{n}a^{2}_{i}\,,
\end{align*}
where the last inequality is because $\vert S_{i}\vert \geq n^{0.55}$. 
Next, we show that
\begin{align*}
\left\vert \left[\Vert v^{*}\Vert^{2}_{2}+\sum_{i=1}^{k}\mathbb{E}[\Vert u_{i}\Vert^{2}_{2}]\right]-\mathbb{E}[\Vert H^{Z}v_{\pi}\Vert^{2}_{2}]  \right\vert\leq o(1)\left[\sum_{i=1}^{n}a^{2}_{i}+\Vert v^{*}\Vert^{2}_{2}\right]\,.
\end{align*}
We directly compute $\mathbb{E}\left[\Vert H^{Z}v_{\pi}\Vert^{2}_{2}\right]$ by:
\begin{align*}
\mathbb{E}\left[\Vert H^{Z}v_{\pi}\Vert^{2}_{2}\right]&=\Vert v^{*}+\sum_{i=1}^{k}\mathbb{E}(u_{i})\Vert^{2}_{2}+\sum_{i=1}^{k}\mathbb{E}[\Vert u_{i}-\mathbb{E}(u_i)\Vert^{2}_{2}]\\
&=\Vert v^{*}+\sum_{i=1}^{k}\mathbb{E}(u_{i})\Vert^{2}_{2}+\sum_{i=1}^{k}\mathbb{E}[\Vert u_{i}\Vert^{2}_{2}]-\sum_{i=1}^{k}\Vert \mathbb{E}(u_i)\Vert^{2}_{2}\,.
\end{align*}
We have shown in the proof of Proposition~\ref{proposition: high probability bound on group decomposition} that $\Vert \sum_{i=1}^{k}\mathbb{E}(u_i)\Vert^{2}_{2}\leq O(1)\sum_{i=1}^{k}\frac{1}{\vert S_{i}\vert}\sum_{i=1}^{n}a^{2}_{i}\leq o(1)\sum_{i=1}^{n}a^{2}_{i}$. Moreover, for any $t > 0$,
$$
\Vert v^{*}+\sum_{i=1}^{k}\mathbb{E}(u_{i})\Vert^{2}_{2}\leq (1+t)\Vert v^{*}\Vert^{2}_{2}+(1+\frac{1}{t})\Vert \sum_{i=1}^{k}\mathbb{E}(u_i)\Vert^{2}_{2}\,,
$$
so we may choose $t \in o(1)$ such that $\frac{1}{t}\Vert \sum_{i=1}^{k}\mathbb{E}(u_i)\Vert^{2}_{2}\leq o(1)\sum_{i=1}^{n}a^{2}_{i}$.
We have also established that $\Vert \mathbb{E}(u_i)\Vert^{2}_{2} \leq O(1)\frac{1}{\vert S_i\vert}\sum_{i=1}^{n}a^{2}_{i}$, which in turn implies
$$
\sum_{i=1}^{k}\Vert \mathbb{E}(u_i)\Vert^{2}_{2}\leq o(1)\sum_{i=1}^{n}a_{i}^{2}\,.
$$
Consequently, we obtain
\begin{align*}
\left\vert \Vert v^{*}\Vert^{2}_{2}+\sum_{i=1}^{n}\mathbb{E}[\Vert u_i\Vert^{2}_{2}]-\mathbb{E}[\Vert H^{Z}v_\pi\Vert^{2}_{2}]   \right\vert\leq o(1)\left(\Vert v^{*}\Vert^{2}+\sum_{i=1}^{n}a^{2}_{i}\right)\,.
\end{align*}

Finally, note that we have shown in Proposition~\ref{proposition: high probability bound on group decomposition} that
\begin{align*}
\mathbb{P}\left[\vert v^{T}_{\pi}v- \mathbb{E}[v^{T}_{\pi}v]\vert\geq \frac{\sum_{i=1}^{n}(v_i-\bar{v})^{2}}{\ln^{0.5}n}\right]\leq e^{1-\ln^{-1.5}(n)}
\end{align*}
Combining this with the fact that $\vert\mathbb{E}[v^{T}_{\pi}v]- n\bar{v}^{2}\vert\leq o(\sum_{i=1}^{n}a^{2}_{i})$, we obtain Lemma~\ref{lem: approximation of optimization objective}.

\subsection{Correctness of Algorithm~\ref{alg structure}}

\subsubsection{Theoretical guarantee for the partition of subsets}

We first prove the correctness of Algorithm~\ref{alg partition}.

\begin{lemma}\label{lem: alg set blocks}
In Algorithm~\ref{alg structure}, if for some $l$ we have:
$$
(\sum_{j=1}^{l}a_{i_j})^{2}+(\sum_{j=1}^{l}b_{i_j})^{2}\geq \sum_{i=1}^{k}(a^2_i+b^{2}_{i})
$$
Then there exists $i_{l+1}\notin \{i_1,i_2,...,i_l\}$ such that:
$$
(\sum_{j=1}^{l+1}a_{i_j})^{2}+(\sum_{j=1}^{l+1}b_{i_j})^{2}\leq(\sum_{j=1}^{l}a_{i_j})^{2}+(\sum_{j=1}^{l}b_{i_j})^{2}-(a^{2}_{i_{l+1}}+b^{2}_{i_{l+1}})\,.
$$
This implies that for any $i$, we have:
\begin{align}
\Vert \sum_{j=1}^{i}(a_{i},b_{i})\Vert_{2}\leq 2\sqrt{\sum_{i=1}^{k}(a^{2}_i+b^{2}_{i})}
\end{align}
\end{lemma}
\begin{proof}
Let $u=\sum_{j=1}^{l}(a_{i_j},b_{i_j})$ and $s=\{i_1,...,i_l\}$. Since $\sum_{j=1}^{k}a_j=\sum_{j=1}^{k}b_j=0$, we have:
$$
\sum_{i\notin s}\langle u,(a_i,b_i)\rangle =-\Vert u\Vert^{2}_{2}\leq-\sum_{i\notin s}(a_i^{2}+b^{2}_{i})\,,
$$
where the inequality is due to the definition of $i_l$ in Algorithm~\ref{alg partition}. Thus, there exists some $i\notin s$ such that 
$$
\langle u,(a_i,b_i)\rangle \leq-(a^{2}_{i}+b^{2}_{i})\,.
$$
This implies that
$$
\Vert u+(a_i,b_i)\Vert^{2}_{2}= \Vert u\Vert^{2}+\Vert (a_i,b_i)\Vert^{2}_{2}+2\langle u,(a_i,b_i)\rangle \leq \Vert u\Vert^{2}_{2}-\Vert (a_i,b_i)\Vert^{2}_{2}\,.
$$
Therefore, for any $i$ satisfying $\Vert \sum_{j=1}^{i}(a_i,b_i)\Vert^{2}_{2}\geq \sum_{i=1}^{k}(a^{2}_{i}+b^{2}_{i})$, we must have $\Vert \sum_{j=1}^{i+1}(a_i,b_i)\Vert^{2}_{2}\leq \Vert \sum_{j=1}^{i}(a_i,b_i)\Vert^{2}_{2}$. Furthermore, let $j_1\leq i$ be the largest integer such that $\Vert \sum_{j=1}^{j_1}(a_i,b_i)\Vert^{2}_{2}\geq \sum_{i=1}^{k}(a^{2}_{i}+b^{2}_{i})$ and $\Vert \sum_{j=1}^{j_1-1}(a_i,b_i)\Vert^{2}_{2}< \sum_{i=1}^{k}(a^{2}_{i}+b^{2}_{i})$. Then we obtain:
\begin{align*}
\Vert \sum_{j=1}^{i}(a_i,b_i)\Vert_{2}\leq \Vert \sum_{j=1}^{j_1}(a_i,b_i)\Vert_{2}\leq \Vert \sum_{j=1}^{j_1-1}(a_i,b_i)\Vert_{2} +\Vert(a_{j_1},b_{j_1})\Vert_2\leq 2\sqrt{\sum_{i=1}^{k}(a^{2}_{i}+b^{2}_{i})}\,.
\end{align*}
\end{proof} 
We next show that for any $S_{i}\subseteq I^{'}_{s}$(s=1,2,3), $\vert S_{i}\vert\geq \Omega(n^{0.55})$.
By Lemma~\ref{lem: alg set blocks}, for any $S_{i}\subseteq I^{'}$, where $I^{'}$ is one of $I^{'}_1,I^{'}_2,I^{'}_3$, which are the set of indices after dividing $\{1,2,...,n\}$ into 3 subsets containing different types of dimensions, as in Algorithm~\ref{alg rearrange}. Then we have:
$$
\vert \sum_{j\in S_{i}}a_{j}\vert\leq 4\sqrt{\sum_{i=1}^{n}a^{2}_{i}}\,,\,\,\,\left\vert \sum_{j\in S_{i}}(c_{j}-\frac{1}{\vert I^{'}\vert}\sum_{l\in I^{'}}c_{l})\right\vert\leq 4\max_{i}(\vert a_{i}\vert)\sqrt{\sum_{i=1}^{n}a^{2}_{i}}
$$
First, we consider $\vert S_{i}\vert$. On one hand, by Algorithm~\ref{alg partition} we have
$$
\sum_{j\in S_{i}}a^{2}_{j}+\frac{1}{M}(c_j-\frac{1}{\vert I^{'}\vert}\sum_{l\in I^{'}}c_{l})^{2}\geq M^{\frac{1}{3}}S^{\frac{2}{3}}\,,
$$
On the other hand, 
\begin{align*}
\sum_{j\in S_{i}}\frac{1}{M}(c_j-\frac{1}{\vert I^{'}\vert}\sum_{l\in I^{'}}c_{l})^{2}
&\leq \frac{1}{M}\sum_{j\in S_{i}} c^{2}_{j}+\frac{1}{M}\cdot 3\max_{i}(\vert a_{i}\vert)\sqrt{\sum_{i=1}^{n}a^{2}_{i}}\frac{1}{\vert I^{'}\vert}\sum_{l\in I^{'}}c_{l}
\\
&\leq \sum_{j\in S_{i}} c_{j}+3\sqrt{MS}\,,
\end{align*}
where, in the first step, we use the fact that $(a-b)^{2}\leq a^{2}+3\max(\vert a\vert,\vert b\vert)\cdot  \vert b\vert$, and the second step is because $c_j\leq M$.
Thus, we obtain:
\begin{align*}
2\sum_{j\in S_{i}}c_{j}+3\sqrt{MS}\geq M^{\frac{1}{3}}S^{\frac{2}{3}}.
\end{align*}
Since $\sum_{l\in I^{'}}c_{l}\leq S$, we also have:
$$
\sum_{j\in S_{i}}c_{j}\leq \vert S_{i}\vert\frac{1}{\vert I^{'}\vert}S+4\sqrt{MS}
$$
By $M\leq o(S)$, we obtain: $\frac{\vert S_{i}\vert}{\vert I^{'}\vert}\geq (\frac{1}{2}-o(1))M^{\frac{1}{3}}S^{-\frac{1}{3}}\in\Omega(n^{-\frac{1}{3}})$, and we finally obtain $\vert S_{i}\vert\in \Omega( n^{0.55})$ by $\vert I^{'}\vert\geq n^{0.9}$. 

On the other hand, for any $S_{i}$, we have $\sum_{j\in S_{i}}a^{2}_{j}\leq 2M^{\frac{1}{3}}S^{\frac{2}{3}}$. 

Thus, when $M \leq \frac{1}{8 \ln^{12} n} S$, condition (2) in Proposition~\ref{proposition: high probability bound on group decomposition} is satisfied, and consequently Lemma~\ref{lem: approximation of optimization objective} holds.

Then in Algorithm~\ref{alg structure}, $\mathbb{E}\left[\frac{1}{2}v^{T}_{\pi_k}+\Vert H^{Z}v_{\pi_k}\Vert^{2}_{2}\right]\leq \sum_{i=1}^{k}\frac{1}{\vert S_{i}\vert}\left(\sum_{j\in S_{i}} a^{2}_{j}\right)\left(\sum_{j\in S_{i}} b_{j}\right)+o(1)\sum_{i=1}^{n}a^{2}_{i}$, which demonstrates the effectiveness of \eqref{optimization objective: final form}.
It remains to compute
$\sum_{i=1}^{k}\frac{1}{\vert S_{i}\vert}\left(\sum_{j\in S_{i}} a^{2}_{j}\right)\left(\sum_{j\in S_{i}} b_{j}\right)$, which can be rewritten as $\sum_{i=1}^{k}\frac{1}{\vert S_{i}\vert}\left[\sum_{j\in S_{i}}(c_{j}-\bar{c})\right]\left[\sum_{j\in S_{i}}(b_{j}-\bar{b})\right]+n\bar{b}\bar{c}$.

\subsubsection{Upper bound for the value of the optimization objective}
Now we compute the value of optimization target~\ref{optimization objective: final form} in our algorithm. For the set $I^{'}_{1}$, which is the set of indices corresponding to $J_1$, we can derive that:
\begin{align*}
&\frac{1}{\vert S_{i}\vert}\left[\sum_{j\in S_{i}}(c_{j}-\bar{c})\right]\left[\sum_{j\in S_{i}}(b_{j}-\bar{b})\right]
\\&\leq \frac{1}{\vert S_{i}\vert}\left[\left\vert\sum_{j\in S_{i}}(c_j-\frac{1}{\vert J_{1}\vert}\sum_{l\in J_1}c_{l})\right\vert +\vert S_{i}\vert \left\vert \frac{1}{\vert J_{1}\vert}\sum_{l\in J_1}c_{l}-\bar{c} \right\vert  \right]\left\vert\sum_{j\in S_{i}}(b_{j}-\bar{b})\right\vert\\
&\leq \left\vert\sum_{j\in S_{i}}(c_j-\frac{1}{\vert J_{1}\vert}\sum_{l\in J_1}c_{l})\right\vert +\vert S_{i}\vert \left\vert \frac{1}{\vert J_{1}\vert}\sum_{l\in J_1}c_{l}-\bar{c} \right\vert\\
&\leq 4\sqrt{MS}+\frac{\vert S_{i}\vert}{\vert J_1\vert} \cdot 4\sqrt{MS}\\
&\leq 8\sqrt{MS}
\end{align*}
Where we use the fact that $\vert b_{j}-\bar{b}\vert\leq 1$ since $b_{j},\bar{b}\in[0,1]$.

By taking the sum over all the $S_{i}\subseteq J_1$, we obtain:
\begin{align*}
\sum_{S_{i}\subseteq J_1}\frac{1}{\vert S_{i}\vert}\left[\sum_{j\in S_{i}}(c_{j}-\bar{c})\right]\left[\sum_{j\in S_{i}}(b_{j}-\bar{b})\right] \leq \sum_{S_i\subseteq J_1}8\sqrt{MS}\leq 8\sqrt{MS}\cdot\frac{S}{M^{\frac{1}{3}}S^{\frac{2}{3}}}\leq o(S)   
\end{align*}
Similarly, if we let $\bar{b}_2=\frac{1}{\vert J_2\vert}\sum_{j\in J_2}b_{j}$, $\bar{b}_3=\frac{1}{\vert J_3\vert}\sum_{j\in J_3}b_{j}$, $\bar{c}_2=\frac{1}{\vert J_2\vert}\sum_{j\in J_2}c_{j}$, $\bar{c}_3=\frac{1}{\vert J_3\vert}\sum_{j\in J_3}c_{j}$, we can obtain:
\begin{align*}
\sum_{S_{i}\subseteq J_2}\frac{1}{\vert S_{i}\vert}\left[\sum_{j\in S_{i}}(c_{j}-\bar{c})\right]\left[\sum_{j\in S_{i}}(b_{j}-\bar{b})\right]\leq \vert J_2\vert(\bar{b}_2-\bar{b})(\bar{c}_2-\bar{c})+o(S)\,,
\end{align*}
\begin{align*}
\sum_{S_{i}\subseteq J_3}\frac{1}{\vert S_{i}\vert}\left[\sum_{j\in S_{i}}(c_{j}-\bar{c})\right]\left[\sum_{j\in S_{i}}(b_{j}-\bar{b})\right]\leq \vert J_3\vert(\bar{b}_3-\bar{b})(\bar{c}_3-\bar{c})+o(S)\,.
\end{align*}
Finally, we summarize the result by
\begin{align}\label{value of final optimization objective}
\sum_{i=1}^{k} \frac{1}{\vert S_{i}\vert}\left[\sum_{j\in S_{i}}(c_{j}-\bar{c})\right]\left[\sum_{j\in S_{i}}(b_{j}-\bar{b})\right]   \leq \vert J_2\vert(\bar{b}_2-\bar{b})(\bar{c}_2-\bar{c})+\vert J_3\vert(\bar{b}_3-\bar{b})(\bar{c}_3-\bar{c})+o(S)\,.
\end{align}
Here $S=\sum_{i=1}^{n}a_{i}^{2}=\Vert v-\bar{v}\vec{1}\Vert^{2}_{2}$.
\subsubsection{Proof of Proposition~\ref{prop: gap between ours and random permutation}}
Now we complete the proof of Proposition~\ref{prop: gap between ours and random permutation}, along with the case of estimating $\lambda^{*},\lambda^{**}$ by $\Omega(1/\alpha^{2})$ samples. 

Lemma~\ref{lem: approximation of optimization objective}, \eqref{value of final optimization objective}, and \Cref{subsubsection: comparison with random permutation} together demonstrate that for $\pi_k$ uniformly sampled from $\mathcal{P}_{K}$ as constructed in Algorithm~\ref{alg structure}, and for $\pi'$ drawn uniformly at random from the full symmetric group, the following two expectations exhibit a provable gap:
\begin{align*}
&\mathbb{E}_{\pi_k}\left[\frac{1}{2}v^{T}_{\pi_k}v+\Vert H^{Z}v_{\pi_k}\Vert^{2}_{2}\right]\\&\leq \mathbb{E}_{\pi^{'}}\left[\frac{1}{2}v^{T}_{\pi^{'}}v+\Vert H^{Z}v_{\pi^{'}}\Vert^{2}_{2}\right]+\vert J_2\vert(\bar{b}_2-\bar{b})(\bar{c}_2-\bar{c})+\vert J_3\vert(\bar{b}_3-\bar{b})(\bar{c}_3-\bar{c})+o(1)\Vert v\Vert^{2}_{2}\,.
\end{align*}
By Lemma~\ref{lem: upper and lower of optimization objective formal} we have $\lambda_2(X,Z,\mathcal{P}_n,\frac{1}{2}\alpha)\geq \mathbb{E}_{\pi^{'}}\left[\frac{1}{2}v^{T}_{\pi^{'}}v+\Vert H^{Z}v_{\pi^{'}}\Vert^{2}_{2}\right]-O(\alpha)\Vert v\Vert^{2}_{2}$, and by Proposition~\ref{proposition: high probability bound on group decomposition} we have:
$\lambda_2(X,Z,\mathcal{P}_K,\frac{1}{4}\alpha)\leq \mathbb{E}_{\pi_k}\left[\frac{1}{2}v^{T}_{\pi_k}v+\Vert H^{Z}v_{\pi_k}\Vert^{2}_{2}\right]+o(1)\Vert v\Vert^{2}_{2}$. This implies that
\begin{align*}
&\lambda_2(X,Z,\mathcal{P}_K,\frac{1}{4}\alpha)-\lambda_2(X,Z,\mathcal{P}_n,\frac{1}{2}\alpha)\\&\leq \vert J_2\vert(\bar{b}_2-\bar{b})(\bar{c}_2-\bar{c})+\vert J_3\vert(\bar{b}_3-\bar{b})(\bar{c}_3-\bar{c})+O(\alpha)\Vert v\Vert^{2}_{2}\,.
\end{align*}
\paragraph{Extension to estimation by finite samples.}
In practice, both $\mathcal{P}_K$ and $\mathcal{P}_{n}$ contain too many permutations to allow direct computation of $\lambda_2(X,Z,\mathcal{P}_n,\frac{1}{4}\alpha)$ and $\lambda_2(X,Z,\mathcal{P}_K,\frac{1}{4}\alpha)$. Therefore, we estimate these quantities using $m \geq 1/\alpha^2$ i.i.d. samples, which yield estimators $\lambda^*$ and $\lambda^{**}$ that satisfy:
\begin{align*}
\lambda^{*}=\inf \lambda : \frac{1}{m}\sum_{i=1}^{m}\mathbbm{1}\left\{\frac{1}{2}v^{T}_{\pi_{i}}v+\Vert H^{Z_{\pi_{i}}}v\Vert^{2}_{2}\leq \lambda  \right\}\geq 1-\frac{1}{2}\alpha (\pi_{i}\text{ is sampled i.i.d. from all permutations})
\end{align*}

\begin{align*}
\lambda^{**}=\inf \lambda : \frac{1}{m}\sum_{i=1}^{m}\mathbbm{1}\left\{\frac{1}{2}v^{T}_{\pi_{i}}v+\Vert H^{Z_{\pi_{i}}}v\Vert^{2}_{2}\leq \lambda  \right\}\geq 1-\frac{1}{4}\alpha (\pi_{i}\text{ is sampled i.i.d. from }\mathcal{P}_{K})
\end{align*}
Then Theorem~\ref{thm: value of optimization objective} and Proposition~\ref{proposition: high probability bound on group decomposition} imply that both $\lambda^{*}\geq \lambda_2(X,Z,\mathcal{P}_n,\frac{1}{4}\alpha)-O(\alpha)\Vert v\Vert^{2}_{2}$ and $\vert \lambda^{**}-\lambda_2(X,Z,\mathcal{P}_K,\frac{1}{4}\alpha)\vert \leq o(1)\Vert v\Vert^{2}_{2}$ hold with probability at least $1-e^{-\Omega(\alpha^{-1})}$. Thus, we also have:
\begin{align*}
\lambda^{**}\leq \lambda^{*}+\vert J_2\vert(\bar{b}_2-\bar{b})(\bar{c}_2-\bar{c})+\vert J_3\vert(\bar{b}_3-\bar{b})(\bar{c}_3-\bar{c})+O(\alpha)\Vert v\Vert^{2}_{2}\,, \, w.p. \text{ at least }1-e^{-\Omega(\alpha^{-1})}\,.
\end{align*}
\subsubsection{Proof of Lemma~\ref{lem: type II error of finite sample}}
For Type I error, we have:
\begin{align*}
\mathbb{P}\left[\mathcal{H}_{0}\text{ is rejected }\vert b=0 \right]&\leq\mathbb{P}\left[\min(\phi_1,\phi_2)\leq (1+c)\alpha\vert b=0 \right]+\mathbb{P}\left[\phi_1,\phi_2>(1+c)\alpha\,, \, \min(\phi^{'}_1,\phi^{'}_2)\leq \alpha\right]
\end{align*}
The first term is no larger than $4(1+c)\alpha$, and it suffices to upper bound $$\mathbb{P}\left[\phi_1,\phi_2>(1+c)\alpha\,, \, \min(\phi^{'}_1,\phi^{'}_2)\leq \alpha\right].$$ Let $\delta_i=1\left\{X^{T}H^{ZZ_{\pi_{i}}}X\leq X^{T}H^{ZZ_{\pi_{i}}}Y \right\}$, then $\mathbb{E}[\delta_i]=\phi^{'}_1\geq (1+c)\alpha.$ Thus for any $\lambda>0$, we have:
\begin{align*}
\mathbb{P}\left[\phi^{'}_1\leq \alpha\right]&\leq e^{\lambda m\alpha}\mathbb{E}\left[e^{-\lambda\sum_{i=1}^{m}\delta_i}\right]\\
&\leq e^{-\lambda m\alpha}\left[(1-\phi^{'}_1)+\phi^{'}_1e^{-\lambda}\right]^{m}\\
&=e^{\lambda m\alpha}\left[1-\phi^{'}_1(1-e^{-\lambda})\right]^{m}\,.
\end{align*}
Since $1-x\leq e^{-x}\leq 1-x+\frac{1}{2}x^{2}(x\geq 0)$, we have: $\left[1-\phi^{'}_1(1-e^{-\lambda})\right]^{m}\leq e^{-m\phi^{'}_1(\lambda-\frac{1}{2}\lambda^{2})}\leq e^{-\lambda(1+c)m\alpha(1-\frac{1}{2}\lambda)}$.
Let $1-\frac{1}{2}\lambda=\frac{1+\frac{1}{2}c}{1+c}$ we obtain:
\begin{align*}
\mathbb{P}\left[\phi^{'}_1\leq \alpha\right]\leq e^{-\frac{1}{2}cm\alpha}\leq e^{-\Omega(\alpha^{-1})}\,.
\end{align*}
Similarly, we can prove that $\mathbb{P}\left[\phi^{'}_2\leq \alpha\vert \phi_2\geq (1+c)\alpha\right]\leq e^{-\Omega(\alpha^{-1})}$. Combining these we conclude that 
\begin{align*}
\mathbb{P}\left[\phi_1,\phi_2>(1+c)\alpha\,, \, \min(\phi^{'}_1,\phi^{'}_2)\leq \alpha\right]\leq e^{-\Omega(\alpha^{-1})}\,.
\end{align*}
For the Type II error, we upper bound it by
\begin{align*}
\mathbb{P}\left[\mathcal{H}_{0}\text{ is accepted }\right]&\leq\mathbb{P}\left[\min(\phi_1,\phi_2)\geq (1-c)\alpha\right]+\mathbb{P}\left[\phi_1,\phi_2<(1-c)\alpha\,, \, \min(\phi^{'}_1,\phi^{'}_2)\geq \alpha\right]\,.
\end{align*}
Similar to the probability bound of Type I error, the second term is upper bounded by $e^{-\Omega(\alpha^{-1})}$, and we complete the proof of Lemma~\ref{lem: type II error of finite sample}.



\section{Proofs of nonexchangeable case}\label{proof_nonex}
\subsection{Proof of \texorpdfstring{\Cref{thm::Theorem nonexchangeable grouped CPT}}{Theroem 11}}
\label{detailed proof of theorem nonexchangeable grouped CPT}


From the nonexchangeable point of view, we may view this as a covariate-shift setting. In this case, the test statistic $R_0$ may behave differently from the other statistics $R_k, k\in\{1,\ldots,K\}$. We set a user-specified weight $w_0\in (0,1)$ for $R_0$, and the other weights $w_i$ are taken to be identical because of the property of the permutation group. Accordingly, we consider the rejection region 
$$ [Q_{1-\alpha}(\sum_{i=0}^{K}w_i\cdot \delta_{R_i}),1],$$
where 
$$w_i=\frac{1-w_0}{K},i=\{1,\cdots,K\}.$$
Correspondingly, define 
\begin{equation}
    \mathcal{K} \sim \sum_{i=0}^{K} w_i \cdot \delta_{\{i\}},
\end{equation}
and we obtain the following lemma
\begin{lemma}
\label{lem::Theorem nonexchangeable grouped CPT}
 Under $H_0$, and $w_0\in[\frac{1}{K+1},1)$, we set 
 $$w_i=\frac{1-w_0}{K},i\in \{1,\cdots,K\},$$
  then we obtain 
    \begin{equation*}
        \begin{aligned}
        P(R_0 \leq Q_{1-\alpha}
    (\sum_{k=0}^{K}w_k\cdot\delta_{R_k}))& \geq 1-\alpha-\sum_{k=1}^Kw_i\,\,(d_{TV}(R(\epsilon),R(\epsilon^k)) \\
    & \geq 1-\alpha-\sum_{k=1}^Kw_i\,\,(d_{TV}(\epsilon,\epsilon_{\pi_k}))
        \end{aligned}
    \end{equation*}
    where $Q_{\tau}(\cdot)$ denotes the $\tau$- quantile of its argument, $\delta_{a}$ denotes the point mass at $a$, and $R(\epsilon)$ is a $K+1$-dimensional vector that $(R(\epsilon))_{i}=R_{i-1}(\epsilon),i\in\{1,\cdots,K+1\}.$
\end{lemma}
We see that when we set $w_0=\frac{1}{K+1}$, Lemma \ref{lem::Theorem nonexchangeable grouped CPT} directly turns to \Cref{thm::Theorem nonexchangeable grouped CPT}, so that we only need to prove the Lemma \ref{lem::Theorem nonexchangeable grouped CPT}.
\subsection{Proof of Lemma \ref{lem::Theorem nonexchangeable grouped CPT}}
For any $k\in [0,K]$, as before, let $\pi_k$ denote the permutation corresponding to $P_k\in \mathcal{P}_K$. Then, for any $k\in [0,K]$, and we denote $(R(Y))_{i}=R_{i-1}=R_{i-1}(Y), i\in \{1,\cdots,K+1\}$, from the definition of (\ref{R_k}), we can calculate
$$(R(Y^k))_{i+1}=\frac{1}{K+1}\sum_{j=0,j\neq i}^K\mathbbm{1}\{(P_iP_kY)^T\eta\leq (P_kY)^T\eta_j\}, i=0,\cdots,K.$$

Additionally, from the construction of the linear model \eqref{eq::model}, it simply satisfies
$$(R(Y^k))_{i+1}=(R(\epsilon^k))_{i+1}=\frac{1}{K+1}\sum_{j=0,j\neq i}^K\mathbbm{1}\{(P_iP_k\epsilon)^T\eta\leq (P_k\epsilon)^T\eta_j\}, i=0,\cdots,K.$$

For simplicity, we initially consider $\mathcal{P}_K$ as the cycle permutation group, for which there exist $P_1\in \mathcal{P}_K$, $\forall k\in\{1,\cdots,K\}$, $P_k=P_1^k,P_0=I$. From the definition above, the cycle permutation group obviously satisfies \Cref{assump::grouppermutation}. Therefore,
\begin{equation}
\label{unchangeable}
    (R(Y^\mathcal{K}))_{i+1} = 
\begin{cases}
R_{\mathcal{K}+i} &  \mathcal{K}+i<K+1, \\
R_{\mathcal{K}+i-K-1}  &  \mathcal{K}+i\geq K+1.
\end{cases}
\end{equation}
Since enlarging a quantity that is already above the quantile does not change the value of the quantile, our rejection region can be rewritten as 
$$\quad R_{0}> Q_{1-\alpha}\{ \sum_{k=0}^{K} w_k \cdot \delta_{R_k}\}  \iff \quad R_{0}> Q_{1-\alpha}\{ \sum_{k=1}^{K} w_k \delta_{R_k}+w_0\cdot \delta_{+\infty}\}.$$
Under the cycle permutation group, the equation ($\ref{unchangeable}$) holds, then we have

\begin{equation}
\begin{aligned}
\label{correlationforRk}
    Q_{1-\alpha}\{ \sum_{i=1}^{K}w_k\delta_{R_k}+w_0 \delta_{+\infty}\} &\geq Q_{1-\alpha}(\sum_{k=1}^K w_k\delta_{(R(Y^\mathcal{K}))_{k+1}}+ w_0\cdot \delta_{(R(Y^\mathcal{K}))_{1}})\\
    &=Q_{1-\alpha}(\sum_{k=1}^K w_k\cdot\delta_{(R(Y^\mathcal{K}))_{k+1}}+ w_0\delta_{R_{\mathcal{K}}})
\end{aligned}
\end{equation}
We now give a detailed proof of this identity. The right-hand side of (\ref{correlationforRk}) equals 
\begin{equation}
\label{analysisforQuantile}
    \begin{aligned}
    Q_{1-\alpha}(\sum_{k=1}^K w_k\cdot\delta_{(R(Y^\mathcal{K}))_{k+1}}+ w_0\cdot \delta_{R_{\mathcal{K}}})&=Q_{1-\alpha}(\sum_{k=1,k\neq K+1-\mathcal{K}}^K w_k\cdot\delta_{(R(Y^{\mathcal{K}}))_{k+1}}\\&\qquad\qquad\qquad+w_{K+1-\mathcal{K}}\cdot \delta_{(R(Y^{\mathcal{K}}))_{(K+1-\mathcal{K})+1}}+ w_0\cdot \delta_{R_{\mathcal{K}}})\\
    &=Q_{1-\alpha}(\sum_{k=1,k\neq K+1-\mathcal{K}}^K w_k\cdot\delta_{(R(Y^{\mathcal{K}}))_{k+1}}+w_{K+1-\mathcal{K}}\cdot \delta_{R_0} \\
     &\qquad\qquad\qquad\qquad\qquad\qquad\qquad\qquad\qquad+ w_0\cdot \delta_{R_{\mathcal{K}}})\\
    &=Q_{1-\alpha}(\sum_{k=1,k\neq K+1-\mathcal{K}}^K w_k\cdot\delta_{(R(Y^{\mathcal{K}}))_{k+1}}\\&\qquad\qquad\qquad+w_{\mathcal{K}}\delta_{R_{\mathcal{K}}}+w_{K+1-\mathcal{K}}\cdot \delta_{R_0}+ (w_0-w_{\mathcal{K}})\cdot \delta_{R_{\mathcal{K}}})\\
    &=Q_{1-\alpha}\{ \sum_{k=1}^{K}w_k\delta_{R_k}+w_{K+1-\mathcal{K}}\cdot \delta_{R_0}+ (w_0-w_{\mathcal{K}})\cdot \delta_{R_{\mathcal{K}}}\}.
\end{aligned}
\end{equation}
Since $\forall k\in\{1,\cdots,K\}$, $w_k$s have the same value, and $w_0\geq w_k$, hence $w_{K+1-\mathcal{K}}+(w_0-w_{\mathcal{K}})=w_0,$ and $w_0-w_{\mathcal{K}}\geq 0.$ From the above results, (\ref{correlationforRk}) must hold.

Not only for the cycle permutation group, if we analyze the general permutation group $\mathcal{P}_K$ that satisfies Assumption \ref{assump::grouppermutation}. Using the same definition as in Proposition~\ref{prop::Proposition for group permutation}, $\mathbbm{P}_{\pi_{k}}(\pi_j):=P_j P_k, k=0,\cdots,K, j=0,\cdots,K$  , we have that
\begin{equation*}
    \begin{aligned}
 (R(Y^{\mathcal{K}}))_{i+1}&=\frac{1}{K+1}\sum_{j=0,j\neq i}^K\mathbbm{1}\{(\mathbbm{P}_{\pi_{\mathcal{K}}}(\pi_i)\cdot\epsilon)^T\eta\leq (P_{\mathcal{K}}\epsilon)^T\eta_j\}\\&=\frac{1}{K+1}\sum_{j=0,j\neq i}^K\mathbbm{1}\{(\mathbbm{P}_{\pi_{\mathcal{K}}}(\pi_i)\cdot\epsilon)^T\eta\leq (\mathbbm{P}_{\pi_{\mathcal{K}}}(\pi_j)\epsilon)^T\eta\}, \\
  &=\frac{1}{K+1}\sum_{j=0,j\neq \text{Index}(\mathbbm{P}_{\pi_{\mathcal{K}}}(\pi_i))}^K\mathbbm{1}\{(\mathbbm{P}_{\pi_{\mathcal{K}}}(\pi_i)\cdot\epsilon)^T\eta\leq (\epsilon)^T\eta_j\},\\
  &=(R(Y))_{\text{Index}(\mathbbm{P}_{\pi_{\mathcal{K}}}(\pi_i))+1}=R_{\text{Index}(\mathbbm{P}_{\pi_{\mathcal{K}}}(\pi_i))+1}, i=0,\cdots,K,
    \end{aligned}
\end{equation*}
where we define the notation index of the permutation matrix $\forall P_i$ in the permutation group $\mathcal{P}_K$ as   $\text{Index}(P_i)$, which satisfies 
\begin{equation} 
\label{abbrevIndex}
    \text{Index}(P_i)=i, P_i \in \mathcal{P}_K, i\in\{0,\cdots,K\}.
\end{equation} and the third equation holds since in  \Cref{prop::Proposition for group permutation}, we have proved that, if we determine the $\mathcal{K}$, $\mathbbm{P}_{\pi_{\mathcal{K}}}$ is a bijection. Also, using the bijection property, we obtain that 
$$\{(R(Y^{\mathcal{K}}))_{i+1}\}_{i=0}^K=\{(R(Y))_{\text{Index}(\mathbbm{P}_{\pi_{\mathcal{K}}}(\pi_i))+1}\}_{i=0}^K=\{R(Y)_{i+1}\}_{i=0}^{K},$$
where the $\{\cdot\}$ represents the unordered set. Also, because of the bijection property, there exist $j\in\{0,\cdots,K\}$, $P_jP_{\mathcal{K}}=P_0$, $j$ can be represented as $\text{Index}(P_\mathcal{K}^TP_0)$. So that with the similar analysis as (\ref{analysisforQuantile}), we obtain that 
\begin{equation}
    \label{analysispermutationgroupquantile}
    \begin{aligned}
    Q_{1-\alpha}(\sum_{k=1}^K w_k\cdot\delta_{(R(Y^\mathcal{K}))_{k+1}}+& w_0\cdot \delta_{R_{\mathcal{K}}})\\
    &=Q_{1-\alpha}\{ \sum_{i=1}^{K}w_k\delta_{R_k}+w_{\text{Index}(P_\mathcal{K}^TP_0)}\cdot \delta_{R_0}+ (w_0-w_{\text{Index}(\mathbbm{P}_{\pi_{\mathcal{K}}}(I))})\cdot \delta_{R_{\mathcal{K}}}\}.
\end{aligned}
\end{equation}

Because of $\forall k\in\{1,\cdots,K\}$, $w_k$s have the same value, and $w_0\geq w_k$, combined with (\ref{analysispermutationgroupquantile}), we have that for a general permutation group $\mathcal{P}_K$ satisfying  \Cref{assump::grouppermutation}, the comparison in (\ref{correlationforRk}) still holds.

Then, combining the following result, we have
$$ 
\begin{aligned}
  R_{0}> Q_{1-\alpha}\{ \sum_{i=0}^{K}w_k\cdot\delta_{R_k}\}  &\Rightarrow R_{0}>Q_{1-\alpha}(\sum_{i=1}^Kw_k\cdot\delta_{(R(Y^\mathcal{K}))_{i+1}}+w_0\cdot\delta_{R_{\mathcal{K}}})\\
  &\Rightarrow R_{0}>Q_{1-\alpha}(\sum_{i=0}^Kw_k\cdot\delta_{(R(Y^\mathcal{K}))_{i+1}}).
\end{aligned}
$$
Equivalently,
$$ R_{0}> Q_{1-\alpha}\{ \sum_{i=0}^{K}w_k\cdot\delta_{R_k}\}  \Rightarrow (R(Y^{\mathcal{K}}))_{\mathcal{K}+1}>Q_{1-\alpha}(\sum_{i=0}^K w_k\cdot\delta_{(R(Y^\mathcal{K}))_{i+1}}).$$ 
Next, to meet the property of the $1-\alpha$ quantile of the distribution, we consider the ``strange points'' and construct a projection $S$ from $\mathbb{R}^{K+1}$ to subsets of ${0,\cdots,K}$, as follows: for any $r \in \mathbb{R}^{K+1},$
$$S(r)=\{i\in\{0,\cdots,K\}:r_{i+1}>Q_{1-\alpha}(\sum_{j=0}^{K}w_j\cdot \delta_{r_{j+1}})\}.$$
The above set contains so-called ``strange'' points-indices $i$ for which $r_i$ is unusually large, relative to the empirical distribution of $r_1,\cdots,r_{K+1}$. A direct argument in \citet[lemma~A.1]{harrison2012importance} shows that
$$\sum_{i\in S(r)} w_k\leq \alpha.$$
That is, the fraction of ``strange'' points cannot exceed $\alpha$. From the above,  let $(\Omega,\mathcal{F},P)$ be a probability space, $\forall w\in \Omega$,
$$
\begin{aligned}
   w\in \{R_{0}> Q_{1-\alpha}\{ \sum_{i=0}^{K}w_k\cdot \delta_{R_{k+1}}\} &\} \Rightarrow w\in \{(R(Y^{\mathcal{K}}))_{\mathcal{K}+1}>Q_{1-\alpha}(\sum_{i=0}^K w_k\cdot\delta_{(R(Y^\mathcal{K}))_{i+1}})\}\\&\iff w\in \{\mathcal{K}\in S(R(Y^{\mathcal{K}}))\}, \\
\end{aligned}
$$
hence,
$P(R_{0}> Q_{1-\alpha}\{ \sum_{i=0}^{K}w_k\cdot\delta_{R_{k+1}}\})\leq P(\mathcal{K}\in S(R(Y^{\mathcal{K}})))$. Next, we analyze the probability of the event $\{\mathcal{K}\in S(R(Y^{\mathcal{K}}))\}.$

In the probability space $(\Omega,\mathcal{F},P)$, we define the probability measure in $\mathbb{R}^{K+1}$ of random vector $R(Y)$ as $\mu_{Y}$, and for $R(Y^{\mathcal{K}})$ as $\mu_{Y^{\mathcal{K}}}$. For a given $i \in \{0,\cdots,K\}$, we denote the event $A:=\{X\in \mathbb{R}^{K+1}: i\in S(X)\},$ then 
$$
\begin{aligned}
P(i \in S(R(Y ^{\mathcal{K}})))=\mu_{Y^{\mathcal{K}}}(A)&\leq \mu_{Y}(A)+\vert \mu_{Y}(A)-\mu_{Y^{\mathcal{K}}}(A)\vert \leq \mu_{Y}(A)+d_{TV}(\mu_{Y},\mu_{Y^{\mathcal{K}}}) \\   
&=P(i \in S(R(Y)))+ d_{TV}(R(Y),R(Y^{\mathcal{K}})).
\end{aligned}
$$
 
Armed with the above discussions, we have 
\begin{equation*}
    \begin{aligned}
        P(\mathcal{K}\in S(R(Y^{\mathcal{K}})))&=\sum_{k=0}^K
P(\mathcal{K}=k\,\,\textit{and}\,\, k \in S(R(Y^{k})))\\
&=\sum_{k=0}^{K} w_k \cdot P(k \in S(R(Y^{k}))\\
&\leq \sum_{k=0}^{K} w_k\cdot(P(k\in S(R(Y)))+d_{TV}(R(Y),R(Y^{k})))\\
&=\mathbb{E}[\sum_{k\in S(R(Y))}w_k]+\sum_{k=0}^Kw_k\cdot d_{TV}(R(\epsilon),R(\epsilon^k))\\
&\leq \mathbb{E}[\alpha] +\sum_{k=1}^{K} w_k\cdot d_{TV}(R(\epsilon),R(\epsilon^k))\\
&\leq \alpha +\sum_{k=1}^{K} w_k \cdot d_{TV}(\epsilon,\epsilon_{\pi_k}),
\end{aligned}
\end{equation*}
Hence, the desired result follows.
\subsection{Proof of Theorem \ref{thm::Theorem nonexchangeable grouped PALMRT}}
\label{sec::Proof of PALMRT beyond exchangeable}
As in \Cref{detailed proof of theorem nonexchangeable grouped CPT}, and using the same definition of $w_0$ and $\{w_i\}_{i=1}^K$, we obtain the extended version of \Cref{thm::Theorem nonexchangeable grouped PALMRT}.
 \begin{lemma}
   \label{lem::Theorem nonexchangeable grouped PALMRT}
Under $H_0$, and $w_0\in[\frac{1}{K+1},1)$, we set 
 $$w_k=\frac{1-w_0}{K},k\in \{1,\cdots,K\},$$
  then we have 
 \begin{equation*}
     \begin{aligned}
         P(\sum_{k=1}^K w_k \cdot 1\{X^{T}(I-H^{ZZ_{\pi_k}})&Y > (X)_{\pi_k}^{T}(I-H^{ZZ_{\pi_k}})Y\}< 1-\alpha)\\&\geq 1-2\alpha-\sum_{k=1}^Kw_k\cdot d_{TV}(T(\epsilon),T((\epsilon)^k))
         \\&\geq 1-2\alpha-\sum_{k=1}^Kw_k \cdot d_{TV}(\epsilon,\epsilon_{\pi_k})
     \end{aligned}
 \end{equation*}
    \end{lemma}
As before, if we set $w_0=\frac{1}{K+1}$, Lemma \ref{lem::Theorem nonexchangeable grouped PALMRT} implies \Cref{thm::Theorem nonexchangeable grouped PALMRT}, so we turn to give a proof of Lemma \ref{lem::Theorem nonexchangeable grouped PALMRT}.
\subsection{Proof of Lemma \ref{lem::Theorem nonexchangeable grouped PALMRT}}
 Initially, recap the definition of $$(F(\epsilon))_{i,j}:=F(\pi_{i-1},\pi_{j-1};x,Z,\epsilon)=X_{\pi_{i-1}}^{T}(I-H^{Z_{\pi_{i-1}}Z_{\pi_{j-1}}})\epsilon, i,j\in\{1,\cdots,K+1\},$$ 
 where $F(\pi_i,\pi_j;x,Z,\epsilon)$ is denoted the same in Section \ref{Grouped PALMRT}, and in \Cref{prop::Finverse} that for any permutation $\pi_1,\pi_2$ of $S_n$,
 $$F(\pi_1,\pi_2;x,Z,\epsilon_{\sigma})=F(\sigma^{-1} \circ \pi_1,\sigma^{-1} \circ \pi_2;x,Z,\epsilon),$$
where $S_n$ is defined as the permutation space of $[n]$. And $\mathcal{K}$ is defined in (\ref{mathcalK}).

 For simplicity, we first consider the cycle permutation group $\mathcal{P}_K$, which satisfies that $P_k=P_1^k, P_k\in \mathcal{P}_K, k\in \{0,\cdots,K\}$, $P_{K+1}=P_0=I$. We also use the abbreviation $\text{Index}(\cdot)$ defined in \eqref{abbrevIndex}. In this setting, we have 
  \begin{equation}
  \label{ijforF}
      \begin{aligned}
      (F(\epsilon^{\mathcal{K}}))_{i+1,j+1}:=(F(\epsilon_{\pi_k}))_{i+1,j+1} &= F(\pi_i,\pi_j;x,Z,\epsilon_{\pi_k})= F(\pi_k^{-1}\circ \pi_i,\pi_k^{-1}\circ\pi_j;x,Z,\epsilon)\\
      &= (F(\epsilon))_{\text{Index}(P_k^TP_i)+1,\text{Index}(P_k^TP_j)+1}\\
      &=  (F(\epsilon))_{(i+K+1-\mathcal{K}\mod K+1)+1, \,(j+K+1-\mathcal{K}\mod K+1)+1} 
  \end{aligned}
  \end{equation}
 where $i,j \in \{0,\cdots,K\}$ and the third equation holds since $P_k$ is a permutation matrix, $P_k^{-1}=P_k^T$, and the last equation holds since in cycle permutation group, $\forall k\in\{0,\cdots,K\}$, $P_{K+1-k}\cdot P_{k}=P_1^{K+1-k}\cdot P_1^{k}=P_1^{K+1}=P_{K+1}=I,$ hence $P_k^{-1}=P_k^T=P_{K+1-k}$. Since $(\mathcal{K}+k\mod K+1),\mathcal{K}\in \{0,\cdots,K\},$ applying (\ref{ijforF}), we can find that 
  $$
  \begin{aligned}
(F(\epsilon^{\mathcal{K}}))_{(\mathcal{K}+k\mod K+1)+1,\mathcal{K}+1}&=(F(\epsilon))_{([(\mathcal{K}+k \mod K+1)+K+1-\mathcal{K}]\mod K+1)+1,\,(\mathcal{K}+K+1-\mathcal{K} \mod K+1)+1}\\
&=(F(\epsilon))_{([(\mathcal{K}+k \mod K+1)-\mathcal{K}]\mod K+1)+1,\,1}\\
&=(F(\epsilon))_{k+1,1}=(X)_{\pi_k}^T(I-H^{ZZ_{\pi_k}})\epsilon, k\in\{0,\cdots,K\},
  \end{aligned}
  $$
  and also
  $$
(F(\epsilon^{\mathcal{K}}))_{\mathcal{K}+1,(\mathcal{K}+k\mod K+1)+1}=(F(\epsilon))_{1,k+1}=X^T(I-H^{ZZ_{\pi_k}})\epsilon,$$
Therefore,
\begin{equation*}
\begin{aligned}
    \sum_{k=1}^K& w_k\mathbbm{1}\{X^{T}(I-H^{ZZ_{\pi_k}})Y > (X)_{\pi_k}^{T}(I-H^{ZZ_{\pi_k}})Y\}\\&=
\sum_{k=1}^K w_k\mathbbm{1}\{(T(\epsilon^{\mathcal{K}}))_{\mathcal{K}+1,(\mathcal{K}+k\mod  K+1)+1} >T(\epsilon^{\mathcal{K}}))_{(\mathcal{K}+k\mod K+1)+1,\mathcal{K}+1}\}\\
&=\sum_{k=0,k\neq \mathcal{K}}^{K} w_{(k-\mathcal{K}\mod K+1)}\mathbbm{1}\{(T(\epsilon^{\mathcal{K}}))_{\mathcal{K}+1,k+1} >T(\epsilon^{\mathcal{K}}))_{k+1,\mathcal{K}+1}\}\\
&\leq \sum_{k=0,k\neq \mathcal{K}}^{K} w_{k}\mathbbm{1}\{(T(\epsilon^{\mathcal{K}}))_{\mathcal{K}+1,k+1} >T(\epsilon^{\mathcal{K}}))_{k+1,\mathcal{K}+1}\}\\
&=\sum_{k=0}^K w_k\mathbbm{1}\{T(\epsilon^{\mathcal{K}}))_{\mathcal{K}+1,k+1} >T(\epsilon^{\mathcal{K}}))_{k+1,\mathcal{K}+1}\},
\end{aligned}
\end{equation*}
 where the inequality means $w_{(k-\mathcal{K}\mod K+1)} \leq w_k ,k\in \{0,\cdot,K\}\backslash \{\mathcal{K}\}$, since by definition, $w_1=\cdots=w_k\leq w_0$. Under the cycle permutation group, we obtain
\begin{equation}
    \label{mathcalkchangeinequ}
     \sum_{k=1}^K w_k\mathbbm{1}\{X^{T}(I-H^{ZZ_{\pi_k}})Y > (X)_{\pi_k}^{T}(I-H^{ZZ_{\pi_k}})Y\}\leq \sum_{k=0}^K w_k\mathbbm{1}\{T(\epsilon^{\mathcal{K}}))_{\mathcal{K}+1,k+1} >T(\epsilon^{\mathcal{K}}))_{k+1,\mathcal{K}+1}\},
\end{equation}
For general permutation group $\mathcal{P}_K$ satisfying \Cref{assump::grouppermutation}, (\ref{mathcalkchangeinequ}) is also satisfied since the bijection property for $\mathcal{P}_K$ from \Cref{prop::Proposition for group permutation}, (\ref{mathcalkchangeinequ}) can be proved through the similar discussion combining the states in Section \ref{detailed proof of theorem nonexchangeable grouped CPT} and above cycle permutation case.
Now for any $r \in \mathbb{R}^{(K+1)\times(K+1)}$, define
$$S(r)=\{i\in \{0,\cdots,K\}:\sum_{j=0}^Kw_k 1\{r_{i+1,j+1}>r_{j+1,i+1}\}\geq 1-\alpha\}$$
an empirical set of ``strange'' points. The Lemma of \cite{lei2021weighted} implies that, for $S(r)$ defined above, $$\sum_{k\in S(r)}w_k\leq 2\alpha,$$
for all $r\in \mathbb{R}^{(n+1)\times(n+1)}.$
In other words, from above we obtain that 
$$\{\sum_{k=1}^K w_k\mathbbm{1}\{X^{T}(I-H^{ZZ_{\pi_k}})Y > (X)_{\pi_k}^{T}(I-H^{ZZ_{\pi_k}})Y\}\geq 1-\alpha\}\Rightarrow \{\mathcal{K}\in S(T(\epsilon^{\mathcal{K}}))\},$$
by the same procedure as in \Cref{detailed proof of theorem nonexchangeable grouped CPT}, we have 
$$P(K\in S(T(\epsilon^{\mathcal{K}})))\leq 2\alpha+\sum_{k=0}^{K}\frac{1}{K+1}d_{TV}(T(\epsilon),T(\epsilon^k)).$$


\end{document}